\documentclass[10pt]{article}
\usepackage{newcent} 
\usepackage{amsmath,amssymb,amsthm,color,enumerate}
\usepackage{soul}

\usepackage[sort&compress]{natbib}
\usepackage{caption}
\usepackage{graphicx}
\usepackage{epstopdf}
\usepackage{hyperref}
\hypersetup{colorlinks=true,allcolors=blue}

\usepackage[section]{placeins}

\def\today{\number\day\ \ifcase\month\or
 January\or February\or March\or April\or May\or June\or
 July\or August\or September\or October\or November\or December\fi
 \space \number\year}
\def\gobble#1#2{}
\def\shortdate{\number\day/\number\month/\expandafter\gobble\number\year}

\def\ifrac#1#2{#1/#2}
\def\N{\mathbb{N}}
\def\R{\mathbb{R}}
\def\p{Pain\-lev\'e}

\def\peqs{\p\ equations}

\def\dPI{\mbox{\rm dP$_{\rm I}$}}

\def\det{\mathop{\rm det}\nolimits}

\def\ds{\displaystyle}

\def\Ai{\mathop{\rm Ai}\nolimits}
\def\Bi{\mathop{\rm Bi}\nolimits}

\newcommand{\HyperpFq}[2]{{}_{#1}F_{#2}}

\newcommand\mycom[2]{\genfrac{}{}{0pt}{}{#1}{#2}}
\def\pFq#1#2#3#4#5{\HyperpFq{#1}{#2}\!\left(\mycom{#3}{#4};#5\right)}
\def\Fpq#1#2#3#4#5{\HyperpFq{#1}{#2}\left({#3;#4};#5\right)}
\def\Fpqq#1#2#3#4#5{\HyperpFq{#1}{#2}\!\left({#3;#4};#5\right)}
\def\rmd{\mathrm{d}}\def\rme{\mathrm{e}}\def\rmi{\mathrm{i}}
\def\a{\alpha}
\def\b{\beta}
\def\ga{\gamma}
\def\la{\lambda}
\def\k{\kappa}
\def\w{\omega}
\def\ep{\varepsilon}

\def\vth{\vartheta}
\def\ph{\varphi}
\def\vph{\varphi}
\def\A{\mathcal{A}}
\def\B{\mathcal{B}}
\def\O{\mathcal{O}}

\newcommand{\Real}{\mathbb{R}}
\newcommand{\Com}{\mathbb{C}}

\newcommand{\Integer}{\mathbb{Z}}

\def\Z{\Integer}
\def\dx{\rmd x}\def\dy{\rmd y}

\newcommand{\deriv}[3][]{\frac{\rmd^{#1}{#2}}{{\rmd{#3}}^{#1}}}
\newcommand{\pderiv}[3][]{\frac{\partial^{#1}{#2}}{{\partial{#3}}^{#1}}}

\newtheorem{theorem}{Theorem}[section]

\newtheorem{lemma}[theorem]{Lemma}

\newtheorem{corollary}[theorem]{Corollary}
\theoremstyle{definition}

\newtheorem{example}[theorem]{Example}
\newtheorem{remark}[theorem]{Remark}
\newtheorem{remarks}[theorem]{Remarks}
\newtheorem{conjecture}[theorem]{Conjecture}
\numberwithin{figure}{section}
\numberwithin{equation}{section}
\numberwithin{table}{section}

\newcommand{\comment}[1]{}

\def\imp{\int_{-\infty}^{\infty}}

\def\beq{\begin{equation}}
\def\eeq{\end{equation}}
\def\ds{\displaystyle}
\definecolor{dkg}{rgb}{0,0.7,0}
\definecolor{dkp}{rgb}{0.8,0.05,0.2}
\definecolor{dkr}{rgb}{0.9,0,0}
\definecolor{dkb}{rgb}{0,0,0.7}
\definecolor{purple}{rgb}{0.5,0,0.7}

\def\blue#1{\textcolor{blue}{#1}}

\definecolor{gold}{rgb}{0.83, 0.69, 0.22}

\def\etal{\textit{et al.}}
\def\red#1{\textcolor{red}{#1}}
\def\purple#1{\textcolor{purple}{#1}}
\def\green#1{\textcolor{dkg}{#1}}

\def\Tr{\mathop{\rm Tr}\nolimits}

\def\fig#1#2{\includegraphics[width=#1in]{Figs/#2}}
\def\tfig#1#2{\includegraphics[width=#1cm]{Figs/#2}}

\def\ww#1{\w(#1;\tau,t)}

\usepackage{ulem}

\begin{document}

\title{Symmetric Sextic Freud Weight}

\author{Peter A. Clarkson$^{1}$, Kerstin Jordaan$^{2}$ and Ana Loureiro$^{1}$\\[5pt]
$^{1}$ School of Engineering, Mathematics and Physics,\\ University of Kent, Canterbury, CT2 7NF, UK\\[2.5pt]
$^{2}$ Department of Decision Sciences,\\
University of South Africa, Pretoria, 0003, South Africa 
\\ Email: {P.A.Clarkson@kent.ac.uk}, {jordakh@unisa.ac.za}, {A.Loureiro@kent.ac.uk}
}
\maketitle



\begin{abstract}
In this study we are concerned with the properties of the sequence of coefficients $(\b_n)_{n\geq0}$ in the recurrence relation satisfied by the sequence of monic symmetric polynomials, orthogonal with respect to the symmetric sextic Freud weight
$$\omega(x;\tau, t) = \exp\!\big(\!-x^6 +\tau x^4 + t x^2\big), \qquad x \in \mathbb{R},$$
with real parameters $\tau$ and $t$. 
It is known that the recurrence coefficients $\b_n$ for the symmetric sextic Freud weight satisfy a fourth-order nonlinear discrete equation, which is a special case of the second member of the discrete Painlev\'{e} I hierarchy, often known as the ``string equation".
The recurrence coefficients have been studied in the context of Hermitian one-matrix models and random symmetric matrix ensembles with researchers in the 1990s observing ``chaotic, pseudo-oscillatory" behaviour. More recently, this ``chaotic phase" was described as a dispersive shockwave in a hydrodynamic chain.

Our emphasis is a comprehensive study of the behaviour of the recurrence coefficients as the parameters $\tau$ and $t$ vary. Extensive
computational analysis is carried out, using Maple, for critical parameter ranges, and graphical plots are presented to illustrate the behaviour of the recurrence coefficients as well as the complexity of the associated Volterra lattice hierarchy. The corresponding symmetric sextic Freud polynomials are shown to satisfy a second-order differential equation with rational coefficients. The moments of the weight are examined in detail, including their integral representations, differential equations, and recursive structure. Closed-form expressions for moments are obtained in several special cases in terms of generalised hypergeometric functions and modified Bessel functions, and asymptotic expansions for the recurrence coefficients are given. The results highlight the rich algebraic and analytic structures underlying the Freud weight and its connections to integrable systems.

\end{abstract}

{\tableofcontents}
 
\section{Introduction}\label{Intro}
The main goal of this paper is to analyse the behaviour of the sequence $(\b_n)_{n\geq 0}$, where $\b_n:= \b_n(\tau,t)$ are the recurrence coefficients in the second order recurrence relation \eqref{eq:srr} satisfied by the sequence of monic orthogonal polynomials with respect to the \textit{symmetric sextic Freud weight}
\beq \label{Freud642}
\w(x;\tau,t)=\exp\{-U(x)\},
\qquad U(x)=x^6-\tau x^4-tx^2,\qquad x\in\R,\eeq
under the assumption that $\tau$ and $t$ are real parameters. 
Much is known regarding the asymptotic behaviour of these coefficients for large $n$. Our goal is to discuss the behaviour of these parameters for finite $n$ whilst discussing hidden structures as the pair of parameters $\tau$ and $t$ vary.
Such coefficients $\b_n$ are known to be solutions of a fourth-order discrete equation \cite{refIK90,refJurk91,refLech92,refLechRR,refSen92}, which is
the second member of the discrete \p\ I hierarchy (dP$_{\rm I}^{(2)}$), as discussed in \S\ref{sec:rec}. The description of their rich structure is the main aim of this paper, where throughout sections \S\ref{Sec:NumComp} and \S\ref{sec:2D} we describe the behaviour of the sequence depending on the choice for the regions of the pair of parameters $\tau$ and $t$. The asymptotic behaviour of $\b_n$ when $n$ is large is described in Theorem \ref{FC} via a cubic, which naturally depends on the pair $\tau$ and $t$. This behaviour is expected and remains consistent regardless of the relationship the values of $\tau$ and $t$. However, the sequence $(\b_n)_{n\geq 0}$ changes for initial values of $n$, up to a critical region, depending on how the ratio $\k :=-\ifrac{t}{\tau^2}$ varies. The parameter values of primary interest are when $\tau>0$ and $-\frac23\leq \k \leq \frac25$ and we analyse these critical regions in \S\ref{Sec:NumComp}. Further illustrations of the {comparative} behaviour of two consecutive terms is {given} in \S\ref{sec:2D}. The initial conditions of the recurrence relation satisfied by $(\b_n)_{n\geq 0}$ depend on the first elements of the moment sequence associated with the weight function \eqref{Freud642} on the real line. Such moments are themselves special functions and solutions to a linear third order recurrence relation of hypergeometric type -- see \S\ref{sec:moments} for properties and \S\ref{sec:closed moments} for explicit expressions. 

The problem we are addressing arises in a variety of contexts, other than orthogonal polynomials. Most notably, this study has a clear motivation within the context of random matrix theory. 
It is known that there is a deep connection between random matrix theory and orthogonal polynomials, providing a framework for understanding the statistical properties of random matrices, cf.~\cite{refBleher,refBleherIts99,refDeift,refMehta,refWVA2020}. Specifically, certain random matrix ensembles have their eigenvalue distributions described by determinants related to orthogonal polynomials.
In the following, we briefly report on that connection and review previous and relevant studies linked to this problem.

An equivalent weight function to \eqref{Freud642} is
\beq \label{Freud642g}
w(x)=\exp\{-N V(x)\},\qquad V(x)=g_6x^6 + g_4x^4 + g_2 x^2,\qquad x\in\R,\eeq
with parameters $N$, $g_2$, $g_4$ and $g_6>0$, {which is the weight that arises in random matrix theory}. The terms of the sequence of monic orthogonal polynomials $\big(P_n(x)\big)_{n=0}^{\infty}$ with respect to this weight can be described via the $n$-fold Heine integral 
\[
P_n(x) =\frac{1}{n!\Delta_n} \int_{\Real^n} \prod_{j=1}^n (x-\la_j) \prod_{1\leq j<k\leq n} (\la_j-\la_k)^2 \prod_{\ell=1}^n\exp\big\{-N V(\la_\ell)\big\}\, \rmd\la_1\,\rmd\la_2 \ldots \rmd\la_n, 
\]
where $\Delta_n$ is the well-known Heine formula
\beq\label{heine}
\Delta_n =\frac{1}{n!} \int_{\Real^n} \prod_{1\leq j<k\leq n} (\la_j-\la_k)^2 \prod_{\ell=1}^n\exp\big\{-N V(\la_\ell)\big\}\, \rmd\la_1\,\rmd\la_2 \ldots \rmd\la_n, 
\eeq
under the assumption that $\la_1< \la_2 <\ldots <\la_n$. 
These polynomials are closely related with the study of unitary ensembles of random matrices associated with a family of probability measures of the form 
\[
\rmd \mu(\mathbf{M}) =\frac{1}{\mathcal{Z}_N} \exp\{-N\Tr [V(\mathbf{M}) ] \}\, \rmd\mathbf{M},
\]
on the space of $N\times N$ Hermitian matrices, $\mathcal{H}_N$. The scalar polynomial function $V(x)$ is referred to as the polynomial of the external field and $\mathcal{Z}_N$ is a normalisation factor in the unitary ensemble measures. Consider the change of variables $\mathbf{M}\mapsto (\boldsymbol{\Lambda},\mathbf{U})$, where $\mathbf{U}$ is a unitary matrix and $\boldsymbol{\Lambda}= \mathrm{diag}(\la_1,\la_2,\ldots,\la_N)$ with $\la_1< \la_2 <\ldots <\la_N$ representing the eigenvalues of the Hermitian matrix $\mathbf{M}=\mathbf{U} \boldsymbol{\Lambda} \mathbf{U}^*$. For any function $f: \mathcal{H}_N \rightarrow \mathcal{C} $ such that $f(\mathbf{U} \mathbf{M} \mathbf{U}^*) = f(\mathbf{M})$, by the Weyl integration formula for class functions, one has 
\[
\int f(\mathbf{M})\, \rmd\mathbf{M} = \pi^{N(N-1)/2}\prod_{j=1}^N\frac{1}{j!} \int_{\Real^N} f(\la_1,\la_2,\ldots,\la_N) \prod_{1\leq j<k\leq n} (\la_j-\la_k)^2\,\rmd\la_1\,\rmd\la_2 \ldots \rmd\la_N. 
\]
Hence, the normalisation factor $\mathcal{Z}_N$, known as partition function, is given by 
\begin{align*}
\mathcal{Z}_N &= \int_{\mathcal{H}_N} \exp\{-N\Tr [V(\mathbf{M}) ] \}\, \rmd\mathbf{M} \\
&= \pi^{N(N-1)/2} \prod_{j=1}^N\frac{1}{j!} \int_{\Real^N} \prod_{\ell=1}^N\exp\big\{-N V(\la_\ell)\big\}
\prod_{1\leq j<k\leq n} (\la_j-\la_k)^2\,\rmd\la_1\,\rmd\la_2 \ldots \rmd\la_N\\
 &=\pi^{N(N-1)/2} \prod_{j=1}^{N-1}\frac{1}{j!}\,\Delta_N. 
\end{align*}
All correlation functions between the eigenvalues $(\la_1,\la_2,\ldots,\la_N)$ can be expressed as determinants
in terms of the standard Christoffel-Darboux kernel formed from them 
\[
\mathcal{K}_N(x,y) :=\sum_{j=0}^{N-1}\frac{1}{h_j}P_j(x) P_j(y) \exp\{-\tfrac{1}{2}N [V(x)+V(y)]\},
\]
where $h_j = \int_\Real P_j^2(x) \exp\{-NV(x)\}\,\rmd x$. 

The simplest choice of $ V(x)=\tfrac{1}{2} x^2$, the classical Gaussian weight, implies the free matrix elements are i.i.d.\ normal variables $\mathcal{N}(0,1/(4N))$ if $j<k$ and i.i.d.\ normal variables $\mathcal{N}(0,1/(2N))$ if $j=k$. This leads to the Gaussian Unitary Ensemble (GUE). In this case, the eigenvalues of a random matrix in GUE form a determinantal point process with the Christoffel-Darboux kernel of Hermite polynomials (modulo a scale). The average characteristic polynomial, $\mathbb{E}(\det(x\mathbf{I}-\mathbf{M}))$ are essentially scaled Hermite polynomials (and therefore, on the average, the eigenvalues behave as the zeros of Hermite polynomials). 

The quartic model 
\beq V(x) = g_2 x^2 + g_4 x^4,\label{Freud4} \eeq
with $g_2$ and $g_4>0$ parameters, was first studied by Br\'ezin \etal\ \cite{refBIPZ} and Bessis \etal\ \cite{refBessis,refBIZ}. Subsequently there have been numerous studies, e.g.~\cite{refBarhoumi,refBT15,refBT16,refBGMcL,refBleherIts99,refBleherIts03,refCGMcL,refDK,refSV14,refWongZhang09}.

The sextic model 
\beq V(x) = g_2 x^2 + g_4 x^4+ g_6 x^6,\label{Freud6}\eeq
with $g_2$, $g_4$ and $g_6>0$ parameters, which gives \eqref{Freud642g}, is the simplest weight to give rise to multi-critical points, cf.~\cite{refBMN,refBK90,refCicuta,refCMM,refDoug,refDSS,refDS,refGMiga}.

Br\'ezin \etal\ \cite{refBMP} consider the weight \eqref{Freud642g} with $g_2=90$, $g_4=-15$ and $g_6=1$, which is a critical case and discussed in \S\ref{casei}.
Fokas \etal\ \cite{refFIK91,refIK90} investigated the weight \eqref{Freud642g}
with 
$g_2=\tfrac{1}{2}$, $g_4<0$ and $0\leq 5g_6<4g_4^2$
in their study of the continuous limit for the Hermitian matrix model in connection with the nonperturbative theory of two-dimensional quantum gravity. 

The behaviour of the recurrence coefficients $\b_n$ {for the weight \eqref{Freud642g}} was studied, primarily by numerical methods in the early 1990s by
Demeterfi \etal\ \cite{refDDJT}, Jurkiewicz \cite{refJurk90}, Lechtenfeld \cite{refLech92,refLech92b,refLechRR}, Sasaki and Suzuki \cite{refSS91} and S\'en\'echal \cite{refSen92}.
The conclusion was that behaviour of the recursion coefficients was ``chaotic”, 
e.g.~Jurkiewicz \cite{refJurk90} states that the recurrence coefficients ``show a chaotic, pseudo-oscillatory behaviour".
As explained in \S\ref{Sec:NumComp}, Jurkiewicz \cite{refJurk91}, Sasaki and Suzuki \cite{refSS91} and S\'en\'echal \cite{refSen92} used one method to numerically compute the recurrence coefficients, whilst Demeterfi \etal\ \cite{refDDJT} and Lechtenfeld \cite{refLech92,refLech92b,refLechRR} used a different approach.
Bonnet \etal\ \cite{refBDE00} in a subsequent study 
concluded that ``in the two-cut case the behaviour is always periodic or quasi-periodic and never chaotic (in the mathematical sense)". Further studies of the quasi-periodic asymptotic behaviour of the recurrence coefficients include
Eynard \cite{refEynard09}, Eynard and Marino \cite{refEynM11} and Borot and Guionnet \cite{refBG24}.
Recently Benassi and Moro \cite{refBM20}, see also \cite{refDellAtti}, interpreted Jurkiewicz's ``chaotic phase" as a dispersive shock propagating through the chain in the thermodynamic limit and explain the complexity of its phase diagram in the context of dispersive hydrodynamics. 

Deift \etal\ \cite{refDKMVZa,refDKMVZb}
discuss the asymptotics of orthogonal polynomials with respect to the weight $\exp\{-Q(x)\}$, where $Q(x)$ is a polynomial of even degree with positive leading coefficient, using a Riemann-Hilbert approach; see also \cite{refEMcL}.
Bertola \etal\ \cite{refBEH03,refBEH06} discuss the relationship between partition functions for matrix models, semi-classical orthogonal polynomials and associated isomonodromic tau functions.

Hermitian random matrix models have a wide variety of physical and mathematical applications including two-dimensional quantum gravity \cite{refDiFGZJ,refWitten}, 
phase transitions \cite{refCicuta,refMFOR},
probability theory \cite{refAGZ,refKonig}, 
graph enumeration \cite{refBGMcL,refBIPZ,refEMcL,refELT,refEMcLV,refZvon} and
number theory, for example the distribution of the zeros of the Riemann $\zeta$-function on the critical line \cite{refKeatSn1,refKeatSn2}.
Recently, Hermitian random matrices were used to study graph enumeration for the sextic model \eqref{Freud6}, with $g_4=0$ \cite{refGLL}; previous studies had considered the quartic model \eqref{Freud4}.

In previous work, we studied the quartic Freud weight \cite{refCJ18,refCJK}
\beq \w(x;\rho,t)=|x|^{\rho}\exp(-x^4+tx^2),\qquad x\in\R,
\label{Freud42}\eeq
with parameters $\rho>-1$ and $t\in\R$,
where the recurrence coefficients are expressed in terms of parabolic cylinder functions $D_{\nu}(z)$,
the sextic Freud weight \cite{refCJ21a,refCJ21b}
\beq \w(x;\rho,t)=|x|^{\rho}\exp(-x^6+tx^2),\qquad x\in\R,
\label{genFreud62}\eeq
with parameters $\rho>-1$ and $t\in\R$,
where the recurrence coefficients are expressed in terms of the hypergeometric functions $\Fpq{1}{2}{a_1}{b_1,b_2}{z}$
and the higher-order Freud weight \cite{refCJL23}
\beq \w(x;\rho,t)=|x|^{\rho}\exp(-x^{2m}+tx^2),\qquad x\in\R,\label{genFreud2n}\eeq
with parameters $\rho>-1$, $t\in\R$ and $m=2,3,\ldots\,$,
where the recurrence coefficients are expressed in terms of the generalised hypergeometric functions $\Fpq{1}{m-1}{a_1}{b_1,b_2,\ldots,b_{m-1}}{z}$.
{Whilst for the weights \eqref{Freud42}, \eqref{genFreud62} and \eqref{genFreud2n}, the recurrence coefficients are expressed in terms of special functions, for the symmetric sextic weight \eqref{Freud642}, we are only able to do so in three cases.}

The novelty of this work is to offer a comprehensive explanation of the intricate structure exhibited by the recurrence coefficients $\b_n$ for the symmetric Freud weight \eqref{Freud642}. In doing so, we reveal some hidden and fundamental properties and structures which are new to the theory. 

The outline of this document is as follows.
After a brief review of some properties of orthogonal polynomials with respect to symmetric weights in \S\ref{sec:op}, we focus on the recurrence relation coefficients associated with the weight \eqref{Freud642} in \S\ref{sec:rec}, presenting key recurrence and differential-difference equations satisfied by these coefficients, including their connection to the discrete \p\ I hierarchy and their asymptotic behaviour. In \S\ref{sec:polynomials}, we investigate the symmetric sextic Freud polynomials themselves, deriving differential-difference and linear second-order differential equations they satisfy. In \S\ref{sec:moments} we consider the moments of the weight \eqref{Freud642}, exploring their properties, differential equations and integral representations. In \S\ref{sec:closed moments} we give closed-form expressions for these moments under specific parameter constraints, with particular attention given to special cases and corresponding series expansions. In \S\ref{Sec:NumComp}, the discussion returns to the recurrence coefficients, going beyond asymptotic properties, to present extensive numerical computations, covering a wide range of parameter values and illustrating the behaviour of the recurrence coefficients under different parameter regimes. In \S\ref{sec:2D} we supplement this with two-dimensional plots to visually interpret the numerical findings, 
in particular, illustrating what happens for large $n$, which differ as the parameters $\tau$ and $t$ vary.
Finally, in \S\ref{sec:Volterra}, we connect the recurrence coefficients to integrable systems by demonstrating how they satisfy equations in the Volterra lattice hierarchy, thus linking the theory to broader mathematical structures.

\section{\label{sec:op}Orthogonal polynomials with symmetric weights}
The sequence of monic polynomials $\big\{P_n(x)\big\}_{n=0}^{\infty}$ of exact degree $n\in\N$ is orthogonal with respect to a positive weight $\w(x)$ on 
the real line $\R$ if
\[\imp P_m(x)P_n(x)\,\w(x)\,\rmd x = h_n\delta_{m,n},\qquad h_n>0,\]
where $\delta_{m,n}$ denotes the Kronecker delta, see, for example \cite{refChihara78,refIsmail,refSzego,refWVAbk}. 
Monic orthogonal polynomials $P_n(x)$, $n\in\N$, satisfy a three-term recurrence relationship of the form
\[
P_{n+1}(x)=xP_n(x)-\a_nP_n(x)-\b_nP_{n-1}(x),
\]
where the coefficients $\a_n$ and $\b_n$ are given by the integrals
\[
\a_n =\frac{1}{h_n}\imp xP_n^2(x)\,\w(x)\,\rmd x,\qquad \b_n =\frac{1}{h_{n-1}}\imp xP_{n-1}(x)P_n(x)\,\w(x)\,\rmd x,
\]
with $P_{-1}(x)=0$ and $P_{0}(x)=1$. 

The weights of classical orthogonal polynomials satisfy a first-order ordinary differential equation, the \textit{Pearson equation}
\begin{equation}\label{eq:Pearson}
\deriv{}{x}[\sigma(x)\,\w(x)]=\vth(x)\,\w(x),
\end{equation}
where $\sigma(x)$ is a monic polynomial of degree at most $2$ and $\vth(x)$ is a polynomial with degree $1$. However for \textit{semi-classical} orthogonal polynomials, the weight function $\w(x)$ satisfies the Pearson equation (\ref{eq:Pearson}) with either deg$(\sigma)>2$ or deg$(\vth)\neq 1$, cf.~\cite{refHvR,refMaroni,refWVAbk}.
For example, the sextic Freud weight \eqref{Freud642}
satisfies the Pearson equation \eqref{eq:Pearson} with \[\sigma(x)=x,\qquad\vth(x)= {-6x^6+4\tau x^4+2tx^2 +1},\]
so is a semi-classical weight.
 
For an orthogonality weight that is symmetric, i.e.\ when $\w(x)=\w(-x)$, it follows that $\a_n \equiv0$ and the monic orthogonal polynomials $P_n(x)$, $n\in\N$, satisfy the simplified three-term recurrence relation
\beq\label{eq:srr}
P_{n+1}(x)=xP_n(x)-\b_nP_{n-1}(x).
\eeq
The $k$th moment, $\mu_k$, associated with the weight $\w(x)$ is given by the integral
\[\mu_k=\imp x^k\w(x)\,\rmd x,\] 
while the determinant of moments, known as Hankel determinant, is \beq\label{eq:detsDn}
\Delta_n=\det\big[\mu_{j+k}\big]_{j,k=0}^{n-1}=\left|\begin{matrix} \mu_0 & \mu_1 & \ldots & \mu_{n-1}\\
\mu_1 & \mu_2 & \ldots & \mu_{n}\\
\vdots & \vdots & \ddots & \vdots \\
\mu_{n-1} & \mu_{n} & \ldots & \mu_{2n-2}\end{matrix}\right|,\qquad n\geq1,\eeq
where $\Delta_0=1$ and $\Delta_{-1}=0$, {which corresponds to the Heine formula \eqref{heine}}. The recurrence coefficient $\b_n$ in \eqref{eq:srr} can be expressed in terms of the Hankel determinant as
\beq\label{def:bn}
\b_n =\frac{\Delta_{n+1}\Delta_{n-1}}{\Delta_{n}^2}.\eeq
 
For symmetric weights $\mu_{2k-1}\equiv 0$, for $k=1,2,\ldots\ $, and so
it is possible to write the Hankel determinant $\Delta_n$ in terms of the product of two Hankel determinants, as given in the following lemma. The decomposition depends on whether $n$ is even or odd.
\begin{lemma}
Suppose that 
$\A_n$ and $\B_n$ are the Hankel determinants given by
\begin{align}\label{def:AnBn} \A_n &
=\left|\begin{matrix} 
\mu_0 & \mu_2 & \ldots & \mu_{2n-2}\\
\mu_2 & \mu_4 & \ldots & \mu_{2n} \\
\vdots & \vdots & \ddots & \vdots \\
 \mu_{2n-2} & \mu_{2n}& \ldots & \mu_{4n-4}
\end{matrix}\right|,\qquad
\B_n =\left|\begin{matrix} 
\mu_2 & \mu_4 & \ldots & \mu_{2n}\\
\mu_4 & \mu_6 & \ldots & \mu_{2n+2} \\
\vdots & \vdots & \ddots & \vdots \\
 \mu_{2n} & \mu_{2n+2}& \ldots & \mu_{4n-2}
\end{matrix}\right|.\end{align}
Then the determinant $\Delta_n$ \eqref{eq:detsDn} is given by
\beq \Delta_{2n}=\A_n\B_n,\qquad \Delta_{2n+1}=\A_{n+1}\B_n.\label{res:lemma21}\eeq
\end{lemma}

\begin{proof}{The result is easily obtained by matrix manipulation interchanging rows and columns.}
\end{proof}
{\begin{remark}{\rm The expression of the Hankel determinant $\Delta_n$ for symmetric weights as a product of two determinants is given in \cite{refCHL,refLCF}.
}\end{remark}}%
\begin{corollary}\label{Cor:2.3}
For a symmetric weight, the recurrence relation coefficient $\b_n$ is given by
\[ \b_{2n} =\frac{\A_{n+1}\B_{n-1}}{\A_n\B_n},\qquad
\b_{2n+1}=\frac{\A_{n}\B_{n+1}}{\A_{n+1}\B_n},\]
where $\A_n$ and $\B_n$ are the Hankel determinants given by \eqref{def:AnBn}, with $\A_0=\B_0=1$.
\end{corollary}
\begin{proof}
{Substituting \eqref{res:lemma21} into \eqref{def:bn} gives the result.} 
\end{proof}%
 
\begin{lemma}\label{lem:34} Let $\w_0(x)$ be a symmetric positive weight on the real line and suppose that \[\w(x;\tau,t)=\exp(tx^2+\tau x^4)\,\w_0(x),\qquad x\in\R,\] 
is a weight such that all the moments also exist. Then the recurrence coefficient $\b_{n}(\tau,t)$ satisfies the Volterra, or the Langmuir lattice, equation
\beq \pderiv{\b_{n}}{t} = \b_{n}(\b_{n+1}-\b_{n-1}),\label{eq:langlat}\eeq
and the differential-difference equation
\beq\pderiv{\b_n}{\tau} =\b_n\left[(\b_{n+2}+\b_{n+1}+\b_n)\b_{n+1} - (\b_{n}+\b_{n-1}+\b_{n-2})\b_{n-1}\right]. \label{eq:langlat4}\eeq
\end{lemma}
\begin{proof} 
{By definition
\beq h_n(\tau,t) = \imp P_n^2(x;\tau,t)\, \w(x;\tau,t)\,\rmd x, \label{def:hn}\eeq
and then differentiating this with respect to $t$ gives
\beq \pderiv{h_n}{t} = 2\imp P_n(x;\tau,t)\pderiv{P_n}{t}(x;\tau,t)\, \w(x;\tau,t)\,\dx + \imp P_n^2(x;\tau,t)\, x^2\w(x;\tau,t)\,\dx .\label{eq:Llat1}\eeq
Since $P_n(x;\tau,t)$ is a monic polynomial of degree $n$ in $x$, then $\ds\pderiv{P_n}{t}(x;\tau,t)$ is a polynomial of degree less than $n$ and so
\beq \imp P_n(x;\tau,t)\pderiv{P_n}{t}(x;\tau,t)\, \w(x;\tau,t)\,\dx =0.\label{eq:Llat2}\eeq
Using this and the recurrence relation \eqref{eq:srr}
in \eqref{eq:Llat1} gives
\begin{align*} 
\pderiv{h_n}{t} &= \imp \left[ P_{n+1}(x;\tau,t) +\b_n(\tau,t)P_{n-1}(x;\tau,t)\right]^2 \w(x;\tau,t)\,\dx \\
&= \imp \left[ P_{n+1}^2(x;\tau,t) +\b_n^2(\tau,t)P_{n-1}^2(x;\tau,t)\right] \w(x;\tau,t)\, \dx 
= h_{n+1}+\b_n^2 h_{n-1}
\end{align*}
using the orthogonality of $P_{n+1}(x;\tau,t)$ and $P_{n-1}(x;\tau,t)$,
and so it follows from 
\beq\label{bn:hn} h_n=\b_n h_{n-1},\eeq that
\beq\pderiv{h_n}{t}= h_{n+1}+\b_nh_n.\label{eq:hnt}\eeq
Differentiating \eqref{bn:hn} with respect to $t$ gives
\[\pderiv{h_n}{t}=\pderiv{\b_n}{t}h_{n-1}+\b_n\pderiv{h_{n-1}}{t},\]
and then using \eqref{eq:hnt} gives
\[ h_{n+1}+\b_nh_n = \pderiv{\b_n}{t}h_{n-1} + \b_n(h_n+\b_{n-1}h_{n-1}).\]
Since $h_{n+1}=\b_{n+1}h_n=\b_{n+1}\b_n h_{n-1}$, then we obtain
\[\pderiv{\b_n}{t}=\b_n(\b_{n+1}-\b_{n-1}),\]
as required.}

{To prove \eqref{eq:langlat4}, differentiating \eqref{def:hn} with respect to $\tau$ gives
\[ \pderiv{h_n}{\tau} = 2\imp P_n(x;\tau,t)\pderiv{P_n}{\tau}(x;\tau,t)\, \w(x;\tau,t)\,\dx + \imp P_n^2(x;\tau,t)\, x^4\w(x;\tau,t)\,\dx.\]
Analogous to \eqref{eq:Llat2} we have
\[\imp P_n(x;\tau,t)\pderiv{P_n}{\tau}(x;\tau,t)\, \w(x;\tau,t)\,\dx =0,\]
and then using the recurrence relation \eqref{eq:srr} gives
\begin{align*} 
\pderiv{h_n}{\tau} &= \imp x^2 \left( P_{n+1} +\b_nP_{n-1}\right)^{\!2} \w(x)\,\dx \\
&= \imp \left( x^2P_{n+1}^2 +2x^2\b_nP_{n+1}P_{n-1} +x^2\b_n^2P_{n-1}^2\right) \w(x)\,\dx \\
&= \imp \left[ \left( P_{n+2} +\b_{n+1}P_{n}\right)^{\!2} + 2\b_n\left( P_{n+2} +\b_{n+1}P_{n}\right)\left( P_{n} +\b_{n-1}P_{n-2}\right)
+\b_n^2\left( P_{n} +\b_{n-1}P_{n-2}\right)^{\!2} \right] \w(x)\,\dx \\
&= \imp \left[P_{n+2}^2 + \b_{n+1}^2P_{n}^2 +2\b_{n+1}\b_nP_{n}^2+\b_n^2P_{n}^2+\b_n^2\b_{n-1}^2P_{n-2}^2\right] \w(x)\,\dx \\
&= h_{n+2}+(\b_{n+1}+\b_n)^2 h_n +\b_n^2\b_{n-1}^2 h_{n-2}.
\end{align*} Consequently
\beq \pderiv{h_n}{\tau} =\left[\b_{n+2}\b_{n+1}+(\b_{n+1}+\b_n)^2 +\b_n\b_{n-1}\right] h_{n},\label{eq:langlat41}\eeq
since 
\[ h_{n+2}=\b_{n+2}h_{n+1}=\b_{n+2}\b_{n+1}h_n,\qquad h_{n-2}=\frac{h_{n-1}}{\b_{n-1}}=\frac{h_{n}}{\b_n\b_{n-1}}.\]
Differentiating \eqref{bn:hn} 
with respect to $\tau$ gives
\begin{align} \pderiv{h_n}{\tau}&=\pderiv{\b_n}{\tau}h_{n-1}+\b_n\pderiv{h_{n-1}}{\tau} 
=\pderiv{\b_n}{\tau}h_{n-1}+\b_{n}\left[\b_{n+1}\b_{n}+(\b_{n}+\b_{n-1})^2 +\b_{n-1}\b_{n-2}\right] h_{n-1}.\label{eq:langlat42}
\end{align}
From \eqref{eq:langlat41} and \eqref{eq:langlat42} we have
\begin{align*} 
\pderiv{\b_n}{\tau} &= \left[\b_{n+2}\b_{n+1}+(\b_{n+1}+\b_n)^2 +\b_n\b_{n-1}\right] \b_n -
\left[\b_{n+1}\b_{n}+(\b_{n}+\b_{n-1})^2 +\b_{n-1}\b_{n-2}\right] \b_n\\
&=(\b_{n+2}+\b_{n+1}+\b_n)\b_{n+1}\b_{n} - (\b_{n}+\b_{n-1}+\b_{n-2})\b_{n}\b_{n-1},
\end{align*}
and therefore
\[\pderiv{\b_n}{\tau} =\b_n\left[(\b_{n+2}+\b_{n+1}+\b_n)\b_{n+1} - (\b_{n}+\b_{n-1}+\b_{n-2})\b_{n-1}\right],\]
 as required.}
\end{proof}

\begin{remark}{\rm The differential-difference equation \eqref{eq:langlat} is also known as the discrete KdV equation, or the Kac-van Moerbeke lattice \cite{refKvM75}.}\end{remark}

\section{Recurrence relation coefficients for the sextic Freud weight}\label{sec:rec}
In this section we discuss properties of the coefficient $\b_n$ in the three-term recurrence relation \eqref{eq:srr} for the symmetric sextic Freud weight \eqref{Freud642}.
\begin{lemma} 
The recurrence relation coefficient $\b_{n}(\tau,t)$ satisfies the recurrence relation
\begin{align}6\b_{n} \big(\b_{n-2} \b_{n-1} &+ \b_{n-1}^2 + 2 \b_{n-1} \b_{n} + \b_{n-1} \b_{n+1} 
+ \b_{n}^2 + 2 \b_{n}\b_{n+1} + \b_{n+1}^2 + \b_{n+1} \b_{n+2}\big) \nonumber\\
& -4\tau\b_n(\b_{n-1}+\b_{n}+\b_{n+1})-2t\b_{n}=n.\label{eq:rr642}
\end{align}
\end{lemma}
\begin{proof}
The fourth-order nonlinear discrete equation \eqref{eq:rr642} when $\tau=0$ and $t=0$ was derived by Freud \cite{refFreud76}.
It is straightforward to modify the proof for the case when $\tau$ and $t$ are nonzero.
\end{proof}

\begin{remarks}{\rm
\begin{enumerate}\item[] \item 
The discrete equation \eqref{eq:rr642} is a special case of dP$_{\rm I}^{(2)}$, the second member of the discrete \p\ I hierarchy which is given by 
\begin{align}c_4\b_{n} \big(\b_{n+2} \b_{n+1} &+ \b_{n+1}^2 + 2 \b_{n+1} \b_{n} + \b_{n+1} \b_{n-1} + \b_{n}^2 + 2 \b_{n}\b_{n-1} + \b_{n-1}^2 + \b_{n-1} \b_{n-2}\big)\nonumber\\
& +c_3\b_{n} \big(\b_{n+1}+\b_n+\b_{n-1})+ c_2\b_{n}= c_1+c_0(-1)^{n}+n,\label{eq:gendPI2}\end{align}
with $c_j$, $j=0,1,\ldots,4$ constants.
Cresswell and Joshi \cite{refCJ99a,refCJ99b} show that if $c_0=0$ then the continuum limit of \eqref{eq:gendPI2} is equivalent to
\[ \nonumber \deriv[4]{w}{z}=10w\deriv[2]{w}{z}+5\left(\deriv{w}{z}\right)^{\!2} -10w^3+z,\]
which is P$_{\rm I}^{(2)}$, the second member of the first \p\ hierarchy \cite{refKud97}, see also \cite{refBMP,refFIK91}.
\item Equation \eqref{eq:rr642} is also known as the ``string equation" and arises in important physical applications such as two-dimensional quantum gravity, cf.~\cite{refDS,refFIK91,refFIK92,refGMiga,refGMigb,refGMigc,refPS}.
\item Equation \eqref{eq:rr642} is also derived in \cite{refMinDing} using the method of ladder operators due to Chen and Ismail \cite{refChenIsmail}.
\item The autonomous analogue of \eqref{eq:rr642} has been studied by Gubbiotti \etal\ \cite[map P.iv]{refGJTV} and Hone \etal~\cite{refHRV,refHRVZ}.
\end{enumerate}
}\end{remarks}

\def\FF{\left[225 n +2\tau(4\tau^2+15 t) +G(n,\tau,t)\right]^{\!1/3}}

\begin{lemma}
The recurrence coefficient $\b_{n}(\tau,t)$ satisfies the system
\begin{subequations}\label{sys642:bn}\begin{align}
\pderiv[2]{\b_{n}}{t}&-3(\b_{n}+\b_{n+1}-\tfrac{2}{9}\tau)\pderiv{\b_{n}}{t}+\b_{n}^3+6\b_{n}^2\b_{n+1}+3\b_{n}\b_{n+1}^2 
\nonumber\\ &-\tfrac{2}{3}\tau\b_n(\b_n+2\b_{n+1})-\tfrac{1}{3}t\b_{n}=\tfrac{1}{6}n,\\
\pderiv[2]{\b_{n+1}}{t}&+3(\b_{n}+\b_{n+1}-\tfrac{2}{9}\tau)\pderiv{\b_{n+1}}{t}+\b_{n+1}^3+6\b_{n+1}^2\b_{n}+3\b_{n+1}\b_{n}^2 \nonumber\\
& -\tfrac{2}{3}\tau\b_{n+1}(2\b_n+\b_{n+1})-\tfrac{1}{3}t\b_{n+1}=\tfrac{1}{6}(n+1).
\end{align}\end{subequations}
\end{lemma}
\begin{proof}This is analogous to that for the generalised sextic Freud weight 
\[ 
\w(x;\tau,t,\rho)=|x|^{\rho}\exp\left(-x^6+tx^2\right),\qquad \rho>-1,\qquad x\in\R,\]
see \cite[Lemma 4.3]{refCJ21b}, with $\rho=0$. 
Following Magnus \cite[Example 5]{refMagnus95}, from the Langmuir lattice \eqref{eq:langlat} we have
\begin{subequations}
\begin{align} 
 \pderiv{\b_{n-1}}{t} &= \b_{n-1}(\b_{n}-\b_{n-2})\nonumber\\
 &= \b_{n-1}^2 + 3 \b_{n-1} \b_{n} + \b_{n-1} \b_{n+1} 
+ \b_{n}^2 + 2 \b_{n}\b_{n+1} + \b_{n+1}^2 + \b_{n+1} \b_{n+2} \nonumber\\ 
&\qquad - \tfrac{2}{3}\tau(\b_{n-1}+\b_{n}+\b_{n+1}) -\tfrac{1}{3}t-\frac{n}{6\b_n},\label{sys:bn2a}\\
 \pderiv{\b_{n}}{t} &= \b_{n}(\b_{n+1}-\b_{n-1}),\label{sys:bn2b}\\
\pderiv{\b_{n+1}}{t} &= \b_{n+1}(\b_{n+2}-\b_{n}),\label{sys:bn2c}\\
 \pderiv{\b_{n+2}}{t} &= \b_{n+2}(\b_{n+3}-\b_{n+1})\nonumber\\
 &=- \b_{n-1} \b_{n} -\b_{n}^2 - 2 \b_{n} \b_{n+1} - \b_{n} \b_{n+2} 
- \b_{n+1}^2 - 3 \b_{n+1}\b_{n+2} - \b_{n+2}^2 \nonumber\\ 
&\qquad + \tfrac{2}{3}\tau(\b_{n}+\b_{n+1}+\b_{n+2})+\tfrac{1}{3}t +\frac{n+1}{6\b_{n+1}}{,}\label{sys:bn2d}
 \end{align}\end{subequations}
where we have used the discrete equation \eqref{eq:rr642} to eliminate $\b_{n+3}$ and $\b_{n-2}$. Solving \eqref{sys:bn2b} and \eqref{sys:bn2c} for $\b_{n+2}$ and $\b_{n-1}$ gives
 \[ \b_{n+2}=\b_n+\frac{1}{\b_{n+1}}\pderiv{\b_{n+1}}{t},\qquad \b_{n-1}=\b_{n+1}-\frac{1}{\b_n}\pderiv{\b_n}{t},\] and substitution into \eqref{sys:bn2a} and \eqref{sys:bn2d}
 yields the system \eqref{sys642:bn} as required.
\end{proof}

Freud \cite{refFreud76} conjectured that the asymptotic behaviour of recurrence coefficients $\b_{n}$ in the recurrence relation \eqref{eq:srr} satisfied by monic polynomials $\big\{P_{n}(x) \big\}_{n\geq0}$ orthogonal with respect to the weight
\begin{equation*}\w(x) = |x|^{\rho}\exp(-|x|^{m}),\end{equation*}
with $x \in \R$, $\rho>-1$, $m>0$ could be described by
\begin{equation}\label{Freudconj}
\ds \lim_{n\rightarrow \infty}\frac{\b_{n}}{n^{2/m}}= \left[\frac{\Gamma(\tfrac{1}{2}m)\, \Gamma(1+\tfrac{1}{2}m)}{\Gamma(m+1)}\right]^{2/m},
\end{equation}
{where $\Gamma(\a)$ is the Gamma function}.
The conjecture was originally stated by Freud for orthonormal polynomials. Freud showed that \eqref{Freudconj} is valid whenever the limit on the left-hand side exists and proved this for $m=2,4,6$. Magnus \cite{refMagnus85} proved \eqref{Freudconj} for recurrence coefficients associated with weights \begin{equation*}w(x)=\exp\{-Q(x)\}, \end{equation*}
where $Q(x)$ is an even degree polynomial with positive leading coefficient. Lubinsky \etal\ \cite{refLubinskyMS,refLubinskyMS86} settled Freud's conjecture as a special case of a general result for recursion coefficients associated with exponential weights. 

{\begin{theorem}\label{FC} For the symmetric sextic Freud weight \eqref{Freud642}, the recurrence coefficients $\b_{n}:=\b_{n}(\tau,t)$ associated with this weight satisfy 
\begin{equation*}
\lim_{n\to \infty}\frac{\b_{n}}{\b(n)} = 1, \end{equation*}
where $\b:=\b(n)$ is the positive curve defined by 
\beq \label{eq:cubic} 
60\b^3-12\tau \b^2-2t\b 
=n.\eeq
\end{theorem}
\begin{proof} For the weight $\exp\{-Q(x)\}$ with $Q(x)=x^{6}-\tau x^4-tx^2$ it follows from \cite[Theorem 2.3]{refLubinskyMS} that, if we define $a_{n}$ as the unique, positive root of the equation 
\begin{equation*}
n=\frac{1}{\pi} \ds \int_{0}^{1}\frac{a_{n}s\, Q'(a_{n}s)}{\sqrt{1-s^2}}\,\rmd{s},
\end{equation*} 
then $\ds\lim_{n\to\infty}\ifrac{\b_n}{a_n^2}=\tfrac14$. (These are known as the \textit{Mhaskar-Rakhmanov-Saff numbers} 
\cite{refMSaff84,refRakhmanov}.) Hence, $a_n$ satisfy 
\[n=\frac{1}{\pi} \ds \int_{0}^{1}\frac{a_n s\left(6 a_n^5 s^5 -4 a_n^3\tau s^3 -2 a_n t s \right)}{\sqrt{1-s^2}}\,\rmd{s},\]
which gives 
\[\frac{15 a_n^6}{16}-\frac{3\tau a_n^4}{4}-\frac{t a_n^2}{2}
= n.\]
Hence, setting $\b(n) =\frac14 a_n^2 $ we conclude the result.
{The function $\b(n)$ satisfying \eqref{eq:cubic} is a positive curve since, for fixed $\tau$ and $t$, its discriminant
\[\Delta=-97200n^{2}-1728\tau(4\tau^{2}+15 t) n +192 t^{2}(3\tau^{2}+10t),\]
is negative for $n$ sufficiently large, so there is only one real solution. We refer to the start of \S\ref{Sec:NumComp} for further details.}
\end{proof}

\begin{lemma}
The recurrence coefficients $\b_n$ satisfy 
\[
\lim_{n\to\infty}\frac{\b_n}{n^{1/3}} =\frac{1}{\sqrt[3]{60}}. 
\]
\end{lemma}

\begin{proof}
{We first prove that the limit of $ {\b_n}/{n^{1/3}}$ as $n\to \infty$ exists. For that we start by showing that the non-negative sequence $(\ifrac{\b_n}{n^{1/3}})_{n\geq0}$ is bounded from above. By taking into account $\b_n\geq 0$, it follows from \eqref{eq:rr642} that 
\begin{align*}
6\b_n^3 -4\tau \b_n^2 -2 t \b_n 
& \leq n + \left(-12\b_n + 4\tau\right) \b_n(\b_{n-1}+\b_{n+1}). 
\end{align*}
Suppose the sequence $(\ifrac{\b_n}{n^{1/3}})_{n\in\mathbb{N}} $ is unbounded, then there exists a positive integer $m$ such that for infinitely many $n\geq m$, one has $\b_n>   |\tau| \,n^{1/3}$. 
Therefore, the latter inequality becomes 
\[
6\b_n^3 -4\tau \b_n^2 -2 t \b_n\leq n,
\] 
and, hence, 
\[
	6 \leq \frac{n}{\b_n^3}+\frac{2t}{\b_n^2}+\frac{4\tau}{\b_n}
    \leq \frac{n}{\b_n^3}+\frac{2|t|}{\b_n^2}+\frac{4|\tau|}{\b_n}
	\leq \frac{1}{|\tau|^3} + \frac{2 |t|}{|\tau|^2n^{2/3}} + \frac{4}{n^{1/3}},
\]
which is only possible for finitely many $n$. Therefore, $(\b_n/n^{1/3})_{n\geq0}$ is a bounded sequence and, for this reason, there exists 
\[
A= \lim_{n\to\infty} \inf\frac{\b_n}{n^{1/3}}, \quad \text{and}\quad
B= \lim_{n\to\infty} \sup\frac{\b_n}{n^{1/3}}. 
\]
Suppose $(\b_{n_k})_{k\in\mathbb{N}}$ is a subsequence such that 
\[
A= \lim_{k\to\infty}\frac{\b_{n_k}}{\sqrt[3]{n_k}}. 
\]
We evaluate relation \eqref{eq:rr642} and we take the limits over this subsequence to get 
\beq\label{ineq1}
1 \leq 6A \left(5B^2 + 2AB + A^2\right). 
\eeq
Similarly, we take a subsequence $(\b_{m_k})_{k\in\mathbb{N}}$ such that 
\[
B= \lim_{k\to\infty}\frac{\b_{m_k}}{\sqrt[3]{m_k}}, 
\]
and, after taking the limits over this subsequence in \eqref{eq:rr642}, we conclude 
\beq\label{ineq2}
6B \left(5A^2 + 2 AB + B^2\right) \leq 1. 
\eeq
The inequalities \eqref{ineq1} and \eqref{ineq2} then give 
\[ 
(B-A)^3 \leq 0,
\]
which implies $B\leq A$. However, by definition $A\leq B$ and hence $A=B$ and we conclude that 
$ \lim_{n\to\infty} \ifrac{\b_n}{n^{1/3}}$ exists and equals $A$. Thus, we conclude 
\[
	 \lim_{n\to\infty}\frac{\b_{n}}{\sqrt[3]{n}} = A. 
\]
We multiply \eqref{eq:rr642} by $1/n$ and then take the limit as $n\to \infty$ to conclude that 
$60 A^3 =1$ and hence the result. 
}
\end{proof}

\comment{
We first prove that the limit of $ \ifrac{\b_n}{n^{1/3}}$ as $n\to \infty$ exists. For that we start by showing that the non-negative sequence $(\b_n/n^{1/3})_{n\geq0}$ is bounded from above. 
From \eqref{eq:rr642}, it follows 
\begin{align*}
6\b_n^3 -4\tau \b_n^2 -2 t \b_n 
&= n - 6\b_{n} \big(\b_{n-2} \b_{n-1} + \b_{n-1}^2 + 2 \b_{n-1} \b_{n} + \b_{n-1} \b_{n+1} 
+ 2 \b_{n}\b_{n+1} + \b_{n+1}^2 + \b_{n+1} \b_{n+2}\big)\\&\qquad + 4\tau \b_n(\b_{n-1}+\b_{n+1}) \\
& \leq n + \left(-12\b_n + 4\tau\right) \b_n(\b_{n-1}+\b_{n+1}), 
\end{align*}
because $\b_n\geq 0$. When $\tau\leq 0$, from the latter inequality it follows
\[
6\b_n^3 -4\tau \b_n^2 -2 t \b_n\leq n,
\] 
which implies that $(\b_n/n^{1/3})_{n\geq0}$ is a positive bounded sequence. 

\def\cube#1{\left(#1\right)^{\!1/3}}
When $\tau> 0$, set $\widehat{\b}_n := \ifrac{\b_n}{n^{1/3}} $, and the inequality 
\begin{align*}
&	6\b_n^3 -4\tau \b_n^2 -2 t \b_n\leq n + \left(-12\b_n + 4\tau\right) \b_n(\b_{n-1}+\b_{n+1}). 
\end{align*}
reads as 
\begin{align*}
&	 6\widehat{\b}_n^3 -\frac{4\tau}{n^{1/3}} \widehat{\b}_n^2 -\frac{2t}{n^{2/3}} \widehat{\b}_n 
	\leq 1 +\left(-12\widehat{\b}_n+\frac{4\tau}{n^{1/3}} \right) \widehat{\b}_n\left\{ \cube{\frac{n-1}{n}} \widehat{\b}_{n-1}+ \cube{\frac{n+1}{n}} \widehat{\b}_{n+1}\right\}. 
\end{align*}
If the sequence $(\widehat{\b}_n)_{n\in\mathbb{N}} $ is unbounded, then there exists an integer $m>1$ such that 
$\widehat{\b}_{m} > \tfrac13\tau$. Thus, for finitely many $n\geq m$, one has $\widehat{\b}_{n} > \tfrac13\tau$ and therefore, taking into account the positivity of $\widehat{\b}_{n\pm 1}$, it follows 
\begin{align*}
\left(-12\widehat{\b}_n+\frac{4\tau}{n^{1/3}} \right) &\widehat{\b}_n\left\{ \cube{\frac{n-1}{n}} \widehat{\b}_{n-1}+ \cube{\frac{n+1}{n}} \widehat{\b}_{n+1}\right\}\\
&\leq -4\tau \left(1-\frac{1}{n^{1/3}} \right) 
\widehat{\b}_n
\left\{ \cube{\frac{n-1}{n}} \widehat{\b}_{n-1}+ \cube{\frac{n+1}{n}} \widehat{\b}_{n+1}\right\} \leq 0. 
\end{align*} 
This implies 
\[ 
6\widehat{\b}_n^3 -\frac{4\tau}{n^{1/3}} \widehat{\b}_n^2 -\frac{2t}{n^{2/3}} \widehat{\b}_n \leq 1,
\]
for finitely many $n\geq m$, but this would contradict the claim that $(\widehat{\b}_n)_{n\in\mathbb{N}} $ is unbounded. Hence, $(\widehat{\b}_n)_{n\in\mathbb{N}} $ is bounded. Therefore, there exists 
\[
A= \lim_{n\to\infty} \inf\frac{\b_n}{n^{1/3}}, \quad \text{and}\quad
B= \lim_{n\to\infty} \sup\frac{\b_n}{n^{1/3}}. 
\]
Suppose $(\b_{n_k})_{k\in\mathbb{N}}$ is a subsequence such that 
\[
A= \lim_{k\to\infty}\frac{\b_{n_k}}{\sqrt[3]{n_k}}. 
\]
We evaluate relation \eqref{eq:rr642} and we take the limits over this subsequence to get 
\beq\label{ineq1}
1 \leq 6A \left( 9 B^2 + A^2 \right). 
\eeq
Similarly, we take a subsequence $(\b_{m_k})_{k\in\mathbb{N}}$ such that 
\[
B= \lim_{k\to\infty}\frac{\b_{m_k}}{\sqrt[3]{m_k}}, 
\]
and, after taking the limits over this subsequence in \eqref{eq:rr642}, we conclude 
\beq\label{ineq2}
6B \left( 9 A^2 + B^2 \right) \leq 1. 
\eeq
The inequalities \eqref{ineq1} and \eqref{ineq2} then give 
\[
(A-B)(A^2-8AB+B^2) \geq 0,
\]
which implies $A\geq B$. However, by definition $A\leq B$ and hence $A=B$ and we conclude that 
$ \lim_{n\to\infty} \ifrac{\b_n}{n^{1/3}}$ exists and equals $A$. 
}

{We have shown that 
\[
\b_n(\tau,t)= \sqrt[3]{\frac{n}{60} }\left(1 + {o}(1)\right), \quad \text{as} \quad n\to\infty. 
\]}
For the purpose of the present research, we do not need further terms. For the record, we hereby note a formal asymptotic expansion for $\b_n$. {Since we do not prove the size of the error term, we say that the asympotic expansion is ``formal".}

\begin{lemma}{\label{betan_asynp}As $n\to\infty$, the recurrence relation coefficient $\b_{n}(\tau,t)$ satisfying \eqref{eq:rr642} has the following formal asymptotic expansion
\begin{equation}\label{bnasymp} \b_n(\tau,t)=\frac{n^{1/3}}{\gamma} +\frac{\tau}{15}+\frac{(2\tau^2+5t)\gamma}{450\,n^{1/3}} +\frac{2\tau(4\tau^2+15t)}{675\,\gamma\,n^{2/3}}
-\frac{\tau(2\tau^2+5t)(4\tau^2+15t)\gamma}{151875\,n^{4/3}} + \O\big(n^{-5/3}\big),\end{equation}
with $\gamma=\sqrt[3]{60}$.
}\end{lemma}}
\begin{proof}
Suppose that as $n\to\infty$
\begin{subequations}\label{bnn}
\beq \b_n = \frac{n^{1/3}}{\ga}+a_0+\frac{a_1}{n^{1/3}}+\frac{a_2}{n^{2/3}}+\frac{a_3}{n}+\frac{a_4}{n^{4/3}}+\O(n^{-5/3}),\eeq
with $\gamma=\sqrt[3]{60}$, then
\begin{align}
\b_{n\pm1}&=\frac{n^{1/3}}{\ga}+a_0+\frac{a_1}{n^{1/3}}+\frac{3\ga a_2\pm1}{2\ga n^{2/3}}+\frac{a_3}{n}+\frac{3a_4\mp a_1}{3n^{4/3}}+\O(n^{-5/3}),\\
\b_{n\pm2}&=\frac{n^{1/3}}{\ga}+a_0+\frac{a_1}{n^{1/3}}+\frac{3\ga a_2\pm2}{2\ga n^{2/3}}+\frac{a_3}{n}+\frac{3a_4\mp 2a_1}{3n^{4/3}}+\O(n^{-5/3}),
\end{align}
\end{subequations}
as $n\to\infty$. Substituting \eqref{bnn} into \eqref{eq:rr642} and equating powers of $n$ gives
\[
a_0=\frac{\tau}{15},\quad a_1=\frac{(2\tau^2+5t)\ga}{450},\quad a_2=\frac{2\tau(4\tau^2+15t)}{675\ga},\quad a_3=0,\quad a_4=-\frac{\tau(2\tau^2+5t)(4\tau^2+15t)\ga}{151875},\]
and so we obtain \eqref{bnasymp}.
\end{proof}

We study the cubic \eqref{eq:cubic} in conjunction with the behaviour of the recurrence coefficients $\b_n$ in more detail in \S \ref{Sec:NumComp}.

\section{Symmetric sextic Freud polynomials}\label{sec:polynomials}
In this section, we consider some properties of the polynomials associated with the symmetric sextic Freud weight (\ref{Freud642}). In particular, we derive a differential-difference equation and differential equation satisfied by symmetric sextic Freud polynomials which are analogous to those for the generalised sextic Freud polynomials discussed in \cite[\S 4]{refCJ21a}.
The coefficients $A_n(x)$ and $B_n(x)$ in the relation
\begin{equation}\label{ddee}\deriv{P_n}{x}=-B_n(x)P_n(x)+A_n(x)P_{n-1}(x),\end{equation}
satisfied by semi-classical orthogonal polynomials are of interest since differentiating this differential-difference equation yields the second order differential equation satisfied by the orthogonal polynomials. Shohat \cite{refShohat39} gave a procedure using quasi-orthogonality to derive \eqref{ddee} for weights $\w(x)$ such that $\displaystyle{\w'(x)}/{\w(x)}$ is a rational function. This technique was rediscovered by several authors including Bonan, Freud, Mhaskar and Nevai approximately 40 years later, see \cite[pp.~126--132]{refNevai86} and the references therein for more detail. The method of ladder operators was introduced by Chen and Ismail in \cite{refChenIsmail}. Related work by various authors can be found in, for example, \cite{refChenIts,refChenZhang,refFvAZ,refMhaskar} and a good summary of the ladder operator technique is provided in \cite[Theorem 3.2.1]{refIsmail}; {see also \cite[Chapter 4]{refWVAbk}}.
\begin{lemma}\cite[Theorem 3.2.1]{refIsmail}\label{Thm:ABn}
Let \begin{equation*}
\w(x)=\exp\{-v(x)\},\qquad x \in\R,\end{equation*} where $v(x)$ is a twice continuously differentiable function on $\R$. Assume that the polynomials $\big\{P_n(x)\big\}_{n=0}^{\infty}$ satisfy the orthogonality relation
\[\int_{-\infty}^{\infty}P_n(x)P_m(x)\w(x)\,\dx=h_n\delta_{mn}.\]
Then $P_{n}(x)$ satisfy the differential-difference equation
\begin{equation}\label{dde}
\deriv{P_n}{x}=-B_n(x)P_n(x)+A_n(x)P_{n-1}(x),
 \end{equation}
where
\begin{subequations}\label{ABn}
\begin{align}
A_n(x)&=\frac{1}{h_{n-1}}\int_{-\infty}^{\infty}P_n^2(y)\mathcal{K}(x,y)\w(y)\,\dy,\\
B_n(x)&=\frac{1}{h_{n-1}}\int_{-\infty}^{\infty}P_n(y)P_{n-1}(y)\mathcal{K}(x,y)\w(y)\,\dy,
\end{align} \end{subequations}
where
\begin{equation}\mathcal{K}(x,y)=\frac{v'(x)-v'(y)}{x-y}.\label{def:Kxy}\end{equation}
\end{lemma}
\begin{proof}Write Theorem 3.2.1 in \cite{refIsmail} for monic orthogonal polynomials on the real line. The result also follows from \cite[Theorem 2]{refCJK} by letting $\ga=0$.
 \end{proof}
Next we derive a differential-difference equation satisfied by generalised Freud polynomials using Theorem \ref{Thm:ABn}.
\begin{theorem}For the symmetric sextic Freud weight \eqref{Freud642}
the monic orthogonal polynomials $P_{n}(x;\tau,t)$ satisfy the differential-difference equation
\begin{equation*}
\deriv{P_n}{x}(x;\tau,t)=-B_n(x;\tau,t)P_n(x;\tau,t)+A_n(x;\tau,t)P_{n-1}(x;\tau,t),
\end{equation*}
where \begin{align*}
A_n(x;\tau,t)&=\b_n\big\{6x^4-4\tau x^2-2t+ (6x^2-4\tau)(\b_n+\b_{n+1})\big\}\nonumber\\ &\qquad+ 6\b_n\big\{\b _n (\b _{n-1}+\b _n+\b _{n+1})+\b_{n+1} (\b_n+\b_{n+1}+\b_{n+2})\big\},\\
\nonumber B_n(x;\tau,t)&= \b_n\big\{6x^3 -4\tau x + 6x(\b_{n-1}+\b_n+\b_{n+1})\big\}, 
\end{align*} with $\b_n$ the recurrence coefficient in the three-term recurrence relation \eqref{eq:srr}.\end{theorem}
 \begin{proof}For the symmetric sextic Freud weight \eqref{Freud642} we have
\[\ww{x}= \exp\left(-x^6+\tau x^4+tx^2\right),\] i.e.\
$v(x;t)=x^6-\tau x^4-tx^2$, and so $\mathcal{K}(x,y)$ defined by \eqref{def:Kxy} is
\[\mathcal{K}(x,y)= 6 (x^4+x^3 y+x^2 y^2+x y^3+y^4)-4\tau (x^2+x y+y^2)-2 t.\]
Hence 
\begin{align*} \int_{-\infty}^\infty \mathcal{K}(x,y) &{P_n^2(y;\tau,t)}\,\ww{y}\,\dy \\&=
(6x^4-4\tau x^2-2t)\int_{-\infty}^\infty {P_n^2(y;\tau,t)}\,\ww{y}\,\dy
+ (6x^3 -4\tau x) \int_{-\infty}^\infty {yP_n^2(y;\tau,t)}\,\ww{y}\,\dy\\ & \qquad\qquad
+ (6x^2-4\tau)\int_{-\infty}^\infty {y^2P_n^2(y;\tau,t)}\,\ww{y}\,\dy
+ 6x\int_{-\infty}^\infty {y^3P_n^2(y;\tau,t)}\,\ww{y}\,\dy\\ & \qquad\qquad + 6\int_{-\infty}^\infty {y^4P_n^2(y;\tau,t)}\,\ww{y}\,\dy\\
 &=(6x^4-4\tau x^2-2t)h_n + (6x^2-4\tau)(\b_n+\b_{n+1})h_{n}\\&\qquad\qquad+ 6\big\{\b_n (\b_{n-1}+\b_n+\b_{n+1})+\b_{n+1} (\b_n+\b_{n+1}+\b_{n+2})\big\} h_n,
\end{align*}
since $\b_n=h_n/h_{n-1}$ and iteration of the three-term recurrence relation \[xP_n(x)=P_{n+1}(x)+\b_n P_{n-1}(x),\] yields
 \begin{align*}
 x^2 P_n(x)= P_{n+2}(x)&+\left(\b_n+\b_{n+1}\right) P_n(x)+\b_{n-1} \b_n P_{n-2}(x), \\
 \nonumber x^3 P_n(x)=P_{n+3}(x)&+\left(\b_n+\b_{n+1}+\b_{n+2}\right) P_{n+1}(x)+\b_n \left(\b_{n-1}+\b_n+\b_{n+1}\right) P_{n-1}(x)\nonumber\\&+\b_{n-2} \b_{n-1} \b_n P_{n-3}(x),\\
 \nonumber x^4P_n(x)=P_{n+4}(x)&+\left(\b_n+\b_{n+1}+\b_{n+2}+\b_{n+3}\right) P_{n+2}(x)\\&+\big\{\b_n \left(\b_{n-1}+\b_n+\b_{n+1}\right)+\b_{n+1} \left(\b_n+\b_{n+1}+\b_{n+2}\right) \big\} P_n(x) \nonumber\\& +\b_{n-1} \b_n \left(\b_{n-2}+\b_{n-1}+\b_n+\b_{n+1}\right) P_{n-2}(x)+\b_{n-3} \b_{n-2} \b_{n-1} \b_n P_{n-4}(x).
 \end{align*}
Also 
\begin{align*} \int_{-\infty}^\infty \mathcal{K}(x,y) &{P_n(y;\tau,t)P_{n-1}(y;\tau,t)}\,\ww{y}\,\dy\nonumber\\ &=
 (6x^4-4\tau x^2-2t)\int_{-\infty}^\infty {P_n(y;\tau,t)P_{n-1}(y;\tau,t)}\,\ww{y}\,\dy\\& \qquad + (6x^3 -4\tau x) \int_{-\infty}^\infty {yP_n(y;\tau,t)P_{n-1}(y;\tau,t)}\,\ww{y}\,\dy\\& \qquad+ (6x^2-4\tau)\int_{-\infty}^\infty {y^2P_n(y;\tau,t)P_{n-1}(y;\tau,t)}\,\ww{y}\,\dy\\ & \qquad + 6x\int_{-\infty}^\infty {y^3P_n(y;\tau,t)P_{n-1}(y;\tau,t)}\,\ww{y}\,\dy \\& \qquad+ 6\int_{-\infty}^\infty {y^4P_n(y;\tau,t)P_n(y;\tau,t)}\,\ww{y}\,\dy\\
 &=
 (6x^3 -4\tau x) \b_n h_{n-1} + 6x(\b_{n-1}+\b_n+\b_{n+1}) \b_nh_{n-1},
\end{align*}
and the result follows.\end{proof}

\begin{theorem}\cite[Theorem 3.2.3]{refIsmail} \label{thm:gende}Let 
$\w(x)= \exp\{-v(x)\}$, for $x\in\R$,
 with $v(x)$ an even, continuously differentiable function on $\R$. Then
\begin{subequations}\begin{equation}\label{eq:gende}
\deriv[2]{P_n}{x}+R_n(x)\deriv{P_n}{x}+T_n(x)P_n(x)=0,
\end{equation}
where
\begin{align}\label{coef1}
R_n(x)&=-\deriv{v}{x}-\frac{1}{A_n(x)}\,\deriv{A_n}{x},\\[5pt]
\label{coef2}T_n(x)&=\frac{A_n(x)A_{n-1}(x)}{\b_{n-1}}+\deriv{B_n}{x} 
-B_n(x)\left[\deriv{v}{x} +B_n(x)\right]-\frac{B_n(x)}{A_n(x)}\deriv{A_n}{x},
\end{align}
with
\begin{align*}
A_n(x)&=\frac{1}{h_{n-1}}\int_{-\infty}^{\infty}P_n^2(y)\mathcal{K}(x,y)\w(y)\,\dy,\qquad
B_n(x)=\frac{1}{h_{n-1}}\int_{-\infty}^{\infty}P_n(y)P_{n-1}(y)\mathcal{K}(x,y)\w(y)\,\dy.
\end{align*}\end{subequations}
\end{theorem}
\begin{proof}The differential equation is given in factored form {for} orthonormal polynomials in \cite{refIsmail} and can be derived by differentiating both sides of \eqref{dde} with respect to $x$ to obtain
\begin{align}\label{eq:dif1}
\deriv[2]{P_n}{x}&=-B_n(x)\deriv{P_n}{x}+\deriv{A_n}{x}P_{n-1}(x)-\deriv{B_n}{x}P_n(x)+A_n(x)\deriv{P_{n-1}}{x}.\end{align}
Substituting 
\begin{align*} \left\{-\deriv{}{x}+B_n(x)+\deriv{v}{x}\right\}&P_{n-1}(x)
=\frac{A_{n-1}(x)}{\b_{n-1}}P_n(x).\end{align*}
into \eqref{eq:dif1} yields
\begin{align}
\deriv[2]{P_n}{x}
&=\nonumber-B_n(x)\deriv{P_n}{x}-\left[\deriv{B_n}{x}+\frac{A_n(x)A_{n-1}(x)}{\b_{n-1}}\right]P_n(x)\\&\qquad+\left\{\deriv{A_n}{x}+A_n(x)\left[B_n(x)+\deriv{v}{x}\right]\right\}P_{n-1}(x),\label{eq:fin}\end{align}
and the results follows by substituting $P_{n-1}(x)$ in \eqref{eq:fin} using \eqref{dde}.
\end{proof}
Finally, we derive a differential equation satisfied by symmetric sextic Freud polynomials.
\begin{theorem}{
For the symmetric sextic Freud weight \eqref{Freud642}
the monic orthogonal polynomials\newline $P_{n}(x;\tau,t)$ satisfy the differential equation
\begin{equation*}\deriv[2]{P_n}{x}(x;\tau,t)+R_n(x;\tau,t)\deriv{P_n}{x}(x;\tau,t)+T_n(x;\tau,t)P_n(x;\tau,t)=0,
\end{equation*}
where
\begin{align*}
R_n(x;\tau,t)
&=2 x \left\{t-3 x^4+2\tau x^2 -\frac{2 \left\{6 x^2-2\tau+3 \left(\b_n+\b_{n+1}\right)\right\}}{6 x^4-4\tau x^2-2 t+6 \b_n C_n+6 \b_{n+1}C_{n+1}+\left(\b_n+\b_{n+1}\right) \left(6 x^2-4\tau \right)}\right\},\\[5pt]
T_n(x;\tau,t)
&= 2 \b_n \left(3 C_n-2\tau +9 x^2\right)-4 x^2 \b_n \left(3 C_n-2\tau +3 x^2\right) \left\{\b_n \left(3 C_n-2\tau +3 x^2\right)-t+3 x^4-2\tau x^2\right\}\nonumber\\
&\qquad+\b_{n-1} \big\{6 C_{n-1} \b_{n-1}+6 C_n \b_n+\left(\b_{n-1}+\b_n\right) \left(6 x^2-4\tau \right)-2 t+6 x^4-4\tau x^2\big\}\nonumber
\\& ~\qquad\qquad \times \big\{6 C_n \b_n+6 C_{n+1} \b_{n+1}+\left(\b_n+\b_{n+1}\right) \left(6 x^2-4\tau \right)-2 t+6 x^4-4\tau x^2\big\}\nonumber\\
&\qquad+\frac{4 x^2 \b_n \left(3 C_n-2\tau +3 x^2\right) \left\{3 \left(\b_n+\b_{n+1}\right)-2\tau +6 x^2\right\}}{2\tau (\b_n+\b_{n+1})-3 x^2 (\b_n+ \b_{n+1})-3 C_n \b_n-3 C_{n+1} \b_{n+1}+t-3 x^4+2\tau x^2},\end{align*}
where
\[C_n=\b_{n-1}+\b_n+\b_{n+1}.\]
}\end{theorem}
\begin{proof}
In Theorem \ref{thm:gende} we showed that the coefficients in the differential equation \eqref{eq:gende} satisfied by polynomials orthogonal with respect to the weight $\w(x)=\exp\{-v(x)\}$, are given by \eqref{coef1} and \eqref{coef2}. For the symmetric sextic Freud weight \eqref{Freud642} we use \eqref{coef1} and \eqref{coef2} with $v(x)=x^6-\tau x^4-tx^2$, and $A_n$ and $B_n$ given by \eqref{ABn} to obtain the stated result.
\end{proof}

\section{Moments of the symmetric sextic Freud weight}\label{sec:moments}
The moments of the symmetric sextic Freud weight \eqref{Freud642} play a fundamental role in the analysis of the recurrence coefficients $\b_n$. These can be expressed as a ratio of Hankel determinants of the moments \eqref{def:bn} or as solutions of the nonlinear recurrence relation \eqref{eq:rr642} with initial conditions 
 \[\b_0=0, \qquad \b_1=\frac{\mu_2}{\mu_0},\qquad \b_2=\frac{\mu_{0} \mu_{4}-\mu_{2}^{2}}{\mu_{0} \mu_{2}}. 
 \]
 As such, a description of the moments is crucial. This section discusses properties of the moments in terms of the pair of parameters $(\tau,t)\in\mathbb{R}^2$, namely $\mu_n:=\mu_n(\tau,t)$. In Lemma \ref{lem:mu0} we describe $\mu_0$ as a solution of a third order differential equation in $t$ subject to the initial conditions $\mu_0(0,t)$, $\frac{\partial\mu_0}{\partial t}(0,t) $ and 
 $\frac{\partial^2\mu_0}{\partial t^2}(0,t)$ given in Lemma \ref{lemma:41} together with Lemma \ref{lem:mu0t}. The two latter results allows one to then describe the moments $\mu_n:=\mu_n(\tau,t)$ via the linear third order differential equation in $t$, given in \eqref{eq:munt}, as well as a linear second order partial differential equation in \eqref{pde munt}. As a consequence, we describe the moments via a linear third order recurrence relation \eqref{mr1}. 

\begin{lemma}\label{lem:mu0} For the weight \eqref{Freud642},
the first moment is
\[\mu_0(\tau,t)=\int_{-\infty}^{\infty} \exp\left(-x^6+\tau x^4+tx^2\right)\dx =\int_0^\infty s^{-1/2}\exp\left(-s^3+\tau s^2+ts\right)\rmd s, \]
which satisfies the equation
\beq \pderiv[3]{\vph}{t} - \tfrac{2}{3}\tau\pderiv[2]{\vph}{t}-\tfrac{1}{3}t\pderiv{\vph}{t}-\tfrac{1}{6}\vph=0. \label{eq6}
\eeq\end{lemma}
\begin{proof}
To show this result, interchanging integration and differentiation gives
\begin{align*}
 \pderiv[3]{\mu_0}{t} &- \tfrac{2}{3}\tau\pderiv[2]{\mu_0}{t}-\tfrac{1}{3}t\pderiv{\mu_0}{t}-\tfrac{1}{6}\mu_0\\
 &=\int_0^\infty \left[s^3-\tfrac{2}{3}\tau s^2-\tfrac{1}{3}ts-\tfrac{1}{6}\right]s^{-1/2}\exp\left(-s^3+\tau s^2+ts\right)\rmd s\\ &=-\tfrac{1}{3} \int_0^\infty \left\{\deriv{}{s}\left[ s^{1/2}\exp\left(-s^3+\tau s^2+ts\right)\right]\!\right\}\rmd s 
 =-\tfrac{1}{3} \Big[s^{1/2}\exp\left(-s^3+\tau s^2+ts\right)\Big]_{0}^{\infty}=0,
\end{align*}
as required.\end{proof}

Although it frequently is simpler to derive properties of a function from the differential equation it satisfies rather than from an integral representation, and even though \eqref{eq6} is a linear, third-order ordinary differential equation, it is not immediately obvious how to obtain a closed form solution to this differential equation. For the case when $\tau=0$, which is the equation associated with the quadratic-sextic Freud weight $\w(x; 0,t)=\exp\left(-x^6+tx^2\right)$, $x\in\R$,
with $t$ a parameter, the first moment $\mu_0(0,t)$ is solvable in terms of Airy functions $\Ai(z)$ and $\Bi(z)$ and $\mu_{2n}(0,t)$
in terms of the generalised hypergeometric function $\Fpq{1}{2}{a_1}{b_1,b_2}{z}$.

\begin{lemma}\label{lemma:41}
For the quadratic-sextic Freud weight $ \w(x;0,t)=\exp\left(-x^6+tx^2\right)$, the moments are
\begin{align}
\mu_0(0,t)&=\int_{-\infty}^{\infty} \exp\left(-x^6+tx^2\right)\dx 
=\int_0^\infty s^{-1/2}\exp\left(-s^3+ts\right)\rmd s\nonumber\\
&= \pi^{3/2}12^{-1/6}\big[\Ai^2(z)+\Bi^2(z)\big],\qquad z = 12^{-1/3}t,
\nonumber\\
\mu_{2n}(0,t)&=\int_{-\infty}^{\infty} x^{2n}\exp\left(-x^6+tx^2\right)\dx =\int_0^\infty s^{n-1/2}\exp\left(-s^3+ts\right)\rmd s
\nonumber\\&=
\tfrac{1}{3}\Gamma\left(\tfrac{1}{3}n+\tfrac{1}{6}\right)\pFq{1}{2}{\tfrac{1}{3}n+\tfrac{1}{6}}{\tfrac{1}{3},\tfrac{2}{3}}{\frac{t^3}{27}}
+\tfrac{1}{3} t\,\Gamma\left(\tfrac{1}{3}n+\tfrac{1}{2}\right)\pFq{1}{2}{\tfrac{1}{3}n+\tfrac{1}{2}}{\tfrac{2}{3},\tfrac{4}{3}}{\frac{t^3}{27}}\nonumber\\
&\qquad\qquad+ \tfrac{1}{6} t^2\Gamma\left(\tfrac{1}{3}n+\tfrac{5}{6}\right)\pFq{1}{2}{\tfrac{1}{3}n+\tfrac{5}{6}}{\tfrac{4}{3},\tfrac{5}{3}}{\frac{t^3}{27}},\label{mu2n0t}
\end{align}
where $\Ai(z)$ and $\Bi(z)$ are the {Airy functions} and $\Fpq{1}{2}{a_1}{b_1,b_2}{z}$ is the generalised hypergeometric function. 
\end{lemma}
\begin{proof}
This result for $\mu_0(0,t)$ is (9.11.4) in the DLMF \cite{refDLMF}, due to Muldoon \cite[p32]{refMul77} and the result for $\mu_{2n}(0,t)$ follows from \cite[Lemma 3.1]{refCJ21b} taking $\lambda=n+2j-\tfrac{1}{2}$.
\end{proof}
When $t=0$, which is the sextic-quartic Freud weight $\w(x;\tau,0)=\exp\left(-x^6+\tau x^4\right)$, $x\in\R$, with $\tau$ a parameter, the moments $\mu_{2n}(\tau,0)$ are solvable in terms of the generalised hypergeometric function $\pFq22{a_1,a_2}{b_1,b_2}{z}$, see Lemmas \ref{qsmoment} and \ref{lem:mun0}.

A formal power series expansion about $\tau=0$ can be straightforwardly derived via the integral series representation.

\begin{lemma}
For the weight \eqref{Freud642}, the moments are formally given by
\begin{align*}
 \mu_{2n}(\tau,t)
 &=\tfrac{1}{3}\sum_{j=0}^{\infty}\frac{\tau^j}{j!} \Bigg\{
 \Gamma \left(\tfrac{2}{3} j+\tfrac{1}{3} n+\tfrac{1}{6}\right)\, \pFq{1}{2}{\tfrac{2}{3}j+\tfrac{1}{3}n+\tfrac{1}{6}}{\tfrac{1}{3},\tfrac{2}{3}}{\frac{t^3}{27}} \\ & \qquad\qquad\qquad 
 -t\,\Gamma \left(\tfrac{2}{6}j+\tfrac{1}{3} n+\tfrac{1}{2}\right)\, \pFq{1}{2}{\tfrac{2}{3}j+\tfrac{1}{3}n+\tfrac{1}{2}}{\tfrac{2}{3},\tfrac{4}{3}}{\frac{t^3}{27}} \\
 & \qquad\qquad\qquad + \tfrac{1}{2} t^2\,\Gamma \left(\tfrac{2}{3} j+\tfrac{1}{3} n+\tfrac{5}{6}\right)\, \pFq{1}{2}{\tfrac{2}{3}j+\tfrac{1}{3}n+\tfrac{5}{6}}{\tfrac{4}{3},\tfrac{5}{3}}{\frac{t^3}{27}}
 \Bigg\}.\end{align*}
\end{lemma}

\begin{proof} By definition, we can formally successively derive
\begin{align*}
 \mu_{2n}(\tau,t)
 &= \int_0^{\infty} s^{n-1/2} \exp\left(-s^3+ts\right) \exp\left(\tau s^2\right)\rmd s \\
 &= \int_0^{\infty} s^{n-1/2} \exp\left(-s^3+ts\right) \left(\sum_{j=0}^{\infty}\frac{\tau^js^{2j}}{j!} \right) \rmd s \\
 &= \sum_{j=0}^{\infty}\frac{\tau^j}{j!} \int_0^{\infty} s^{n+2j-1/2} \exp\left(-s^3+ts\right) \rmd s \\
 &= \sum_{j=0}^{\infty}\frac{\tau^j}{j!} \mu_{2n+4j}(0,t),
\end{align*}
and so the result follows using \eqref{mu2n0t}. 
The interchanging of the integral and sum is justified by the Lebesgue dominated convergence theorem.\end{proof}

Higher order moments $\mu_{2n}(\tau,t)$ can be obtained after differentiation of the expression for first moment with respect to $t$. More precisely, we have the following result:

\begin{lemma}\label{lem:mu0t}
For the weight \eqref{Freud642}, the even moments can be written in terms of derivatives of the first moment, as follows 
\[
\mu_{2n}(\tau,t)= \pderiv[n]{}{t}\mu_0(\tau,t),\qquad n=0,1,2,\ldots \]
\end{lemma}
\begin{proof} This follows immediately from the integral representation \begin{align*}\mu_{2n}(\tau,t)&=\imp x^{2n} \exp\left(-x^6+\tau x^4+tx^2\right)\dx \\&= \pderiv{}{t}\imp x^{2n-2}\exp\left(-x^6+\tau x^4+tx^2\right)\dx \\&=\pderiv{}{t} \mu_{2n-2}(\tau,t)\\
&=\pderiv[n]{}{t}\mu_0(\tau,t),\qquad n=0,1,2,\ldots\,\end{align*}
where, as before, the interchange of integration and differentiation is justified by Lebesgue's Dominated Convergence Theorem.\end{proof}

\begin{lemma}
The moment $\mu_{2n}(\tau,t)$ satisfies the differential equation
\beq\label{eq:munt} \pderiv[3]{\mu_{2n}}{t} - \tfrac{2}{3}\tau\pderiv[2]{\mu_{2n}}{t}-\tfrac{1}{3}t\pderiv{\mu_{2n}}{t}-\tfrac{1}{6}(2n+1)\mu_{2n}=0.\eeq
\end{lemma}
\begin{proof} This is easily proved using induction. Equation \eqref{eq:munt} holds when $n=0$ from Lemma \ref{lem:mu0}. Differentiating \eqref{eq:munt} with respect to $t$ and using \[\mu_{2n+2}(\tau,t)=\pderiv{}{t}\mu_{2n}(\tau,t), \]
as shown in the proof of Lemma \ref{lem:mu0t}, gives
\[\pderiv[3]{\mu_{2n+2}}{t} - \tfrac{2}{3}\tau\pderiv[2]{\mu_{2n+2}}{t}-\tfrac{1}{3}t\pderiv{\mu_{2n+2}}{t}-\tfrac{1}{6}(2n+3)\mu_{2n+2}=0,\]
and so the result follows by induction.
\end{proof}

\begin{lemma}
The moment $\mu_{2n}(\tau,t)$ satisfies the partial differential equation
\beq\frac{\partial^2 \mu_{2n}}{\partial\tau\partial t} -\tfrac{2}{3}\tau\pderiv{\mu_{2n}}{\tau} -\tfrac{1}{3}t\pderiv{\mu_{2n}}{t} -\tfrac{1}{6}(2n+1)\mu_{2n}=0.
\label{pde munt}
\eeq
\end{lemma}
\begin{proof}
Since
\[ \mu_{2n}(\tau,t)=\int_0^\infty s^{n-1/2}\exp\left(-s^3+\tau s^2+ts\right)\rmd s,\]
then
\begin{align*}
\frac{\partial \mu_{2n}}{\partial\tau\partial t} &-\tfrac{2}{3}\tau\pderiv{\mu_{2n}}{\tau} -\tfrac{1}{3}t\pderiv{\mu_{2n}}{t} -\tfrac{1}{6}(2n+1)\mu_{2n}\\
&=\int_0^\infty s^{n-1/2}\left(s^3-\tfrac{2}{3}\tau s^2- \tfrac{1}{3}ts-\tfrac{1}{3}n-\tfrac{1}{6}\right)\exp\left(-s^3+\tau s^2+ts\right)\rmd s\\
&=-\tfrac{1}{3}\int_0^\infty \deriv{}{s}\left[ s^{n+1/2}\exp\left(-s^3+\tau s^2+ts\right)\right]\rmd s=0,
\end{align*}
as required. As before, the interchange of integration and differentiation is justified by Lebesgue's Dominated Convergence Theorem.
\end{proof}

{In addition to the ordinary and partial differential equations above, the sequence of moments can be generated recursively.}

\begin{lemma}
 For the weight \eqref{Freud642} the moments satisfy the discrete equation 
\begin{equation}\label{mr1}
 3\mu_{2n+6}(\tau,t) - 2\tau \mu_{2n+4}(\tau,t)
 -t \mu_{2n+2}(\tau,t)-(n+\tfrac{1}{2}) \mu_{2n}(\tau,t)=0.
\end{equation}
\end{lemma}
\begin{proof}
 The result follows from the integral representation using integration by parts to obtain
 \begin{align*} \mu_{2n}(\tau,t)&=2\int_0^\infty x^{2n}\exp\left(-x^6+\tau x^4+tx^2\right)\dx \\&=-\frac{2}{2n+1} \int_0^\infty x^{2n+1}\left(-6x^5+4\tau x^3+2tx\right) \exp\left(-x^6+\tau x^4+tx^2\right)\dx \\&=\frac{2}{2n+1}(3\mu_{2n+6}-2\tau \mu_{2n+4}-t\mu_{2n+2}),
 \end{align*}
 and so obtain \eqref{mr1}, as required.
 \end{proof}

The differential equation \eqref{eq:munt} in $t$ and the discrete equation \eqref{mr1} in the index both reveal the hypergeometric structure of the moments $\mu_{2n}(\tau,t)$. In the next section, the hypergeometric structure is explicitly shown for some values of the parameters. In the general case, we can only provide a series expansion in terms of Laguerre polynomials.

\section{Closed form expressions for moments}\label{sec:closed moments}
In this section we derive some closed form expressions of some moments for the sextic Freud weight \eqref{Freud642} with $t=-\k\tau^2$, i.e.
\beq \w(x;\tau,\k)=\exp\left\{-\left(x^6-\tau x^4+\k\tau^2x^2\right)\right\},\label{Freud642k}\eeq
with $\tau$ and $\k$ parameters. To justify this, consider
the {potential}
\beq U(x;\tau,t)=x^6-\tau x^4-tx^2, \label{Upoly}\eeq
with parameters $\tau$ and $t$, which can be written as 
\[U(x;\tau,t)=x^2 \left\{(x^2-\tfrac{1}{2}\tau)^2-(t+\tfrac{1}{4}\tau^2)\right\}. \]
This has a double root at $x=0$ and the other roots can be real or complex depending on the values of the parameters $\tau$ and $t$, with $t=-\tfrac{1}{4}\tau^2$ being critical. So it is convenient to define $\k=-t/\tau^2$, and write 
\beq U(x;\tau,\k) =x^2\left(x^4-\tau x^2+\k\tau^2\right)
= x^2 \left\{(x^2-\tfrac{1}{2}\tau)^2+ (\k- \tfrac{1}{4})\tau^2\right\}.
\label{Uployk}\eeq

To obtain closed form expressions we derive the differential equations with respect to $\tau$ satisfied by the associated moments
\[\mu_{2n}(\tau;\k)=\int_{-\infty}^\infty x^{2n}\exp\left\{-\left(x^6-\tau x^4+\k\tau^2x^2\right)\right\}\dx 
=\int_0^\infty s^{n-1/2}\exp\left\{-\left(s^3-\tau s^2+\k\tau^2s\right)\right\}\rmd s,\]
with $\k$ considered a parameter,
which satisfy the initial conditions
\begin{subequations}
\begin{align} \mu_{2n}(0;\k)&=\int_0^\infty s^{n-1/2}\exp\left(-s^3\right)\rmd s
=\tfrac{1}{3}\Gamma\left(\tfrac{1}{3}{n}+\tfrac{1}{6}\right),\label{mu0:kappa}\\ 
\deriv{\mu_{2n}}{\tau}(0;\k)&=\int_0^\infty s^{n+3/2}\exp\left(-s^3\right)\rmd s
=\tfrac{1}{3}\Gamma\left(\tfrac{1}{3}{n}+\tfrac{5}{6}\right),\\
\deriv[2]{\mu_{2n}}{\tau}(0;\k)&=\int_0^\infty (s^3-2\k)s^{n+1/2}\exp\left(-s^3\right)\rmd s
=\frac{2n+3-12\k}{18}\,\Gamma\left(\tfrac{1}{3}{n}+\tfrac{1}{2}\right).
\end{align}\end{subequations}

\subsection{Three special cases with $n=0$.}
First we consider the case when $n=0$, in the cases when $\k=\tfrac{1}{4}$, $\k=\tfrac{1}{3}$ and $\k=0$. 
\begin{lemma}
The first moment $\mu_0(\tau;\tfrac{1}{4})$ is given by
\begin{align}
 \mu_0(\tau;\tfrac{1}{4})
 &=\int_{-\infty}^\infty\exp\left\{-\left(x^6-\tau x^4+\tfrac{1}{4}\tau^2x^2\right)\right\}\dx =\int_0^\infty s^{-1/2}\exp\left\{-s(s-\tfrac{1}{2}\tau)^2\right\}\rmd s\nonumber\\
&=\frac{\pi\sqrt{6\tau}}{9}\left\{I_{1/6}\left(\frac{\tau^3}{108}\right) +I_{-1/6}\left(\frac{\tau^3}{108}\right) \right\}\exp\left(-\frac{\tau^3}{108}\right),
\label{sol:mu014}\end{align}
where $I_\nu(z)$ is the modified Bessel function.
\end{lemma}
\begin{proof}
Assuming we can interchange integration and differentiation
\begin{align*}
18\deriv[2]{\mu_0}{\tau}+\tau^2\deriv{\mu_0}{\tau}+\tau\mu_0
&=2 \int_0^\infty \deriv{}{s}\left[ (\tau-3s)s^{1/2}\exp\left\{-s(s-\tfrac{1}{2}\tau)^2\right\}\right]\rmd s\\
&=2\left[ (\tau-3s)s^{1/2}\exp\left\{-s(s-\tfrac{1}{2}\tau)^2\right\}\right]_{s=0}^\infty=0.
\end{align*}
Hence $\mu_0(\tau;\tfrac{1}{4})$ satisfies the second order equation
\beq 18\deriv[2]{\mu_0}{\tau}+\tau^2\deriv{\mu_0}{\tau}+\tau\mu_0=0,\label{eq:mu014}\eeq
which has general solution
\[ \mu_0(\tau;\tfrac{1}{4}) =\sqrt{\tau}
\left\{c_1 I_{1/6}\left(\frac{\tau^3}{108}\right) +c_2I_{-1/6}\left(\frac{\tau^3}{108}\right) \right\}\exp\left(-\frac{\tau^3}{108}\right),
\]
with $I_\nu(z)$ the modified Bessel function and $c_1$ and $c_2$ constants.
The initial conditions are
\[\mu_0(0;\tfrac{1}{4})=\int_0^\infty s^{-1/2}\exp(-s^3)\,\rmd s= \tfrac{1}{3}\Gamma\!\left(\tfrac{1}{6}\right), \qquad
\deriv{\mu_0}{\tau}(0;\tfrac{1}{4})=\int_0^\infty s^{3/2}\exp(-s^3)\,\rmd s=\tfrac{1}{3}\Gamma\!\left(\tfrac{5}{6}\right),\]
so since as $\tau\to0$
\begin{subequations}\label{modBasymp}\begin{align} 
I_{1/6}\left(\frac{\tau^3}{108}\right)&=
\frac{\Gamma\!\left(\tfrac{5}{6}\right)}{2 \pi}
\sqrt{\frac{6}{\tau}}\left\{\tau -\frac{\tau^4}{108} +\mathcal{O}\! \left(\tau^{7}\right)\right\}
\exp\left(\frac{\tau^3}{108}\right),\\
I_{-1/6}\left(\frac{\tau^3}{108}\right)&=
\frac{\Gamma\!\left(\tfrac{1}{6}\right)}{2 \pi}\sqrt{\frac{6}{\tau}} \left\{1 -\frac{\tau^3}{108} +\mathcal{O}\! \left(\tau^{6}\right)\right\}
\exp\left(\frac{\tau^3}{108}\right),
\end{align}\end{subequations}
and $\Gamma\!\left(\tfrac{1}{6}\right)\Gamma\!\left(\tfrac{5}{6}\right)=2\pi$, then
$c_1=c_2=\tfrac{1}{9}{\pi\sqrt{6}}$,
and therefore we obtain the solution \eqref{sol:mu014}
as required.
\end{proof}

\begin{lemma}
The first moment $\mu_0(\tau;\tfrac{1}{3})$ is given by
\begin{align}
\mu_0(\tau;\tfrac{1}{3})
&=\int_{-\infty}^\infty\exp\left\{-\left(x^6-\tau x^4+\tfrac{1}{3}\tau^2x^2\right)\right\}\dx =\int_0^\infty s^{-1/2}\exp\left\{-\left(s^3-\tau s^2+\tfrac{1}{3}\tau^2s\right)\right\}\rmd s\nonumber\\
&= \left\{\tfrac{1}{3}\Gamma\!\left(\tfrac{1}{6}\right) \pFq{2}{2}{\tfrac{1}{6},\tfrac{1}{2}}{\tfrac{1}{3},\tfrac{2}{3}}{\frac{\tau^3}{27}} +\tfrac{1}{3}{\tau}\Gamma\!\left(\tfrac{5}{6}\right)\pFq{2}{2}{\tfrac{1}{2},\tfrac{5}{6}}{\tfrac{2}{3},\tfrac{4}{3}}{\frac{\tau^3}{27}}
-\frac{\tau^2\sqrt{\pi}}{36}\,\pFq{2}{2}{\tfrac{5}{6},\tfrac{7}{6}}{\tfrac{4}{3},\tfrac{5}{3}}{\frac{\tau^3}{27}} \right\}
\exp\left(-\frac{\tau^3}{27}\right), \label{sol:mu013}
\end{align}
where $\pFq{2}{2}{a_1,a_2}{b_1,b_2}{z}$ is the hypergeometric function.
\end{lemma}
\begin{proof}
This result is proved using
\begin{align*}
\deriv[3]{\mu_{0}}{\tau}&+\frac{2\tau^2}{9}\deriv[2]{\mu_{0}}{\tau}+\frac{\tau(\tau^3+27)}{81}\deriv{\mu_{0}}{\tau}+\frac{4\tau^3+45}{324}\mu_{0}\\
&=-\frac{1}{162}\int_0^\infty 
\deriv{}{s}\left\{\left(54s^3-72\tau s^2+30\tau^2 s-4\tau^3-45\right)s^{1/2}
\exp\left(-s^3+\tau s^2-\tfrac{1}{3}\tau^2 s\right)\right\}\rmd s=0.
\end{align*}
Hence the moment $\mu_{0}(\tau;\tfrac{1}{3})$ satisfies the third order equation
\beq\deriv[3]{\mu_{0}}{\tau}+\frac{2\tau^2}{9}\deriv[2]{\mu_{0}}{\tau}+\frac{\tau(\tau^3+27)}{81}\deriv{\mu_{0}}{\tau}+\frac{4\tau^3+45}{324}\mu_{0}=0,\nonumber\eeq
which has general solution
\begin{align*}\mu_{0}(\tau;\tfrac{1}{3})&= \left\{
c_1\, \pFq{2}{2}{\tfrac{1}{6},\tfrac{1}{2}}{\tfrac{1}{3},\tfrac{2}{3}}{\frac{\tau^3}{27}} +c_2\,\tau\,\pFq{2}{2} {\tfrac{1}{2},\tfrac{5}{6}}{\tfrac{2}{3},\tfrac{4}{3}}{\frac{\tau^3}{27}} 
+c_3\,\tau^2\,\pFq{2}{2}{\tfrac{5}{6},\tfrac{7}{6}}{\tfrac{4}{3},\tfrac{5}{3}}{\frac{\tau^3}{27}} \right\}\exp\left(-\frac{\tau^3}{27}\right),
\end{align*}
with $c_1$, $c_2$ and $c_3$ constants.
The initial conditions are
\[\mu_0(0;\tfrac{1}{3})=\tfrac{1}{3}\Gamma\!\left(\tfrac{1}{6}\right), 
\qquad\deriv{\mu_0}{\tau}(0;\tfrac{1}{3})=\tfrac{1}{3}\Gamma\!\left(\tfrac{5}{6}\right), 
\qquad\deriv[2]{\mu_0}{\tau}(0;\tfrac{1}{3})
=-\frac{\sqrt{\pi}}{36},\]
and since $\pFq22{a_1,a_2}{b_1,b_2}{0}=1$, then we obtain the solution \eqref{sol:mu013}, as required.
\end{proof}

\begin{lemma}\label{qsmoment}
For the sextic-quartic Freud weight $\w(x;\tau,0)=\exp\left(-x^6+\tau x^4\right)$, the first moment is 
\begin{align}
\mu_0(\tau;0) & =\int_{-\infty}^{\infty}\exp\left(-x^6+\tau x^4\right)\dx 
 \nonumber\\
& =\tfrac{1}{3}\Gamma(\tfrac{1}{6})\; \pFq22{\tfrac{1}{12},\tfrac7{12}}{\tfrac{1}{3},\tfrac{2}{3}}{\frac{4\tau^3}{27}} + \tfrac{1}{3}{\tau
\,\Gamma(\tfrac{5}{6})} \; \pFq22{\tfrac5{12},\tfrac{11}{12}}{\tfrac{2}{3},\tfrac{4}{3}}{\frac{4\tau^3}{27}}
+\frac{\tau^2\sqrt{\pi}}{12} \; \pFq22{\tfrac{3}{4},\tfrac{5}{4}}{\tfrac{4}{3},\tfrac{5}{3}}{\frac{4\tau^3}{27}}, \label{sol:mu00}
\end{align}
where $\pFq22{a_1,a_2}{b_1,b_2}{z}$ is the generalised hypergeometric function. 
\end{lemma}

\begin{proof} This result is proved using
\[\deriv[3]{\mu_0}{\tau}-\frac{4}{9}\tau^2\deriv[2]{\mu_0}{\tau}-\frac{4}{3}\tau\deriv{\mu_0}{\tau}-\frac{7}{36}\mu_0 
=-\tfrac{1}{18}\int_0^\infty \deriv{}{s}\left\{\left(6s^3+4\tau s^2+7\right) s^{1/2}\exp\left(-s^3+\tau s^2\right)\right\}\rmd s=0.\]
Hence $\mu_0(\tau)$ satisfies the third order equation
\[\deriv[3]{\mu_0}{\tau}-\frac{4}{9}\tau^2\deriv[2]{\mu_0}{\tau}-\frac{4}{3}\tau\deriv{\mu_0}{\tau}-\frac{7}{36}\mu_0=0.\]
which has general solution
\begin{align*}\mu_0(\tau;0)&=c_1\; \pFq22{\tfrac{1}{12},\tfrac7{12}}{\tfrac{1}{3},\tfrac{2}{3}}{\frac{4\tau^3}{27}} + c_2\tau\; \pFq22{\tfrac5{12},\tfrac{11}{12}}{\tfrac{2}{3},\tfrac{4}{3}}{\frac{4\tau^3}{27}}
+ c_3\tau^2 \; \pFq22{\tfrac{3}{4},\tfrac{5}{4}}{\tfrac{4}{3},\tfrac{5}{3}}{\frac{4\tau^3}{27}},
\end{align*}
with $c_1$, $c_2$ and $c_3$ arbitrary constants. 
The initial conditions are
\[\mu_0(0;0)=\tfrac{1}{3}\Gamma(\tfrac{1}{6}),\qquad\deriv{\mu_0}{\tau}(0;0)=\tfrac{1}{3}\Gamma(\tfrac{5}{6}),\qquad
\deriv[2]{\mu_0}{\tau}(0;0)=\tfrac{1}{6}\sqrt{\pi},\]
and so we obtain the solution \eqref{sol:mu00}, as required.
\end{proof}

\begin{remarks}{\rm 
\begin{enumerate}\item[]
\item For general $\k$, it can be shown that $\ph(\tau)=\mu_0(\tau;\k)$ satisfies the third order equation
\begin{align}
\deriv[3]{\ph}{\tau}
&+\frac{2\tau^2}{9}\left\{9\k-2 -\frac{54\k(3\k-1)}{4\k(3\k-1)\tau^3-3}\right\}\deriv[2]{\ph}{\tau}
+\tau\left\{\frac{(4\k-1)\k^2\tau^3}{3}+\frac{36\k^2-27\k+4}{4\k(3\k-1)\tau^3-3}\right\}\deriv{\ph}{\tau}\nonumber\\&
+\left\{\frac{(4\k-1)\k^2\tau^3}{3} -\k+\frac{5}{36}+\frac{1-6\k}{4\k(3\k-1)\tau^3-3}\right\}\ph=0.
\label{eq:mu0gen}\end{align}
It is clear that \eqref{eq:mu0gen} simplifies in the cases when $\k=\tfrac{1}{4}$, $\k=\tfrac{1}{3}$ and $\k=0$. At present we have no closed form solution for general $\k$. We note that unless $\k=\tfrac{1}{3}$ and $\k=0$, \eqref{eq:mu0gen} has three regular singular points at the roots of the cubic $4\k(3\k-1)\tau^3-3=0$. 
It is straightforward to show that \eqref{eq:mu0gen} has an irregular singular point at $\tau=\infty$ for all values of $\k$.
\item If $\k=\tfrac{1}{4}$ then \eqref{eq:mu0gen} simplifies to
\beq \deriv[3]{\ph}{\tau} +\frac{\tau^2(\tau^3-42)}{18(\tau^3+12)} \deriv[2]{\ph}{\tau} 
+\frac{2\tau}{\tau^3+12} \deriv{\ph}{\tau} -\frac{\tau^3-6}{9(\tau^3+12)}\ph=0,
\label{eq:mu014a}\eeq
which has general solution
\[ \ph(\tau) =\sqrt{\tau}
\left\{c_1 I_{1/6}\left(\frac{\tau^3}{108}\right) +c_2I_{-1/6}\left(\frac{\tau^3}{108}\right)\right\}\exp\left(-\frac{\tau^3}{108}\right)+c_3\tau^2,
\]
with $c_1$, $c_2$ and $c_3$ constants. The initial conditions are 
\[ \ph(0)= \tfrac{1}{3}\Gamma\!\left(\tfrac{1}{6}\right), \qquad
\deriv{\ph}{\tau}(0)=\tfrac{1}{3}\Gamma\!\left(\tfrac{5}{6}\right),\qquad 
\deriv[2]{\ph}{\tau}(0)=0,\]
and so, using \eqref{modBasymp},
$c_1=c_2=\tfrac{1}{9}{\pi\sqrt{6}}$ and $c_3=0$.
Equation \eqref{eq:mu014a} is related to \eqref{eq:mu014} as follows
\begin{align*}
\deriv[3]{\ph}{\tau} &+\frac{\tau^2(\tau^3-42)}{18(\tau^3+12)} \deriv[2]{\ph}{\tau} 
+\frac{2\tau}{\tau^3+12} \deriv{\ph}{\tau} -\frac{\tau^2-6}{9(\tau^3+12)}\ph \\ &
=\left(\deriv{}{\tau}-\frac{3\tau^2}{\tau^3+12}\right)\left(\deriv[2]{\ph}{\tau}+\tfrac{1}{18}{\tau^2}\deriv{\ph}{\tau}+\frac{\tau}{18}\ph \right)\\
&=(\tau^3+12)\deriv{}{\tau}\left\{\frac{1}{\tau^3+12}\left(\deriv[2]{\ph}{\tau}+\tfrac{1}{18}{\tau^2}\deriv{\ph}{\tau}+\frac{\tau}{18}\ph \right)\right\}.
\end{align*}
\item We note that
\begin{align*}\mu_{0}(\tau;\tfrac{1}{3})&
=\exp\left\{-\big(\tfrac{1}{3}\tau\big)^{3}\right\}\int_{-\infty}^\infty \exp\left\{-\big(x^2-\tfrac{1}{3}\tau\big)^{3}\right\}\dx ,
\end{align*}
and
$\widetilde{\mu}_{0}(\tau;\tfrac{1}{3})=\mu_{0}(\tau;\tfrac{1}{3})\exp\left\{\big(\tfrac{1}{3}\tau\big)^{3}\right\}$,
satisfies
\[\deriv[3]{\widetilde{\mu}_{0}}{\tau}-\frac{\tau^2}{9}\deriv[2]{\widetilde{\mu}_{0}}{\tau}-\frac{\tau}{3}\deriv{\widetilde{\mu}_{0}}{\tau}+\frac{1}{12}\widetilde{\mu}_{0}=0.\]
\item In the three cases $\k=\tfrac{1}{4}$, $\k=\tfrac{1}{3}$ and $\k=0$ the weight 
$\w(x;\tau,\k)= \exp\left\{-\left(x^6 -\tau x^4+\k\tau^2x^2\right)\right\}$,
takes a special form, i.e.\ the polynomial has a multiple root
\[\begin{array}{|c@{\quad}| @{\quad}c|} \hline
\k & \w(x;\tau,\k)\\ \hline
\tfrac{1}{4} & \exp\big\{-x^2\left(x^2-\tfrac{1}{2}\tau\big)^{2}\right\} \\[5pt] 
\tfrac{1}{3} & \exp\left\{-\big(x^2-\tfrac{1}{3}\tau\big)^{3}\right\} \exp\left\{-\big(\tfrac{1}{3}\tau\big)^{3}\right\}\\[5pt] 
0 & \exp\left\{-x^4(x^2-\tau)\right\} \\[2.5pt] \hline
\end{array}\]
\end{enumerate}
}\end{remarks}

\subsection{Moments of higher order: three special cases.}
Next we derive closed form expressions for the moments $\mu_{2n}(\tau;\k)$ in the special cases when $\k=\tfrac{1}{4}$, $\k=\tfrac{1}{3}$ and $\k=0$.

\begin{lemma}\label{lem:mun14}
The moment $\mu_{2n}(\tau;\tfrac{1}{4})$ is given by
\begin{align}\mu_{2n}(\tau;\tfrac{1}{4})=
\tfrac{1}3\Gamma\!\left(\tfrac{1}{3}{n}+\tfrac{1}{6}\right)& \pFq{2}{2}{\tfrac{1}{3}-\tfrac{1}{3}{n},\tfrac{1}{3}+\tfrac{2}{3}{n}}{\tfrac{1}{3},\tfrac{2}{3}}{-\frac{\tau^3}{54}} +\tfrac{1}{3}{\tau}\Gamma\!\left(\tfrac{1}{3}{n}+\tfrac{5}{6}\right)\pFq{2}{2}{\tfrac{2}{3}-\tfrac{1}{3}{n},\tfrac{2}{3}+\tfrac{2}{3}{n}}{\tfrac{2}{3},\tfrac{4}{3}}{-\frac{\tau^3}{54}}\nonumber\\
& +\frac{n\tau^2}{18}
\Gamma\!\left(\tfrac{1}{3}{n}+\tfrac{1}{2}\right)\pFq{2}{2}{1-\tfrac{1}{3}{n},1+\tfrac{2}{3}{n}}{\tfrac{4}{3},\tfrac{5}{3}}{-\frac{\tau^3}{54}}.
\label{sol:mun14}
\end{align}
\end{lemma}
\begin{proof}
This result is proved using
\begin{align*}
18\deriv[3]{\mu_{2n}}{\tau}&+{\tau^2}\deriv[2]{\mu_{2n}}{\tau}+{(n+3)\tau}\deriv{\mu_{2n}}{\tau}-{(2n+1)(n-1)}\mu_{2n}\\
&=-\int_0^\infty 
\deriv{}{s}\left\{\left(6s^3-5\tau s^2+\tau^2 s+2n-2\right)s^{n+1/2}
\exp\left(-s^3+\tau s^2-\tfrac{1}{4}\tau^2 s\right)\right\}\rmd s=0.
\end{align*}
Hence the moment $\mu_{2n}(\tau;\tfrac{1}{4})$ satisfies the third order equation
\beq 18\deriv[3]{\mu_{2n}}{\tau}+{\tau^2}\deriv[2]{\mu_{2n}}{\tau}+{(n+3)\tau}\deriv{\mu_{2n}}{\tau}-{(2n+1)(n-1)}\mu_{2n}=0,\label{eq:mu2n14}\eeq
which has general solution
\begin{align*}\mu_{2n}(\tau;\tfrac{1}{4})&= c_1\, \pFq{2}{2}{\tfrac{1}{3}-\tfrac{1}{3}{n},\tfrac{1}{3}+\tfrac{2}{3}{n}}{\tfrac{1}{3},\tfrac{2}{3}}{-\frac{\tau^3}{54}} +c_2\,\tau\,\pFq{2}{2}{\tfrac{2}{3}-\tfrac{1}{3}{n},\tfrac{2}{3}+\tfrac{2}{3}{n}}{\tfrac{2}{3},\tfrac{4}{3}}{-\frac{\tau^3}{54}} \nonumber\\ &\qquad 
+c_3\,\tau^2\,\pFq{2}{2}{1-\tfrac{1}{3}{n},1+\tfrac{2}{3}{n}}{\tfrac{4}{3},\tfrac{5}{3}}{-\frac{\tau^3}{54}},
\end{align*}
with $c_1$, $c_2$ and $c_3$ constants.
The initial conditions are
\begin{align*}&\mu_{2n}(0;\tfrac{1}{4})
= \tfrac{1}3\Gamma\!\left(\tfrac{1}{3}{n}+\tfrac{1}6\right), \qquad
\deriv{\mu_{2n}}{\tau}(0;\tfrac{1}{4})
=\tfrac{1}3\Gamma\!\left(\tfrac{1}{3}{n}+\tfrac{5}6\right),
\qquad\deriv[2]{\mu_{2n}}{\tau}(0;\tfrac{1}{4})
=\tfrac{1}{9}{n}\Gamma\!\left(\tfrac{1}{3}{n}+\tfrac{1}{2}\right),
\end{align*}
and so since $\pFq{2}{2}{a_1,a_2}{b_1,b_2}{0}=1$, then we obtain the solution \eqref{sol:mun14}, as required.
\end{proof}

\begin{remarks}{\rm 
\begin{enumerate}\item[]
\item Note that for all $n\in\Z^+$, one of the hypergeometric functions in $\mu_{2n}(\tau;\tfrac{1}{4})$ {given by} \eqref{sol:mun14} will be a polynomial since $\pFq{2}{2}{a_1,a_2}{b_1,b_2}{z}$ is a polynomial if one of $a_1$ or $a_2$ is a nonpositive integer.
\item When $n=0$, \eqref{eq:mu2n14} reduces to
\begin{align*}
18\deriv[3]{\mu_{0}}{\tau}&+{\tau^2}\deriv[2]{\mu_{0}}{\tau}+3{\tau}\deriv{\mu_{0}}{\tau}+\mu_{0} 
=\deriv{}{\tau}\left( 18\deriv[2]{\mu_{0}}{\tau}+{\tau^2}\deriv{\mu_{0}}{\tau}+{\tau}\mu_0\right),
\end{align*}
and \eqref{sol:mun14} becomes
\begin{align*}
\mu_{0}(\tau;\tfrac{1}{4})
&=\tfrac{1}{3}\Gamma\!\left(\tfrac{1}{6}\right) \pFq{2}{2}{\tfrac{1}{3},\tfrac{1}{3}}{\tfrac{1}{3},\tfrac{2}{3}}{-\frac{\tau^3}{54}} +\tfrac{1}{3}{\tau}\Gamma\!\left(\tfrac{5}{6}\right)\pFq{2}{2}{\tfrac{2}{3},\tfrac{2}{3}}{\tfrac{2}{3},\tfrac{4}{3}}{-\frac{\tau^3}{54}}\\
&=\tfrac{1}{3}\Gamma\!\left(\tfrac{1}{6}\right) \Fpqq{1}{1}{\tfrac{1}{3}}{\tfrac{2}{3}}{-\frac{\tau^3}{54}}+\tfrac{1}{3}{\tau}\Gamma\!\left(\tfrac{5}{6}\right)
\Fpqq{1}{1}{\tfrac{2}{3}}{\tfrac{4}{3}}{-\frac{\tau^3}{54}}\\
& =\frac{\pi\sqrt{6\tau}}{9} \left\{I_{1/6}\left(\frac{\tau^3}{108}\right) +I_{-1/6}\left(\frac{\tau^3}{108}\right) \right\}\exp\left(-\frac{\tau^3}{108}\right),
\end{align*}
since
\beq \pFq{2}{2}{a_1,a_2}{a_2,b_2}{z}=\Fpq{1}{1}{a_1}{b_2}{z}
\equiv M(a_1,b_2,z),\label{2F2_1F1}\eeq
with $M(a,b,z)$ the Kummer function and 
\beq \Fpq{1}{1}{\nu+\tfrac{1}{2}}{2\nu+1}{-2z}=
M(\nu+\tfrac{1}{2},2\nu+1,-2z)=\Gamma(\nu+1)\left(\tfrac{1}{2}z\right)^{\!\nu} \rme^{-z} I_\nu(z),\label{1F1_modB}\eeq
cf.~\cite[equation 10.39.5]{refDLMF}.
\end{enumerate}
}\end{remarks}

Whilst in Lemma \ref{lem:mun14} the moments $\mu_{2n}(\tau;\tfrac{1}{4})$ were expressed in terms of $\pFq{2}{2}{a_1,a_2}{b_1,b_2}{z}$ functions, the moments can be expressed in terms of modified Bessel functions as shown in the following Lemma.
\begin{lemma} The moment $\mu_{2n}(\tau;\tfrac{1}{4})$ has the form
\begin{align}\mu_{2n}(\tau;\tfrac{1}{4}) =\tfrac{1}{9}{\sqrt{6}\,\pi\tau^{1/2}}& \Big\{f_n(\tau)\big[I_{1/6}(\xi)+I_{-1/6}(\xi)\big]+g_n(\tau)\big[I_{5/6}(\xi)+I_{-5/6}(\xi)\big]\Big\}\,\rme^{-\xi} 
+\tfrac{1}{3}\sqrt{\pi}\,h_{n-1}(\tau),
\label{mun14bessel}\end{align}
with $\xi=\tau^3/108$,
where $I_{\nu}(z)$ is the modified Bessel function, and 
$f_n(\tau)$, $g_n(\tau)$ and $h_{n}(\tau)$ are polynomials of degree $n$. 
\end{lemma}
\begin{proof}We will prove this result using induction and the discrete equation \eqref{mr1}, with $t=-\tfrac{1}{4}\tau^2$, i.e.
\beq 3\mu_{2n+6}(\tau;\tfrac{1}{4}) - 2\tau \mu_{2n+4}(\tau;\tfrac{1}{4}) +\tfrac{1}{4}\tau^2 \mu_{2n+2}(\tau;\tfrac{1}{4})-(n+\tfrac{1}{2}) \mu_{2n}(\tau;\tfrac{1}{4}) = 0.
\label{mr1k} \eeq
First we need to show that \eqref{mun14bessel} holds when $n=0,1,2$.
From \eqref{sol:mu014}, it is clear that \eqref{mun14bessel} holds when $n=0$, with
\beq f_0(\tau)=1,\qquad g_0(\tau)=h_{-1}(\tau)=0.\label{ic:mu014}\eeq
If $n=1$ then from \eqref{sol:mun14}
\begin{align}
\mu_2(\tau;\tfrac{1}{4}) 
 &=\frac{\tau\Gamma(\tfrac{1}{6})}{18}\pFq{2}{2}{\tfrac{1}{3},\tfrac{4}{3}}{\tfrac{2}{3},\tfrac{4}{3}}{-2\xi}+\frac{\tau^2\Gamma(\tfrac{5}{6})}{18}\pFq{2}{2}{\tfrac{2}{3},\tfrac{5}{3}}{\tfrac{4}{3},\tfrac{5}{3}}{-2\xi}+\tfrac{1}{3}\sqrt{\pi},\qquad \xi=\frac{\tau^3}{108} \nonumber\\
&=\frac{\tau\Gamma(\tfrac{1}{6})}{18}\Fpqq{1}{1}{\tfrac{1}{3}}{\tfrac{2}{3}}{-2\xi}+\frac{\tau^2\Gamma(\tfrac{5}{6})}{18}\Fpqq{1}{1}{\tfrac{2}{3}}{\tfrac{4}{3}}{-2\xi}+\tfrac{1}{3}\sqrt{\pi}\nonumber\\
&=\tfrac{1}{9}{\sqrt{6}\,\pi\tau^{1/2}} \left\{\tfrac{1}{6}\tau[I_{1/6}(\xi) +I_{-1/6}(\xi)]\right\}{\rme^{-\xi}}+\tfrac{1}{3}\sqrt{\pi},\label{sol:mu214}
\end{align}
using \eqref{2F2_1F1} and \eqref{1F1_modB}, and therefore \eqref{mun14bessel} holds when $n=1$, with
\beq f_1(\tau)=\tfrac{1}{6}\tau,\qquad g_1(\tau)=0,\qquad h_0(\tau)=1.\label{ic:mu214}\eeq
If $n=2$ then from \eqref{sol:mun14}
\begin{align*}
\mu_4(\tau;\tfrac{1}{4}) &= \tfrac{1}{3}{\Gamma(\tfrac{5}{6})}\,\pFq{2}{2}{-\tfrac{1}{3},\tfrac{5}{3}}{\tfrac{1}{3},\tfrac{2}{3}}{-2\xi}+\frac{\tau^2\Gamma(\tfrac{1}{6})}{54}\pFq{2}{2}{\tfrac{1}{3},\tfrac73}{\tfrac{4}{3},\tfrac{5}{3}}{-2\xi}+\tfrac{1}{6}\sqrt{\pi}\,\tau,\qquad \xi=\frac{\tau^3}{108}\\
&= \tfrac{1}{3}{\Gamma(\tfrac{5}{6})}\left\{M\! \left(-\tfrac{1}{3}, \tfrac{1}{3}, -2\xi\right)+\frac{\tau^{3}}{36} M\! \left(\tfrac{2}{3}, \tfrac{4}{3}, -2\xi\right)\right\}
\\&\qquad 
+\frac{\tau^2\Gamma(\tfrac{1}{6})}{54}\left\{M\! \left(\tfrac{1}{3}, \tfrac{5}{3}, -2\xi\right)-\frac{\tau^{3}}{360}M\! \left(\tfrac{4}{3}, \tfrac{8}{3}, -2\xi\right)\right\}+\tfrac{1}{6}\sqrt{\pi}\,\tau\\ 
&=\frac{\sqrt{6}\,\pi\tau^{5/2}}{324}\left\{2I_{1/6}(\xi)+I_{-5/6}(\xi)\right\}\rme^{-\xi} +\frac{\sqrt{6}\,\pi\tau^{5/2}}{324} \left\{2I_{-1/6}(\xi)+I_{5/6}(\xi)\right\}\rme^{-\xi}
+\tfrac{1}{6}\sqrt{\pi}\,\tau\\
&=\frac{\sqrt{6}\,\pi\tau^{5/2}}{324}\left\{2\!\left[I_{1/6}(\xi)+I_{-1/6}(\xi)\right]+I_{5/6}(\xi)+I_{-5/6}(\xi)\right\}\rme^{-\xi}+\tfrac{1}{6}\sqrt{\pi}\,\tau,
\end{align*}
using
\begin{subequations}\label{2F2-1F1}\begin{align}
\pFq{2}{2}{-\tfrac{1}{3},\tfrac{5}{3}}{\tfrac{1}{3},\tfrac{2}{3}}{-2\xi}&=M\! \left(-\tfrac{1}{3}, \tfrac{1}{3}, -2\xi\right)+3\xi M\! \left(\tfrac{2}{3}, \tfrac{4}{3}, -2\xi\right),\\ 
\pFq{2}{2}{\tfrac{1}{3},\tfrac73}{\tfrac{4}{3},\tfrac{5}{3}}{-2\xi}&=M\! \left(\tfrac{1}{3}, \tfrac{5}{3}, -2\xi\right)-\frac{3\xi}{10}M\! \left(\tfrac{4}{3}, \tfrac{8}{3}, -2\xi\right),
\end{align}\end{subequations} and
\begin{subequations}\label{DLMF13.6.11a}\begin{align}
M\! \left(-\tfrac{1}{3}, \tfrac{1}{3}, -2 \xi \right) &
=\Gamma\!\left(\tfrac{1}{6}\right) (\tfrac{1}{2}\xi)^{5/6} \left\{I_{1/6}(\xi)+I_{-5/6}(\xi)\right\}\rme^{-\xi},\\
M\! \left(\tfrac{1}{3}, \tfrac{5}{3}, -2 \xi \right) &
=\Gamma\!\left(\tfrac{5}{6}\right) (\tfrac{1}{2}\xi)^{1/6} \left\{I_{5/6}(\xi)+I_{-1/6}(\xi)\right\}\rme^{-\xi},
\end{align}\end{subequations}
together with \eqref{1F1_modB}.
The identities \eqref{2F2-1F1} follow from 
\beq\pFq{2}{2}{a,c+1}{b,c}{-2z} = M\! \left(a, b, -2z\right)-\frac{2a z}{b\, c} M\! \left(a+1, b+1, -2z\right),\nonumber\eeq
see \cite[\S7.12.1]{refPrudBM3},
which also is a special case of equation (2.7) in \cite[Lemma 4]{refMillerP}.
The identities \eqref{DLMF13.6.11a} follow from
\beq M\left(\nu+\tfrac{1}{2},2\nu+2,-2z\right)= \Gamma\!\left(\nu +1\right) (\tfrac{1}{2}z)^{-\nu} \left\{I_{\nu}(z)+I_{\nu +1}(z)\right\}\rme^{-z},\nonumber\eeq
which is a special case of equation (13.6.11\_1) in the DLMF \cite{refDLMF}, i.e.
\beq M\left(\nu+\tfrac{1}{2},2\nu+1+n,-2z\right)=\Gamma\left(\nu\right)\rme^{-z}\left(\tfrac{1}{2}{z}\right)^{-\nu}\sum_{k=0}^{n}\frac{{\left(n\right)_{k}}{\left(2\nu\right)_{k}}(\nu+k)}{{\left(2\nu+1+n\right)_{k}}k!}I_{\nu+k}(z).\nonumber\eeq
Hence
\begin{align}\mu_4(\tau;\tfrac{1}{4}) 
&=\tfrac{1}{9}{\sqrt{6}\,\pi\tau^{1/2}} \left\{\tfrac{1}{18}\tau^2\!\left[I_{1/6}(\xi) +I_{-1/6}(\xi)\right] 
+\tfrac{1}{36}\tau^2\left[ I_{5/6}(\xi) +I_{-5/6}(\xi) \right]\!\right\}{\rme^{-\xi}}+\tfrac{1}{6}\sqrt{\pi}\,\tau,\label{sol:mu414}
\end{align}
and therefore \eqref{mun14bessel} holds when $n=2$, with
\beq f_2(\tau)=\tfrac{1}{18}{\tau^2},\qquad g_2(\tau)=\tfrac{1}{36}{\tau^2},\qquad h_1(\tau)=\tfrac{1}{2}{\tau}.\label{ic:mu414}\eeq

Substituting \eqref{mun14bessel} into the discrete equation \eqref{mr1k} shows that $f_n(\tau)$, $g_n(\tau)$ and $h_{n}(\tau)$ satisfy the discrete equations
\begin{subequations}\label{sys:fgh} \begin{align}
&3f_{n+3}(\tau) - 2\tau f_{n+2}(\tau) +\tfrac{1}{4}\tau^2 f_{n+1}(\tau)-(n+\tfrac{1}{2})f_{n}(\tau) = 0,\\
&3g_{n+3}(\tau) - 2\tau g_{n+2}(\tau) +\tfrac{1}{4}\tau^2 g_{n+1}(\tau)-(n+\tfrac{1}{2})g_{n}(\tau) = 0,\\
&3h_{n+2}(\tau) - 2\tau h_{n+1}(\tau) +\tfrac{1}{4}\tau^2 h_{n}(\tau)-(n+\tfrac{1}{2})h_{n-1}(\tau) = 0.
\end{align}\end{subequations}
From \eqref{ic:mu014}, \eqref{ic:mu214} and \eqref{ic:mu414}, we have the initial conditions
\beq f_0=1,\quad f_1=\tfrac{1}{6}{\tau},\quad f_2=\tfrac{1}{18}{\tau^2},\qquad g_0=g_1=0,\quad g_2=\tfrac{1}{36}{\tau^2},\qquad h_{-1}=0,\quad h_0=1,\quad h_1=\tfrac{1}{2}{\tau}.
\label{ics:fgh}\eeq
Therefore from \eqref{sys:fgh} and \eqref{ics:fgh} it follows that $f_n(\tau)$, $g_n(\tau)$ and $h_{n}(\tau)$ are polynomials of degree $n$, provided that the coefficient of $\tau^n$, for $n\geq2$, is nonzero for these polynomials. To show this, suppose that 
\[ f_n(\tau)=\sum_{j=0}^n a_{n,j}\tau^j,\qquad g_n(\tau)=\sum_{j=0}^n b_{n,j}\tau^j,\qquad h_n(\tau)=\sum_{j=0}^n c_{n,j}\tau^j,\]
then from the coefficient of the highest power of $\tau$ in \eqref{sys:fgh}, it follows that $a_{n,n}$, $b_{n,n}$ and $c_{n,n}$ respectively satisfy
\begin{subequations}\label{sys:abc} 
\begin{align}
&3a_{n+3,n+3}-2a_{n+2,n+2}+\tfrac{1}{4}a_{n+1,n+1}=0,\\
&3b_{n+3,n+3}-2b_{n+2,n+2}+\tfrac{1}{4}b_{n+1,n+1}=0,\\
&3c_{n+2,n+2}-2c_{n+1,n+1}+\tfrac{1}{4}c_{n,n}=0,
\end{align}
with
\beq a_{1,1}=\tfrac{1}{6},\quad a_{2,2}=\tfrac{1}{18},\qquad b_{1,1}=0,\quad b_{2,2}=\tfrac{1}{36},\qquad c_{0,0}=1,\quad c_{1,1}=\tfrac{1}{2},\eeq\end{subequations}
and solving \eqref{sys:abc} gives
\[
a_{n,n}=\tfrac{1}{6}\!\left(\tfrac{1}{2}\right)^n+\tfrac{1}{2}\!\left(\tfrac{1}{6}\right)^n,\qquad b_{n,n}=\tfrac{1}{6}\!\left(\tfrac{1}{2}\right)^n-\tfrac{1}{2}\!\left(\tfrac{1}{6}\right)^n,\qquad c_{n,n}=\left(\tfrac{1}{2}\right)^n,\]
which are nonzero for $n\geq2$, as required. Hence the result follows by induction.
\end{proof}

\begin{example}{\rm
Using the discrete equation \eqref{mr1k} with $\mu_{0}(\tau;\tfrac{1}{4})$, $\mu_{2}(\tau;\tfrac{1}{4})$ and $\mu_{4}(\tau;\tfrac{1}{4})$ given by \eqref{sol:mu014}, \eqref{sol:mu214} and \eqref{sol:mu414} respectively, we obtain
\begin{align*} \mu_6(\tau;\tfrac{1}{4}) &=
\frac{\Gamma(\tfrac{1}{6})}{18}\,\pFq{2}{2}{-\tfrac{2}{3},\tfrac73}{\tfrac{1}{3},\tfrac{2}{3}}{-2\xi}+\frac{5\tau\Gamma(\tfrac{5}{6})}{18}\pFq{2}{2}{-\tfrac{1}{3},\tfrac83}{\tfrac{2}{3},\tfrac{4}{3}}{-2\xi}+\frac{\sqrt{\pi}\,\tau^2}{12}\\
&=\frac{\sqrt{6}\,\pi\tau^{1/2}}{9} \left\{\frac{5\tau^{3}+36}{216}\!\left[I_{1/6}(\xi) +I_{-1/6}(\xi) \right]
+\frac{4\tau^3}{54}\!\left[I_{5/6}(\xi) +I_{-5/6}(\xi) \right]\!\right\}\rme^{-\xi}+\frac{\sqrt{\pi}\,\tau^2}{12},\\[5pt]
\mu_8(\tau;\tfrac{1}{4})&=\frac{7\tau\Gamma(\tfrac{1}{6})}{108}\pFq{2}{2}{-\tfrac{2}{3},\tfrac{10}3}{\tfrac{2}{3},\tfrac{4}{3}}{-2\xi}
+\frac{5\tau^2\Gamma(\tfrac{5}{6})}{27}\,\pFq{2}{2}{-\tfrac{1}{3},\tfrac{11}3}{\tfrac{4}{3},\tfrac{5}{3}}{-2\xi}+\frac{\sqrt{\pi}\left(\tau^3+4\right)}{24}\\
&=\frac{\sqrt{6}\,\pi\tau^{1/2}}{9} \left\{\frac{7\tau(\tau^{3}+18)}{648}\!\!\left[I_{1/6}(\xi) +I_{-1/6}(\xi) \right]
+\frac{13\tau^4}{1296}\!\!\left[I_{5/6}(\xi) +I_{-5/6}(\xi)\right] \!\right\}\rme^{-\xi}+\frac{\sqrt{\pi}(\tau^3+4)}{24},\\[5pt]
\mu_{10}(\tau;\tfrac{1}{4})&=
\frac{5\Gamma(\tfrac{5}{6})}{18}\pFq{2}{2}{-\tfrac{4}{3},\tfrac{11}3}{\tfrac{1}{3},\tfrac{2}{3}}{-2\xi}
+\frac{35\tau^2\Gamma(\tfrac{1}{6})}{648}\,\pFq{2}{2}{-\tfrac{2}{3},\tfrac{13}3}{\tfrac{4}{3},\tfrac{5}{3}}{-2\xi}
+\frac{\sqrt{\pi}\,\tau\left(\tau^3+12\right)}{48}\\
&=\frac{\sqrt{6}\,\pi\tau^{1/2}}{9} \left\{\frac{\tau^2(41\tau^{3}+1260)}{7776}\!\left[I_{1/6}(\xi) +I_{-1/6}(\xi) \right]
+\frac{5\tau^2(2\tau^3+9)}{1944}\!\left[I_{5/6}(\xi) +I_{-5/6}(\xi)\right] \!\right\}\rme^{-\xi}\\ &\qquad
+\frac{\sqrt{\pi}\,\tau(\tau^3+12)}{48},
\end{align*} with $\xi=\tau^3/108$.
}\end{example}

In the next two lemmas we derive closed form expressions for $\mu_{2n}(\tau;\tfrac{1}{3})$ and $\mu_{2n}(\tau;0)$.
\begin{lemma}\label{lem:mun13}
The moment $\mu_{2n}(\tau;\tfrac{1}{3})$ is given by
\begin{align*}\mu_{2n}(\tau;\tfrac{1}{3})= \Bigg\{
\tfrac{1}{3}&\Gamma \!\left(\tfrac{1}{3}{n}+\tfrac{1}{6}\right) \pFq{2}{2}{\tfrac{1}{6}-\tfrac{1}{3}{n},\tfrac{1}{2}-\tfrac{1}{3}{n}}{\tfrac{1}{3},\tfrac{2}{3}}{\frac{\tau^3}{27}} +\tfrac{1}{3}{\tau}\Gamma \!\left(\tfrac{1}{3}{n}+\tfrac{5}{6}\right)\pFq{2}{2} {\tfrac{1}{2}-\tfrac{1}{3}{n},\tfrac{5}{6}-\tfrac{1}{3}{n}}{\tfrac{2}{3},\tfrac{4}{3}}{\frac{\tau^3}{27}} \nonumber\\
&\quad +\frac{(2n-1)\tau^2}{36}
\Gamma \!\left(\tfrac{1}{3}{n}+\tfrac{1}{2}\right)\pFq{2}{2}{\tfrac{5}{6}-\tfrac{1}{3}{n},\tfrac{7}{6}-\tfrac{1}{3}{n}}{\tfrac{4}{3},\tfrac{5}{3}}{\frac{\tau^3}{27}} \Bigg\}\exp\left(-\frac{\tau^3}{27}\right).
\end{align*}
\end{lemma}
\begin{proof}
This result is proved using
\begin{align*}
\deriv[3]{\mu_{2n}}{\tau}&+\frac{2\tau^2}{9}\deriv[2]{\mu_{2n}}{\tau}+\frac{\tau(\tau^3+18n+27)}{81}\deriv{\mu_{2n}}{\tau}+\frac{(2n+1)(4\tau^3-18n+45)}{324}\mu_{2n}\\
&=-\frac{1}{162}\int_0^\infty 
\deriv{}{s}\left\{\left(54s^3-72\tau s^2+30\tau^2 s-4\tau^3+18n-45\right)s^{n+1/2}
\exp\left(-s^3+\tau s^2-\tfrac{1}{3}\tau^2 s\right)\right\}\rmd s=0,
\end{align*}
and the initial conditions
\begin{align*}&\mu_{2n}(0;\tfrac{1}{3})
= \tfrac{1}{3}\Gamma\!\left(\tfrac{1}{3}{n}+\tfrac{1}{6}\right), \qquad
\deriv{\mu_{2n}}{\tau}(0;\tfrac{1}{3})
=\tfrac{1}{3}\Gamma\!\left(\tfrac{1}{3}{n}+\tfrac{5}{6}\right),
\qquad\deriv[2]{\mu_{2n}}{\tau}(0;\tfrac{1}{3})
=\frac{2n-1}{18} \Gamma \left(\tfrac{1}{3}{n}+\tfrac{1}{2}\right).
\end{align*}
\end{proof}

\begin{lemma}\label{lem:mun0}
The moment $\mu_{2n}(\tau;0)$ is given by
\begin{align*}\mu_{2n}(\tau;0)=
\tfrac{1}{3}&\Gamma \!\left(\tfrac{1}{3}{n}+\tfrac{1}{6}\right) \pFq{2}{2}{\tfrac{1}{6}{n}+\tfrac{1}{12},\tfrac{1}{6}{n}+\tfrac{7}{12}}{\tfrac{1}{3},\tfrac{2}{3}}{\frac{4\tau^3}{27}} +\tfrac{1}{3}{\tau}\Gamma \!\left(\tfrac{1}{3}{n}+\tfrac{5}{6}\right)\pFq{2}{2}{\tfrac{1}{6}{n}+\tfrac{5}{12},\tfrac{1}{6}{n}+\tfrac{11}{12}}{\tfrac{2}{3},\tfrac{4}{3}}{\frac{4\tau^3}{27}}\nonumber\\ &\qquad\qquad
+\tfrac{1}{6}\tau^2
\Gamma \!\left(\tfrac{1}{3}{n}+\tfrac{3}{2}\right)\pFq{2}{2}{\tfrac{1}{6}{n}+\tfrac{3}{4},\tfrac{1}{6}{n}+\tfrac{5}{4}}{\tfrac{4}{3},\tfrac{5}{3}}{\frac{4\tau^3}{27}}. 
\end{align*}
\end{lemma}
\begin{proof}
This result is proved using
\begin{align*}
\deriv[3]{\mu_{2n}}{\tau}&-\frac{4\tau^2}{9}\deriv[2]{\mu_{2n}}{\tau}-\frac{4(n+3)\tau}{9}\deriv{\mu_{2n}}{\tau}-\frac{(2n+1)(2n+7)}{36}\mu_{2n}\\
&=-\tfrac{1}{18}\int_0^\infty 
\deriv{}{s}\left\{(6s^3+4\tau s^2+2n+7)\,s^{n+1/2}\exp\left(-s^3+\tau s^2\right)\right\}\rmd s=0,
\end{align*}
and the initial conditions
\begin{align*}&\mu_{2n}(0;0)
= \tfrac{1}{3}\Gamma\!\left(\tfrac{1}{3}{n}+\tfrac{1}{6}\right), \qquad
\deriv{\mu_{2n}}{\tau}(0;0)
=\tfrac{1}{3}\Gamma\!\left(\tfrac{1}{3}{n}+\tfrac{5}{6}\right),
\qquad\deriv[2]{\mu_{2n}}{\tau}(0;0)
=\tfrac{1}{3}\,\Gamma \!\left(\tfrac{1}{3}{n}+\tfrac{3}{2}\right).
\end{align*}
\end{proof}

\begin{remark}{\rm For general $\k$, it can be shown that $\ph_{n}(\tau)=\mu_{2n}(\tau;\k)$ satisfies the third order equation
\begin{align}
\deriv[3]{\ph_{n}}{\tau}&+
\frac{2\tau^{2}}{9} \left\{9 \k -2-\frac{54 \k (4\k-1) (3\k-1)}{4 \k (4\k-1) (3\k-1)\tau^{3}-12 \k +2 n +3}
\right\}\deriv[2]{\ph_{n}}{\tau}\nonumber\\
&+\frac{\tau}{9}\left\{3 \k^{2}\tau^{3} (4\k-1)+2 n (9 \k -2)+\frac{3 (36 \k^{2}-27 \k +4) (12 \k -2 n -3)}{4 \k (4\k-1) (3\k-1)\tau^{3}-12 \k +2 n +3}
\right\}\deriv{\ph_{n}}{\tau}
\nonumber\\
&+\frac{2 n +1}{36} \left\{
12 \k^{2}\tau^{3} (4\k-1)-36 \k -2 n +5-\frac{12 (6 \k -1) (12 \k -2 n -3)}{4 \k (4\k-1) (3\k-1)\tau^{3}-12 \k +2 n +3}
\right\}\ph_{n}=0.
\label{eq:mungen}
\end{align}
At present we have no closed form solution for this equation, except in the three cases discussed above.
We note that unless $\k=\tfrac{1}{3}$, $\k=\tfrac{1}{4}$ or $\k=0$, \eqref{eq:mungen} has three regular singular points at the roots of the cubic \[4 \k (4\k-1) (3\k-1)\tau^{3}-12 \k +2 n +3=0.\] 
Further, \eqref{eq:mungen} has an irregular singular point at $\tau=\infty$ for all values of $\k$.
}\end{remark}

\subsection{Higher order moments: series expansions for general $\k$.}
For values of $\k$ other than $0$, $\tfrac{1}{4}$ or $\tfrac{1}{3}$ we have a series representation for the moments $\mu_n(\tau;\kappa)$ in terms of Laguerre polynomials as well as in terms of Jacobi polynomials with varying parameters. 

\begin{theorem} For $\tau,\,\k \in \Real$, we have 
\begin{align}\label{mu2n series2}
\mu_{2n}(\tau;\k) 
 &= \tfrac{1}{3} \sum_{j=0}^{\infty} \left\{\frac{\Gamma\left(\tfrac{2}{3}j+\tfrac{1}{3}n+\tfrac{1}{6}\right)}{\left(\tfrac{1}{2}\right)_j}\,\tau^j\,L_j^{(-1/2)}(\zeta)
 - \k\frac{\Gamma\left(\tfrac{2}{3}j+\tfrac{1}{3}n+\tfrac{1}{2}\right)}{\left(\tfrac32\right)_j}\,\tau^{j+2}\, L_j^{(1/2)}(\zeta) \right\},
\end{align}
where $L_j^{(\alpha)}\left(\zeta\right)$ 
are the Laguerre polynomials of parameter $\alpha$ and $\zeta=-\frac{1}{4} \k^2\tau^3$. 
\end{theorem}
\begin{proof}
Consider the Taylor series expansion about $s=0$ of the function 
\[
\exp\left( -\k\tau^2 s +\tau s^2\right)
= \sum_{m=0}^\infty C_{m}(\tau;\k)\frac{s^m}{m!}
\]
where $C_{m}(\tau;\k)$ are polynomials in $\tau$ and $\k$ given by 
\begin{align*}
C_{m}(\tau;\k) &=\sum_{j=\lceil m/2 \rceil}^m\frac{(-1)^m m! \k^{2 j-m}\tau ^{3 j-m}}{(2 j-m)!\, (m-j)!} \\
 & =\frac{(-1)^m m! \k^{2 \left\lceil \tfrac{1}{2}m\right\rceil -m}\tau 
 ^{3 \left\lceil \tfrac{1}{2}m\right\rceil -m}}{\left(m-\left\lceil \tfrac{1}{2}m\right\rceil \right)! \left(2
 \left\lceil \tfrac{1}{2}m\right\rceil -m\right)!} 
\, \pFq{2}{2}{1,\left\lceil
 \tfrac{1}{2}m\right\rceil -m}{-\tfrac{1}{2}m+\left\lceil
 \tfrac{1}{2}m\right\rceil +\tfrac{1}{2},-\tfrac{1}{2}m+\left\lceil
 \tfrac{1}{2}m\right\rceil +1}{-\tfrac{1}{4}{\k^2\tau^3}}. 
\end{align*}
Hence, we have 
\begin{align*}
C_{m}(\tau;\k)
 &
 = \begin{cases}
 j!\,2^{2j}\,\tau^{j}
\, L_j^{(-1/2)}\left(-\tfrac{1}{4}{\k^2\tau^3} \right),&\quad \text{if} \quad m=2j, \\
 - j!\,2^{2j}\,\k\,\tau^{j+2}
\, L_j^{(1/2)}\left(-\tfrac{1}{4}{\k^2\tau^3} \right),&\quad \text{if} \quad m=2j +1, 
 \end{cases}
\end{align*}
where $L_j^{(\alpha)}\left(z\right) =\frac{(\alpha)_j}{j!}\Fpq{1}{1}{-j}{\alpha+1}{z}$ are the Laguerre polynomials of parameter $\alpha>-1$. 

In order to study the radius of convergence of the series 
\begin{equation}\label{mu2n series proof}
 \sum_{m=0}^{\infty}\frac{C_{m}(\tau;\k)}{m!} \int_0^\infty s^{m+n-1/2} \exp(-s^3)\,\rmd s
= \tfrac{1}{3} \sum_{m=0}^{\infty} \Gamma\left(\tfrac{1}{3}m+\tfrac{1}{3}n+\tfrac{1}{6}\right)\frac{C_{m}(\tau;\k)}{m!}, 
\end{equation}
we analyse the ratio 
\[
\rho_m =\left|\frac{\Gamma\left(\tfrac{1}{3}m+\tfrac{1}{3}n+\tfrac{1}{2}\right) C_{m+1}(\tau;\k)}{(m+1)\,\Gamma\left(\tfrac{1}{3}m+\tfrac{1}{3}n+\tfrac{1}{6}\right) C_{m}(\tau;\k)}\right| 
\]
as $m\to \infty$. Observe that the two subsequences of $(\rho_m)_{m\geq 0}$ of even and odd order are respectively given by 
\begin{align*}
\rho_{2j} 
&= \left|\k\,\tau^{2}\,\frac{\Gamma\left(\tfrac{2}{3}j+\tfrac{1}{3}n+\tfrac{1}{2}\right)}{(2j+1)\Gamma\left(\tfrac{2}{3}j+\tfrac{1}{3}n+\tfrac{1}{6}\right)}
\frac{L_j^{(1/2)}\left(-\tfrac{1}{4}{\k^2\tau^3} \right)}
 {L_j^{(-1/2)}\left(-\tfrac{1}{4}{\k^2\tau^3} \right)}\right| 
\end{align*}
and 
\begin{align*}
\rho_{2j+1} 
= & 
 \left|\frac{2}{\k\,\tau}\,\frac{\Gamma\left(\tfrac{2}{3}j+\tfrac{1}{3}n+\tfrac{5}{6}\right)}{\Gamma\left(\tfrac{2}{3}j+\tfrac{1}{3}n+\tfrac{1}{2}\right)}
\frac{L_{j+1}^{(-1/2)}\left(-\tfrac{1}{4}{\k^2\tau^3} \right)}
 {L_j^{(1/2)}\left(-\tfrac{1}{4}{\k^2\tau^3} \right)}\right|.
\end{align*}
Recall the asymptotic behaviour for the Gamma function, see e.g.\ \cite[Eq.5.11.12]{refDLMF}, to conclude that 
\[
\k\,\tau^{2}\,\frac{\Gamma\left(\tfrac{2}{3}j+\tfrac{1}{3}n+\tfrac{1}{2}\right)}{(2j+1)\Gamma\left(\tfrac{2}{3}j+\tfrac{1}{3}n+\tfrac{1}{6}\right)}
\sim\frac{\k\,\tau^{2}}{\sqrt[3]{12}\,j^{2/3}}
\quad \text{and}\quad 
\frac{2}{\k\,\tau}\,\frac{\Gamma\left(\tfrac{2}{3}j+\tfrac{1}{3}n+\tfrac{5}{6}\right)}{\Gamma\left(\tfrac{2}{3}j+\tfrac{1}{3}n+\tfrac{1}{2}\right)}
 \sim\frac{2}{\k\,\tau}\left(\frac{2j}{3} \right)^{\!1/3} 
\]
as $j\to\infty$. We use \cite[Theorem 8.22.2]{refSzego} for $\tau<0$ and \cite[Theorem 8.22.3]{refSzego} for $\tau>0$ to obtain the following asymptotic behaviour when $\k\neq 0$ 
\[
\frac{L_j^{(1/2)}\left(-\tfrac{1}{4}{\k^2\tau^3} \right)}{L_j^{(-1/2)}\left(-\tfrac{1}{4}{\k^2\tau^3} \right)}
 \sim \left(\frac{\k^2 |\tau|^3}{4j}\right)^{\!-{1/2}} =\frac{2\sqrt{j}}{|\k|\,|\tau|^{3/2}},
 \qquad\text{as}\quad j\to\infty,
\]
and 
\[
\frac{L_{j+1}^{(-1/2)}\left(-\tfrac{1}{4}{\k^2\tau^3} \right)}
 {L_j^{(1/2)}\left(-\tfrac{1}{4}{\k^2\tau^3} \right)}
 \sim \left(\frac{\k^2 |\tau|^3}{4j}\right)^{\!1/2} =\frac{|\k|\,|\tau|^{3/2}}{2\sqrt{j}},
 \qquad\text{as}\quad j\to\infty.
\]
Hence, for fixed $\tau\neq0$ and $\k\neq 0$, we have 
\begin{align*}
\rho_{2j} 
&= \left|\k\,\tau^{2}\,\frac{\Gamma\left(\tfrac{2}{3}j+\tfrac{1}{3}n+\tfrac{1}{2}\right)}{(2j+1)\Gamma\left(\tfrac{2}{3}j+\tfrac{1}{3}n+\tfrac{1}{6}\right)}
\frac{L_j^{(1/2)}\left(-\tfrac{1}{4}{\k^2\tau^3} \right)}
 {L_j^{(-1/2)}\left(-\tfrac{1}{4}{\k^2\tau^3} \right)}\right| 
\\ &\sim\left|\frac{\k\,\tau^{2}}{\sqrt[3]{12}\,j^{2/3}}\,\frac{2\sqrt{j}}{\k\,|\tau|^{3/2}}\right| 
= \left|\frac{\sqrt[3]{\tfrac{2}{3}}\,\sqrt{|\tau|}}{\,j^{1/6}}\right| \longrightarrow 0,\qquad\text{as}\quad j\to\infty, 
\end{align*} 
and 
\begin{align*}
\rho_{2j+1} 
&= \left|\frac{2}{\k\,\tau}\,\frac{\Gamma\left(\tfrac{2}{3}j+\tfrac{1}{3}n+\tfrac{5}{6}\right)}{\Gamma\left(\tfrac{2}{3}j+\tfrac{1}{3}n+\tfrac{1}{2}\right)}
\frac{L_{j+1}^{(-1/2)}\left(-\tfrac{1}{4}{\k^2\tau^3} \right)}
 {L_j^{(1/2)}\left(-\tfrac{1}{4}{\k^2\tau^3} \right)}\right| 
\\ &\sim\left|\frac{2}{\k\,\tau}\left(\frac{2j}{3} \right)^{\!1/3}\frac{\k\,|\tau|^{3/2}}{2\sqrt{j}} \right| 
= \left|\frac{\sqrt[3]{\tfrac{2}{3}}\,\sqrt{|\tau|}}{\,j^{1/6}}\right|
 \longrightarrow 0, \qquad\text{as}\quad j\to\infty. 
\end{align*}
By the ratio test, the series \eqref{mu2n series proof} 
converges absolutely for any $\tau,\, \k \in \Real\backslash\{0\}$. Therefore, by Lebesgue dominated convergence theorem, it follows 
\begin{align*}
\mu_{2n}(\tau;\k) & = \sum_{m=0}^{\infty}\frac{C_{m}(\tau;\k)}{m!} \int_0^\infty s^{m+n-1/2} \exp(-s^3)\,\rmd s 
= \tfrac{1}{3} \sum_{m=0}^{\infty} \Gamma\left(\tfrac{1}{3}m+\tfrac{1}{3}n+\tfrac{1}{6}\right)\frac{C_{m}(\tau;\k)}{m!}. 
\end{align*}
When $\k = 0$ and $\tau\neq0$, one has 
\begin{align*}
\mu_{2n}(\tau;0) &= \tfrac{1}{3} \sum_{m=0}^{\infty}\Gamma\left(\tfrac{1}{3}n+\tfrac{2}{3}m+\tfrac{1}{6}\right) 
\frac{\tau^m}{m!}\\
&= \tfrac{1}{3} \sum_{r=0}^{2} \sum_{j=0}^{\infty}\Gamma\left(\tfrac{1}{3}n+\tfrac{2}{3}r+\tfrac{1}{6}+2j\right) 
\frac{\tau^{3j+r}}{(3j+r)!} \\
&= \tfrac{1}{3} \sum_{r=0}^{2} \Gamma\left(\tfrac{1}{3}n+\tfrac{2}{3}r+\tfrac{1}{6}\right) \sum_{j=0}^{\infty} 
\left(\tfrac{1}{6}n+\tfrac{1}{3}r+\tfrac{1}{12}\right)_j\left(\tfrac{1}{6}n+\tfrac{1}{3}r+\tfrac7{12}\right)_j\,2^{2j}
\frac{\tau^{3j+r}}{(3j+r)!},
\end{align*}
which, after using the Legendre duplication formula for the Gamma function, can be written as in Lemma \ref{lem:mun0}. Hence, the result holds for any $\k\in\Real$.

Finally, \eqref{mu2n series2} also holds when $\tau =0$, since it gives \eqref{mu0:kappa}.
\end{proof} 

As a straightforward consequence of the latter result, one has 
\begin{align*}
 \mu_{2n}(\tau;\k) +\mu_{2n}(\tau;-\k) 
 &= \tfrac{2}{3}\sum_{j=0}^{\infty} 
\frac{\Gamma\left(\tfrac{2}{3}j+\tfrac{1}{3}n+\tfrac{1}{6}\right)}{\left(\tfrac{1}{2}\right)_j}\,\tau^j\,
 L_j^{(-1/2)}(-\tfrac{1}{4}{\k^2\tau^3})\\
 \mu_{2n}(\tau;-\k) -\mu_{2n}(\tau;\k) 
 &= \tfrac{2}{3}\k\tau^2\sum_{j=0}^{\infty}
\frac{\Gamma\left(\tfrac{2}{3}j+\tfrac{1}{3}n+\tfrac{1}{2}\right)}{\left(\tfrac32\right)_j}\tau^{j}\, L_j^{(1/2)}(-\tfrac{1}{4}{\k^2\tau^3}).
\end{align*}
Note that the series expansion obtained above, written in terms of Laguerre polynomials, could of course be written using Hermite polynomials. 

The expressions given in the latter result have a clear $3$-fold decomposition in $\tau$, and in fact \eqref{mu2n series2} reads as: 
 \begin{align*}
 \mu_{2n}(\tau;\k) 
 &= F_n(\tau;\k) +\tau G_n(\tau;\k) +\tau^2 H_n(\tau;\k), 
 \end{align*}
where 
\begin{align*}
 F_n(\tau;\k)
 &= \tfrac{1}{3} \sum_{\ell=0}^{\infty}\frac{\Gamma\left(2\ell+\tfrac{1}{3}n+\tfrac{1}{6}\right)\tau^{3\ell}}{\left(\tfrac32\right)_{3\ell+1}}
 \Bigg\{\tfrac32 ( 2 \ell+1) (6 \ell+1) L_{3\ell}^{(-1/2)}(\zeta)
 - \left( 2\ell+\tfrac{1}{3}n+\tfrac{1}{6}\right) 
 \k\tau^{3}\, L_{3\ell+1}^{(1/2)}(\zeta) \Bigg\} \\
 G_n(\tau;\k)
 &= \tfrac{1}{3} \sum_{\ell=0}^{\infty}\frac{\Gamma\left(2\ell+\tfrac{1}{3}n+\tfrac{5}{6}\right)\tau^{3\ell}}{\left(\tfrac32\right)_{3\ell+2}}
 \Bigg\{\tfrac32(2\ell+1)(6\ell+5) L_{3\ell+1}^{(-1/2)}(\zeta)
 - {\left(2\ell+\tfrac{1}{3}n+\tfrac{5}{6}\right)} \k\tau^{3} L_{3\ell+2}^{(1/2)}(\zeta)\Bigg\} \\
 H_n(\tau;\k)
 &= \tfrac{1}{3} \sum_{\ell=0}^{\infty}
\frac{\Gamma\left(2\ell+\tfrac{1}{3}n+\tfrac{1}{2}\right)\tau^{3\ell}}{\left(\tfrac{1}{2}\right)_{3\ell+2}}
 \Bigg\{\left(2\ell+\tfrac{1}{3}n+\tfrac{1}{2}\right) L_{3\ell+2}^{(-1/2)}(\zeta)
 - \tfrac{3}{4} (2\ell+1) \k\, L_{3\ell}^{(1/2)}(\zeta)\Bigg\} 
\end{align*}
with $\zeta=-\tfrac{1}{4}{\k^2\tau^3}$.

Besides, the latter result states that 
\[ 
\mu_{2n}(\tau;\k) 
= \tfrac{1}{3} \sum_{m=0}^{\infty} 
\sum_{j=\lceil m/2 \rceil}^m\frac{(-1)^m \k^{2 j-m}\tau ^{3 j-m}}{(2 j-m)!\, (m-j)!}\Gamma\left(\tfrac{1}{3}m+\tfrac{1}{3}n+\tfrac{1}{6}\right).
\] 
A swap of the order of summation gives
\[ 
\mu_{2n}(\tau;\k) =
 \tfrac{1}{3} \sum_{j=0}^{\infty} 
\sum_{m=j}^{2j}\frac{(-1)^m \k^{2 j-m}\tau ^{3 j-m}}{(2 j-m)!\, (m-j)!}\Gamma\left(\tfrac{1}{3}m+\tfrac{1}{3}n+\tfrac{1}{6}\right). 
\]
The change of variables $(j,m)\mapsto (\ell,3\ell-m)$ followed by a change in the order of summation corresponds to 
\begin{align*}
\mu_{2n}(\tau;\k) 
&= \tfrac{1}{3} \sum_{m=0}^{\infty} \sum_{\ell= \left\lceil m/2\right\rceil}^{m}\frac{(-\k)^{m-\ell}}{(m-\ell)!\, (2\ell-m)!} 
\Gamma\left(\ell-\tfrac{1}{3}{m}+\tfrac{1}{3}{n}+\tfrac{1}{6}\right)
\tau^m\\
 &= \tfrac{1}{3} \sum_{m=0}^{\infty} \sum_{\ell= 0}^{\left\lfloor m/2\right\rfloor} 
\frac{(-\k)^{\ell}}{\ell!\, (m-2\ell)!} 
\Gamma\left(-\ell +\tfrac{2}{3}{m}+\tfrac{1}{3}{n}+\tfrac{1}{6}\right)
\tau^m
\\
 &= \tfrac{1}{3} \sum_{m=0}^{\infty} \sum_{\ell= 0}^{\left\lfloor m/2\right\rfloor} 
\frac{(-1)^{\ell}2^{2\ell}\left(-\tfrac{1}{2}m\right)_{\!\ell}\left(-\tfrac{1}{2}m+\tfrac{1}{2}\right)_{\!\ell} (-\k)^{\ell}}
{\left(-\tfrac{1}{3}{n}-\tfrac{2}{3}{m}+\tfrac{5}{6}\right)_{\!\ell}\ell!}
\,\frac{\Gamma\left(\tfrac{2}{3}{m}+\tfrac{1}{3}{n}+\tfrac{1}{6}\right)}{m!}\,\tau^m, 
\end{align*}
where we used in the last identity 
\[
\frac{1}{(m-2\ell)!}=\frac{(-m)_{2\ell}}{m!}=\frac{2^{2\ell}\left(-\tfrac{1}{2}m\right)_{\!\ell}\left(-\tfrac{1}{2}m+\tfrac{1}{2}\right)_{\!\ell}}{m!}
\]
and 
\[
\Gamma\left(-\ell +\tfrac{2}{3}{m}+\tfrac{1}{3}{n}+\tfrac{1}{6}\right)
= \Gamma\left(\tfrac{2}{3}{m}+\tfrac{1}{3}{n}+\tfrac{1}{6}\right)\frac{(-1)^{\ell}}{\left(-\tfrac{1}{3}{n}-\tfrac{2}{3}{m}+\tfrac{5}{6}\right)_{\!\ell}}. 
\]
Hence, we obtain the series expansion 
\begin{align*}
\mu_{2n}(\tau;\k) 
&= \tfrac{1}{3} \sum_{m=0}^{\infty} \pFq{2}{1}{-\tfrac{1}{2}m,-\tfrac{1}{2}m+\tfrac{1}{2}}{-\tfrac{1}{3}{n}-\tfrac{2}{3}{m}+\tfrac{5}{6}}{4\k} 
\,\frac{\Gamma\left(\tfrac{2}{3}{m}+\tfrac{1}{3}{n}+\tfrac{1}{6}\right)}{m!}\tau^m. 
\end{align*}
Using the Gauss formula for the hypergeometric function, see \cite[\S 15.4]{refNevai86}, the latter expression evaluated at $\k=\tfrac{1}{4}$ gives \eqref{sol:mun14}. Similarly, series expansions about $\k=\tfrac{1}{3}$ can be obtained, from which the expression in Lemma \ref{lem:mun13} appears as a particular case.

\section{Numerical computations}\label{Sec:NumComp}
In this section we plot numerically the coefficient $\b_n$ in the three-term recurrence relation \eqref{eq:srr} for the symmetric sextic Freud weight \eqref{Freud642k}. As we explain below, these computations were done in Maple using the discrete equation \eqref{eq:rr642}, treating it as an \textit{initial value problem}, sometimes referred to as the ``orthogonal polynomial method". The earlier calculations in the 1990s by Jurkiewicz \cite{refJurk91}, Sasaki and Suzuki \cite{refSS91} and S\'en\'echal \cite{refSen92} solved \eqref{eq:rr642} as a \textit{discrete boundary problem}. They use the cubic \eqref{eq:cubic} to provide an estimate for $\b_n$ for large $n$. At a similar time, Demeterfi \etal\ \cite{refDDJT} and Lechtenfeld \cite{refLech92,refLech92b,refLechRR} {also solved} \eqref{eq:rr642}, though as a discrete initial value problem, but were only able calculate $\b_n$ for small values of $n$, up to $n=25$. The development of computers and Maple during the subsequent years has meant that we are now able to calculate $\b_n$ for large values of $n$ through a discrete initial value problem.
{We remark that none of the results in this section, or the next two sections, have been proved rigorously. The discussion of the behaviour of the recurrence coefficients $\b_n$ is based solely on the numerical calculations.}

By Theorem \ref{FC}, for large values of $n$, we have $\b_{n\pm k} \sim \b(n)$, for $k=0,1,2$.
Setting $t=-\k\tau^2$ in the cubic \eqref{eq:cubic} gives
\beq 60\b^3-12\tau\b^2+2\k\tau^2\b=n.\label{eq:cubici}\eeq
Differentiating \eqref{eq:cubici} with respect to $n$ gives
\beq 2\left(90\b^2-12\tau\b+\k\tau^2\right)\deriv{\b}{n}=1,\label{eq:cubicia}\eeq
which can be written as 
\beq \left\{\left(\b-\frac{\tau}{15}\right)^{\!2} -\frac{1}{90} \left(\frac{2}{5} - \k\right)\tau^2\right\}\deriv{\b}{n}=\frac{1}{180}. 
\label{eq:cubicia2}\eeq
Differentiating \eqref{eq:cubicia} with respect to $n$ gives
\beq 2 \left(\b-\frac{\tau}{15}\right)\left(\deriv{\b}{n}\right)^{\!2}+ \left\{\left(\b-\frac{\tau}{15}\right)^{\!2} -\frac{1}{90} \left(\frac{2}{5} - \k\right)\tau^2\right\}\deriv[2]{\b}{n}=0.\label{eq:cubicib}\eeq
It is therefore clear that $\k=\tfrac{2}{5}$ is a critical point. 

For $\k<\tfrac{2}{5}$, then $\b(n)$ is multivalued for 
$n_-<n<n_+$, where 
\beq 
n_\pm=\frac{2\tau^3}{225}\left[15\k - 4\pm\sqrt{2}(2-5\k)^{3/2}\right].\label{eq:npm}\eeq
We observe that for $\k>\tfrac3{10}$ then $ \b(n)$ intersects the line $n=0$ only at $(0,0)$. For values of $\k=\tfrac3{10}$, $ \b(n)$ intersects the line $n=0$ at $(0,0)$ and at $\left(0,\tfrac{1}{10}\right)$. 
When $\k<\tfrac3{10}$, then $\b(n)$ 
intersects $n=0$ at 
\[\left(0,\tfrac{1}{30} \left(3\pm\sqrt{30}\sqrt{\tfrac{3}{10}-\k}\right)\right).\]
Plots of the real solution $\b(n)$ of the cubic \eqref{eq:cubici}
for $\k$ such that $\tfrac{1}{4}\leq\k\leq\tfrac{2}{5}$ are given in Figure \ref{Fig:3.3}. 

\begin{figure}[ht!]
\[\begin{array}{c@{\quad}c@{\quad}c@{\quad}c}
\fig{2}{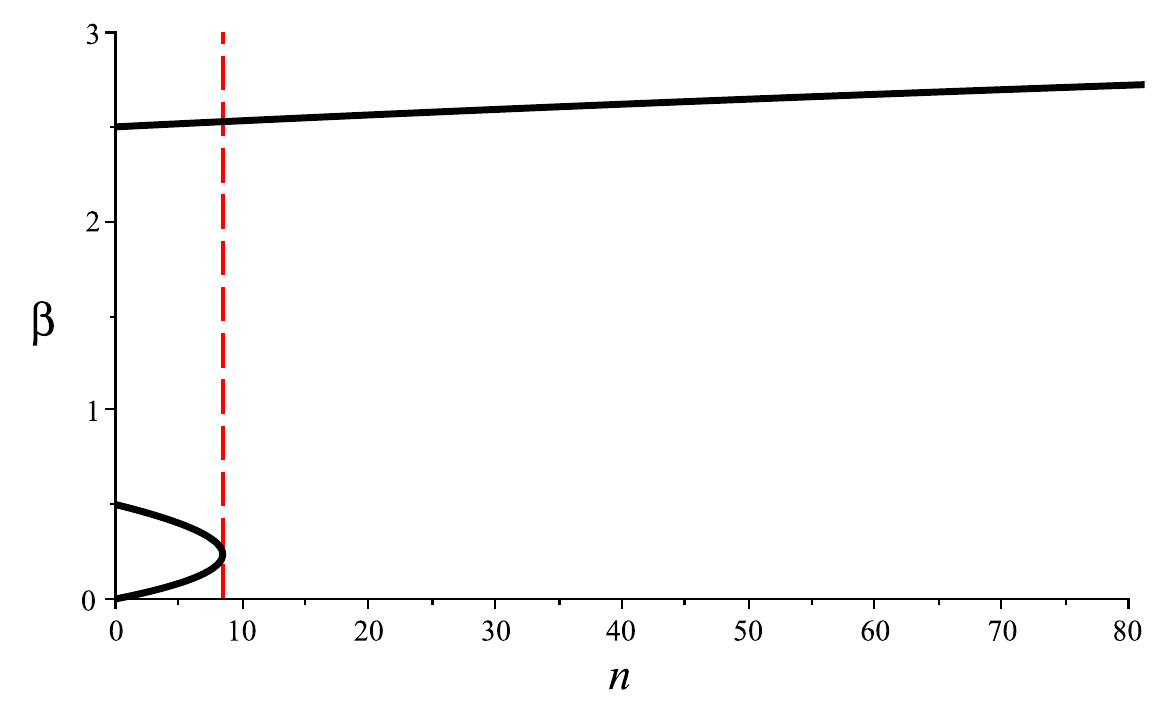} & \fig{2}{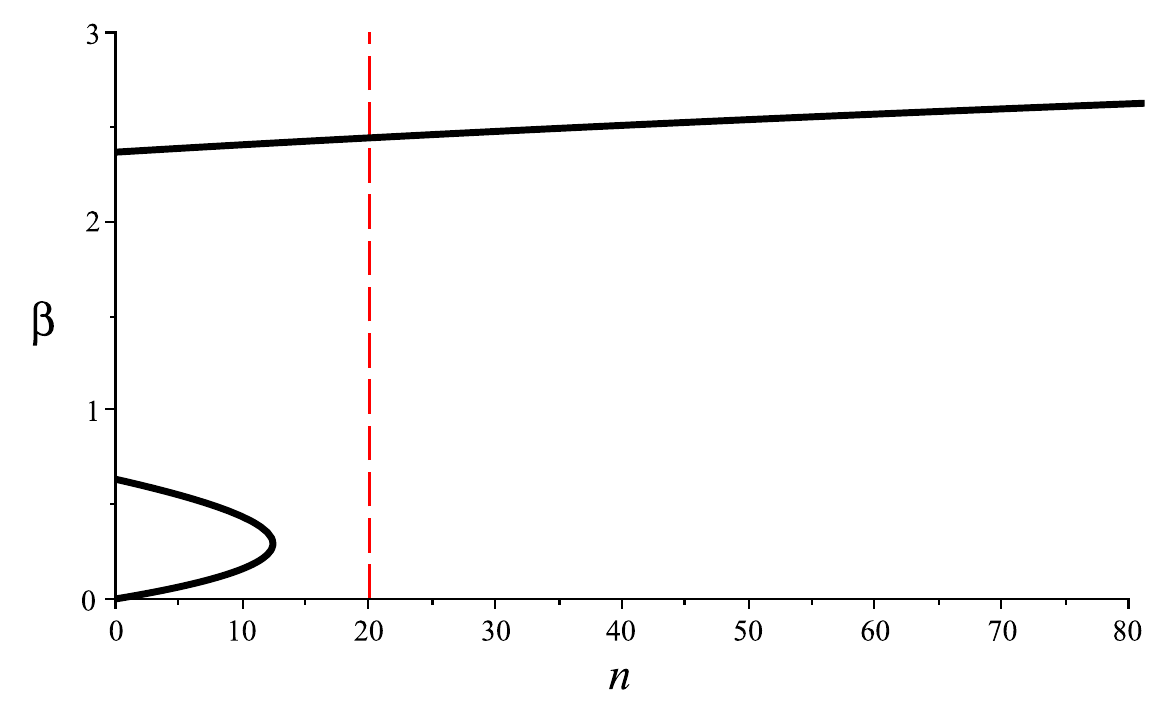} &\fig{2}{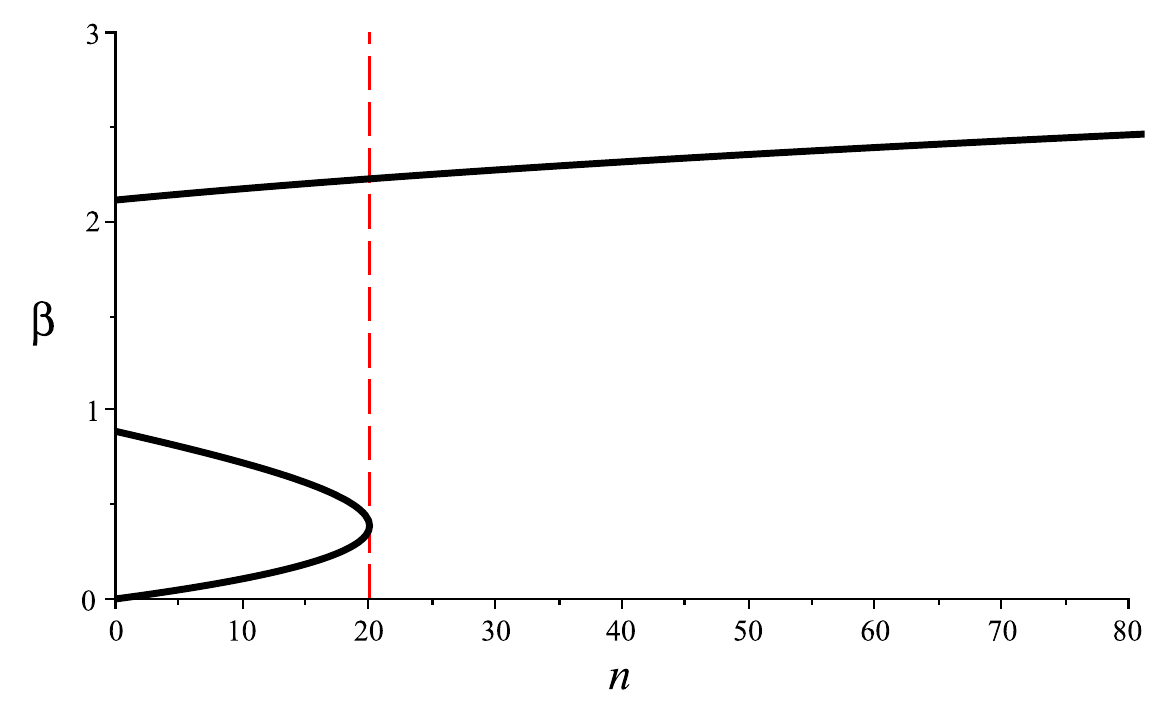} \\
\tau=15,\enskip\k=\tfrac{1}{6} & \tau=15,\enskip\k=\tfrac{1}{5} & \tau=15,\enskip\k=\tfrac{1}{4} \\[5pt]
\fig{2}{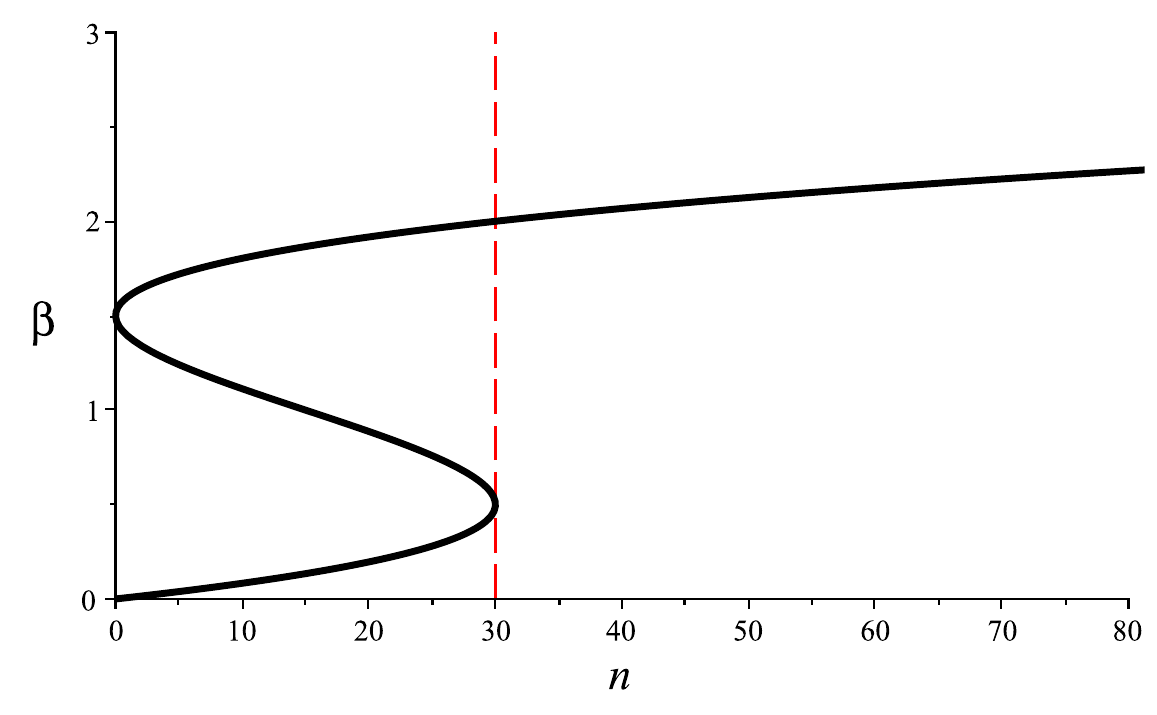} & \fig{2}{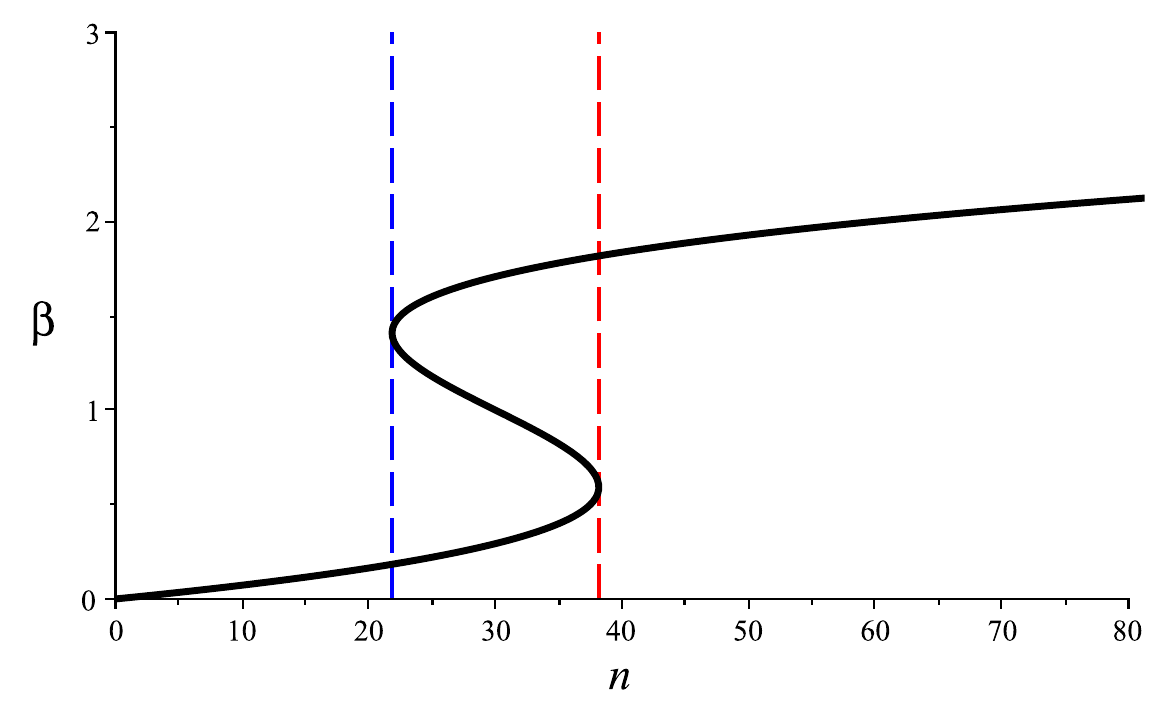} & \fig{2}{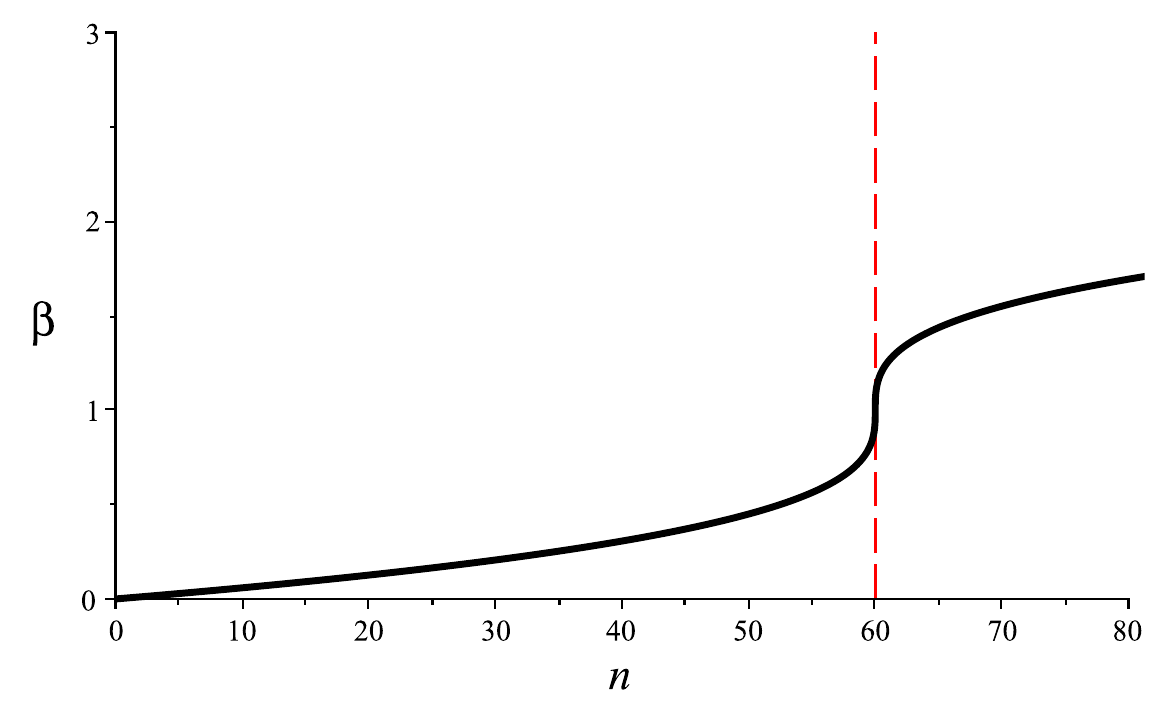} \\
\tau=15,\enskip\k=\tfrac{3}{10} & \tau=15,\enskip\k=\tfrac{1}{3} & \tau=15,\enskip\k=\tfrac{2}{5} 
\end{array}\]
\caption{\label{Fig:3.3}Plots of the real solution $\b(n)$ of the cubic \eqref{eq:cubici} for $\k$ such that $\tfrac{1}{6}\leq\k\leq\tfrac{2}{5}$, {with $\tau=15$}. The vertical lines are $n_\pm$ given by \eqref{eq:npm}.}
\end{figure}

Next we classify the roots of $U(x;\tau,\k)$ given by \eqref{Uployk}.
The value $\k=\tfrac{1}{4}$ is a critical one, when $\tau>0$. 
As such, the following hold: 
\begin{enumerate}[(i)]
\item if $\tau>0$ and $\k>\tfrac{1}{4}$, then $U(x)$ has four complex roots and a double root at $x=0$;
\item if $\tau>0$ and $0<\k<\tfrac{1}{4}$, then $U(x)$ has four real roots and a double root at $x=0$;
\item if $\tau>0$ and $\k=\tfrac{1}{4}$, then $U(x)=x^2(x^2-\tfrac{1}{2}\tau)^2$ which has three double roots at $x=\pm{\sqrt{\tfrac{1}{2}\tau}}$ and $x=0$;
\item if $\tau>0$ and $\k =0$, then $U(x)= x^4(x^2 -\tau)$, which has two real roots and a quadruple root at $x=0$.
\item if $\tau>0$ and $\k<0$, then $U(x)$ has two real roots, two purely imaginary roots and a double root at $x=0$;
\item if $\tau=0$, then for $t>0$, then $U(x)$ has two real roots, two purely imaginary roots and a double root at $x=0$ and for $t<0$, $U(x)$ has four complex roots and a double root at $x=0$;
\item if $\tau< 0$, then for $\k>0$, $U(x)$ has four complex roots and a double root at $x=0$, for $\k=0$, $U(x)$ has two purely imaginary roots and a quadruple root at $x=0$ and for $\k<0$, then $U(x)$ has two real roots, two purely imaginary roots and a double root at $x=0$;
\item if $\tau=\k=0$, then $U(x)$ has a sextuple root at $x=0$.
\end{enumerate}

\begin{figure}[ht!]
\[\begin{array}{c@{\qquad\qquad}c}
\fig{2}{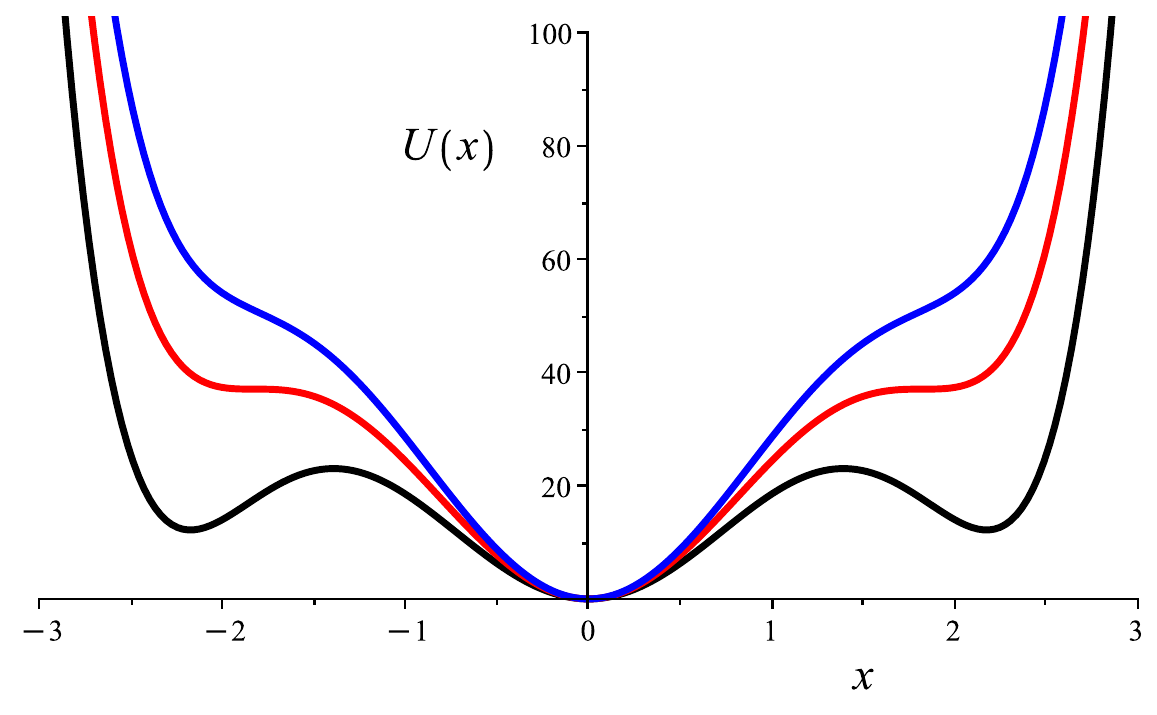}& \fig{2}{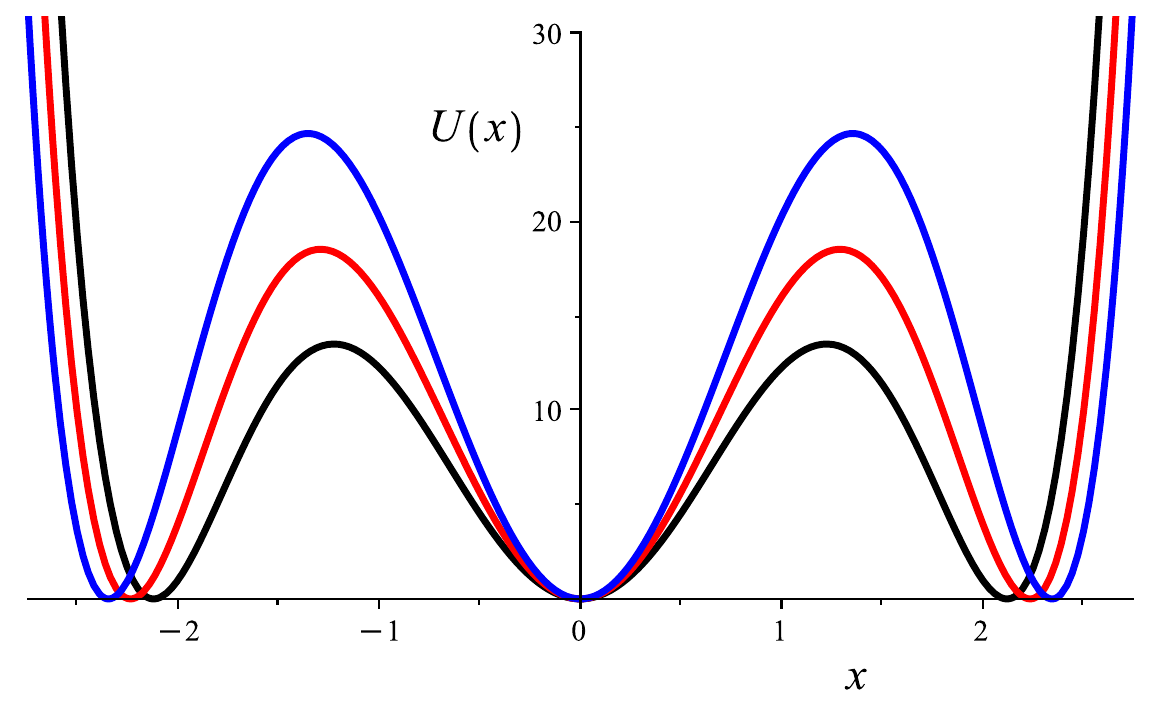} \\ \text{(i)},\quad\tau=10,\enskip\k=\tfrac{11}{40},\red{\tfrac{1}{3}},\blue{\tfrac{3}{8}} & \text{(iii)},\quad\tau=9,\red{10},\blue{11},\enskip\k=\tfrac{1}{4}
\\[10pt] 
\fig{2}{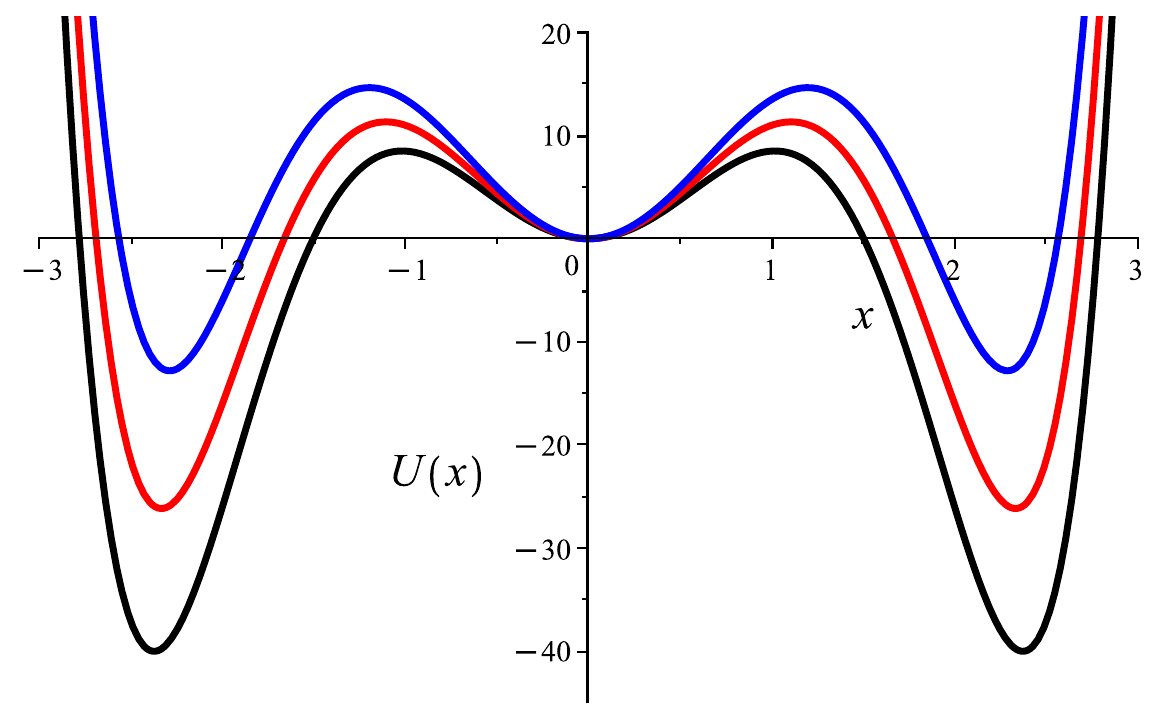} & \fig{2}{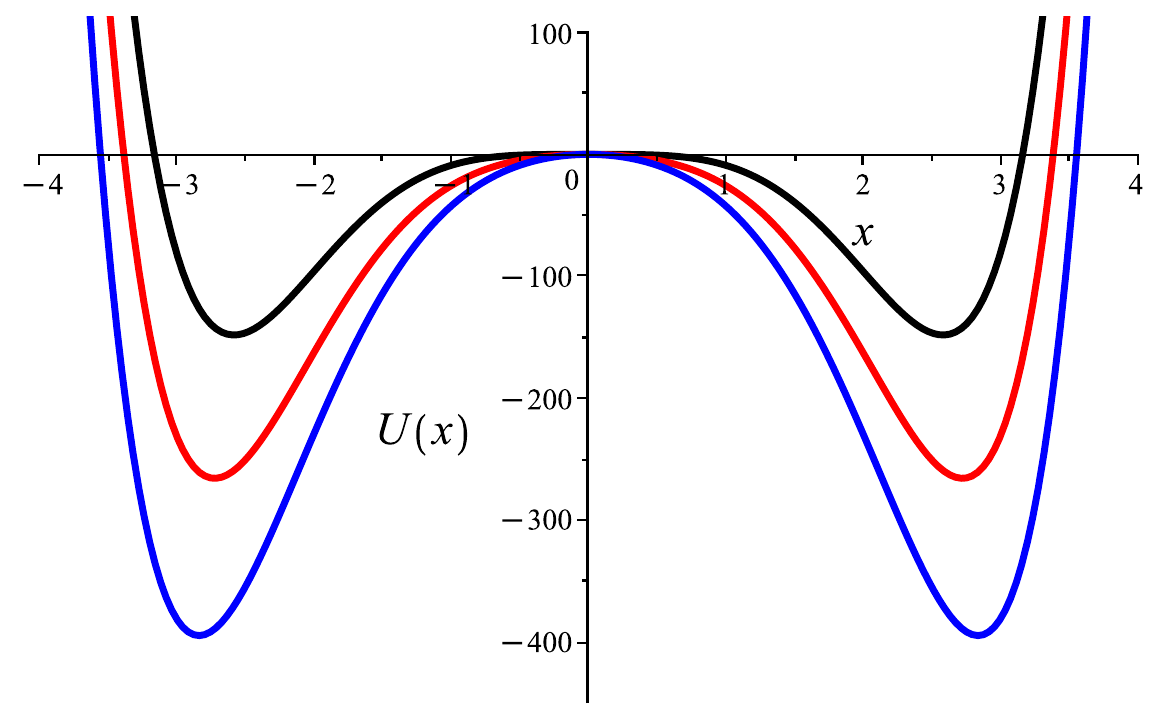} \\ 
\text{(ii)},\quad\tau=10,\enskip\k=\tfrac{7}{40},\red{\tfrac{1}{5}},\blue{\tfrac{9}{40}}
& \text{(iv)},\quad\tau=10,\enskip\k=0,\red{-\tfrac{1}{6}},\blue{-\tfrac{1}{3}} 
\end{array}\]
\caption{Plots of the sextic polynomial \eqref{Upoly} in the cases when
(i), $\tau>0$ and $\k>\tfrac{1}{4}$, (ii), $\tau>0$ and $\k=\tfrac{1}{4}$, (iii), $\tau>0$ and $0<\k<\tfrac{1}{4}$, and (iv), $\tau>0$ and $\k\leq0$.}
\end{figure}

The numerical computations were done in Maple using the discrete equation \eqref{eq:rr642}. Solving \eqref{eq:rr642} with $t=-\k\tau^2$ for $\b_{n+2}$ gives
\begin{align} \b_{n+2}&=\frac{n-2\k\tau^2\b_{n}+4\tau\b_n(\b_{n-1}+\b_{n}+\b_{n+1})}{6\b_n\b_{n+1}}\nonumber\\ &\qquad\qquad
-\frac{6\big(\b_{n-2} \b_{n-1} + \b_{n-1}^2 + 2 \b_{n-1} \b_{n} + \b_{n-1} \b_{n+1} + \b_{n}^2 + 2 \b_{n}\b_{n+1} + \b_{n+1}^2 \big)}{6\b_{n+1}}.\label{bn2}\end{align}
The initial conditions are
\beq \b_{-1}=0,\qquad \b_0=0, \qquad \b_1=\frac{\mu_2}{\mu_0},\qquad \b_2=\frac{\mu_{0} \mu_{4}-\mu_{2}^{2}}{\mu_{0} \mu_{2}},\label{beta_ics}\eeq
where $\mu_k(\tau;\k)$ is the $k$th moment
\[\mu_k(\tau;\k)=\int_{-\infty}^{\infty} x^k\exp\big(-x^6+\tau x^4-\k\tau^2x^2\big)\,\dx ,\]
and so using \eqref{bn2} we can evaluate $\b_n$ for $n\geq3$. 
The moments $\mu_0$, $\mu_2$ and $\mu_4$ are computed numerically using Maple.
Since the discrete equation \eqref{bn2} is highly sensitive to the initial conditions then it is necessary to use a high number of digits in Maple, usually several hundred, sometimes thousands, to do the computations; {see \S\ref{subsec710} for a discussion of how sensitive the problem is with an illustration in Figure \ref{fig:tau40_pert}}.

In the following subsections, we discuss the behaviour of the recurrence coefficients {for} the various cases mentioned above. 

\subsection{\label{casei}Case (i): $\tau>0$ and $\k>\tfrac{1}{4}$} 
In this case $U(x)$ has four complex roots and a double root at $x=0$ and is sometimes known as the ``one-branch case", cf.~\cite{refSen92}.

\begin{figure}[ht!]
\[\begin{array}{c@{\quad}c@{\quad}c}
\fig{2}{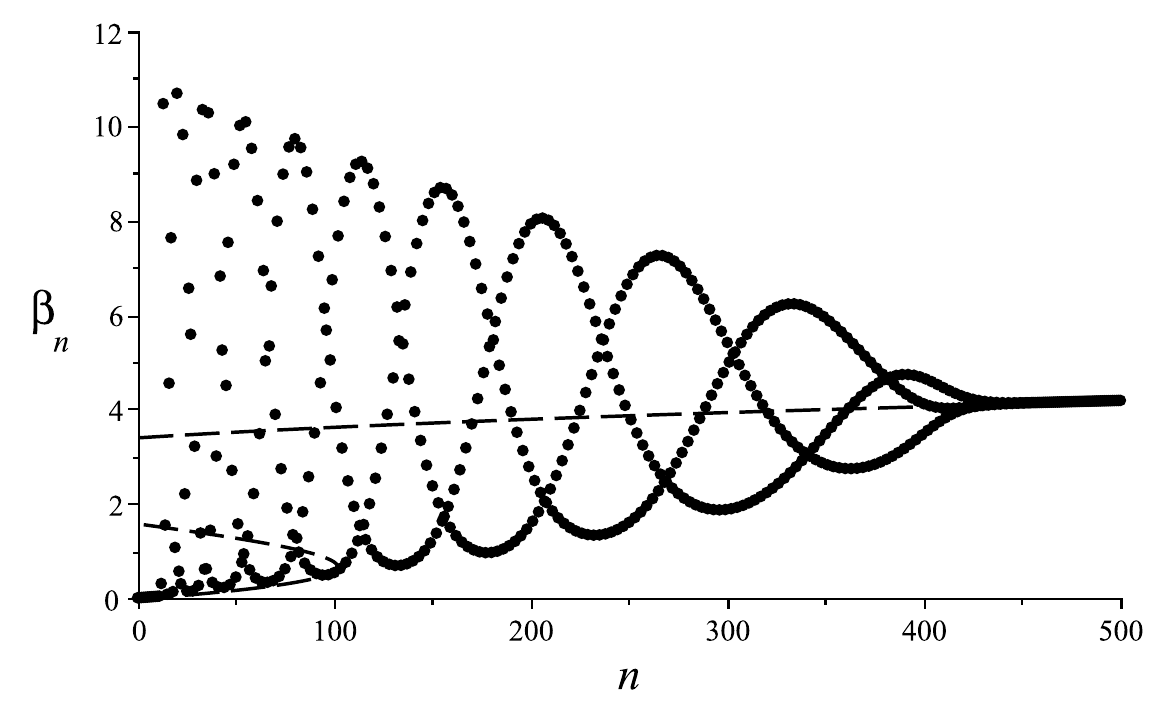} & \fig{2}{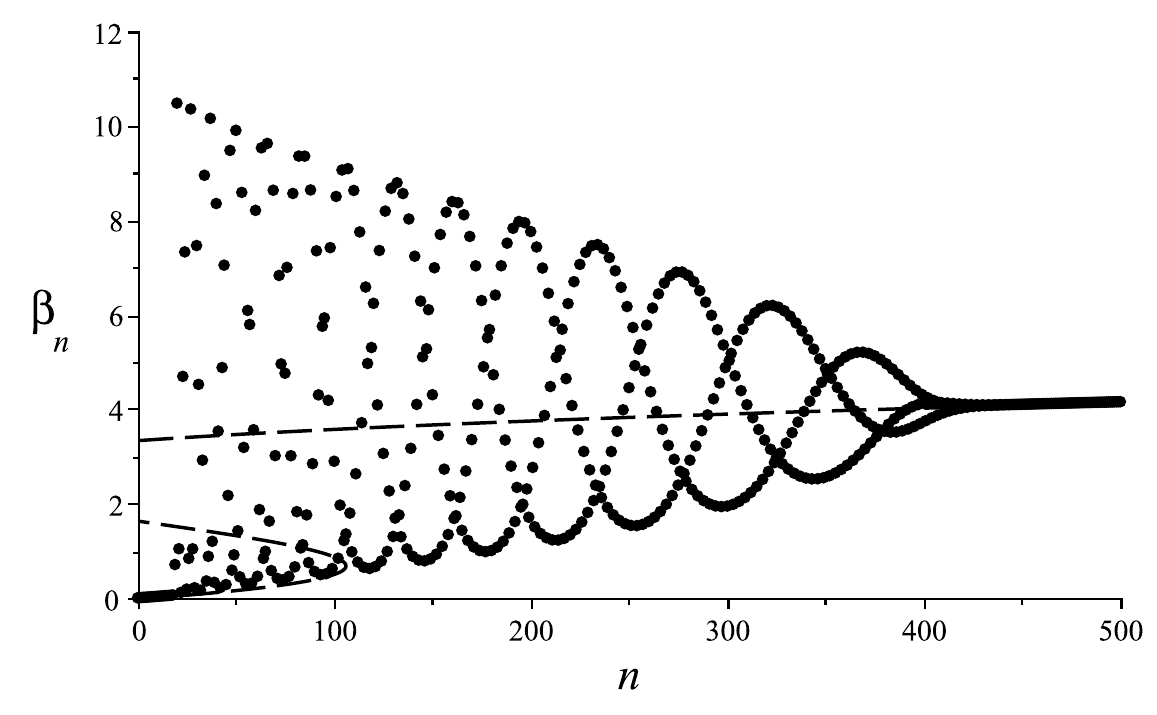}&
\fig{2}{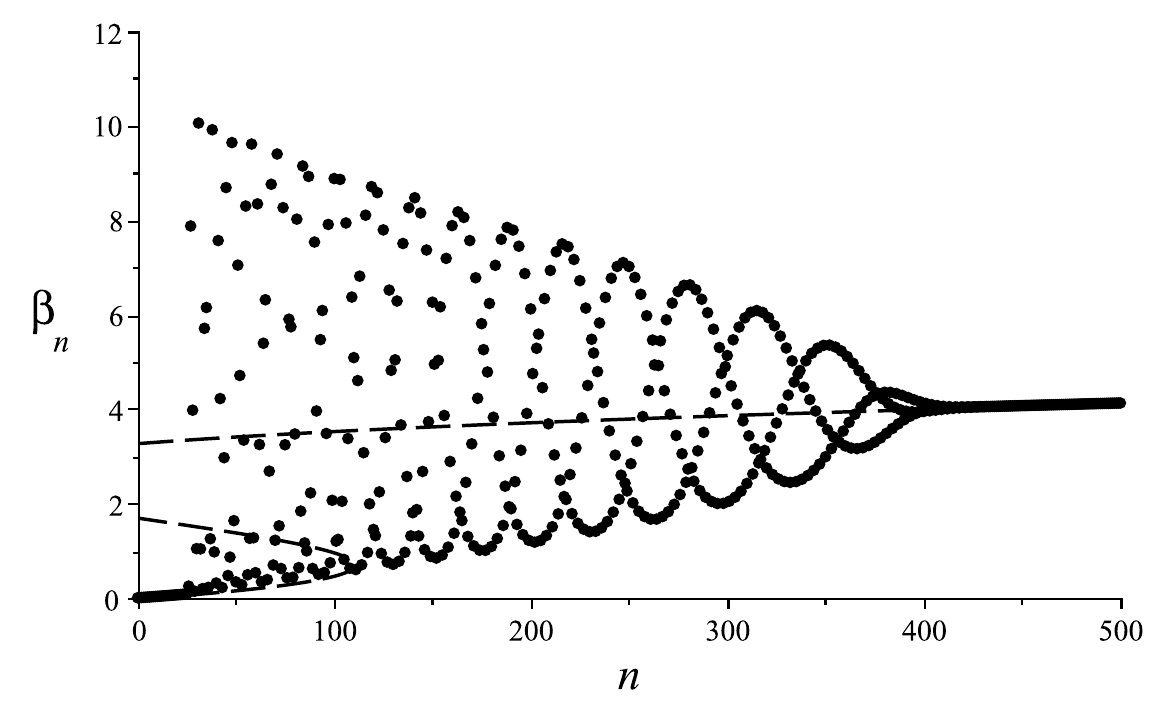}\\
\tau=25,\enskip \k=0.26 &\tau=25,\enskip \k=0.265 &\tau=25,\enskip \k=0.27\\[5pt]
\fig{2}{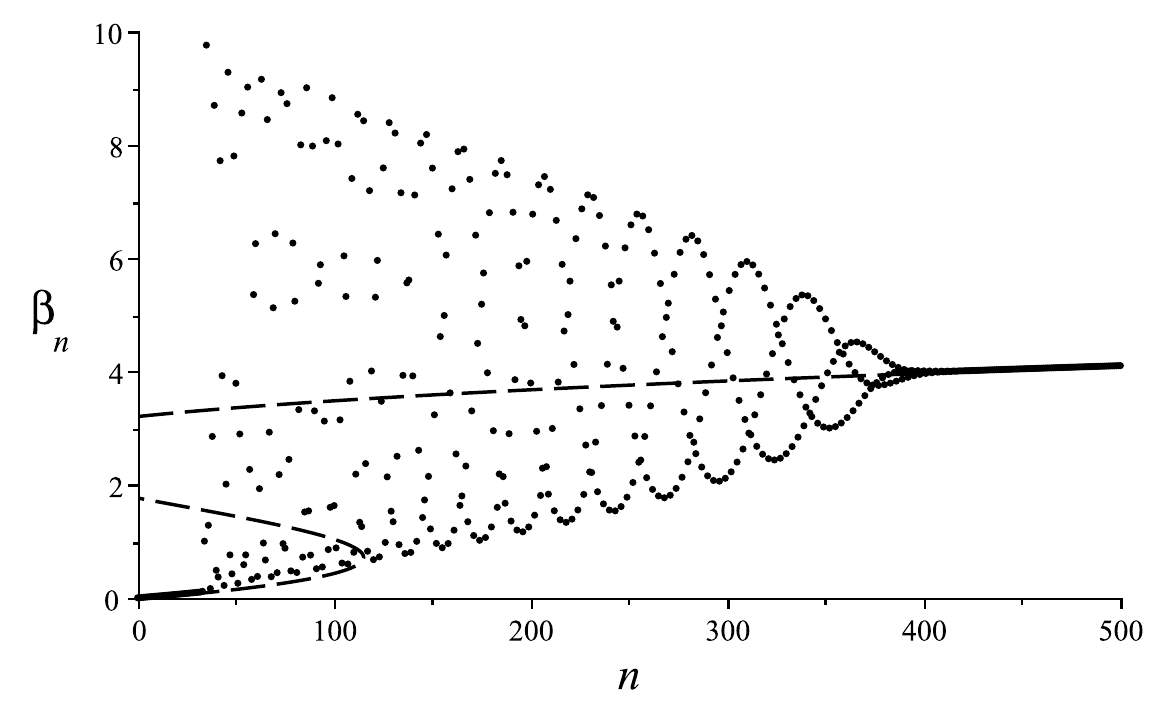} & \fig{2}{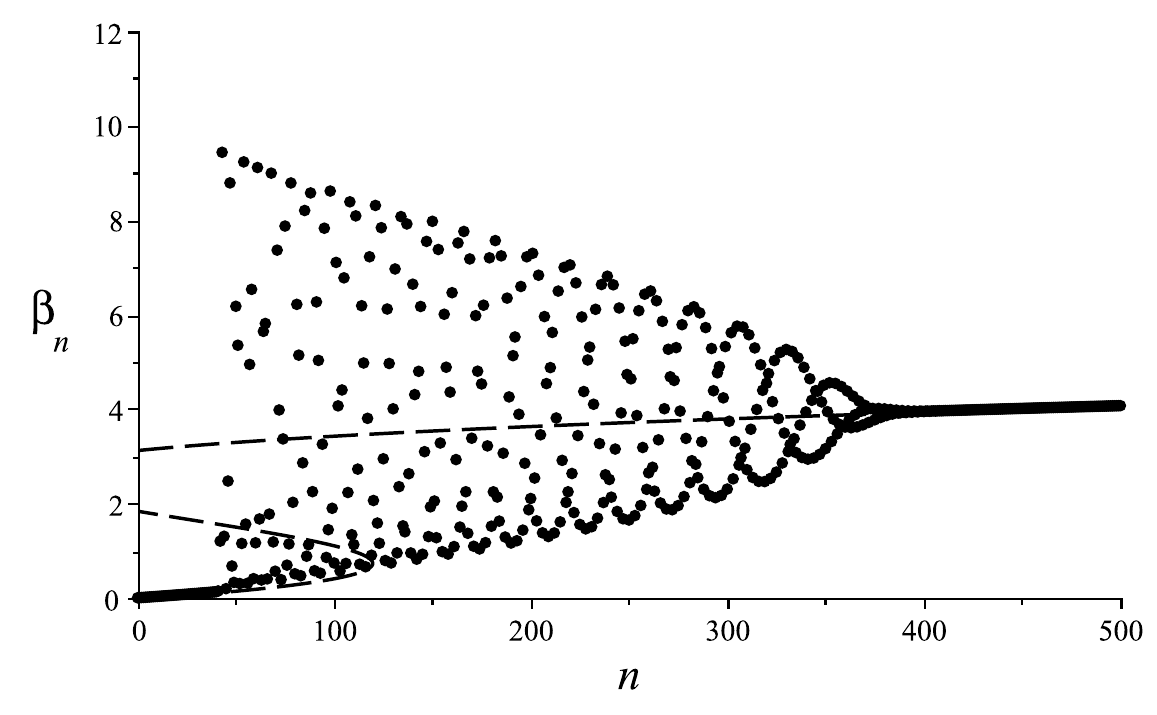}&
\fig{2}{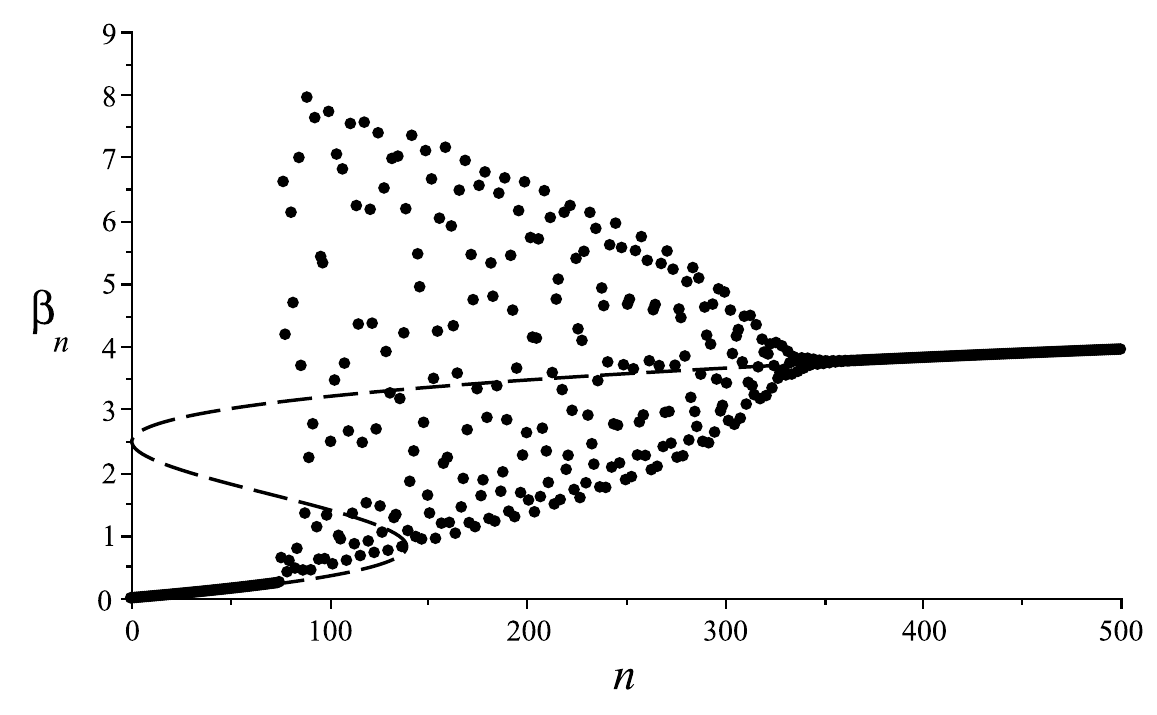}\\
\tau=25,\enskip \k=0.275 &\tau=25,\enskip \k=0.28 &\tau=25,\enskip \k=0.3\\[5pt]
\fig{2}{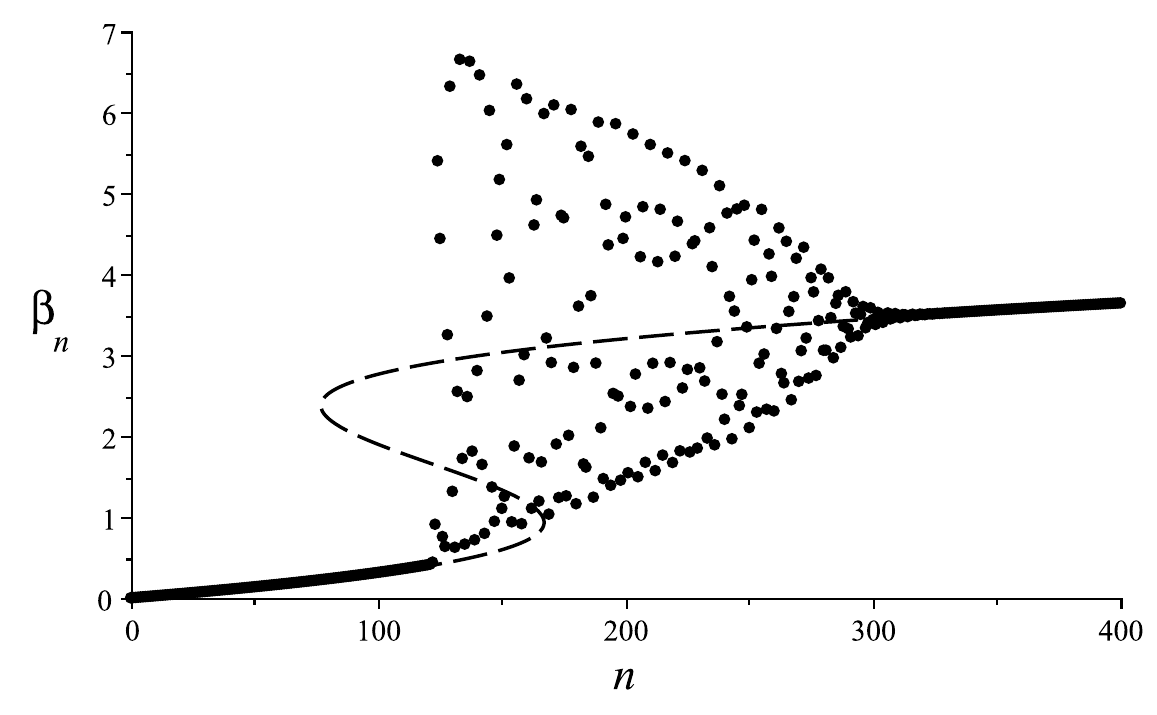} & \fig{2}{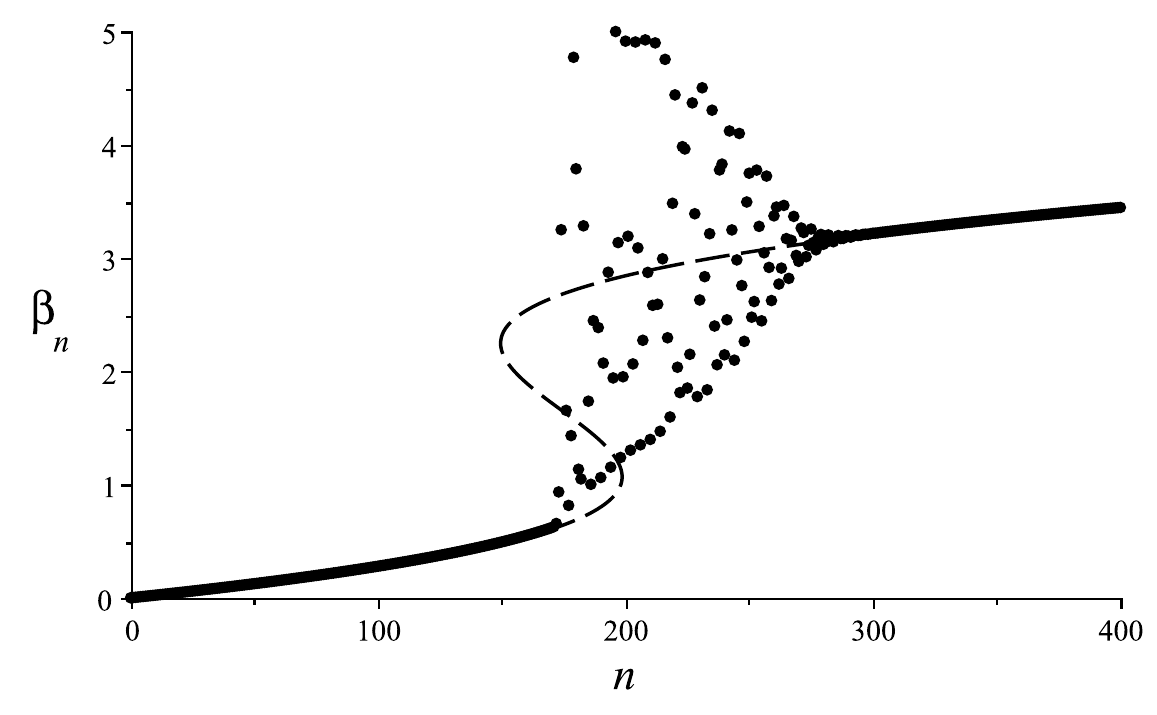}& \fig{2}{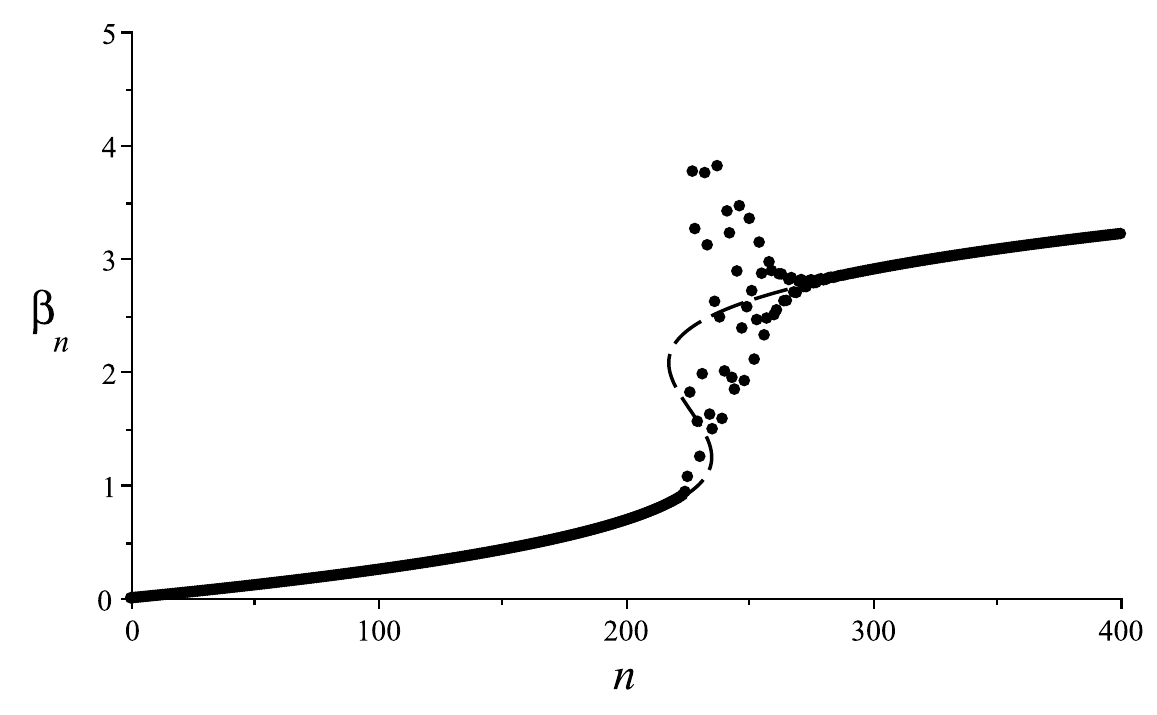}\\
\tau=25,\enskip \k=0.325 &\tau=25,\enskip \k=0.35 &\tau=25,\enskip \k=0.375
\end{array}\]
\caption{\label{fig:casei_tau25}Plots of the recurrence coefficient $\b_n$ when $\tau=25$, for various $\k$ such that $\tfrac{1}{4}<\k<\tfrac{2}{5}$, together with the real solution of the cubic \eqref{eq:cubici} (dashed line).}
\end{figure}

Plots of the recurrence coefficient $\b_n$ when $\tau=25$, for various $\k$ such that $\tfrac{1}{4}<\k<\tfrac{2}{5}$ are given in Figure \ref{fig:casei_tau25}, 
together with the real solution of the cubic \eqref{eq:cubici}.
For both small values of $n$ and for large $n$, $\b_n$ is approximately given by the real solution of the cubic \eqref{eq:cubici}. Whilst it might be expected that $\b_n$ tends to the cubic \eqref{eq:cubici} for large $n$, it is rather surprising that the cubic also gives a good approximation for small values of $n$. We also note that as $\k$ increases, the size of the ``transition region" decreases {and the value of $n$ for which $\b_n$ does not follow the cubic increases as $\k$ increases. Further for $\k$ just above $\tfrac{1}{4}$, there is evidence of a three-fold structure.} 

When $\k=\tfrac{2}{5}$, then \eqref{eq:cubicia2} and \eqref{eq:cubicib} respectively give 
\[ \left(\b-\frac{\tau}{15}\right)^{\!2}\deriv{\b}{n}=\frac{1}{180},
\qquad 
 \left(\b-\frac{\tau}{15}\right) \left\{2\left(\deriv{\b}{n}\right)^{\!2}+ \left(\b-\frac{\tau}{15}\right) \deriv[2]{\b}{n}\right\}=0. 
 \]
Hence a ``gradient catastrophe", i.e.\ the slope of the curve $\b(n)$ becomes infinite, occurs when
\[ \k=\frac{2}5, \qquad \b=\frac{\tau}{15}, \qquad n=\frac{4\tau^3}{225}.\] 
This is the case originally discussed by Br\'ezin \etal\ \cite{refBMP}. The recurrence coefficients $\b_n$ are plotted when $\tau=25$ in Figure \ref{Fig:casei2}(a).

When $\k>\tfrac{2}{5}$, from \eqref{eq:cubicia2} it follows that $\ds\deriv{\b}{n}>0$ for all $n\geq 0$. Note that $\ds\deriv[2]{\b}{n}=0$ when $\b(n)=\tfrac\tau{15}$ which, on account of \eqref{eq:cubici}, occurs when
$n=\frac{2}{15} \left(\k -\frac{4}{15}\right)\tau^3$. 
Hence, $\b(n)$ is a monotonically increasing function with an inflection point at $n=\frac{2}{15} \left(\k -\frac{4}{15}\right)\tau^3$. 
In the case when $\tau=25$ and 
$\k=0.425$, the recurrence coefficients $\b_n$ are plotted in Figure \ref{Fig:casei2}(b) and a plot of $\b_n-\b(n)$, where $\b(n)$ is the real solution of the cubic \eqref{eq:cubici}
is given in Figure \ref{Fig:casei2}(c).

\begin{figure}[ht!]
\[\begin{array}{c@{\quad}c@{\quad}c}
\fig{2}{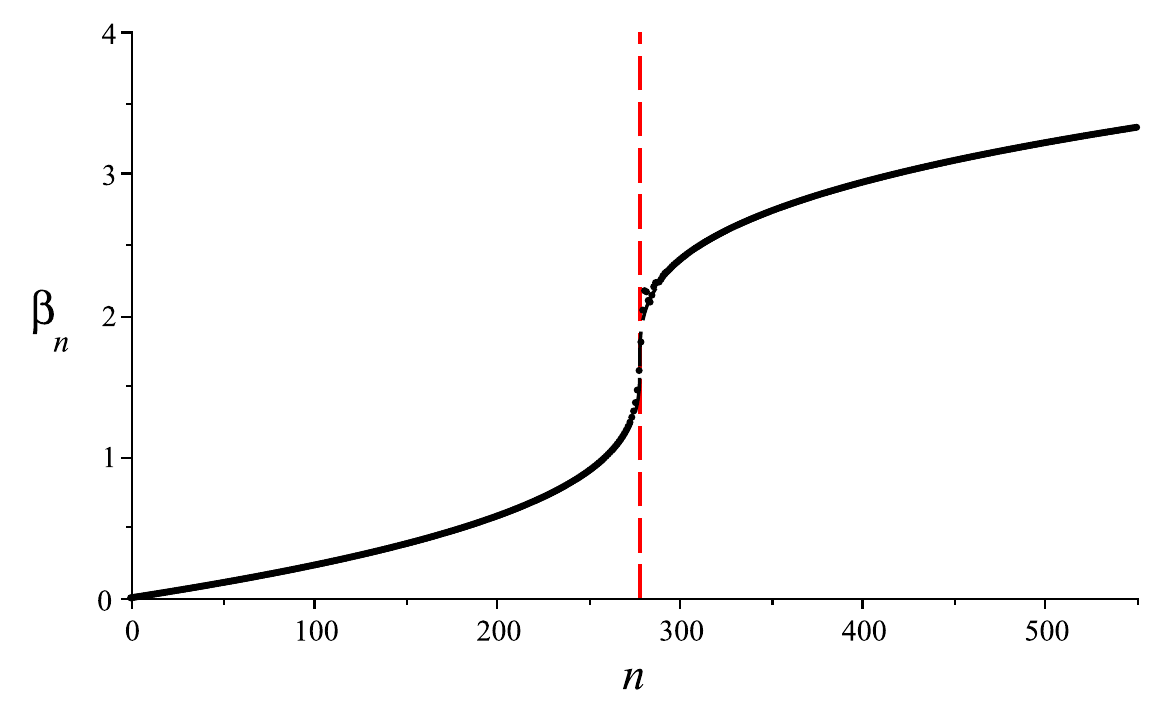} & \fig{2}{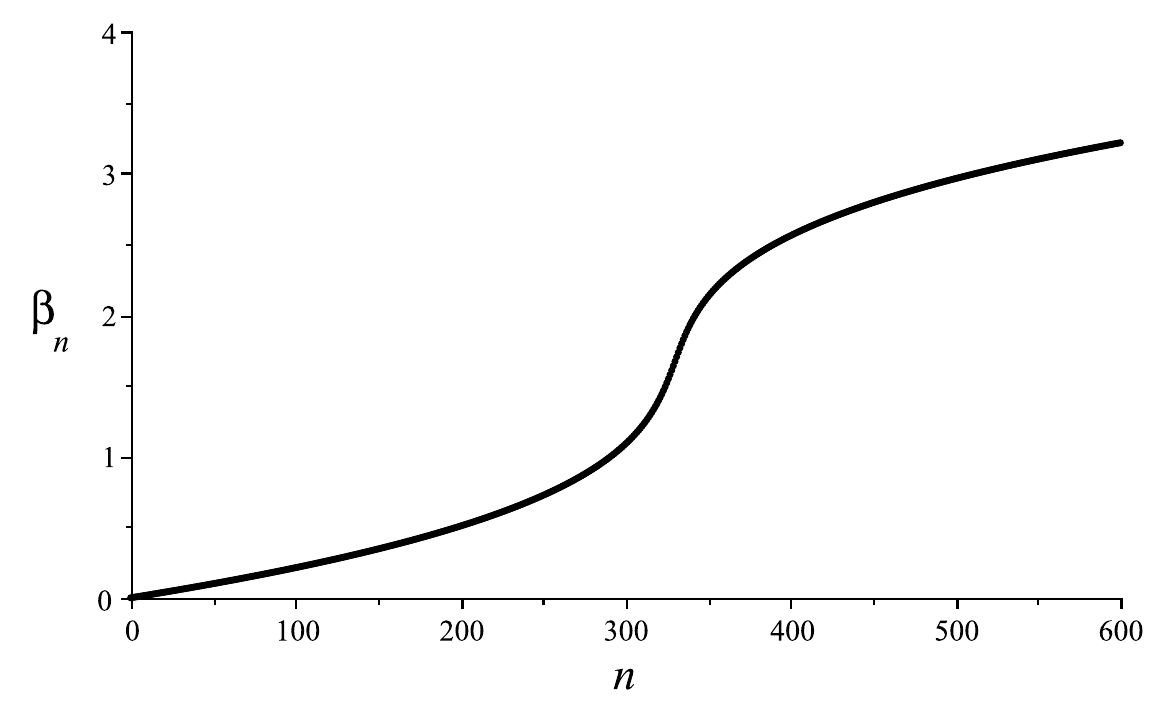} & \fig{2}{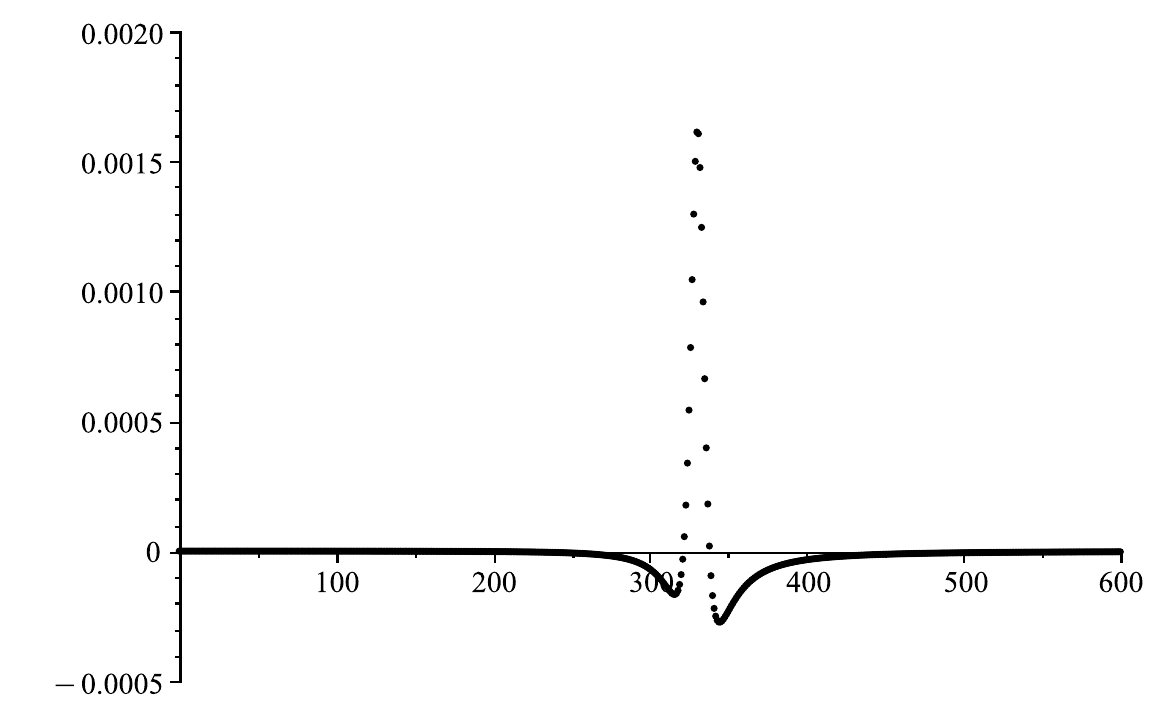}\\ 
\text{(a)},\quad\tau=25,\enskip\k=\tfrac{2}{5}& \text{(b)},\quad\tau=25,\enskip\k=0.425 & \text{(c)},\quad\b_n-\b(n)
\end{array}\]
\caption{\label{Fig:casei2}(a), A plot of the recurrence coefficient $\b_n$ for $\tau=25$ and $\k=\tfrac{2}{5}$. 
(b), A plot of the recurrence coefficient $\b_n$ when $\tau=25$ and $\k=0.425$. 
(c), A plot of $\b_n-\b(n)$ when $\tau=25$ and $\k=0.425$, where $\b(n)$ is the real solution of the cubic \eqref{eq:cubici}.} 
\end{figure}

\subsection{\label{caseii}Case (ii): $\tau>0$ and $0<\k<\tfrac{1}{4}$} In this case $U(x)$ has four real roots and a double root at $x=0$ and is sometimes known as the ``two-branch case", cf.~\cite{refSen92}.

To investigate this case, we set $\b_{2n}=u$, $\b_{2n+1}=v$ and $t=-\k\tau^2$ in \eqref{eq:rr642} which gives the system
\begin{subequations} \label{UVsys}
\begin{align}&6 u\!\left(u^{2}+6 u v +3 v^{2}\right)-4\tau u\!\left(u+2 v \right)+2\k\tau^2u=n,\label{UVsysa}\\
&6 v\!\left(3 u^{2}+6 u v +v^{2}\right)-4\tau v\!\left(2 u +v \right)+2\k\tau^2 v =n, \label{UVsysb}
\end{align}\end{subequations}
{i.e.\ we have replaced $\b_n$ with $n$ even by $u$ and $\b_n$ with $n$ odd by $v$}.
Letting $u=\xi-\eta$ and $v=\xi+\eta$, with $\eta\geq0$, in \eqref{UVsys} gives
\begin{subequations} 
\begin{align}&144 \xi^{3}-72\tau \xi^{2}+4(2+3\k)\tau^2 \xi -2 \k\tau^3+3n=0, \label{sol:UVsysa}\\
&\eta^2 ={3 \xi^{2} -\tfrac{2}{3}\tau \xi +\tfrac{1}{6}\k\tau^2}.\label{sol:UVsysb}
\end{align}\end{subequations}
In Figure \ref{fig:uvplot} we plot the solutions $u(n)$ and $v(n)$ of the system \eqref{UVsys} for various $\k$, with $u(n)$ plotted in \blue{blue} and $v(n)$ in \red{red}. 
\begin{figure}[ht!]
\[\begin{array}{c@{\quad}c@{\quad}c}
\fig{2}{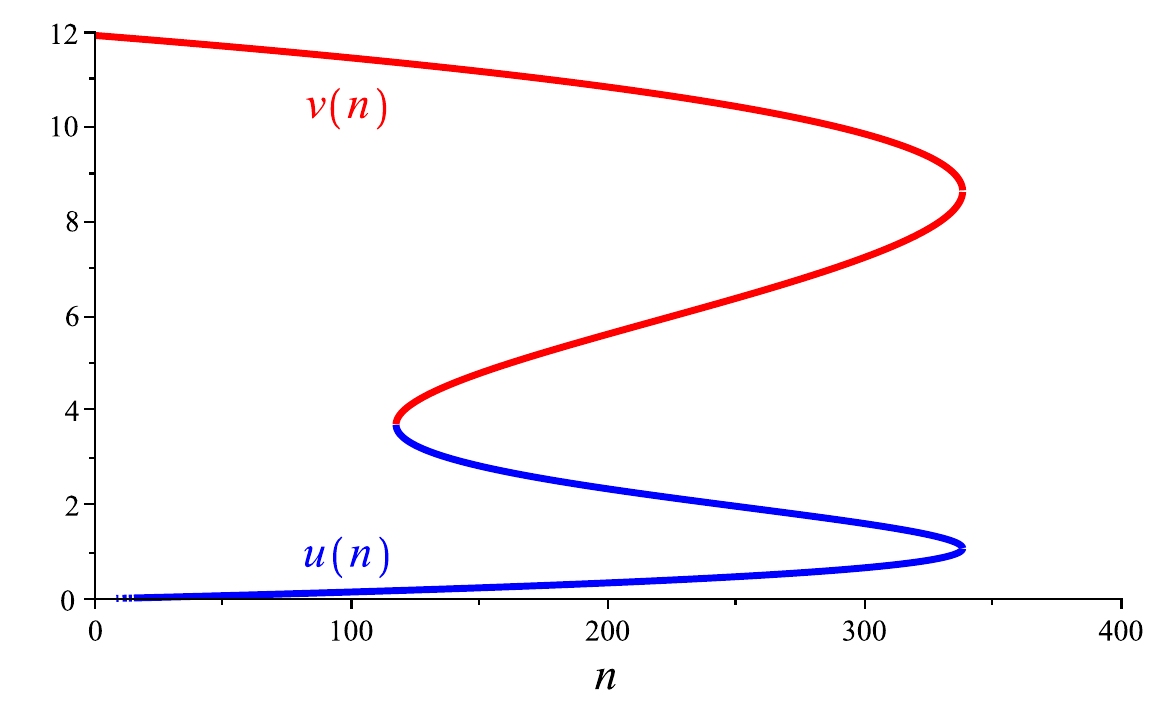} &\fig{2}{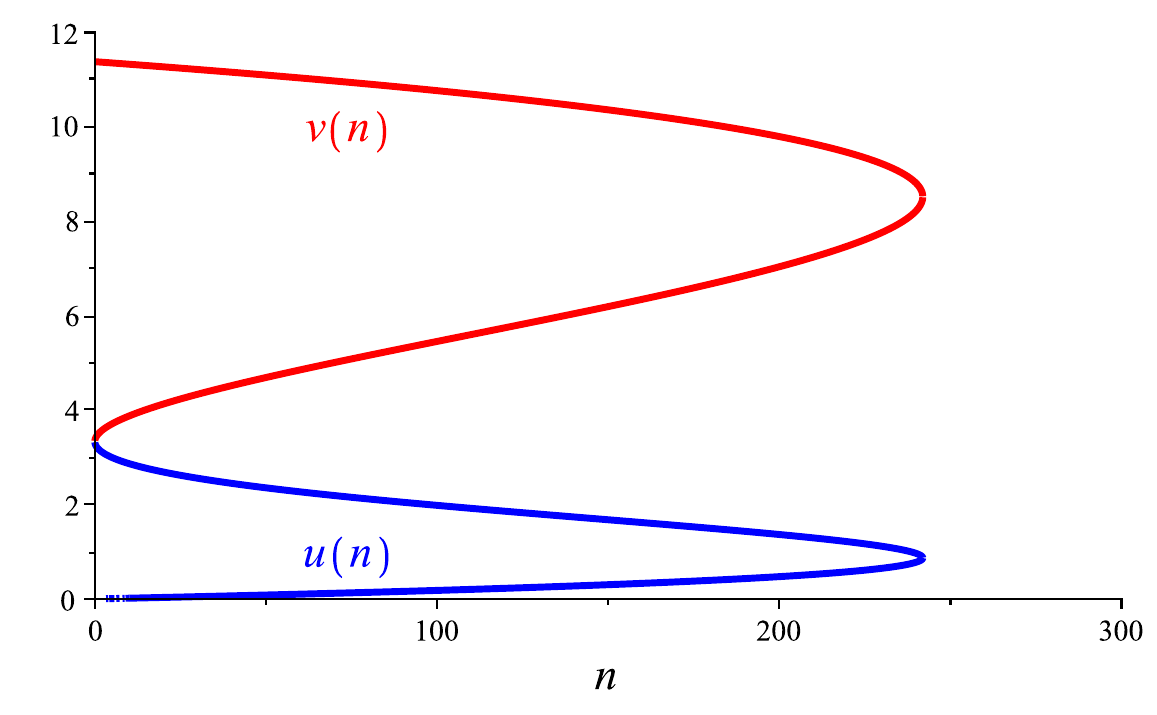} &\fig{2}{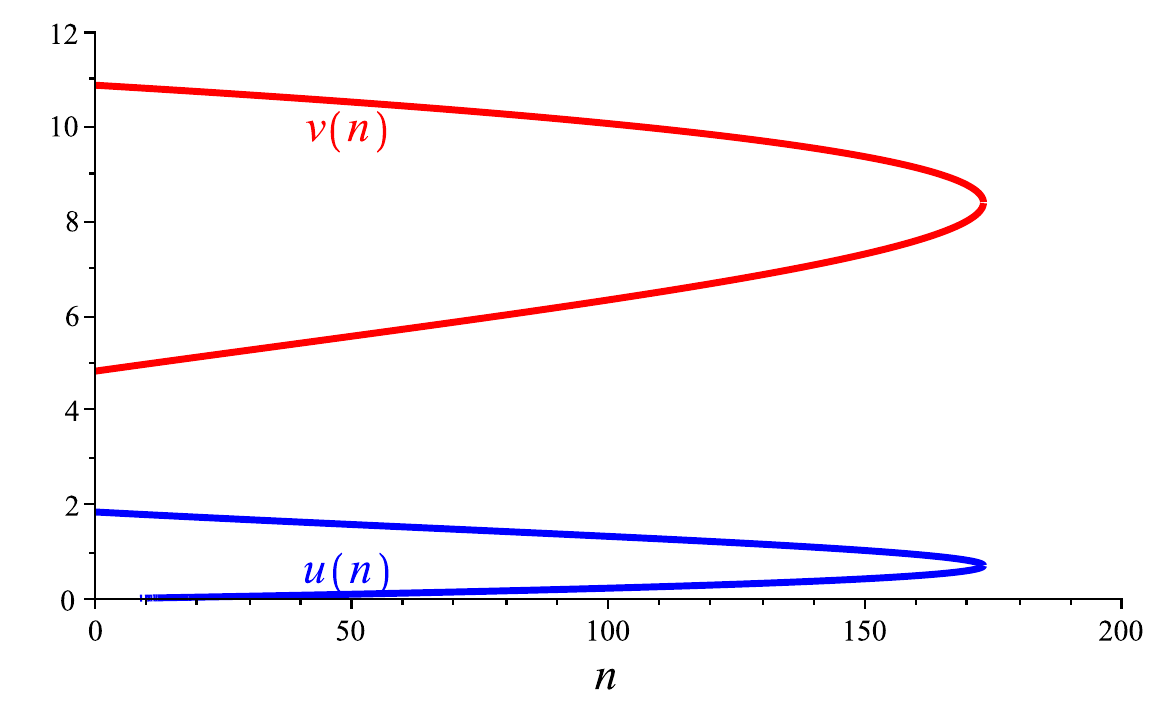} \\
{\tau=10},\enskip\k=\tfrac{1}{8} & {\tau=10},\enskip\k=\tfrac{1}{6} & {\tau=10},\enskip\k=\tfrac{1}{5}
\end{array}\]
\caption{\label{fig:uvplot}Plots of solutions of the system \eqref{UVsys} for {$\tau=10$ in the cases when $\k=\tfrac{1}{8}$, $\k=\tfrac{1}{6}$ and $\k=\tfrac{1}{5}$},
with {$u(n)$} plotted in \blue{blue} and {$v(n)$} in \red{red}.} 
\end{figure}

\begin{figure}[ht!]
\[\begin{array}{c@{\quad}c@{\quad}c}
\fig{2}{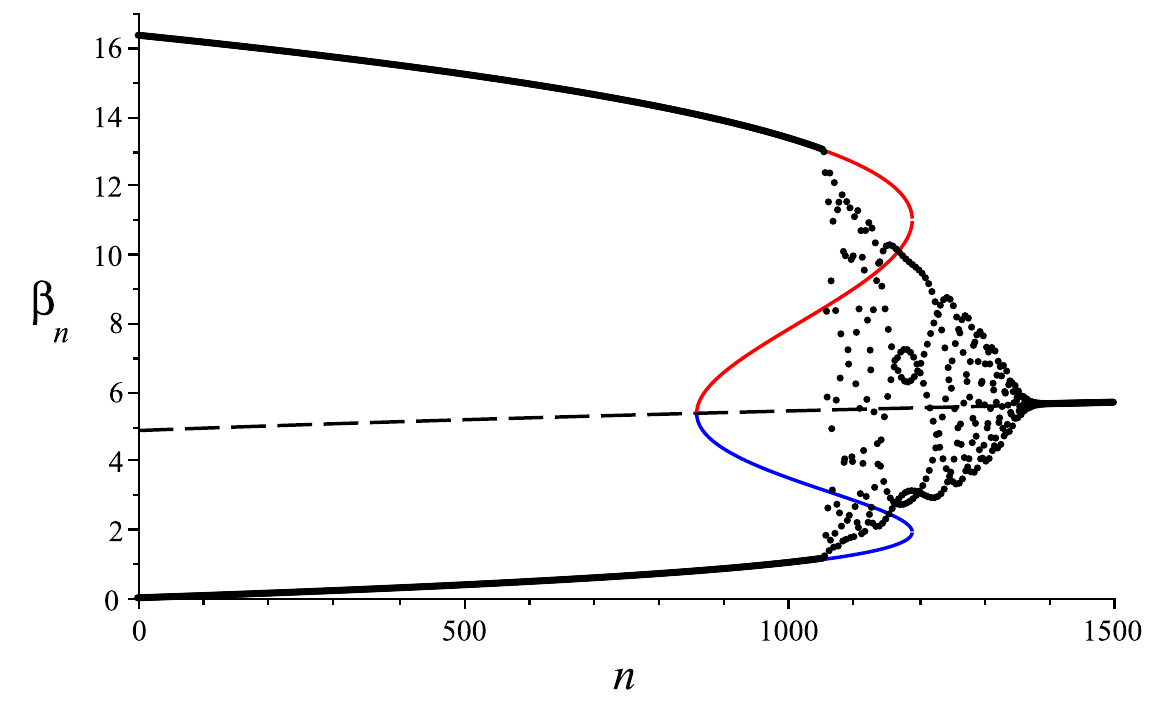} & \fig{2}{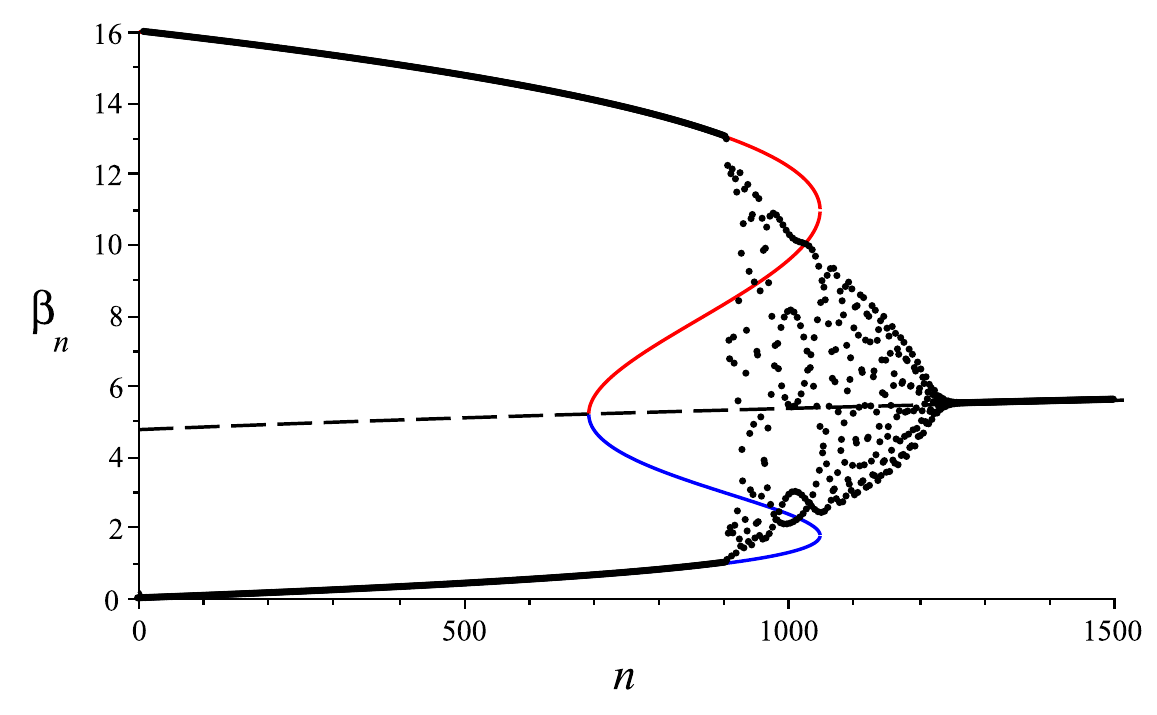} & \fig{2}{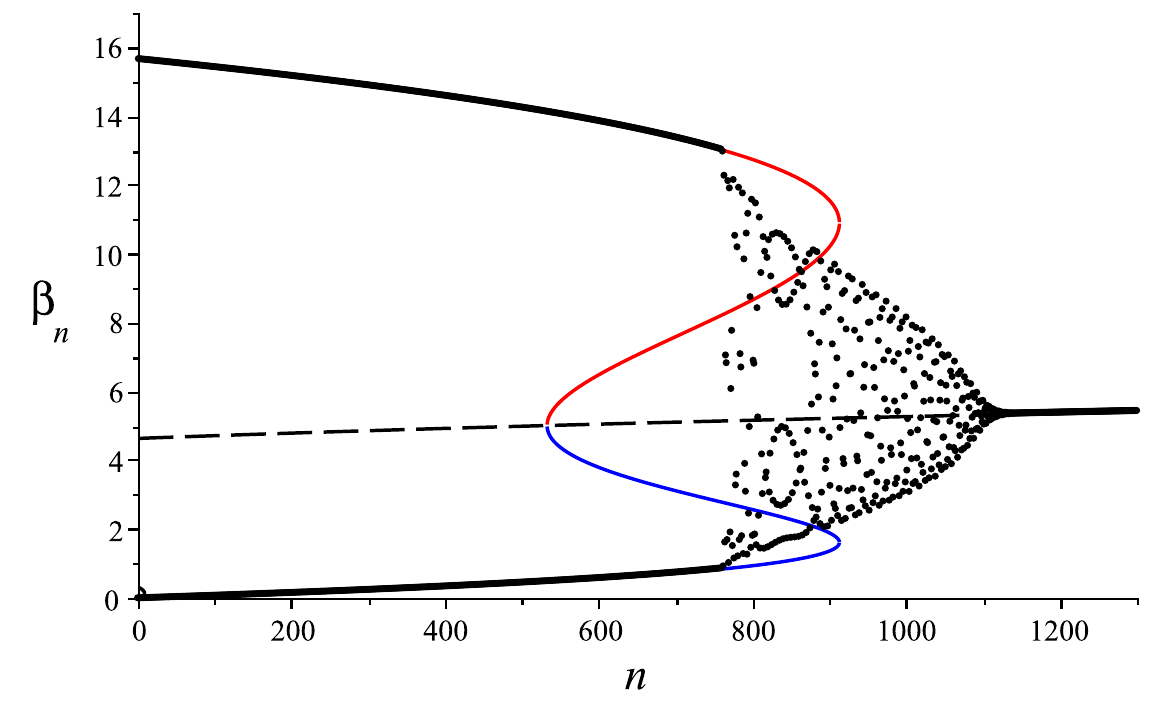}\\
\tau=25, \enskip\k=0.025 &\tau=25, \enskip\k=0.05
&\tau=25, \enskip\k=0.07 5\\[5pt]
\fig{2}{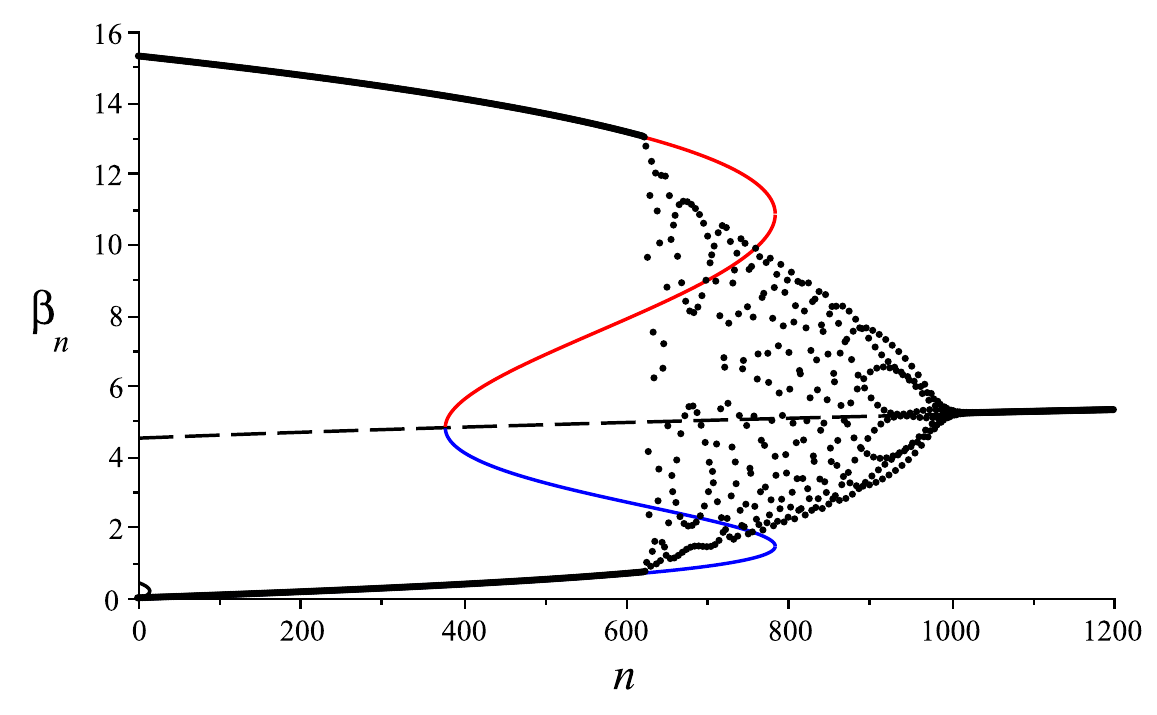}&\fig{2}{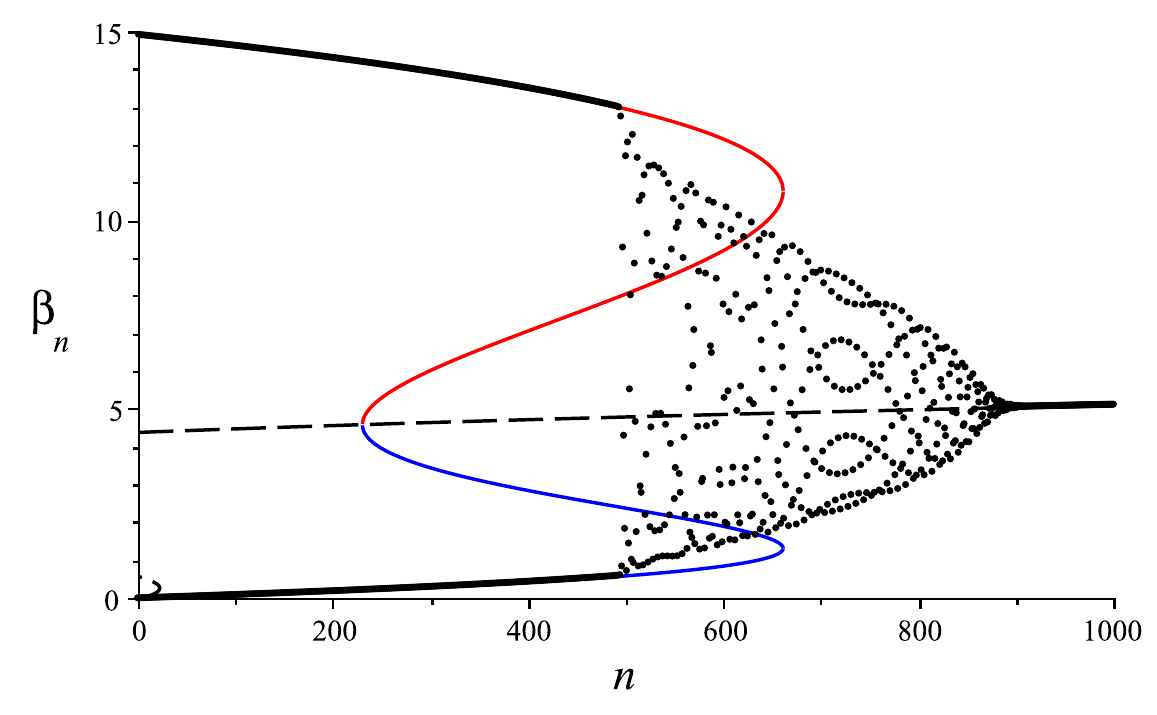}& \fig{2}{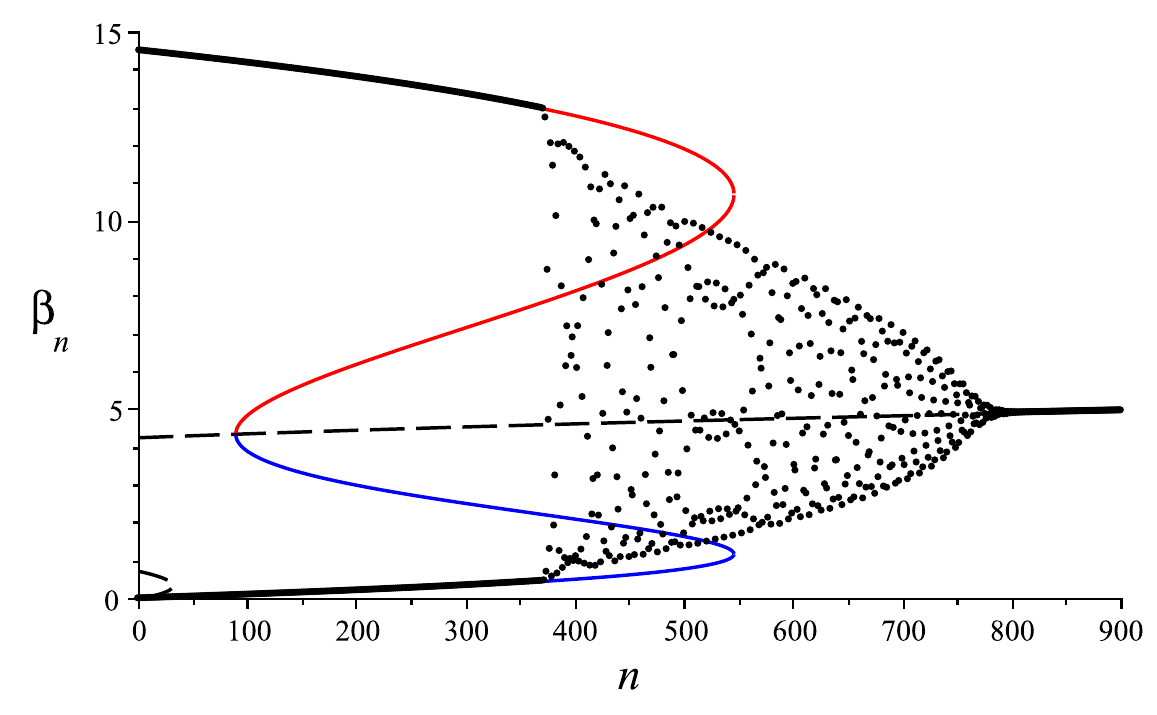} \\
\tau=25,\enskip \k=0.1 &
\tau=25, \enskip\k=0.125&\tau=25, \enskip\k=0.15\\[5pt]
\fig{2}{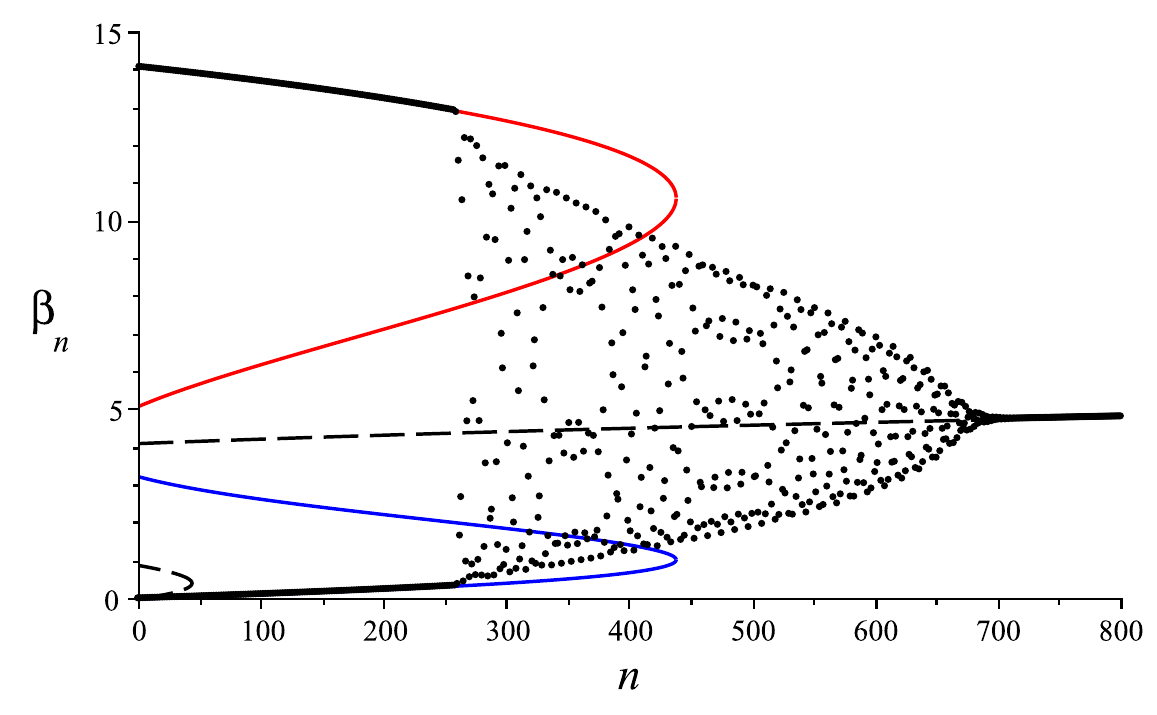} & \fig{2}{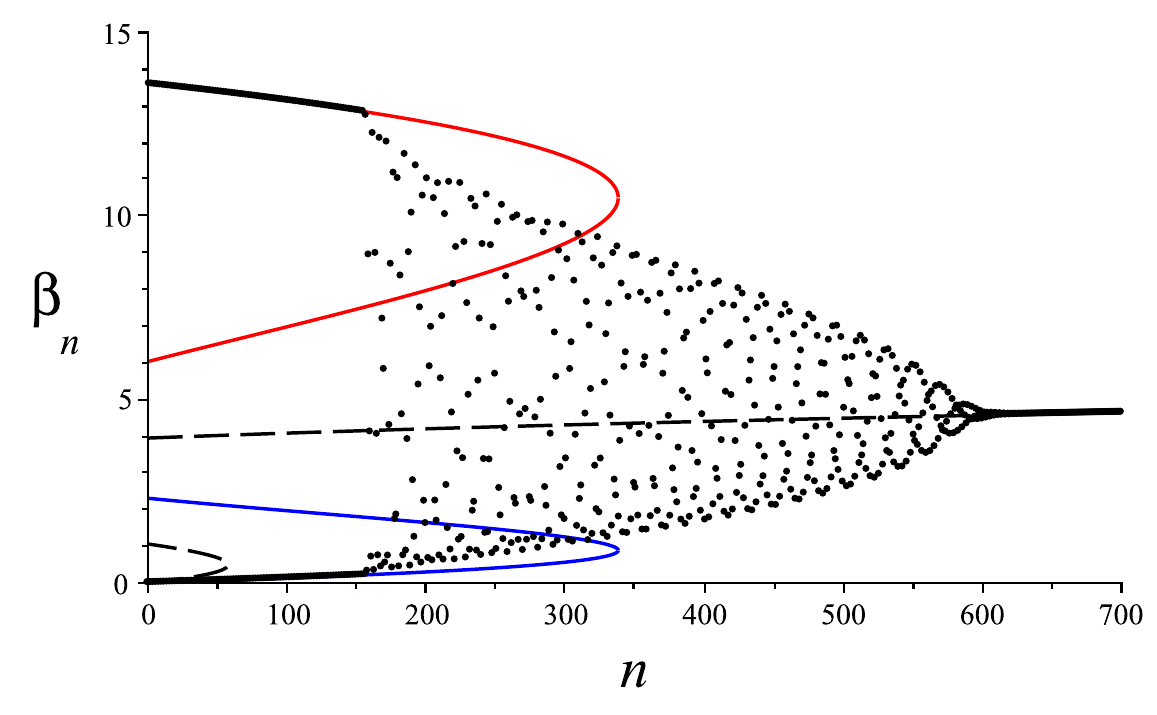} & \fig{2}{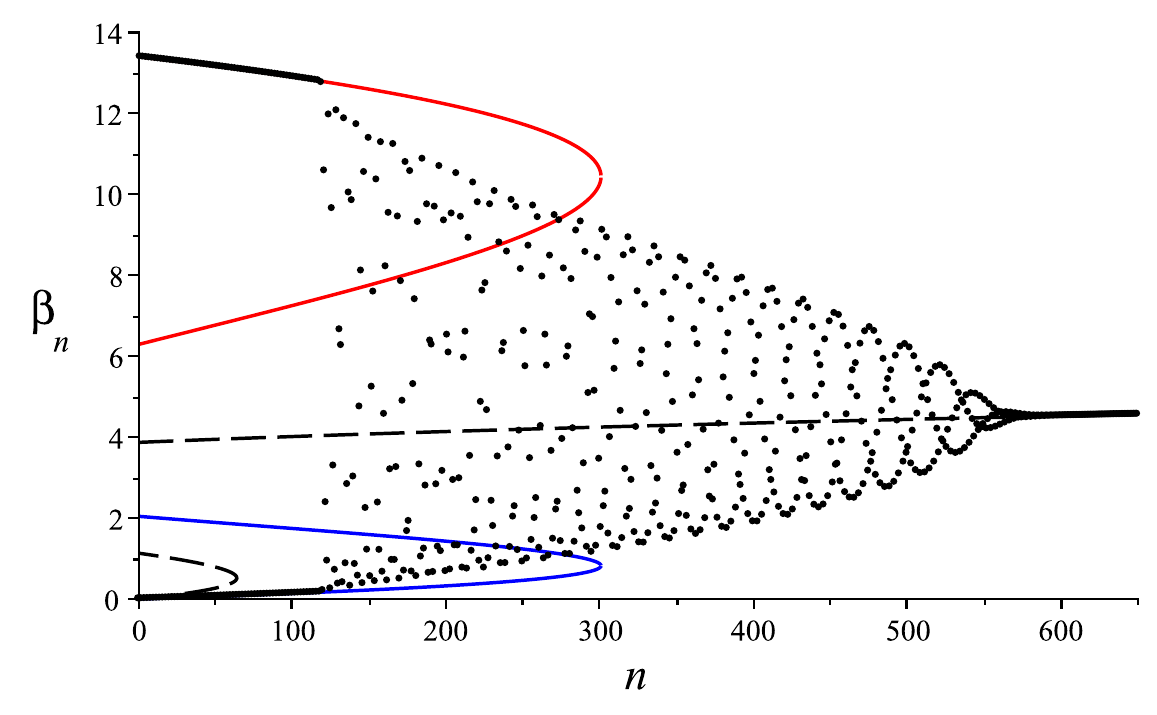}\\[5pt]
\tau=25,\enskip \k=0.175 &\tau=25,\enskip \k=0.2
&\tau=25,\enskip \k=0.21\\[5pt]
\fig{2}{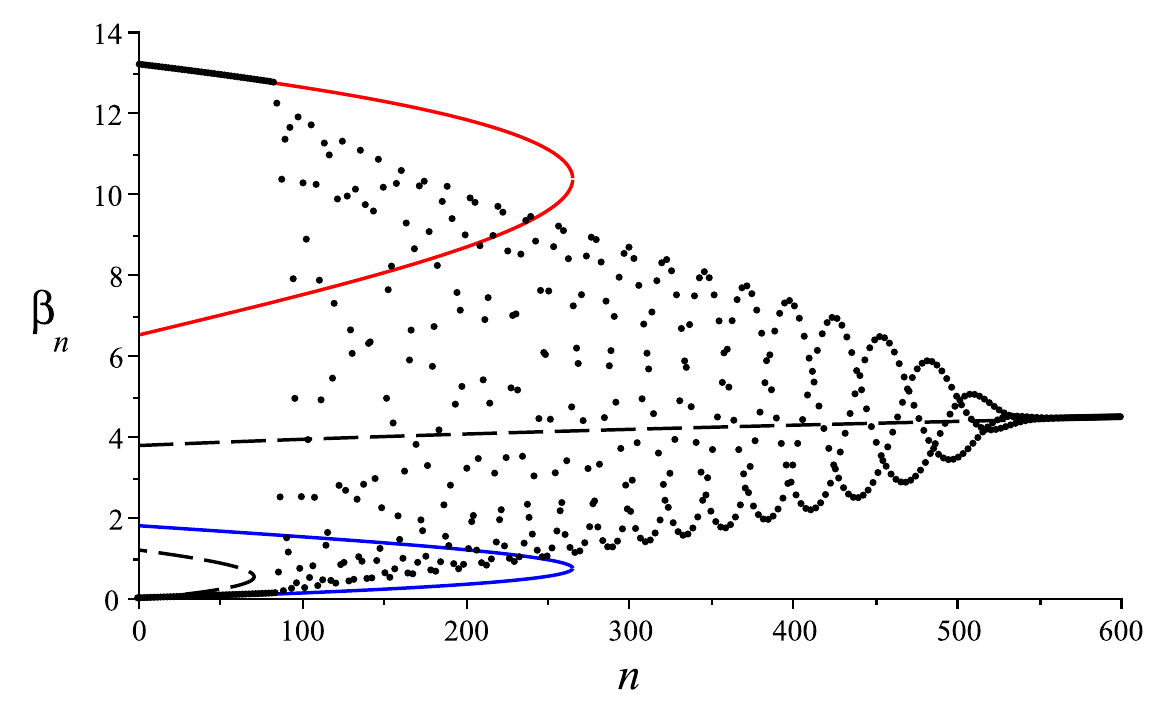} & \fig{2}{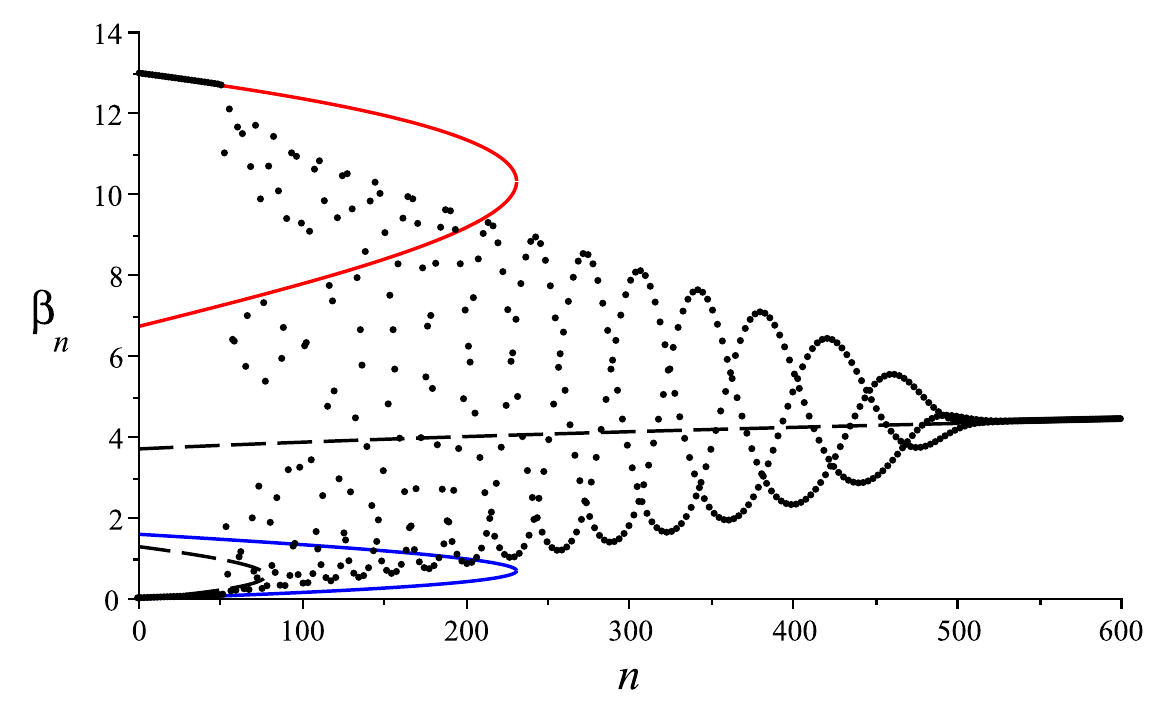} & \fig{2}{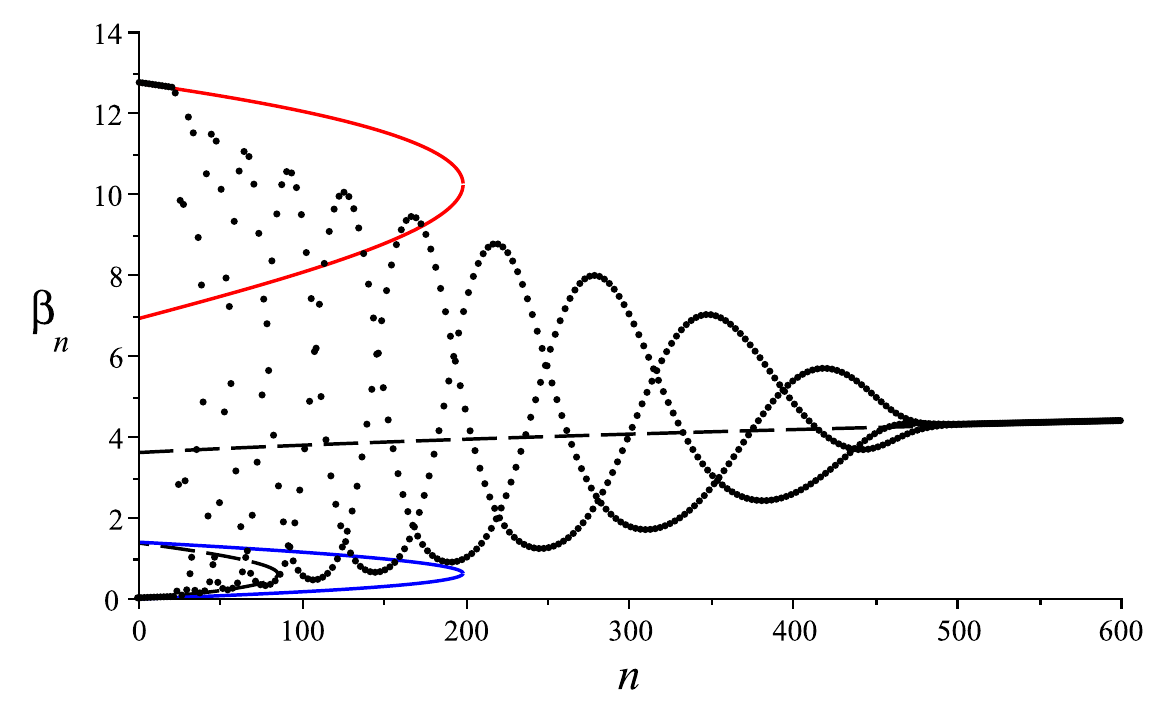}\\[5pt]
\tau=25,\enskip \k=0.22 &\tau=25,\enskip \k=0.23
&\tau=25,\enskip \k=0.24
\end{array}\]
\caption{\label{Figure76}Plots of the recurrence coefficient $\b_n$ when $\tau=25$, for various $\k$ such that $0<\k<\tfrac{1}{4}$, solutions of the system \eqref{UVsys}, with $u(n)$ plotted in \blue{blue} and $v(n)$ in \red{red} and the real solution of the cubic \eqref{eq:cubici} (dashed line).}
\end{figure}

The discriminant of \eqref{sol:UVsysa}
\[ {\Delta}=36864{\tau^6}(1-3\k)^3-5038848n^2,\]
so $\Delta=0$ when
\beq
n_1=\tfrac{4}{9}\,(\tfrac{1}{3}-\k)^{3/2}\tau^3.\nonumber\eeq
Also $u=v=\xi$ when $\eta=0$, so from \eqref{sol:UVsysb}
\[ \xi=\tfrac{1}{9}\big(1+\tfrac{1}{2}\sqrt{4-18\k}\,\big)\tau,\]
and hence from \eqref{sol:UVsysa}
\beq n_2=\frac{\left[2(4-27\k)+(4-18\k)^{3/2}\right]\tau^3}{243}
.\nonumber\eeq
We note that $n_2=0$ when $\k=\tfrac{1}{6}$.
Also $n_1=n_2$ when
\[\frac{4\sqrt{3}}{81}\,(1-3\k)^{3/2}=\frac{2(4-27\k)+(4-18\k)^{3/2}}{243},\]
which has solution $\k=-\tfrac{2}{3}$.

Subtracting the equations in the system \eqref{UVsys} yields
\beq (u-v)\big[3(u^2+4uv+v^2)-2\tau(u+v)+\k\tau^2\big]=0,\label{UVsys1}\eeq
and multiplying \eqref{UVsysa} by $v$, \eqref{UVsysb} by $u$ and subtracting yields
\beq (u-v)\big[12uv(u+v)-4\tau uv-n\big]=0.\label{UVsys2}\eeq
Assuming $u\not=v$, solving \eqref{UVsys1} for $u$ gives
\[ u= -2v+\tfrac{1}{3}\tau\pm\tfrac{1}{3}\sqrt{27v^2-6\tau v+\tau^2(1-3\k)},\]
and then substituting this into \eqref{UVsys2} gives
\[180 v^{3}-36\tau v^{2}+4\tau^{2}(3 \k-1)v\pm4(\tau -9 v) v\sqrt{27v^2-6\tau v+\tau^2(1-3\k)} = 3 n.\]

{In Figure \ref{Figure76} plots of the recurrence coefficient $\b_n$ when $\tau=25$, for various $\k$ such that $0<\k<\tfrac{1}{4}$, solutions of the system \eqref{UVsys}, with $u(n)$ plotted in \blue{blue} and $v(n)$ in \red{red} and the real solution of the cubic \eqref{eq:cubici} (dashed line). Initially $\b_{2n}$ and $\b_{2n+1}$ follow the system \eqref{UVsys}, $\b_{2n}$ follows $u(n)$ and $\b_{2n+1}$ follow $v(n)$. After the ``transition region", $\b_n$ follows the cubic \eqref{eq:cubici}. The size of the ``transition region" decreases as $\k$ decreases. Analogous to the previous case, for $\k$ just below $\tfrac{1}{4}$, there is evidence of a three-fold structure.}

\subsection{\label{caseiii}Case (iii): $\tau>0$ and $\k=\tfrac{1}{4}$} 
In the previous two subsections we saw that the behaviour of $\b_n$ for small $n$ was quite different depending whether $\k>\tfrac{1}{4}$ or $0,\k<\tfrac{1}{4}$.
In the case when $t=-\tfrac{1}{4}\tau^2$ 
the weight is given by
\beq \w(x;\tau)=\exp\left\{-x^2(x^2-\tfrac{1}{2}\tau)^2\right\},\label{w642qr}\eeq
and $U(x)= x^2(x^2-\tfrac{1}{2}\tau)^2$, which has three double roots at $x=0$, $x=\pm\sqrt{\tfrac{1}{2}\tau}$.
This case was discussed by S\'en\'echal \cite[Figures 5, 6]{refSen92} who noted that ``the upper branch contains twice as many points as the lower one"; see also Demeterfi \etal\ \cite[Figure 7]{refDDJT}, Boobna and Ghosh \cite[Figure 2]{refBG13}. 

In Figure \ref{fig1_caseiii} the recurrence coefficients $\b_n$ for the weight \eqref{w642qr} are plotted in the cases when $\tau=20$, $\tau=25$ and $\tau=30$. The recurrence coefficients $\b_{3n}$ are plotted in \blue{blue}, $\b_{3n+1}$ in \green{green} and $\b_{3n-1}$ in \red{red}. These show that initially the recurrence coefficients $\b_{3n}$ follow one curve whilst $\b_{3n+1}$ and $\b_{3n-1}$ appear to follow the same curve. Then for $n$ sufficiently large, all the recurrence coefficients $\b_n$ follow the same curve. {We remark that as $\tau$ increases it becomes more difficult to distinguish between the $\b_{3n+1}$ (in \green{green}) and $\b_{3n-1}$ (in \red{red}) recurrence coefficients.}
\begin{figure}[ht!]
\[ \begin{array}{c@{\quad}c@{\quad}c}
\fig{2}{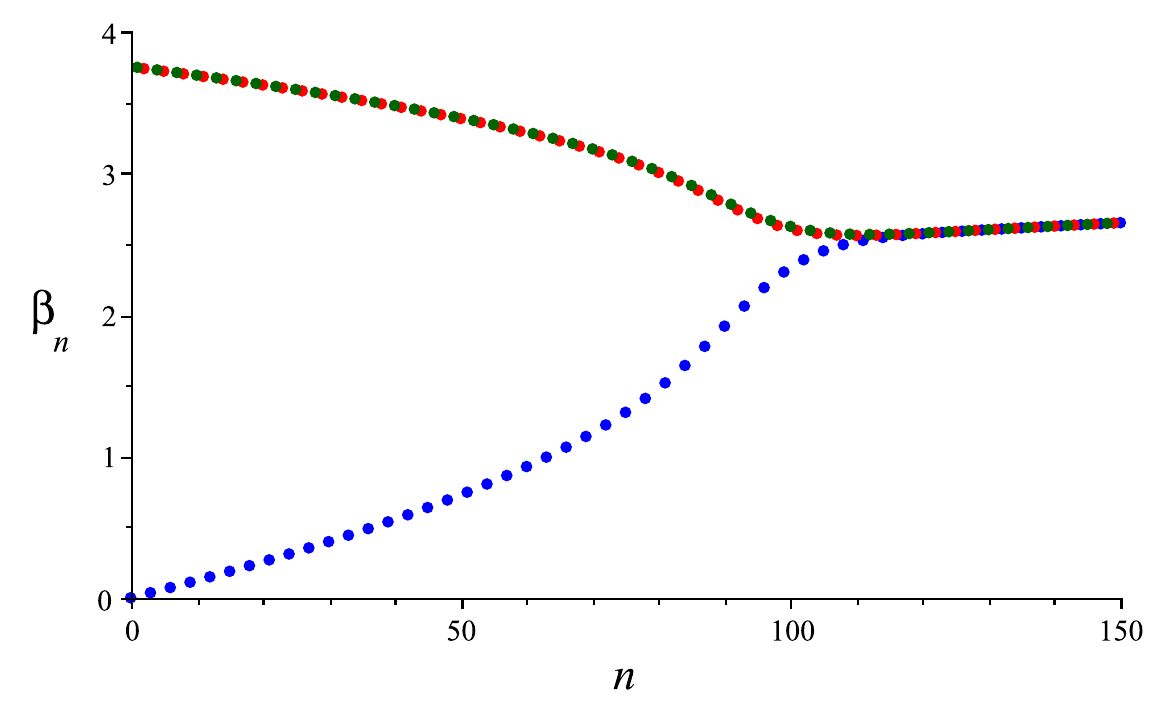} & 
\fig{2}{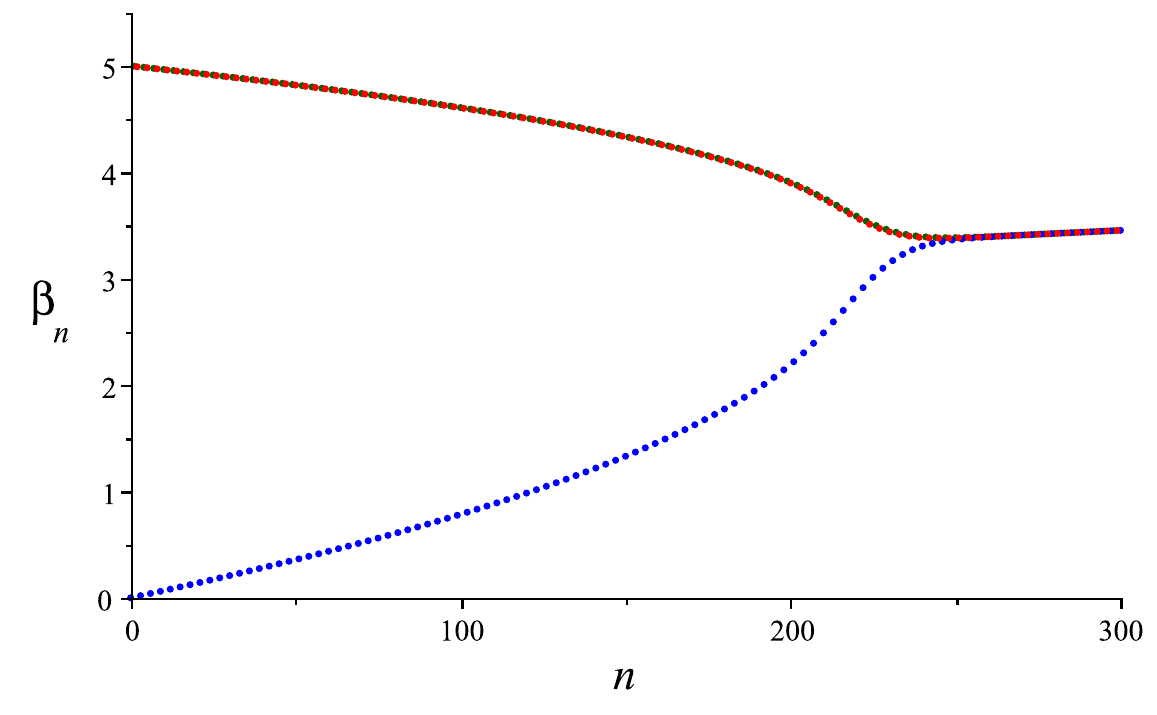} & \fig{2}{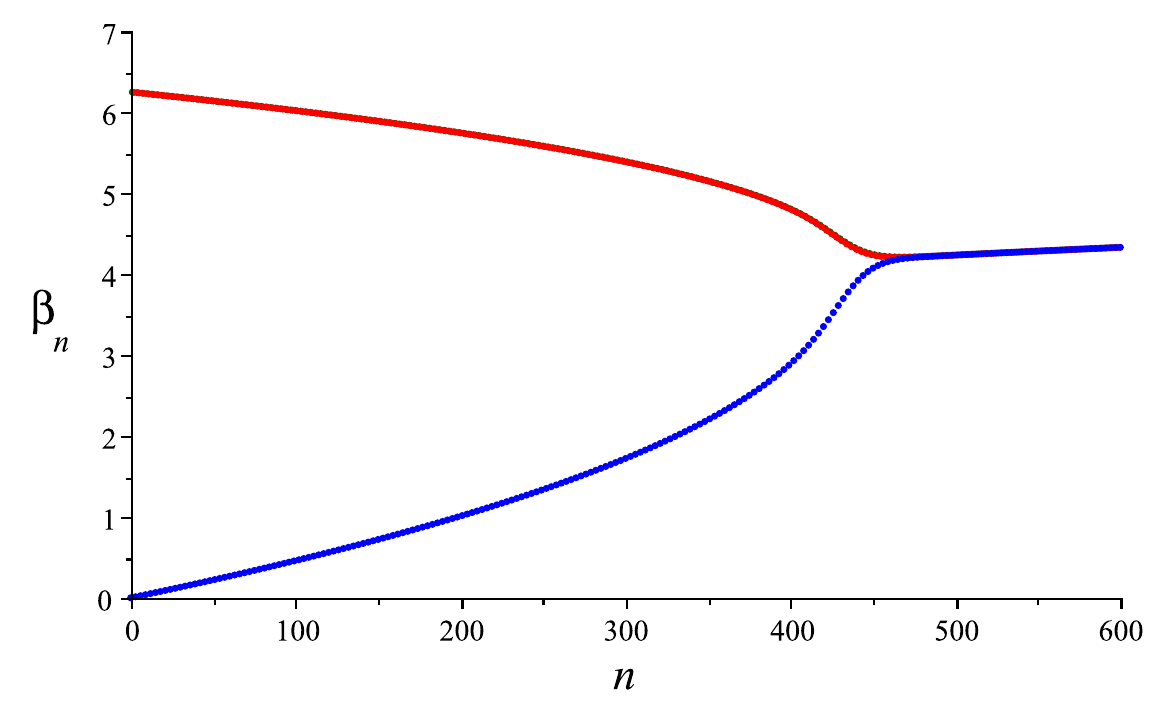}\\
\tau=15 & \tau=20 &\tau=25
\end{array} \]
\caption{\label{fig1_caseiii}Plots of the recurrence coefficients $\b_n$ for the weight \eqref{w642qr}, in the cases when $\tau=15$, $\tau=20$ and $\tau=25$. The recurrence coefficients $\b_{3n}$ are plotted in \blue{blue}, $\b_{3n+1}$ in \green{green} and $\b_{3n-1}$ in \red{red}.} 
\end{figure}

To investigate this case, we set $\b_{3n}=x$, $\b_{3n\pm1}=y$ and $t=-\tfrac{1}{4}\tau^2$ in \eqref{eq:rr642} which gives the system
\begin{subequations} 
\begin{align}&6 x\!\left(x^{2}+4 x y +5 y^{2}\right)-4\tau x\!\left(x+2 y \right)+\tfrac{1}{2}\tau^2x=n,\label{UVsys3a}\\
&6 y\!\left(x^{2}+ 5x y +4y^{2}\right)-4\tau y\!\left(x +2y \right)+\tfrac{1}{2}\tau^2y =n,\label{UVsys3b}
\end{align}\end{subequations}
{i.e.\ we have replaced $\b_n$, with $n$ divisible by $3$, by $x$ and the other $\b_n$ by $y$}.
Multiplying \eqref{UVsys3a} by $y$, \eqref{UVsys3b} by $x$ and subtracting gives
\beq (x-y)(6xy^2-n)=0.\label{Case3:eq1}\eeq
Also subtracting \eqref{UVsys3a} from \eqref{UVsys3b} gives
\beq (x-y)(2x+4y-\tau)(6x+12y-\tau)=0.\label{Case3:eq2}\eeq
If $x=y=\b$ then we obtain the cubic
\beq 60\b^3-12\tau \b^2+\tfrac{1}{2}\tau^2\b -n=0, \label{eq:cubic3}\eeq
which is \eqref{eq:cubici} with $\k=\tfrac{1}{4}$. 
{If $x\not= y$, then solving \eqref{Case3:eq2} for $x$ and substituting into \eqref{Case3:eq1} gives the two cubics
\begin{subequations}\label{cubicV}\begin{align} 12y^3-3\tau y^2+n&=0,\label{cubicVa}\\ 12y^3-\tau y^2+n&=0.\label{cubicVb}\end{align}\end{subequations}
Similarly, when $x\not=y$, solving \eqref{Case3:eq2} for $x$ and substituting into \eqref{Case3:eq1} gives the two cubics
\begin{subequations}\label{cubicU}\begin{align}
12x^3-12\tau x^2+3\tau^2x-8n=0,\label{cubicUa}\\ 
12x^3-4\tau x^2+\tfrac{1}{3}\tau^2x-8n=0.
\label{cubicUb}\end{align}\end{subequations}
The cubics \eqref{eq:cubic3}, \eqref{cubicVa} and \eqref{cubicUa} meet at the point $(\tau^3/36,\tau/6)$, as well as the origin, whereas the cubics \eqref{eq:cubic3}, \eqref{cubicVb} and \eqref{cubicUb} meet at the point $(\tau^3/972,\tau/18)$, as well as the origin. From equation \eqref{eq:cubicia} with $\k=\tfrac{1}{4}$, the cubic \eqref{eq:cubic3} has a positive gradient at $(\tau^3/36,\tau/6)$ whilst it has a negative gradient at $(\tau^3/972,\tau/18)$, so \eqref{cubicVa} and \eqref{cubicUa} are the relevant cubics. This is illustrated in Figure \ref{fig:uv3}(i).}

The real solutions of the cubics \eqref{eq:cubic3}, \eqref{cubicVa} and \eqref{cubicUa} are plotted in Figure \ref{fig:uv3}(ii). Plots of the recurrence coefficients $\b_n$ for the weight \eqref{w642qr}, in the cases when $\tau=20$, $\tau=25$ and $\tau=30$, together with the real solutions of the cubics \eqref{eq:cubic3}, \eqref{cubicV} and \eqref{cubicU}, are given in Figure \ref{fig2_caseiii}. {In these plots, it seems that the coefficients $\b_n$ lie on the curve \eqref{cubicU} when $n\equiv 0\mod 3$ and $\b_n$ lie on the curve \eqref{cubicV} when $n\not\equiv 0\mod 3$. The differences between those values are illustrated in Figure \ref{cubicVerror}.}
\begin{figure}[ht!]
\[\begin{array}{c@{\qquad}c}
\fig{2.75}{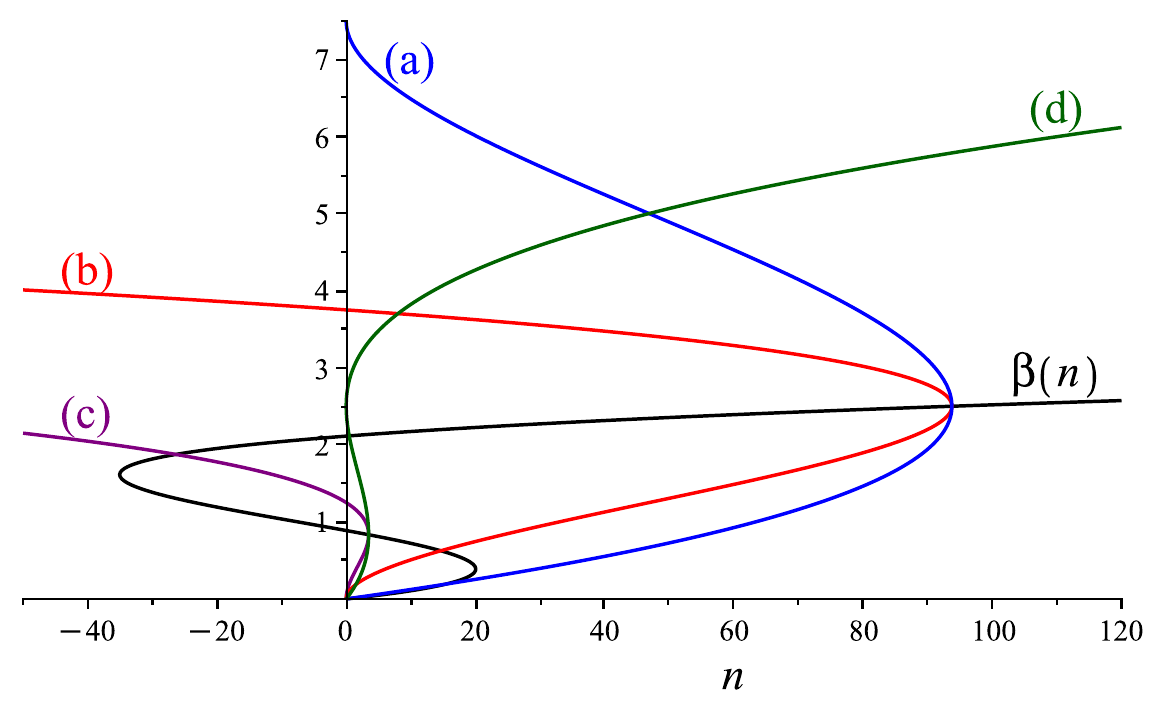} &\fig{2.75}{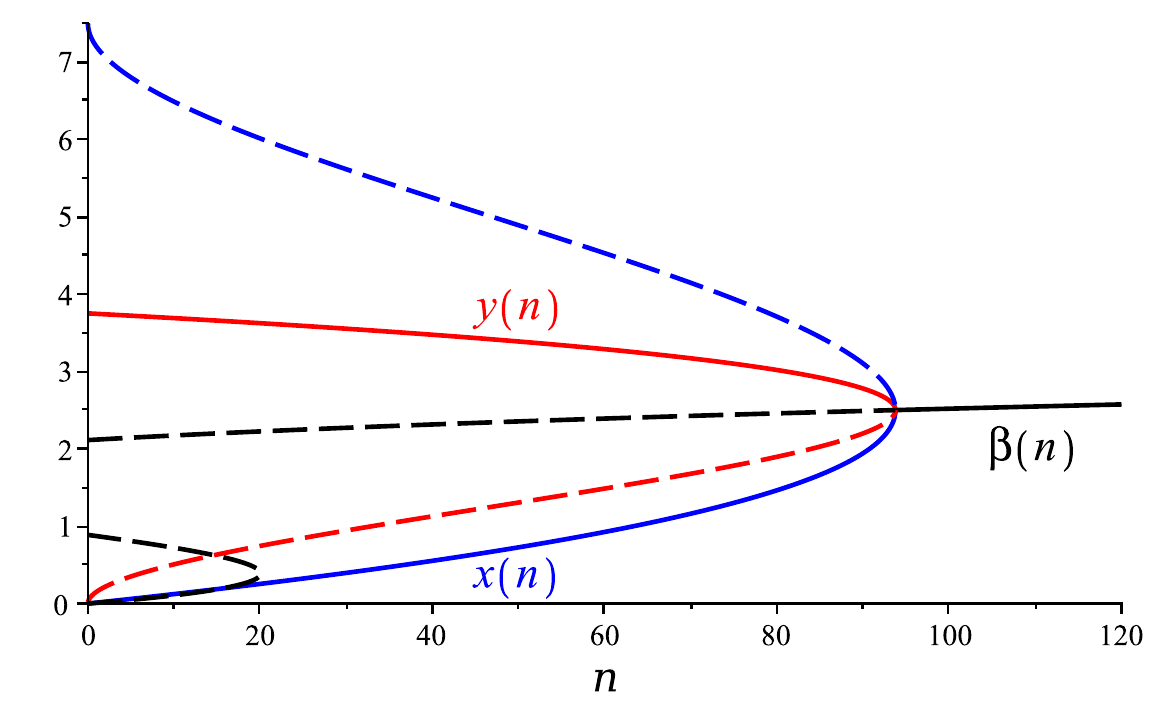}\\
\text{(i)} & \text{(ii)}
\end{array}\]
\caption{{(i), Plots of the real solutions of the cubics (a) \eqref{cubicVa}, (b) \eqref{cubicUa}, (c) \eqref{cubicVb}, (d) \eqref{cubicUb}, in \red{red}, \blue{blue}, \purple{purple} and \green{green}, respectively, together with the cubic \eqref{eq:cubic3} in black}.
(ii), Plots of the real solutions of the cubics \eqref{eq:cubic3}, \eqref{cubicVa} and \eqref{cubicUa}, in black, \red{red} and \blue{blue}, respectively. The solid lines are the sections of the cubics which the recurrence coefficients approximately follow and the dashed lines other sections of the cubics in the positive quadrant.}
\label{fig:uv3}
\end{figure}

\begin{figure}[ht!]
\[ \begin{array}{c@{\quad}c@{\quad}c}
\fig{2}{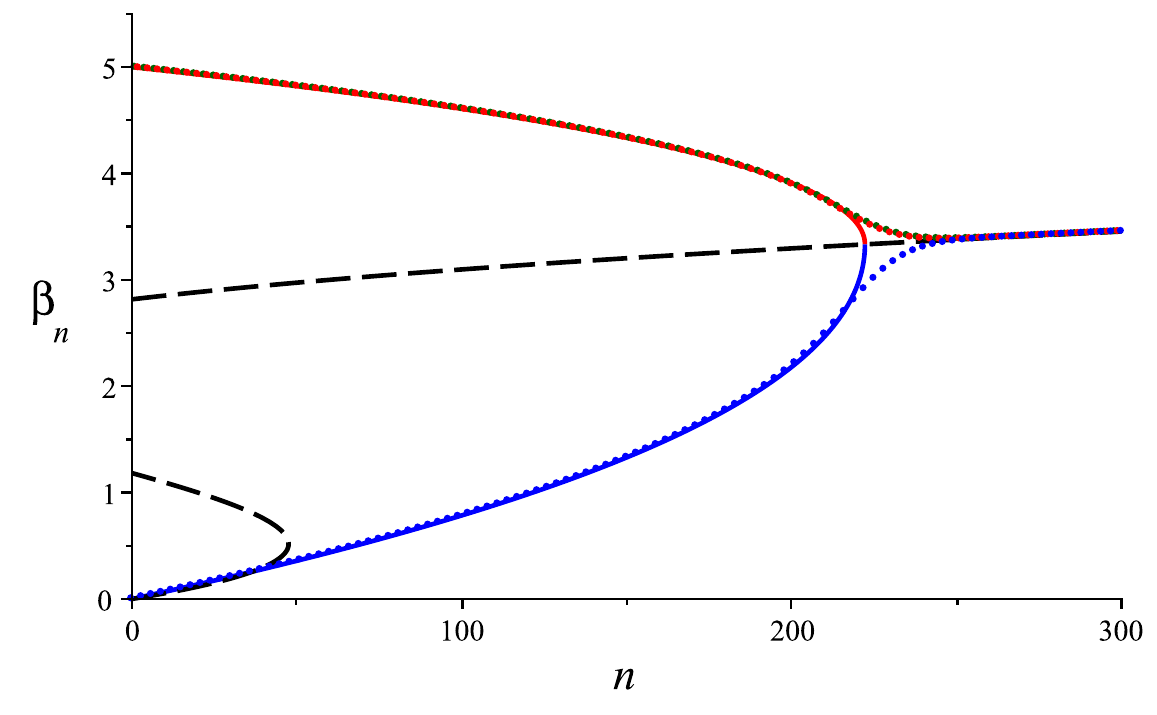} & \fig{2}{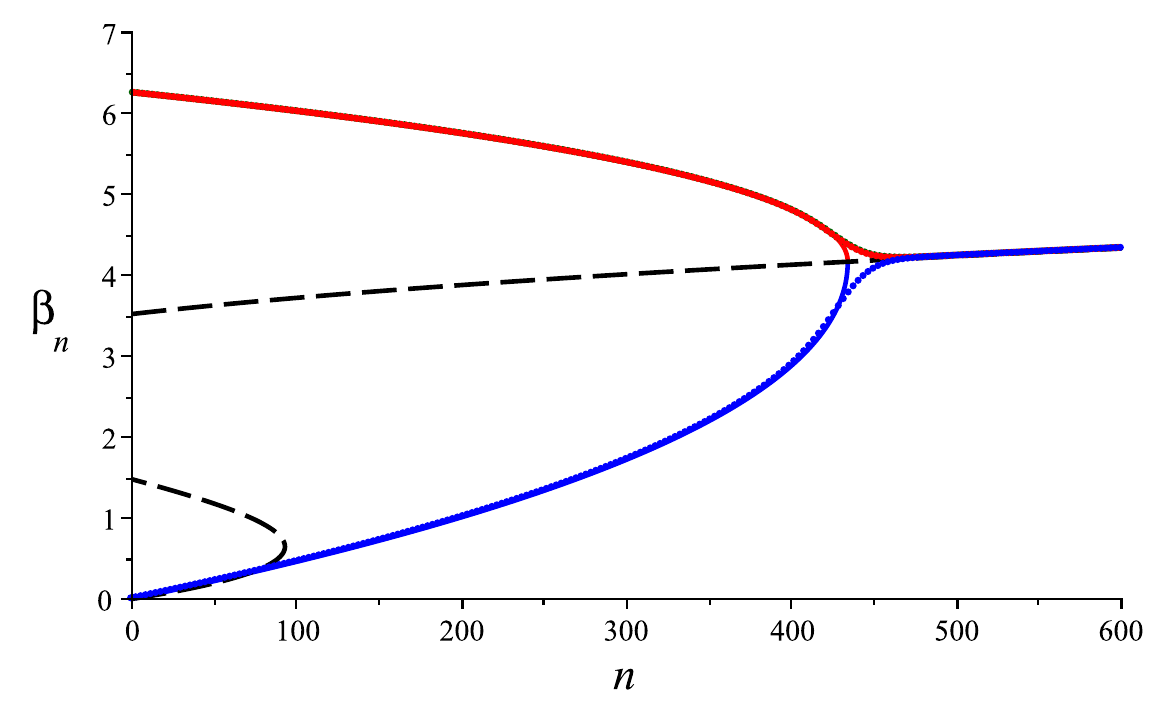} &\fig{2}{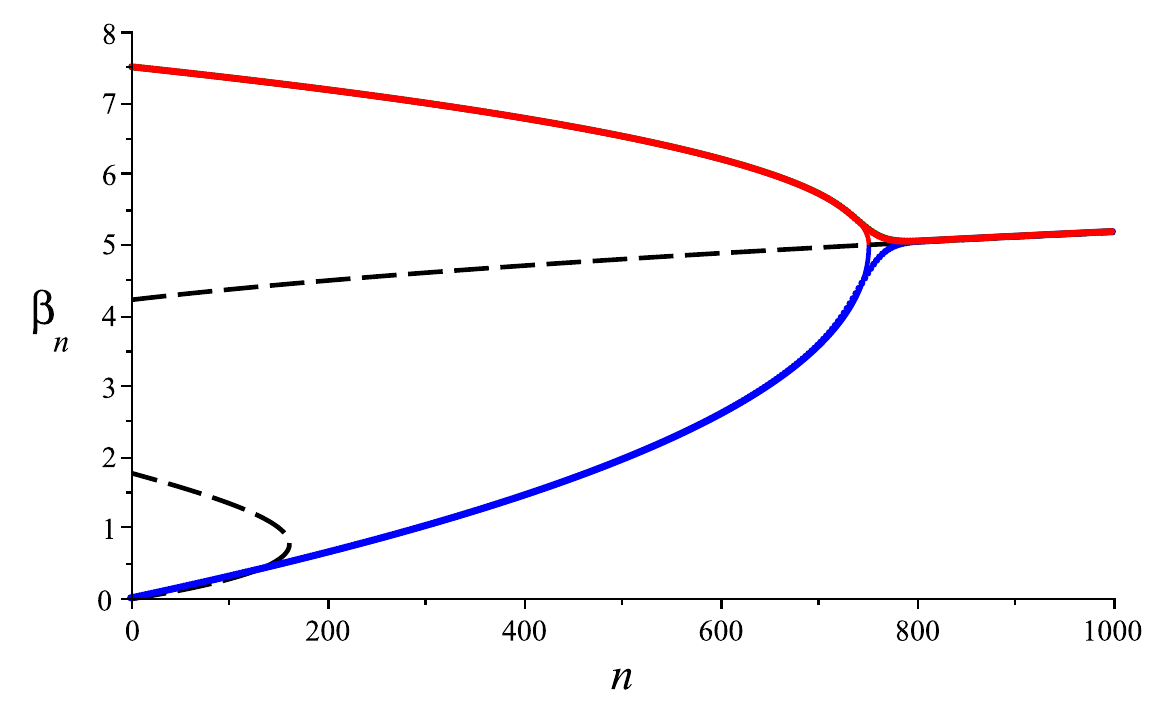}\\
\tau=20&\tau=25&\tau=30
\end{array} \]
\caption{\label{fig2_caseiii}Plots of the recurrence coefficients $\b_n$ for the weight \eqref{w642qr}, in the cases when $\tau=20$, $\tau=25$ and $\tau=30$, together with the real solutions of the cubics \eqref{eq:cubic3}, \eqref{cubicVa} and \eqref{cubicUa}, which are plotted in black, \red{red} and \blue{blue}, respectively. The recurrence coefficients $\b_{3n}$ are plotted in \blue{blue}, $\b_{3n+1}$ in \green{green} and $\b_{3n-1}$ in \red{red}.} 
\end{figure}

\begin{figure}[ht!]
\[ \begin{array}{c@{\quad}c@{\quad}c}
\fig{2}{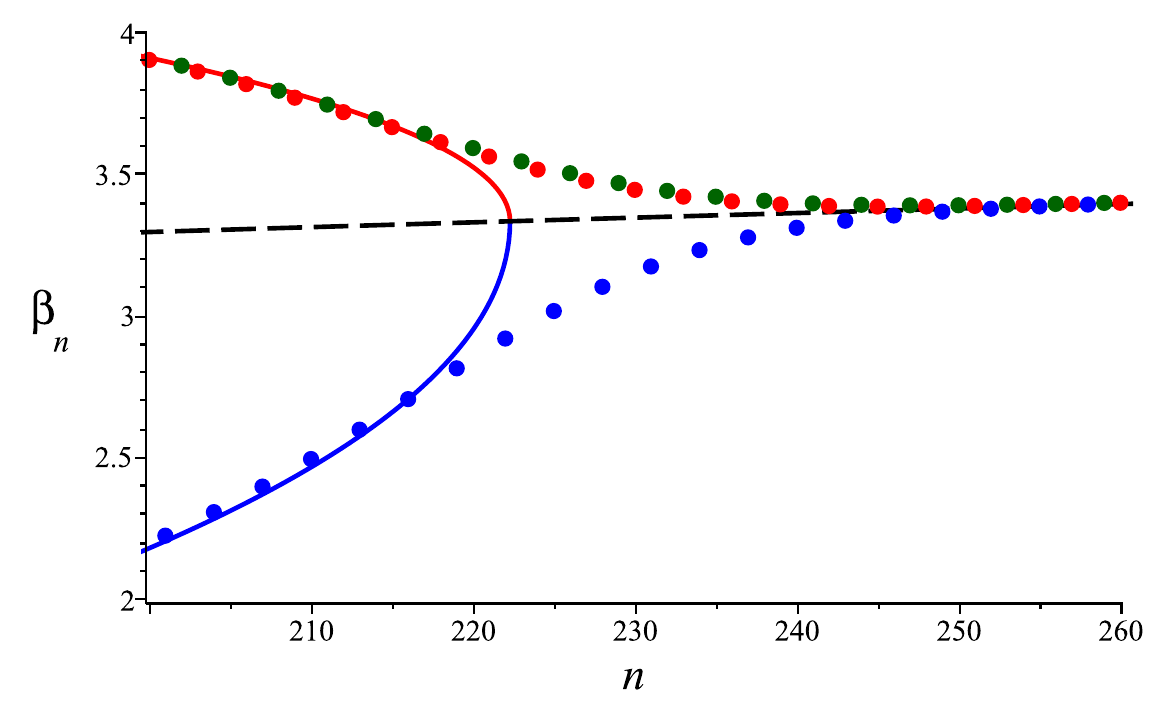} & \fig{2}{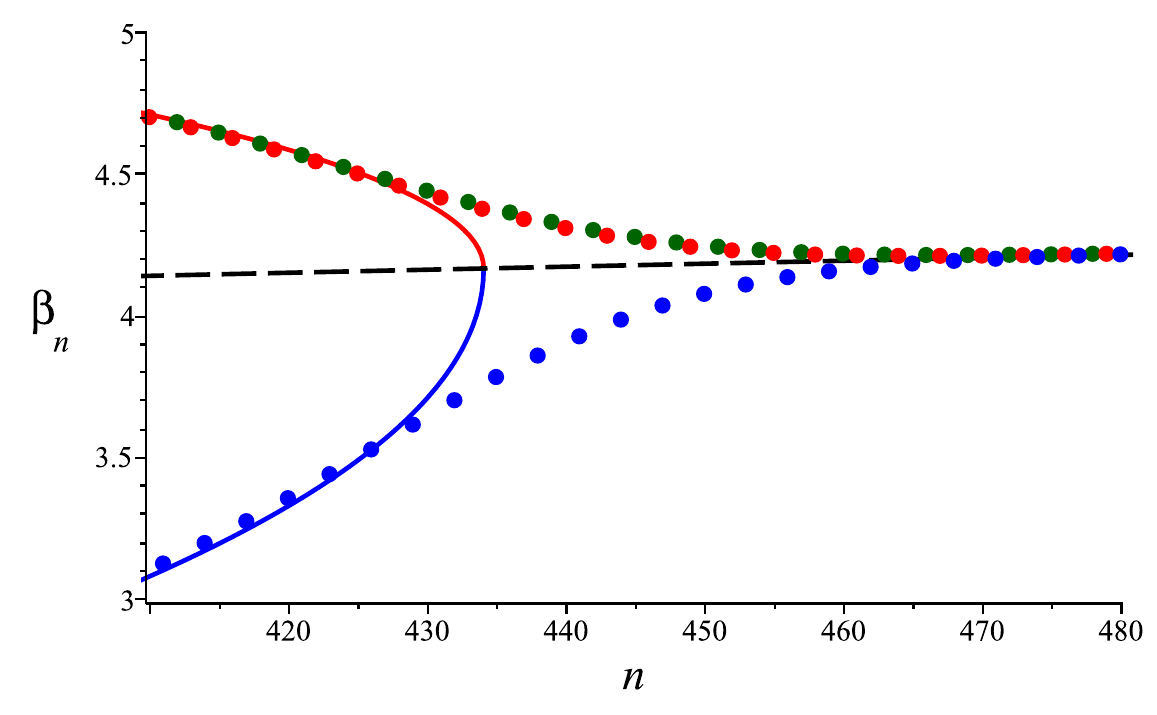} &\fig{2}{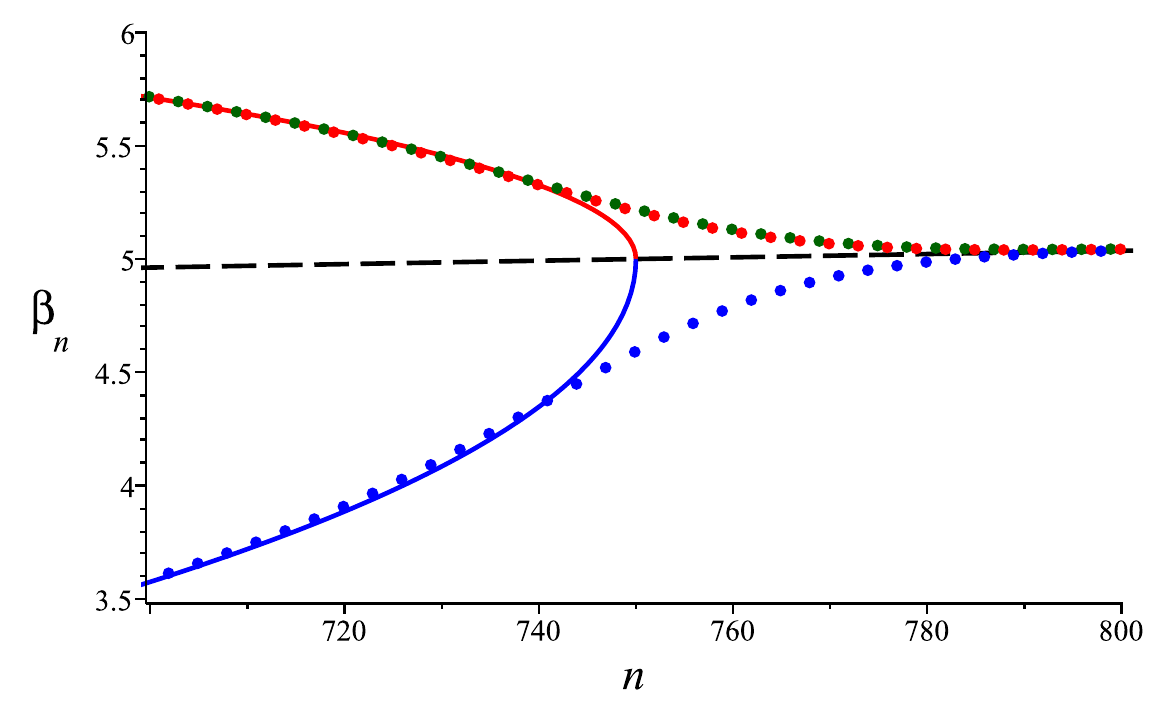}\\
\tau=20&\tau=25&\tau=30
\end{array} \]
\caption{{Plots of the recurrence coefficients $\b_n$ for the weight \eqref{w642qr}, in the cases when $\tau=20$, $\tau=25$ and $\tau=30$, together with the real solutions of the cubics \eqref{eq:cubic3}, \eqref{cubicVa} and \eqref{cubicUa}, which are plotted in black, \red{red} and \blue{blue}, respectively, showing in detail the region where the recurrence coefficients start following \eqref{eq:cubic3}. The recurrence coefficients $\b_{3n}$ are plotted in \blue{blue}, $\b_{3n+1}$ in \green{green} and $\b_{3n-1}$ in \red{red}.}}
\end{figure}

\begin{figure}[ht!]
\[ \begin{array}{c@{\quad}c@{\quad}c}
\fig{2}{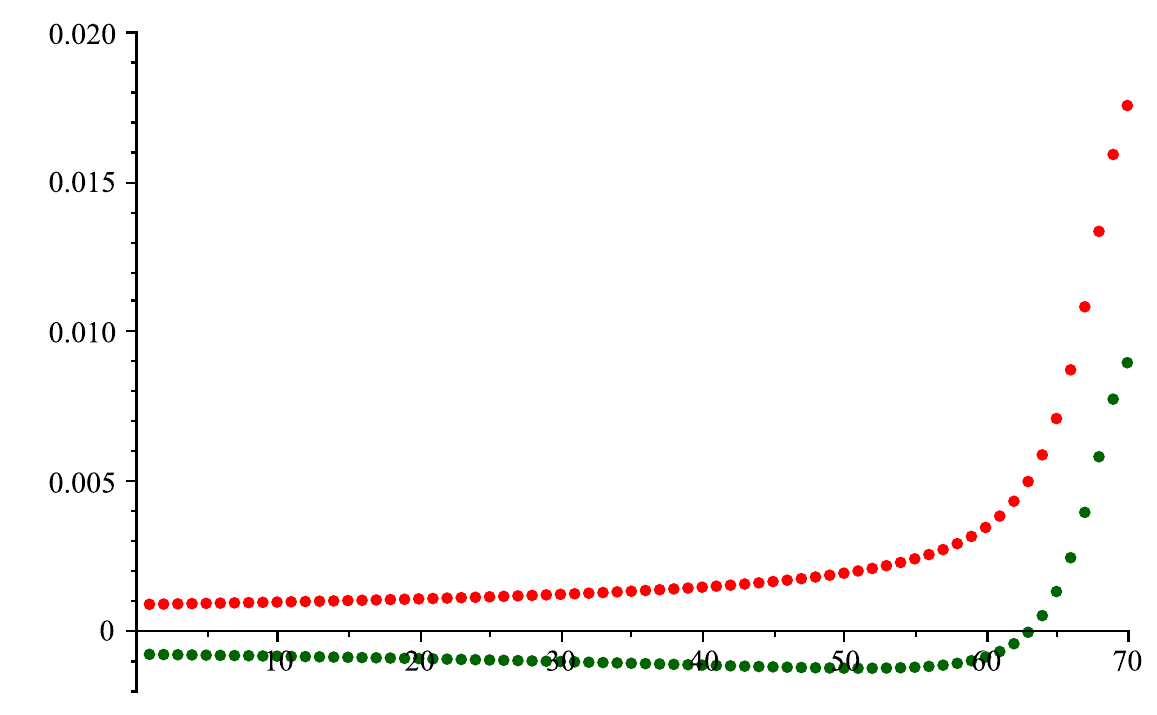} & \fig{2}{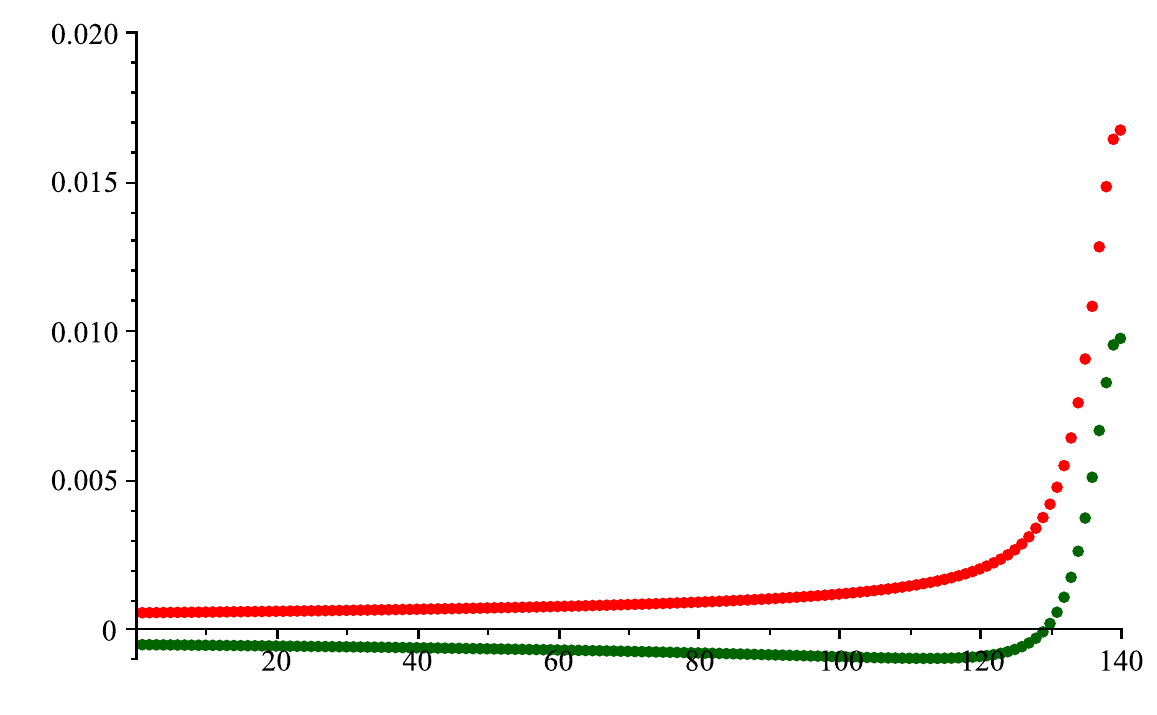} &\fig{2}{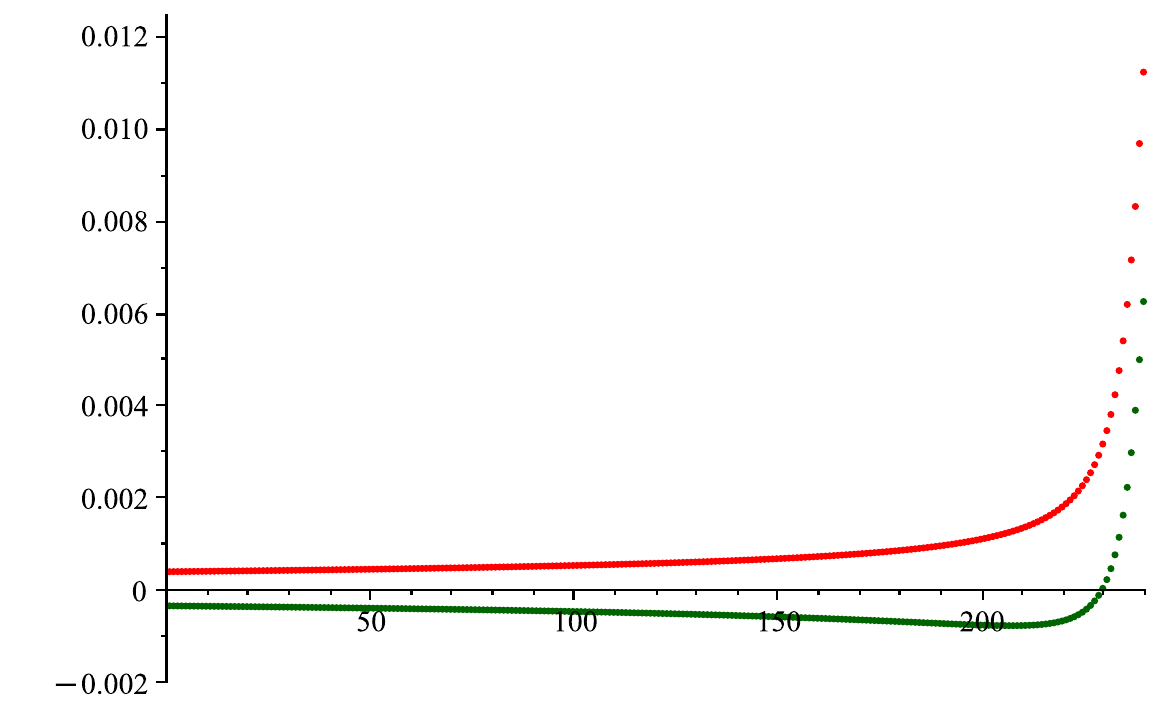}\\
\tau=20&\tau=25&\tau=30
\end{array} \]
\caption{\label{cubicVerror} Plots of the $\b_{3n+1}-y(3n+1)$ in \green{green} and $\b_{3n-1}-y(3n-1)$ in \red{red}, with $y(n)$ the real solution of \eqref{cubicV}, in the cases when $\tau=20$, $\tau=25$ and $\tau=30$.}
\end{figure}

\begin{figure}[ht!]
\[\begin{array}{c@{\quad}c@{\quad}c}
\fig{2}{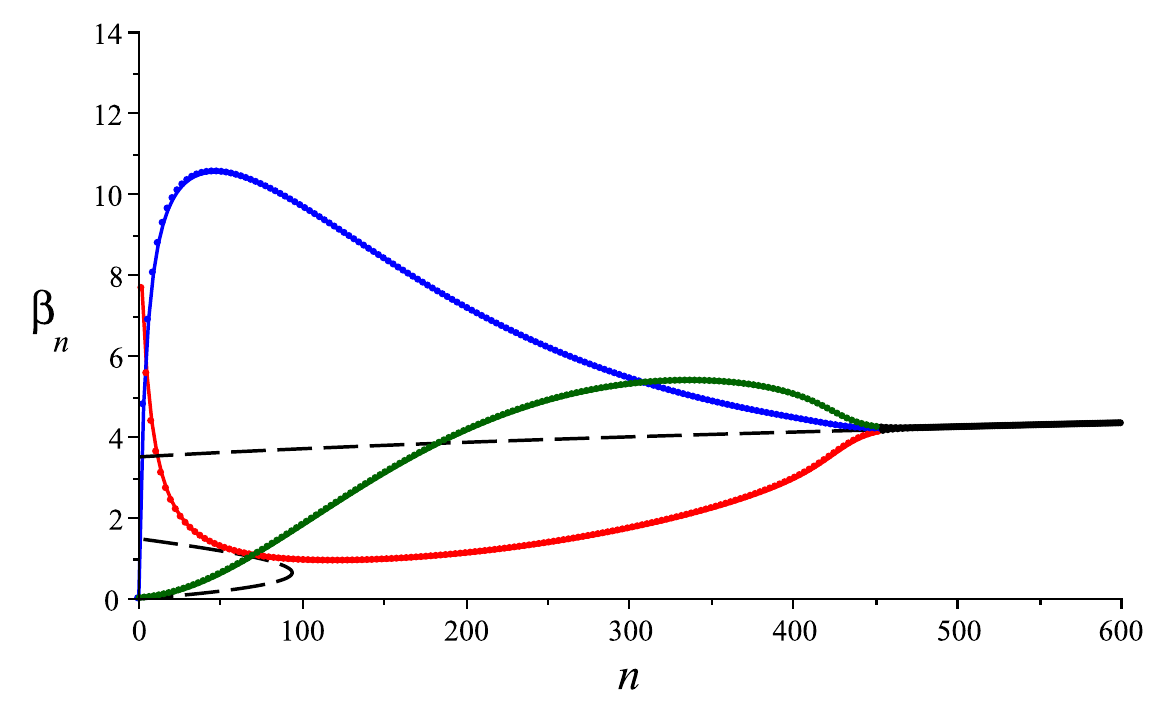} & \fig{2}{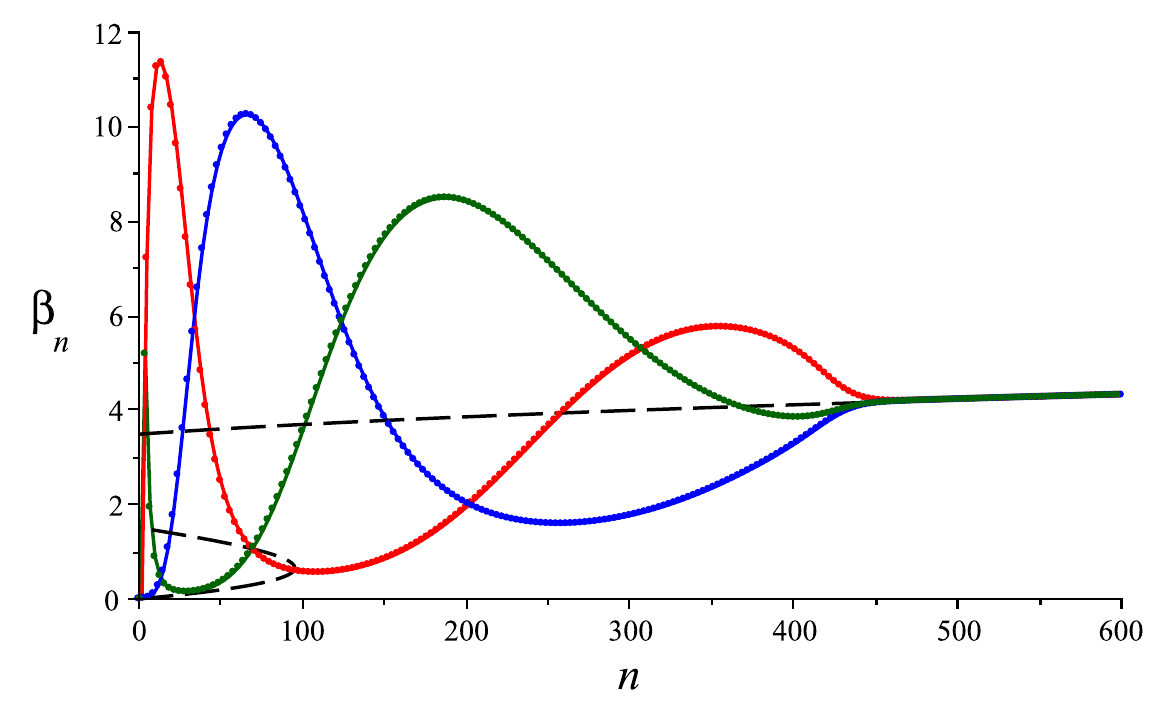} &\fig{2}{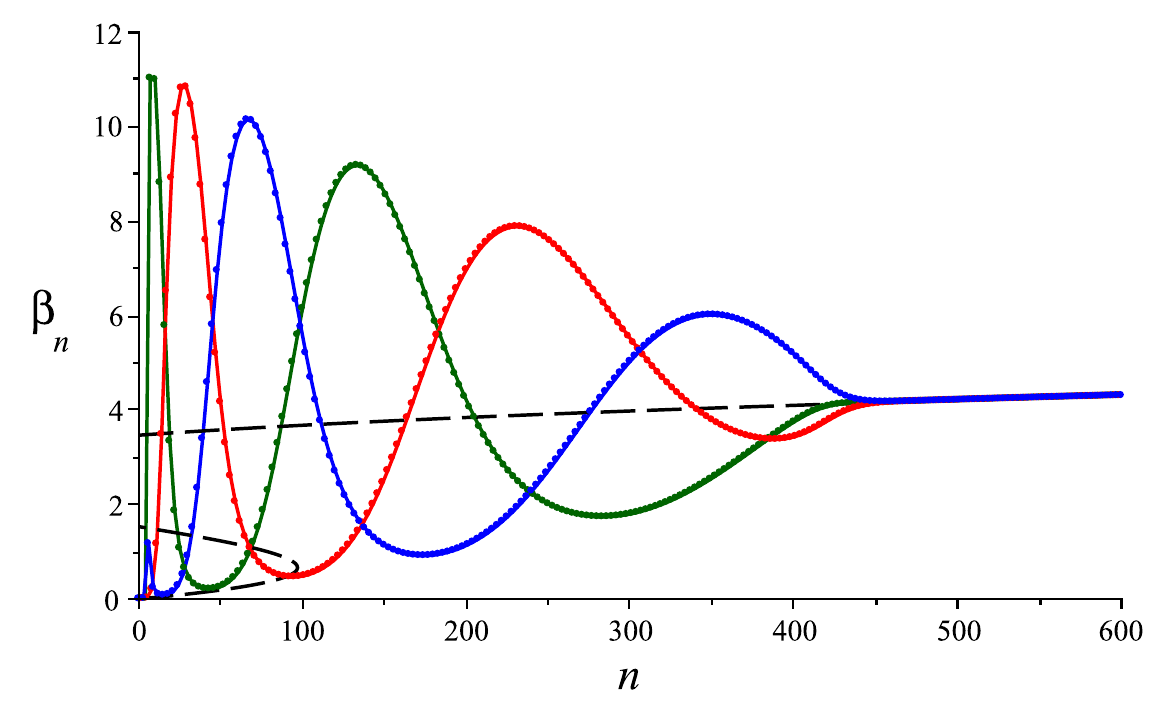} \\
\tau=25,\enskip\k=0.251&\tau=25,\enskip\k=0.253 &\tau=25,\enskip\k=0.255\\[5pt]
\fig{2}{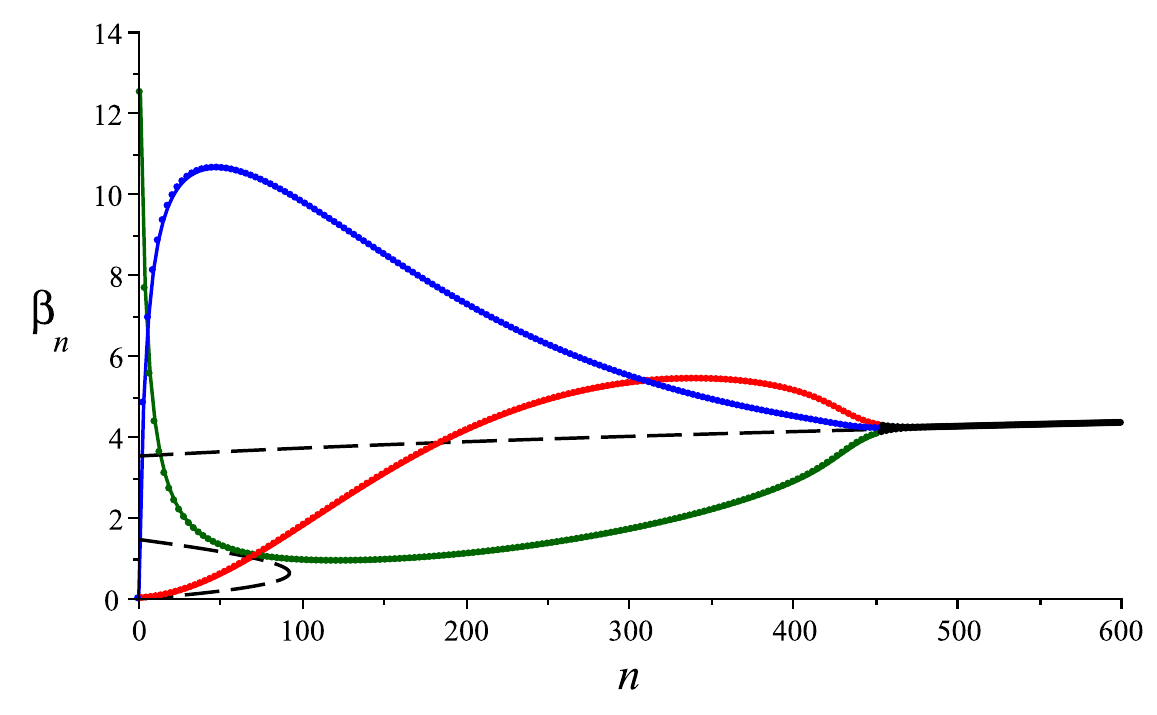} & \fig{2}{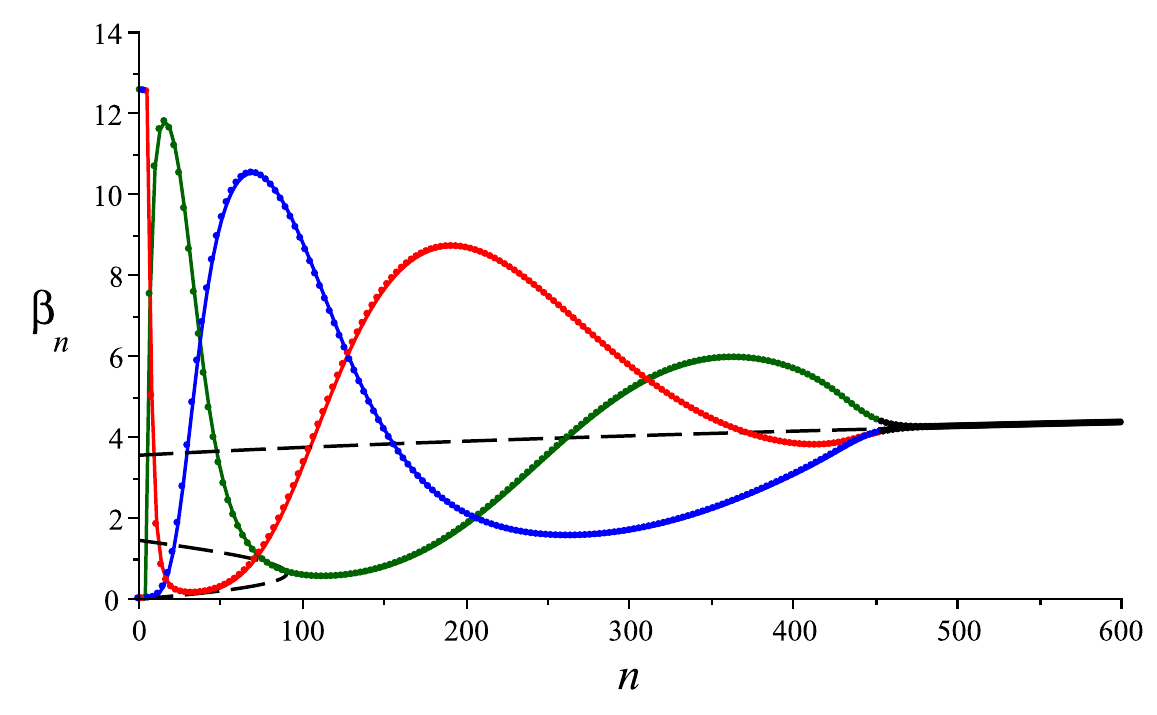} & \fig{2}{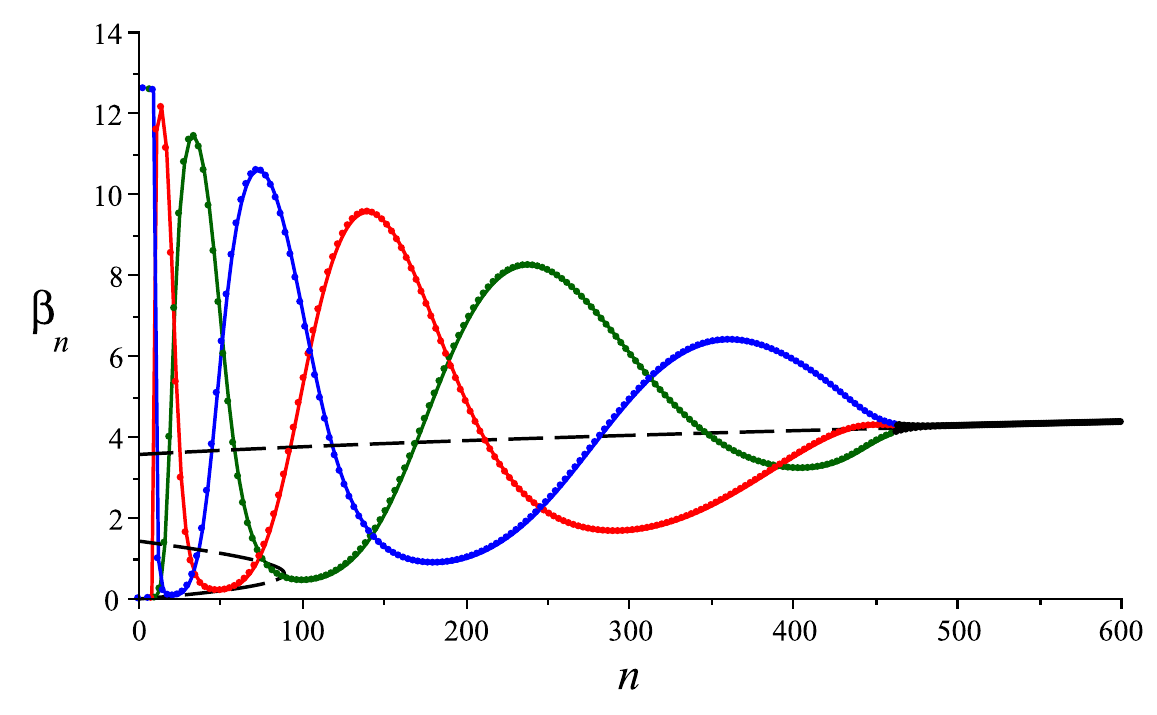}\\
\tau=25,\enskip\k=0.249 &\tau=25,\enskip\k=0.247 &\tau=25,\enskip\k=0.245
\end{array}\]
\caption{\label{Fig810}Plots of $\b_n$ for $\tau=25$, in the cases when $0.245\leq\k\leq0.255$. The recurrence coefficients $\b_{3n}$ are plotted in \blue{blue}, {$\b_{3n+1}$} in \green{green} and $\b_{3n-1}$ in \red{red}. As $|\k-\tfrac{1}{4}|$ increases the number of oscillations increases.}
\end{figure}

In Figure \ref{Fig810} plots of $\b_n$ for $\tau=25$, in the cases when $0.245\leq\k\leq0.255$. The recurrence coefficients $\b_{3n}$ are plotted in \blue{blue}, {$\b_{3n+1}$} in \green{green} and $\b_{3n-1}$ in \red{red}. As $|\k-\tfrac{1}{4}|$ increases we see that the number of oscillations increases. We note that in these plots, when $\k$ is close to $\tfrac{1}{4}$, then $\b_{3n+1}$ and $\b_{3n-1}$ essentially are interchanged as $\k$ passes through $\tfrac{1}{4}$.

\subsection{\label{caseiv}Case (iv): $\tau>0$ and $\k=0$}
In this case the weight is sextic-quartic Freud weight
\beq \w(x;\tau,0)=\exp\left(-x^6+\tau x^4\right),\label{w64}\eeq
with $\tau$ a parameter for which we obtained a closed form expressions for the moments in Lemmas \ref{qsmoment} and \ref{lem:mun13}. 

When $\k=0$ the cubic \eqref{eq:cubici} becomes
\[ 60\b^3-12\tau \b^2-n=0.\]
This is case (ii) with $\k=0$, i.e.\ $t=0$, discussed above. 
Setting $\k=0$ in \eqref{UVsys}, gives
\begin{subequations} \label{UVsys64} \begin{align}&6 u\!\left(u^{2}+6 u v +3 v^{2}\right)-4\tau u\!\left(u+2 v \right)=n,\\
&6 v\!\left(3 u^{2}+6 u v +v^{2}\right)-4\tau v\!\left(2 u +v \right)=n.
\end{align}\end{subequations}
Letting $u=\xi-\eta$ and $v=\xi+\eta$, with $\eta\geq0$, in \eqref{UVsys64} gives
\begin{align*} &144 \xi^{3}-72\tau \xi^{2}+8\tau^2\xi +3n=0, \qquad \eta^2 ={3 \xi^{2} -\tfrac{2}{3}\tau \xi}. \end{align*}
This is illustrated in Figure \ref{Fig1:case_vi} in the cases when $\tau=20$, $\tau=25$ and $\tau=30$. In Figure \ref{Fig1:case_vi2} the ``transition region" is plotted in more detail showing a five-fold structure which becomes more prominent as $\tau$ increases.

\begin{figure}[ht!]
\[\begin{array}{c@{\quad}c@{\quad}c}
\fig{2}{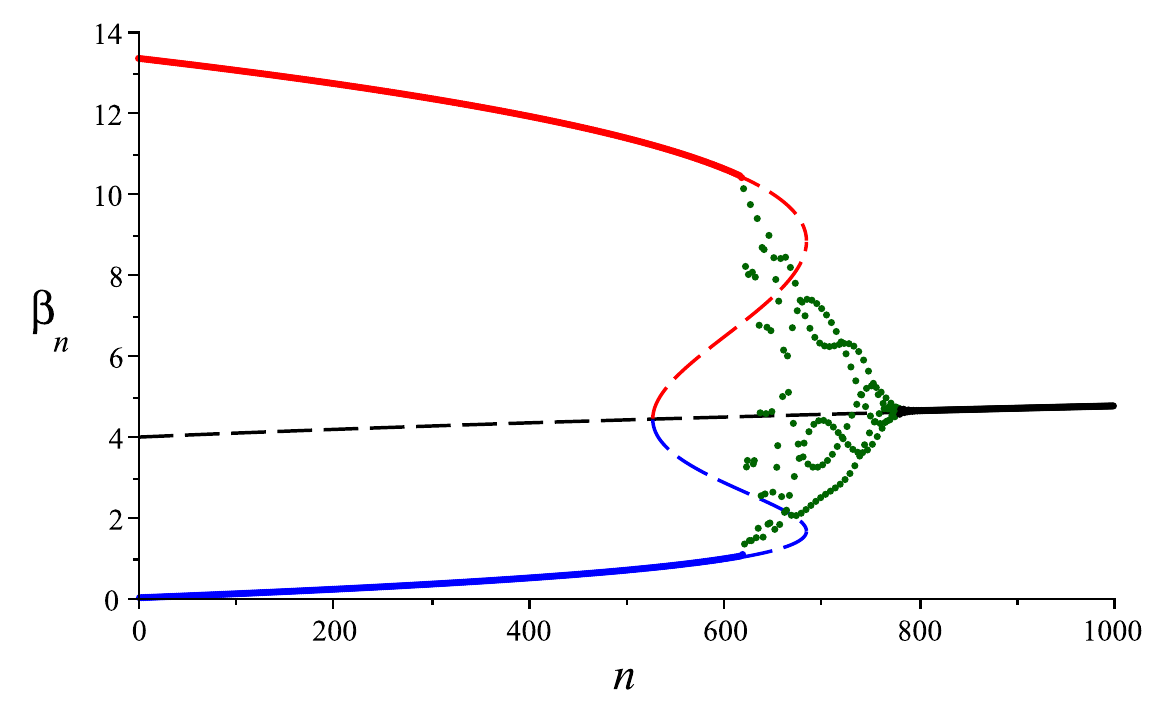} &
\fig{2}{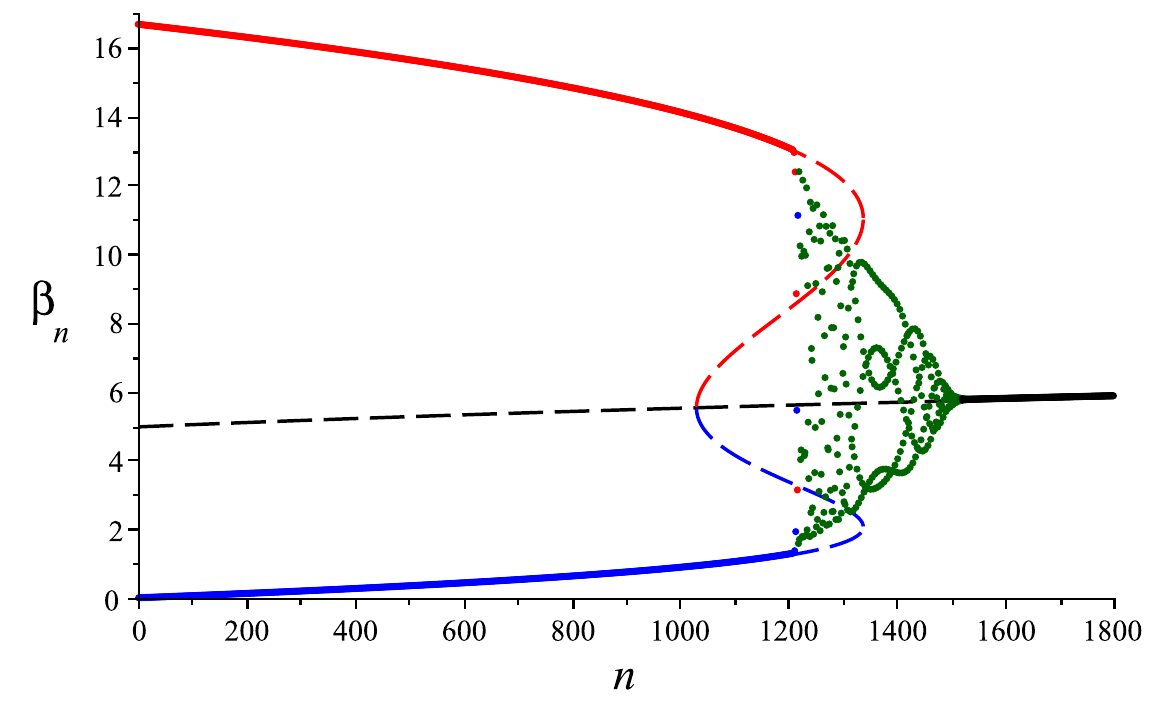} & \fig{2}{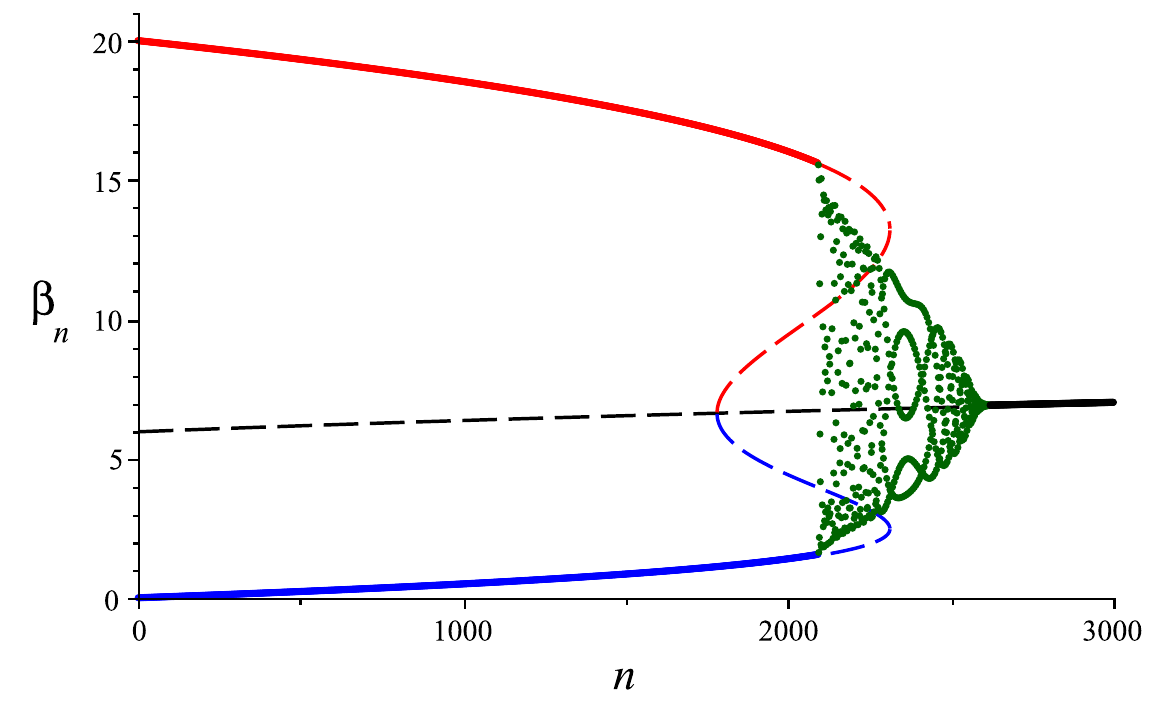}\\
\tau=20 &\tau=25 &\tau=30
\end{array}\]
\caption{\label{Fig1:case_vi}Plots of the recurrence coefficients $\b_n$ for the sextic-quartic Freud weight \eqref{w64} in the cases when 
$\tau=20$, $\tau=25$ and $\tau=30$.}
\end{figure}

\begin{figure}[ht!]
\[\begin{array}{c@{\quad}c@{\quad}c}
\fig{2}{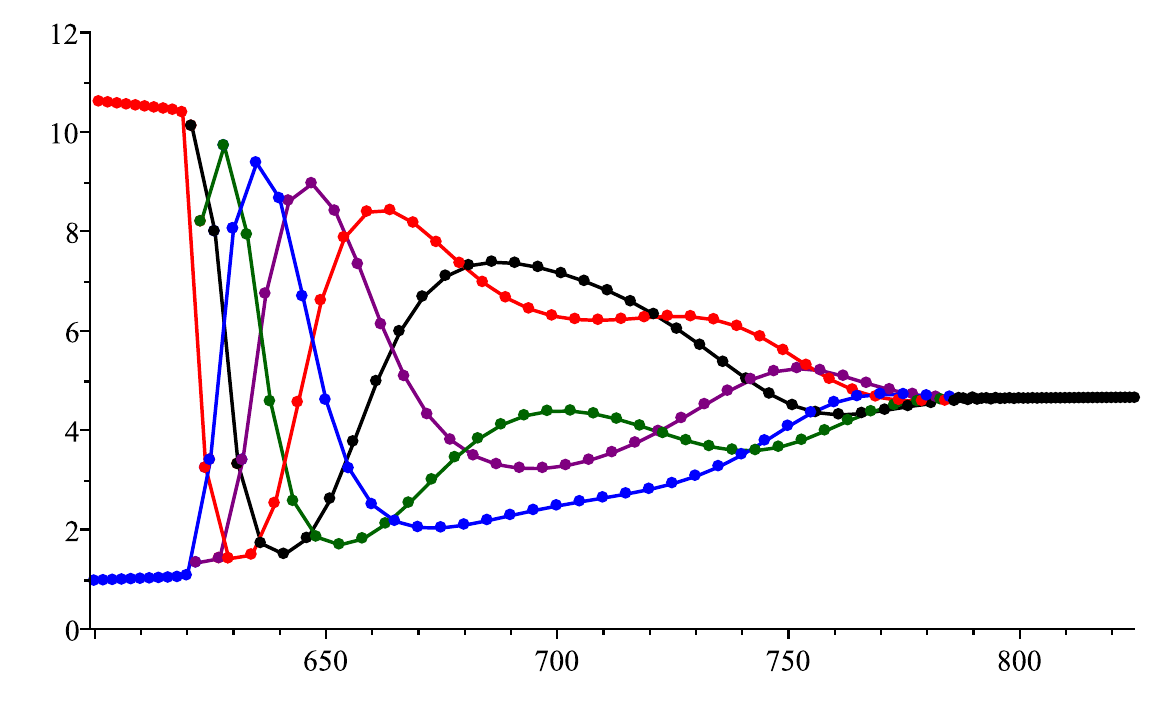} & \fig{2}{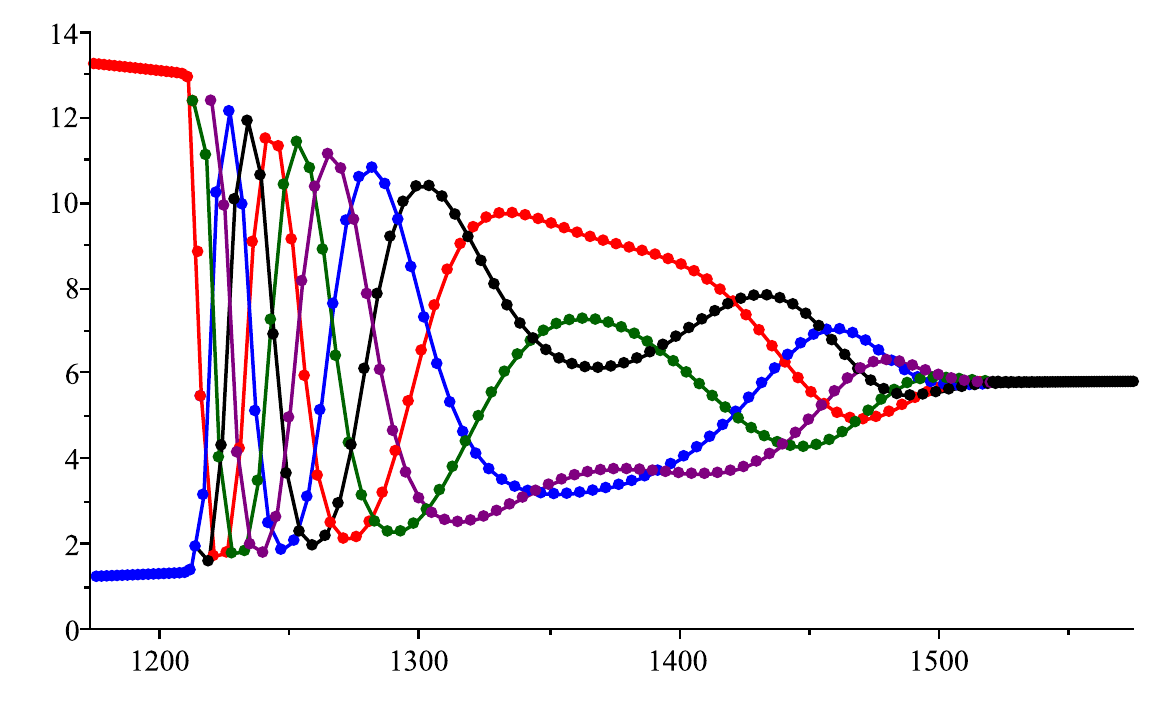} & \fig{2}{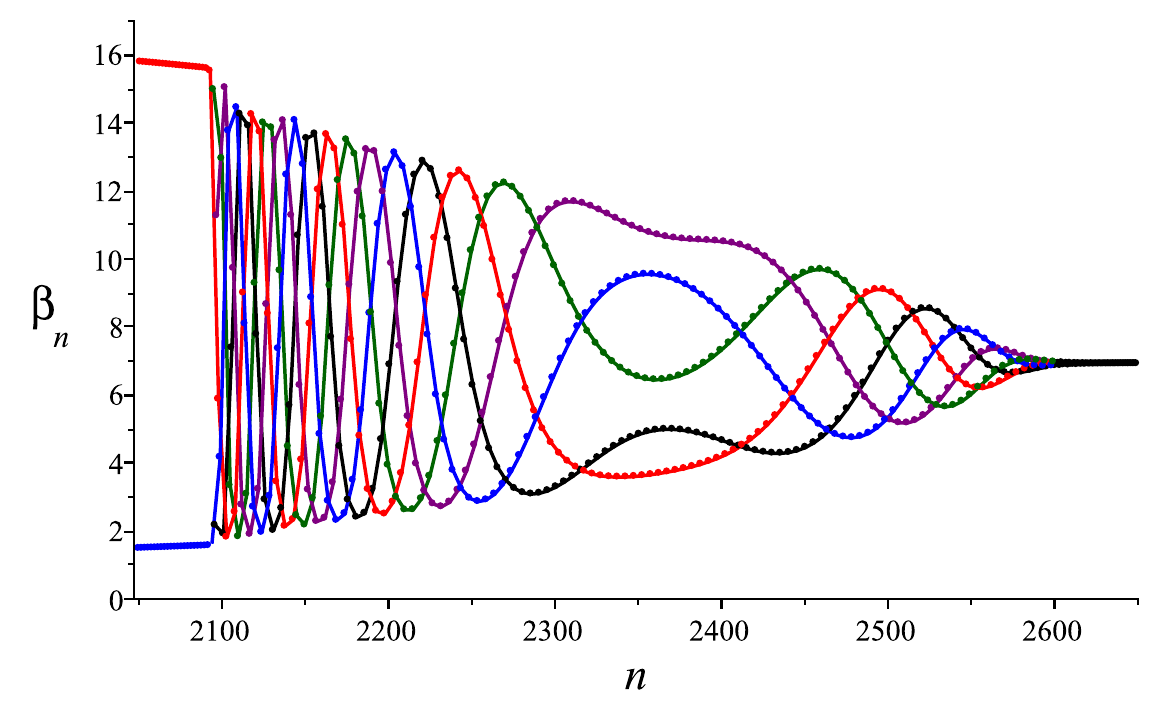}\\
\tau=20 &\tau=25 &\tau=30
\end{array}\]
\caption{\label{Fig1:case_vi2}Plots of the recurrence coefficients $\b_n$ for the sextic-quartic Freud weight \eqref{w64} in the cases when 
$\tau=20$, $\tau=25$ and $\tau=30$ showing that there is a five-fold structure in the ``transition region".}
\end{figure}

\subsection{\label{casev}Case (v): $\tau>0$ and $\k<0$} 
In this case $U(x)$ has two real roots, two purely imaginary roots and a double root at $x=0$. This case splits into two subcases: (a), when $-\tfrac{2}{3}<\k<0$; and (b), when $\k\leq-\tfrac{2}{3}$. This is due to when system \eqref{UVsys} has multivalued solutions as illustrated in Figure \ref{fig:uvplot2} where solutions of the system \eqref{UVsys} for {$\tau=10$ in the cases when} $\k=-\tfrac{1}{6}$, $\k=-\tfrac{2}{3}$ and $\k=-1$, are plotted.

\begin{enumerate}[(a)]
\item{If $-\tfrac{2}{3}<\k<0$, then as in \S\ref{caseii}, $\b_{2n}$ and $\b_{2n+1}$ follow the system \eqref{UVsys} until there is a ``transition region", which decreases in size as $\k$ decreases, then both follow the cubic \eqref{eq:cubici}.}

\item{If $\k\leq-\tfrac{2}{3}$, then there is no ``transition region", with $\b_{2n}$ and $\b_{2n+1}$ following the system \eqref{UVsys} until they switch to follow the cubic \eqref{eq:cubici}.}
\end{enumerate}

This is illustrated in Figure \ref{Fig715} where plots of the recurrence coefficients for $\tau=15$ and various $\k<0$ are given. We remark that in the cases when $\k=-\tfrac{2}{3}$ and $\k=-1$ there is no ``transition region".

\begin{figure}[ht!]
\[\begin{array}{c@{\quad}c@{\quad}c}
\fig{2}{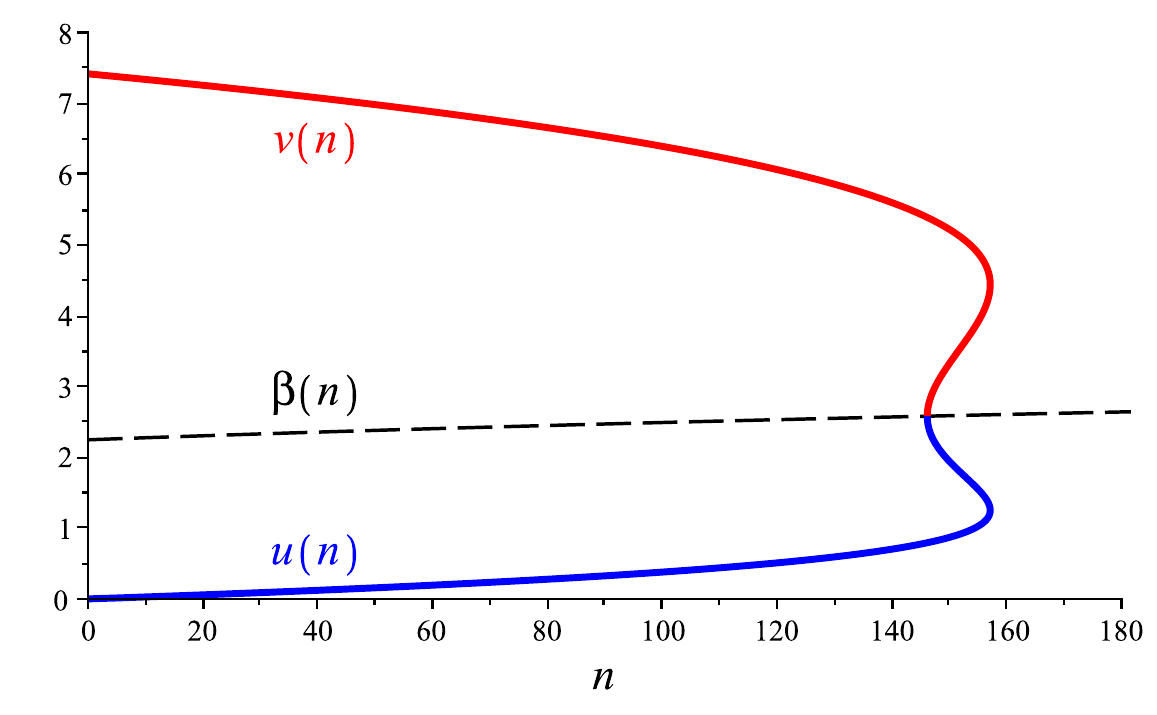} &\fig{2}{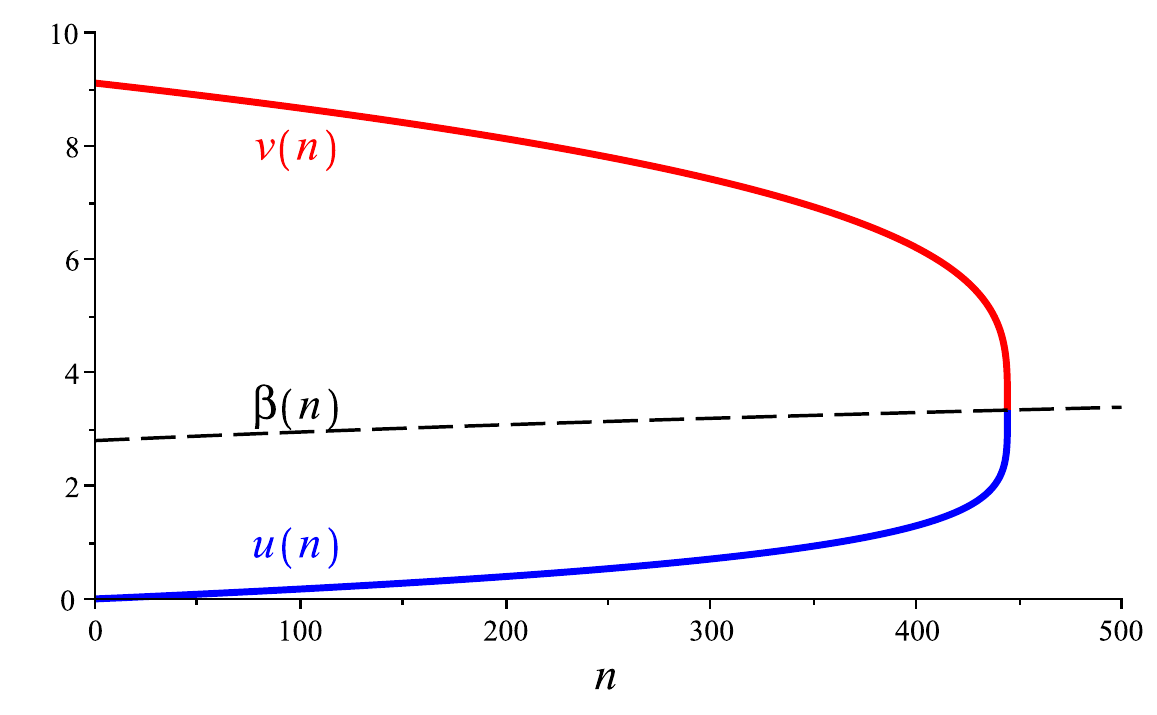} &\fig{2}{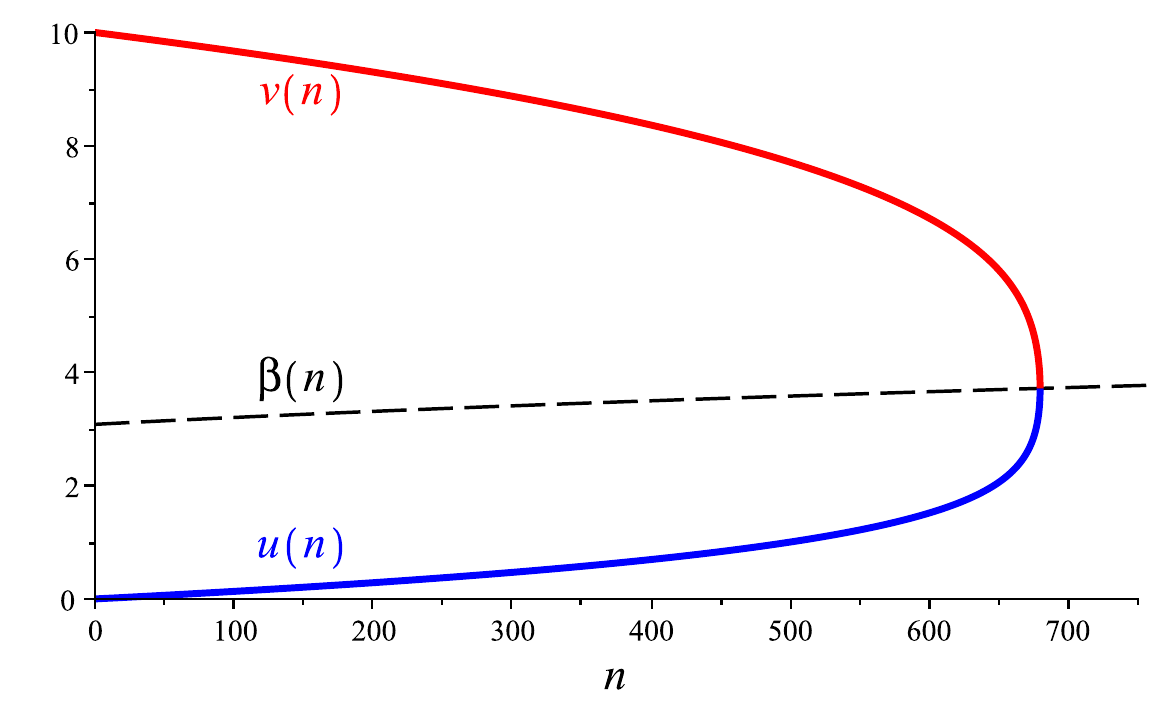}\\
{\tau=10},\enskip\k=-\tfrac{1}{6} & {\tau=10},\enskip\k=-\tfrac{2}{3} & {\tau=10},\enskip\k=-1 
\end{array}\]
\caption{\label{fig:uvplot2}Plots of solutions of the system \eqref{UVsys} for {$\tau=10$ in the cases} $\k=-\tfrac{1}{6}$, $\k=-\tfrac{2}{3}$ and $\k=-1$, with $u(n)$ plotted in \blue{blue} and $v(n)$ in \red{red}. The dashed line is the real solution of the cubic \eqref{eq:cubici}.}
\end{figure}
\begin{figure}[ht!]
\[\begin{array}{c@{\quad}c@{\quad}c}
\fig{2}{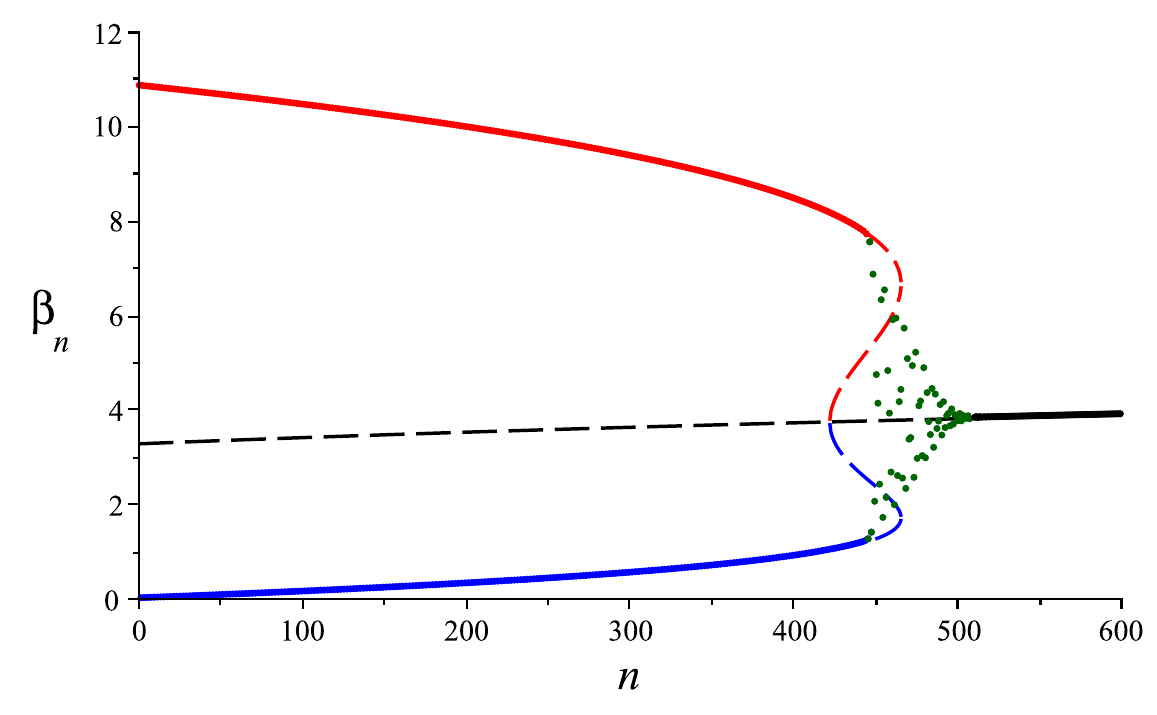} & \fig{2}{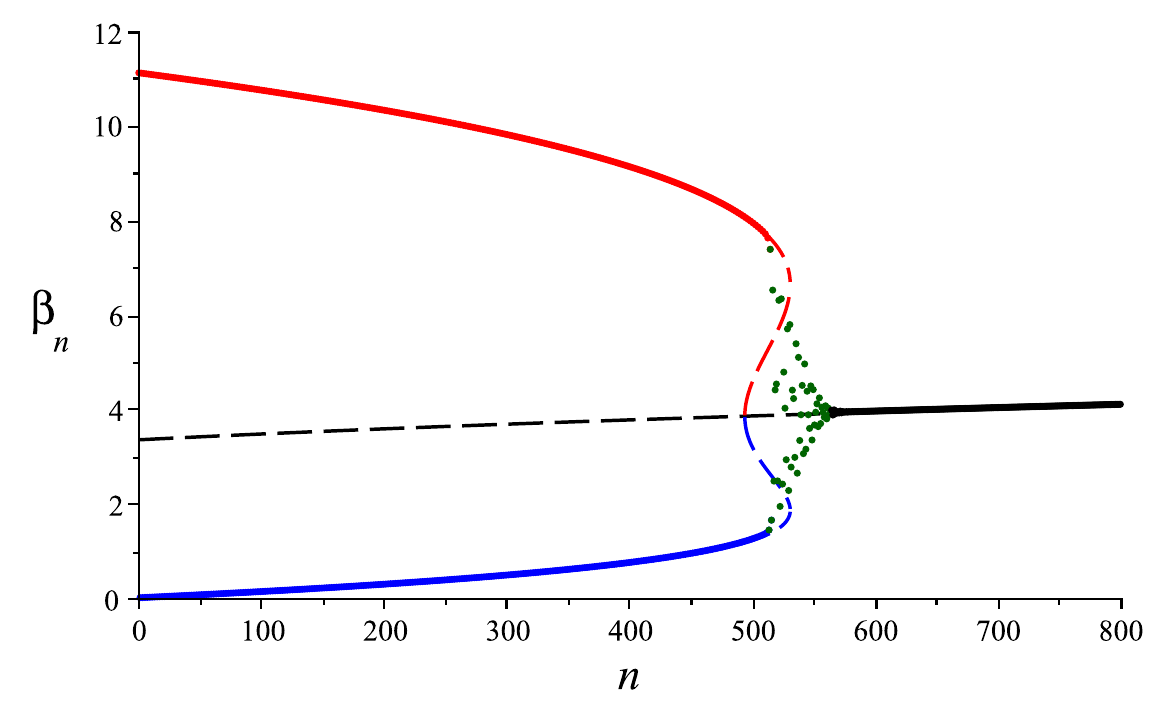}& \fig{2}{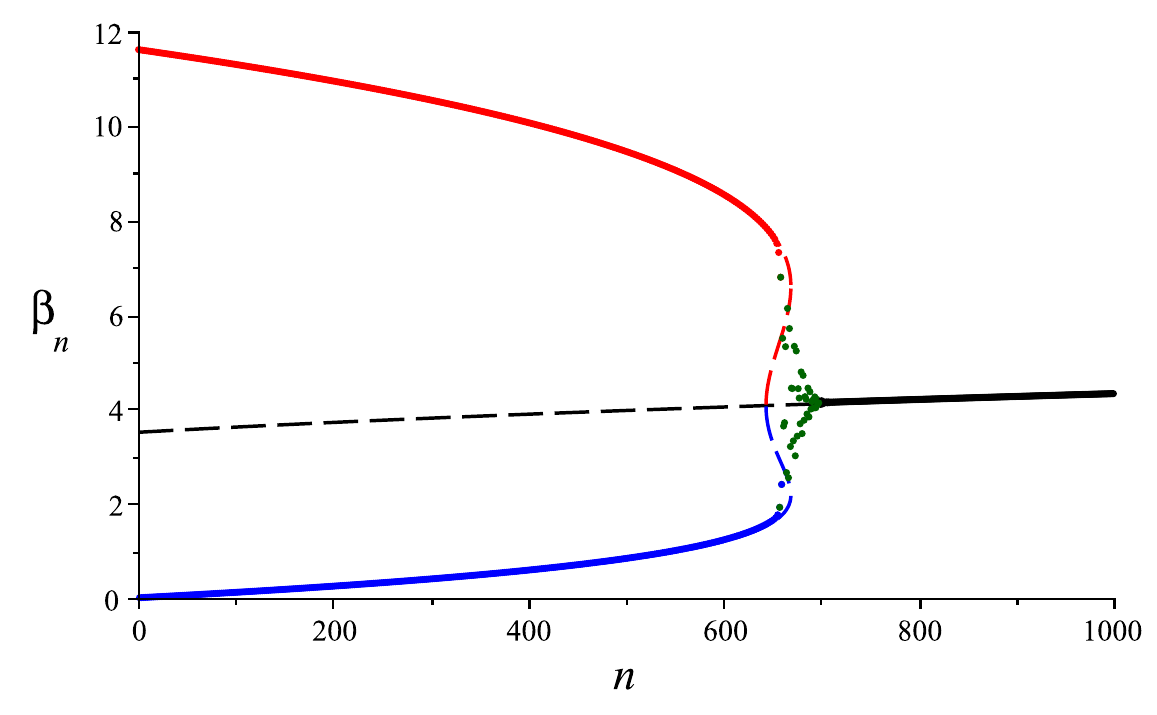}\\
\tau=15,\enskip\k=-\tfrac{1}{8} & \tau=15,\enskip\k=-\tfrac{1}{6}& \tau=15,\enskip\k=-\tfrac{1}{4}\\
 \fig{2}{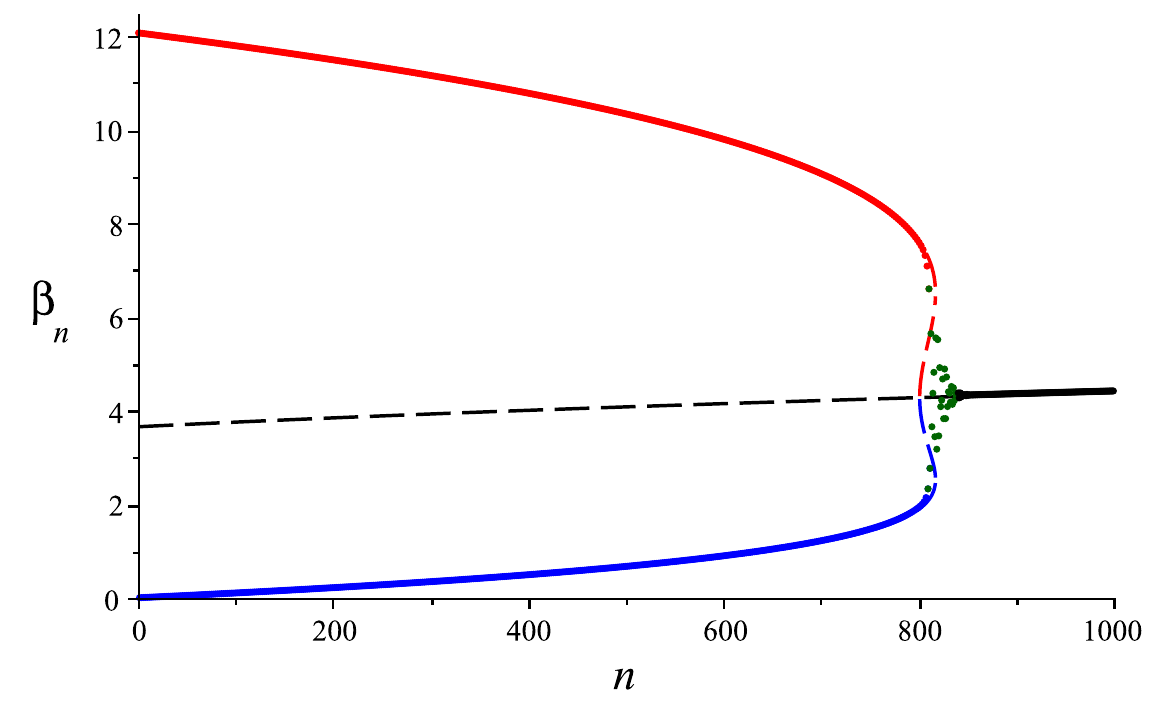}& \fig{2}{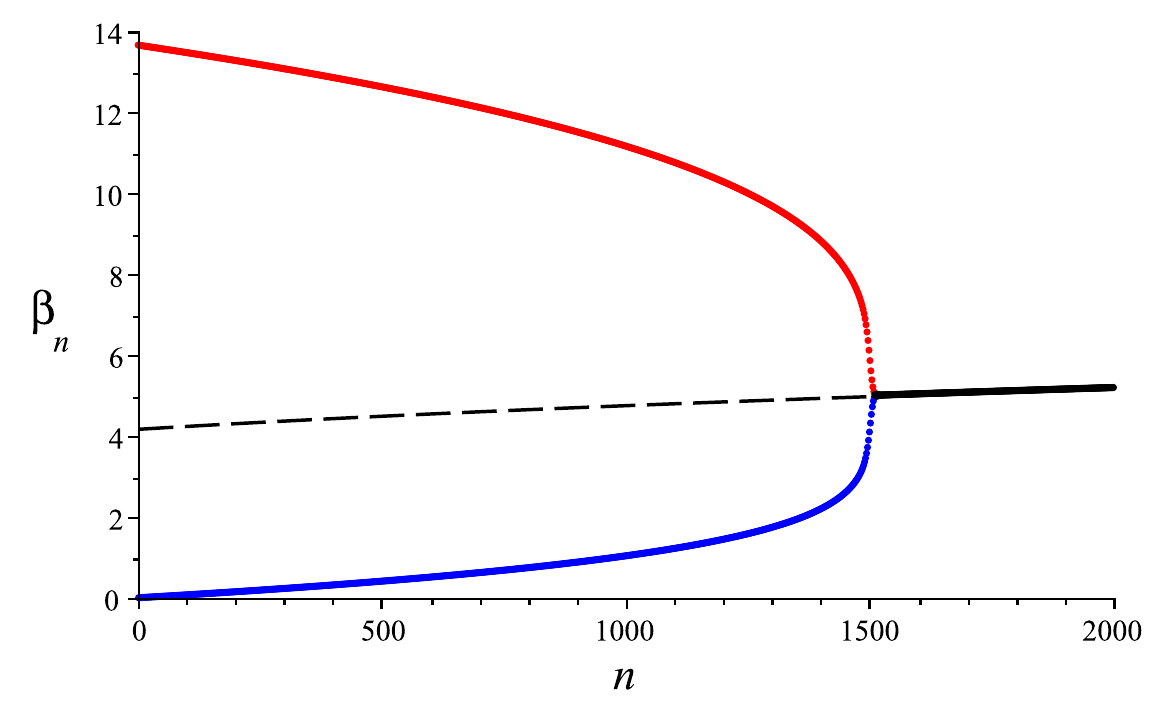}& \fig{2}{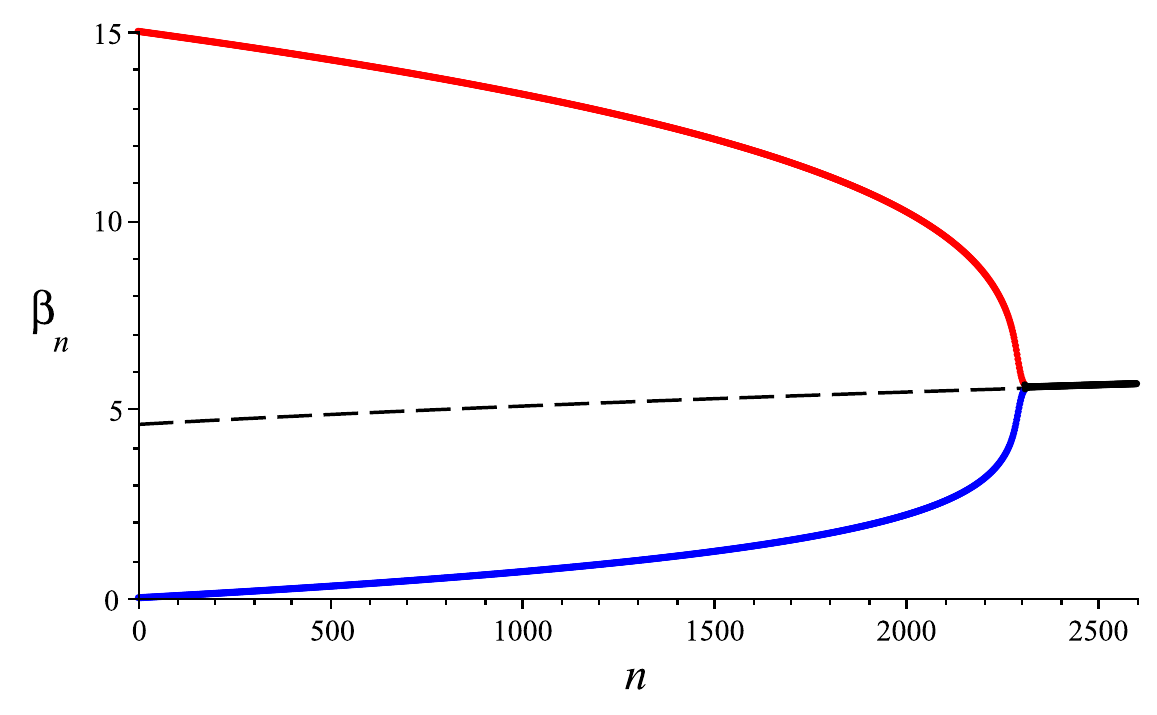}\\
\tau=15,\enskip\k=-\tfrac{1}{3}& \tau=15,\enskip\k=-\tfrac{2}{3}& \tau=15,\enskip\k=-1
\end{array}\]
\caption{\label{Fig715}Plots of the recurrence coefficients for $\tau=15$ and various $\k<0$. Note that for $\k=-\tfrac{2}{3}$ and $\k=-1$ there is no ``transition region".}
\end{figure}

\subsection{\label{casevi}Case (vi): $\tau=0$ and $t\not=0$}
In this case the weight is quadratic-sextic Freud weight
\beq \w(x;0,t)=\exp\left(-x^6+tx^2\right),\label{Freud62}\eeq
with $t$ a parameter, which is a special case of the generalised sextic Freud weight discussed in \cite{refCJ21b}. For the weight \eqref{Freud62} we derived a closed form expression for the first moment in Lemma \ref{lemma:41}.

When $\tau=0$, the cubic \eqref{eq:cubic} becomes
\beq 60\b^3-2t \b-n=0.\label{eq:cubic62}\eeq
There are two scenarios for the recurrence coefficients, (a), $t>0$ and (b), $t<0$. 

\begin{enumerate}[(a)]\item 
When $t>0$, we set $\b_{2n}=u$, $\b_{2n+1}=v$ and $\tau=0$ in \eqref{eq:rr642}, giving the system
\begin{subequations}\label{UVsys62}
\begin{align} &6 u\!\left(u^{2}+6 u v +3 v^{2}\right)-2tu=n,\\ &6 v\!\left(3 u^{2}+6 u v +v^{2}\right)-2t v =n. 
\end{align}\end{subequations}
Then letting $u=\xi-\eta$ and $v=\xi+\eta$, with $\eta\geq0$, gives
\[48 \xi^{3}-4 t \xi + n =0, \qquad\eta^2 ={3 \xi^{2} - \tfrac{1}{6} t}.\]
The three functions $\b(n)$, $u(n)$ and $v(n)$ meet at the point
$\left(\tfrac{1}{9}(2t)^{3/2},\tfrac{1}{6}(2t)^{1/2}\right)$.
In Figure \ref{Fig716} the even recurrence coefficients $\b_{2n}$ are plotted in \blue{blue} and the odd recurrence coefficients $\b_{2n+1}$ are plotted in \red{red}, together with the real solutions of \eqref{UVsys62} and the real solution of the cubic \eqref{eq:cubic62}, when $t=30$, $t=40$ and $t=50$. These show that initially $\b_{2n}$ follow a curve approximated by $u(n)$, $\b_{2n+1}$ follow a curve approximated by $v(n)$ and for $n$ sufficiently large, all recurrence coefficients $\b_{n}$ follow a curve approximated by $\b(n)$.

\begin{figure}[ht!]
\[\begin{array}{c@{\quad}c@{\quad}c}
\fig{2}{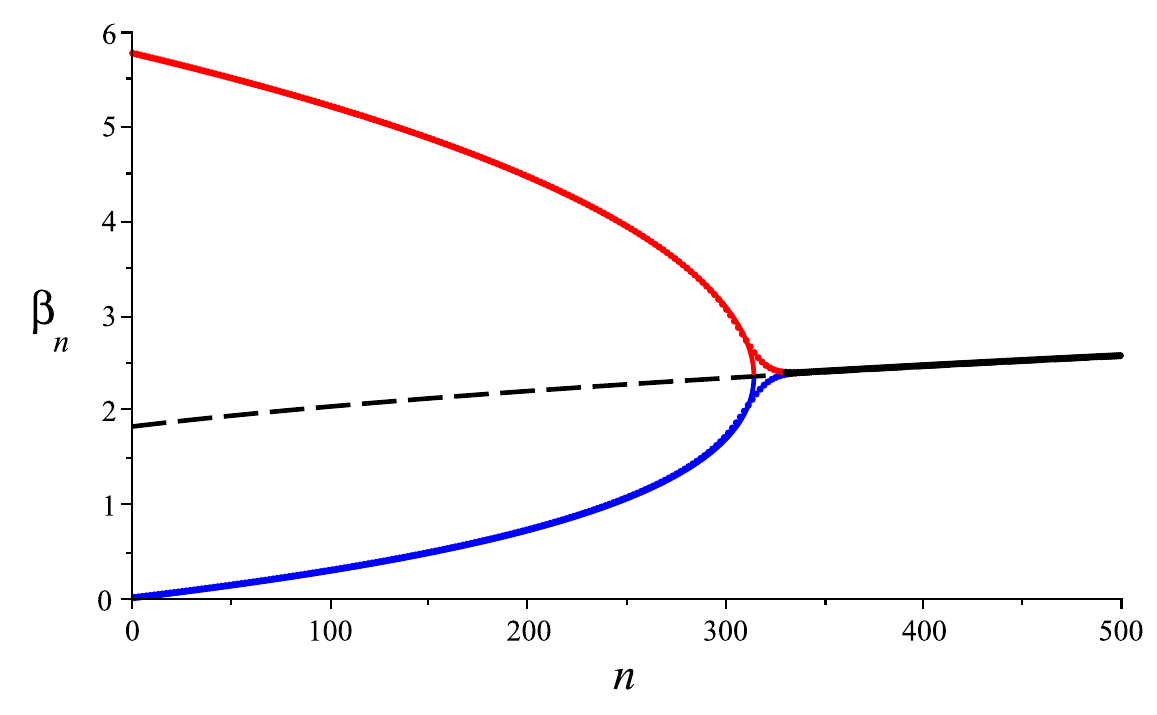}&\fig{2}{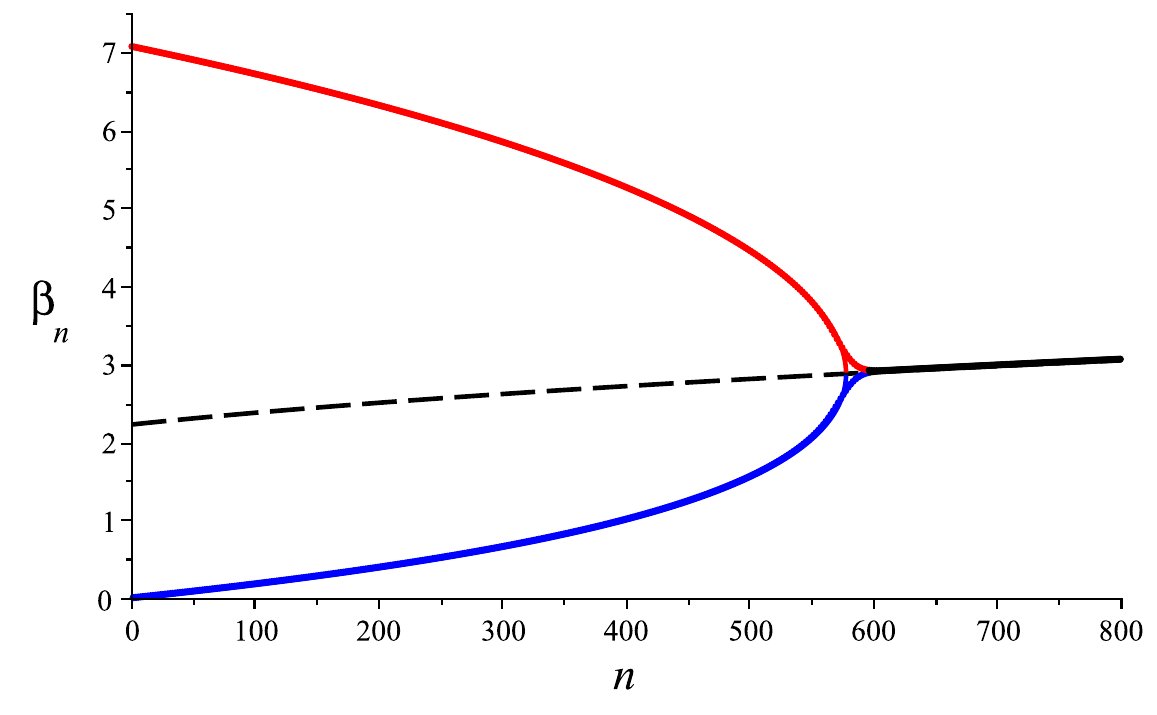}& \fig{2}{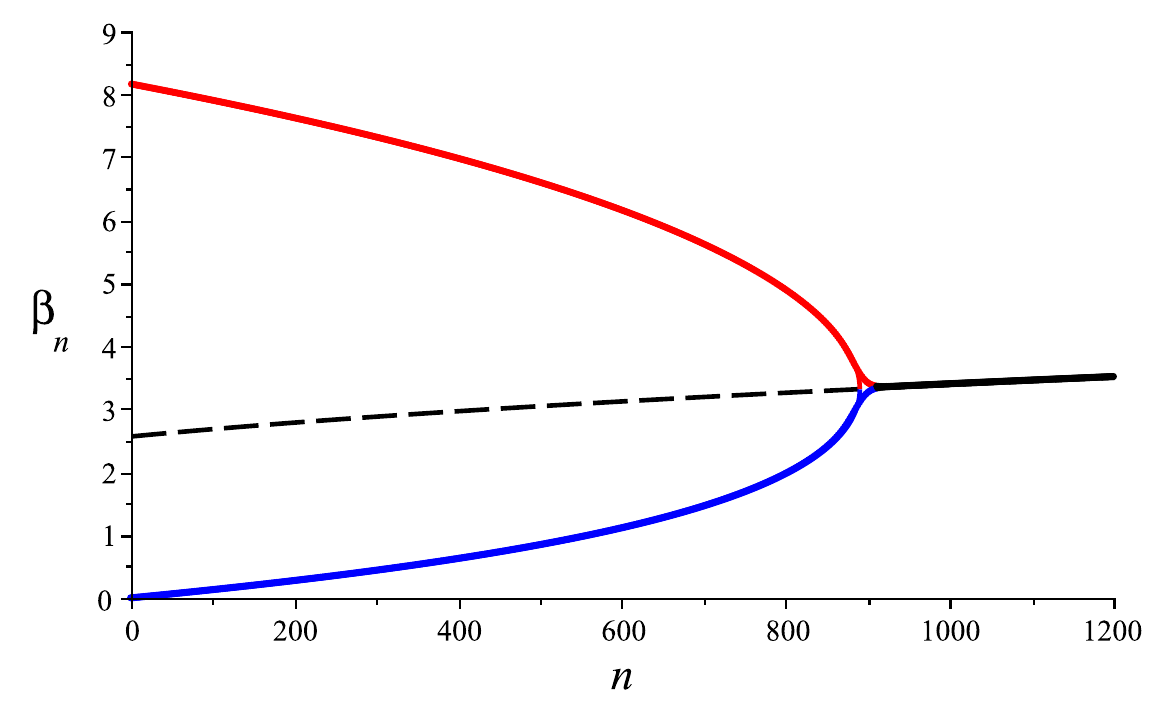}\\
\tau=0,\enskip t=100 &\tau=0,\enskip t=150 &\tau=0,\enskip t=200 
\end{array} \]
\caption{\label{Fig716}Plots of the recurrence coefficients $\b_{n}$ for the quadratic-sextic Freud weight \eqref{Freud62}, in the cases when 
$t=100$, $t=150$ and $t=200$, together with the curves \eqref{UVsys62} and the cubic \eqref{eq:cubic62}. The recurrence coefficients $\b_{2n}$ are plotted in \blue{blue} and $\b_{2n+1}$ in \red{red}.}
\end{figure}
\item When $t<0$, then the recurrence coefficients $\b_n$ increase monotonically {and closely follow the real solution of the cubic \eqref{eq:cubic62}}. 
\end{enumerate}

\subsection{\label{casevii}Case (vii): $\tau<0$}
This case is similar to the previous case when $\tau=0$, except there are no closed form expressions for the moments.
There are two scenarios for the recurrence coefficients, (a), $\k<0$ (i.e.\ $t>0$) and (b), $\k>0$ (i.e.\ $t<0$). 
\begin{enumerate}[(a)]\item 
{When $\k<0$, 
$\b_{2n}$ and $\b_{2n+1}$ follow the system \eqref{UVsys} until they switch to follow the cubic \eqref{eq:cubici}.
This is illustrated in Figure \ref{Fig718}. }
\item When $\k\geq 0$, then the recurrence coefficients $\b_n$ increase monotonically and closely follow the real solution of the cubic \eqref{eq:cubici}.
\end{enumerate}
\begin{figure}[ht!]
\[\begin{array}{c@{\quad}c@{\quad}c}
\fig{2}{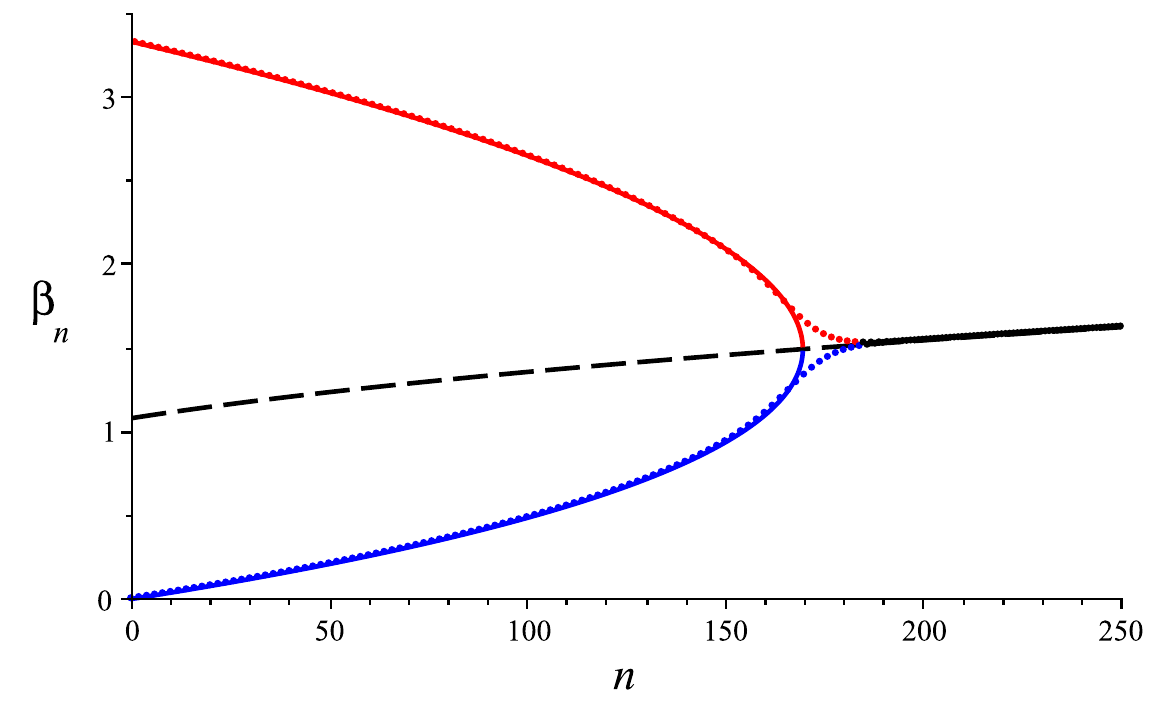} & \fig{2}{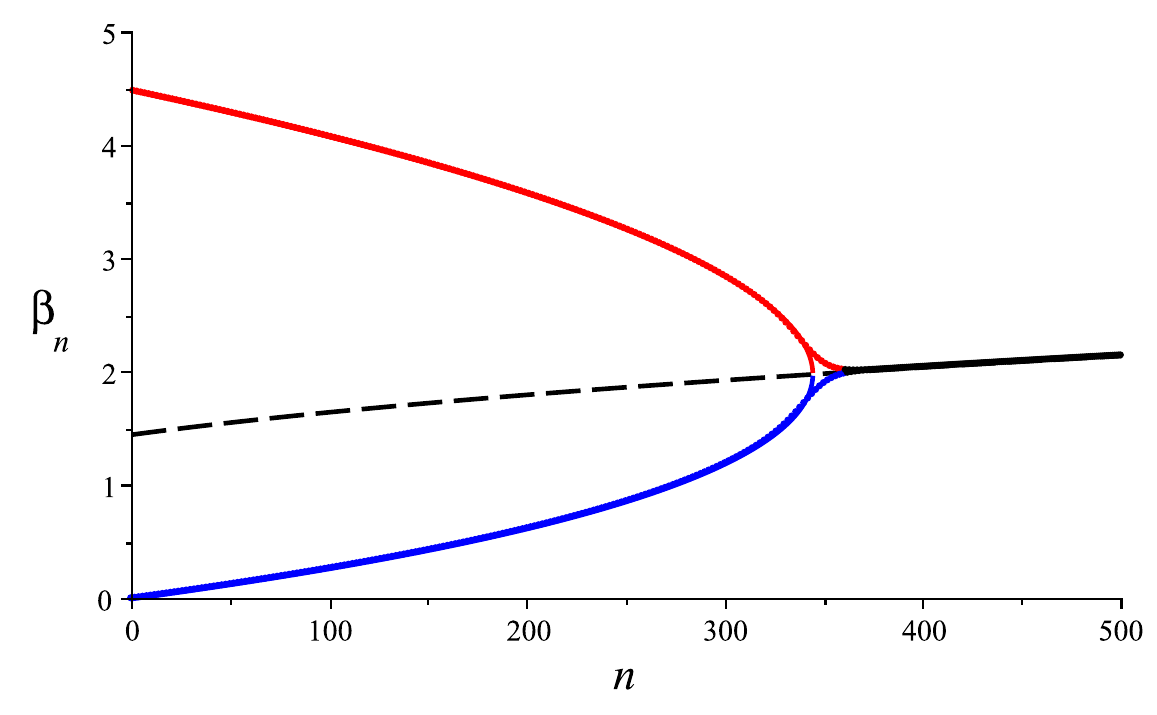} & \fig{2}{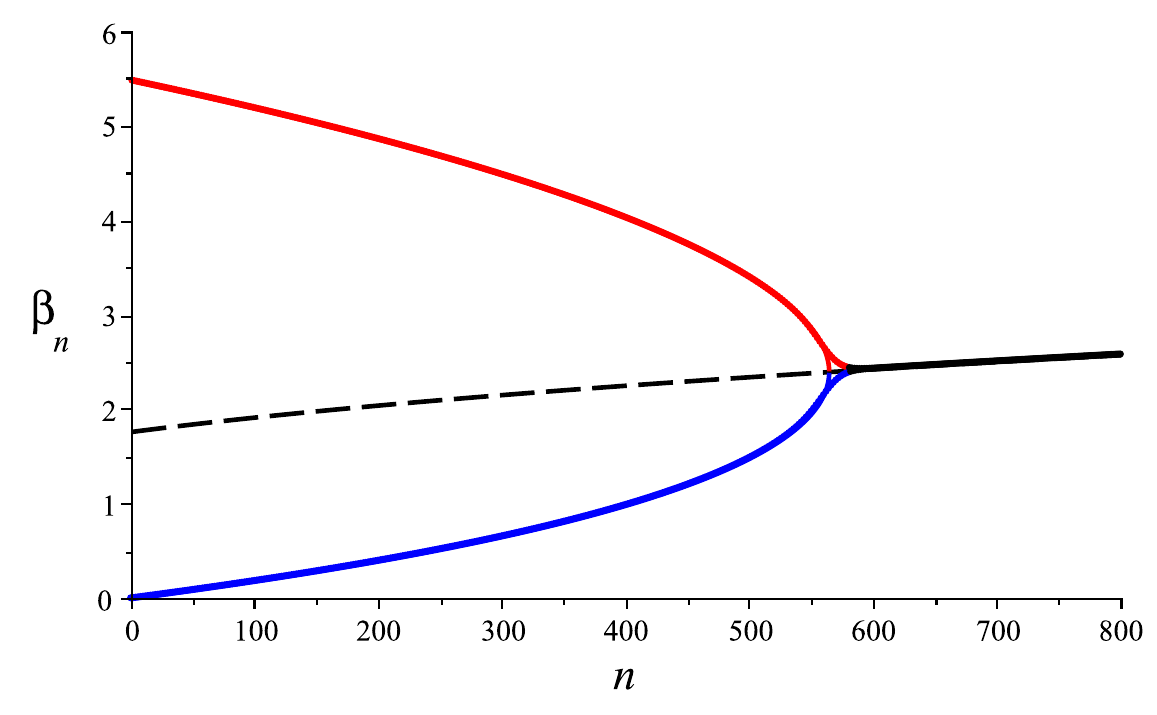}\\
\tau=-10,\enskip \k=-1 &\tau=-10,\enskip \k=-\tfrac32 &\tau=-10,\enskip \k=-2 
\end{array}\]
\caption{\label{Fig718}
Plots of the recurrence coefficients $\b_{n}$ for 
$\tau=-10$ in the cases when 
$\k=-1$, $\k=-\tfrac32$ and $\k=-2$, together with the curves \eqref{UVsys} and the cubic \eqref{eq:cubici}. The recurrence coefficients $\b_{2n}$ are plotted in \blue{blue} and $\b_{2n+1}$ in \red{red} and $\b(n)$ the dashed black line.}
\end{figure}

\subsection{\label{caseviii}Case (viii): $\tau=0$ and $t=0$}
In this case the weight is $\w(x;0,0)=\exp(-x^6)$, so the moments are given by
\[ \mu_{2k}=\imp x^{2k}\exp(-x^6)\,\dx = \tfrac{1}{3}\Gamma(\tfrac{1}{3}k+\tfrac{1}{6}),\qquad \mu_{2k+1}=0,\]
and hence the recurrence coefficients $\b_n$ are expressible in terms of Gamma functions.

\subsection{\label{subsec711}Large values of $\tau$}
In this subsection we illustrate the effect of increasing the value of $\tau$, keeping $\k$ fixed. The numerical computations suggest that the regions of quasi-periodicity are more prominent as $\tau$ increases.
In Figure \ref{Fig817} $\b_n$ is plotted for $\tau=40$, $45$, $50$, in the cases when $\k=0.275$, $\k=0.3$, $\k=0.335$ and $\k=0.365$.
In Figure \ref{Fig818} $\b_n$ is plotted for $\tau=30$, $35$, $40$, in the cases when $\k=\tfrac{1}{5}$, $\k=\tfrac{1}{6}$ and $\k=\tfrac{1}{8}$.
In Figure {\ref{Fig819}} $\b_n$ is plotted for $\tau=40$, $45$, $50$, in the cases when $\k=0.249$ and $\k=0.251$. The recurrence coefficients $\b_{3n}$ are plotted in \blue{blue}, {$\b_{3n+1}$} in \green{green} and $\b_{3n-1}$ in \red{red}. As $\tau$ increases we see that the number of oscillations increases and also that $\b_{3n+1}$ and $\b_{3n-1}$ interchange between the two cases.
\begin{figure}[ht!]
\[\begin{array}{c@{\quad}c@{\quad}c}
\fig{2}{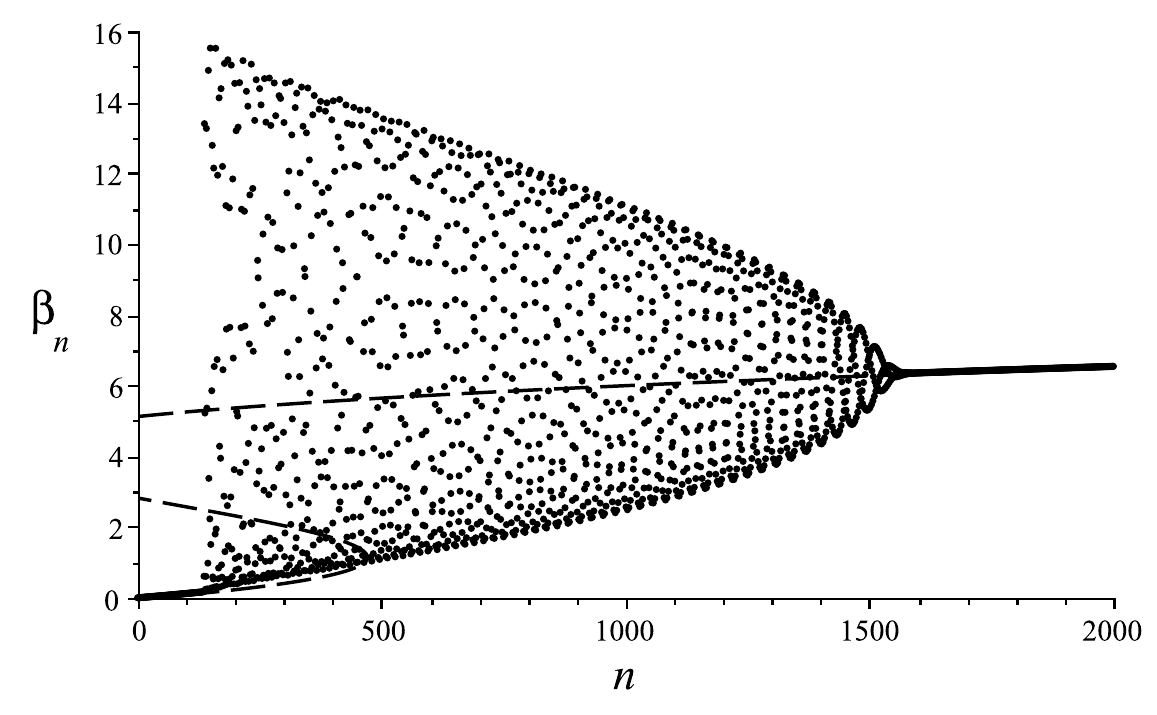}& \fig{2}{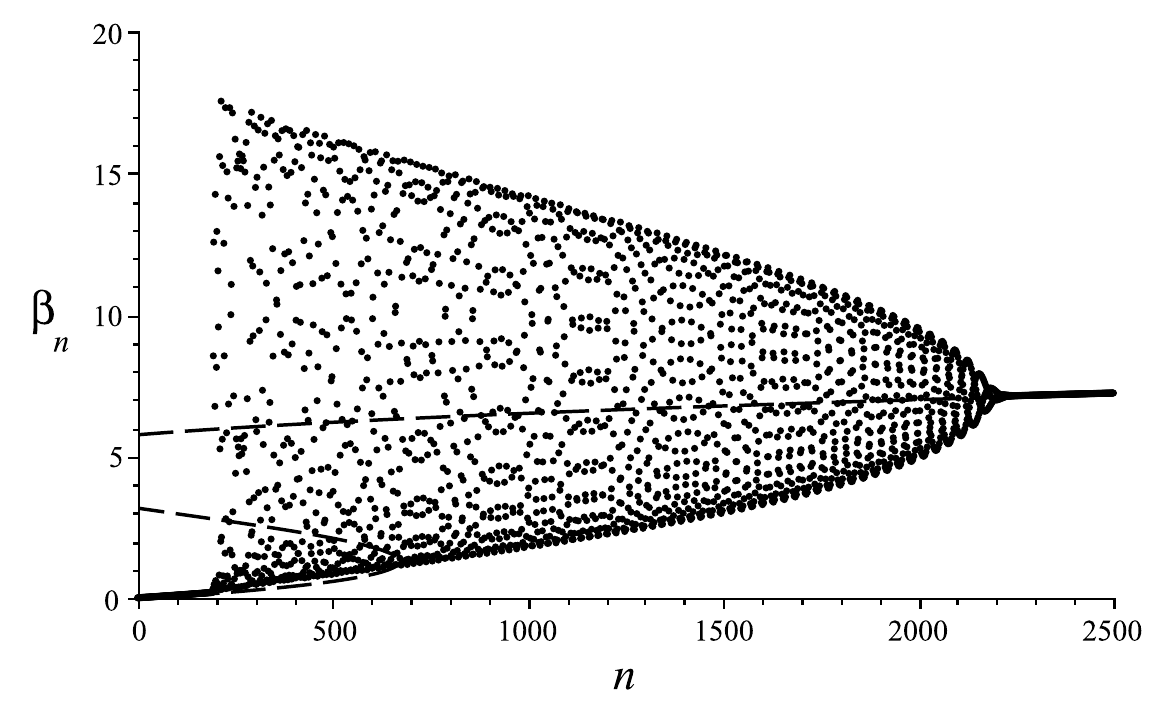} & \fig{2}{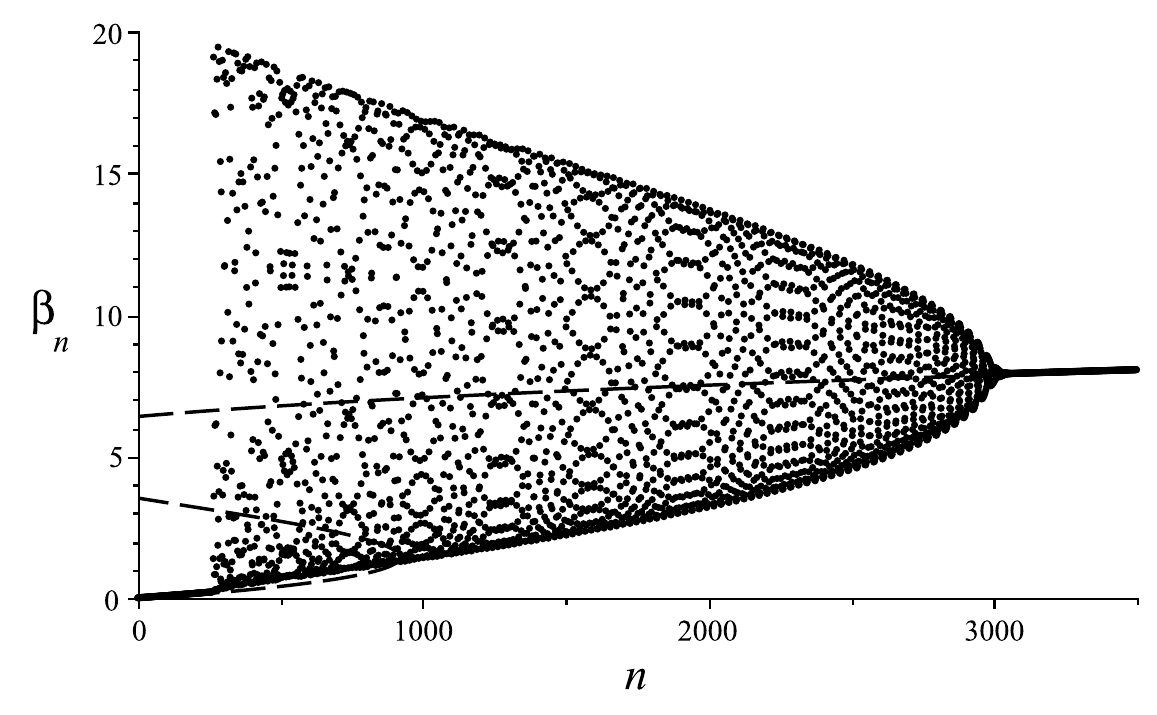}\\
\tau=40,\enskip \k=0.275 &\tau=45,\enskip \k=0.275 &\tau=50,\enskip \k=0.275\\[5pt]
\fig{2}{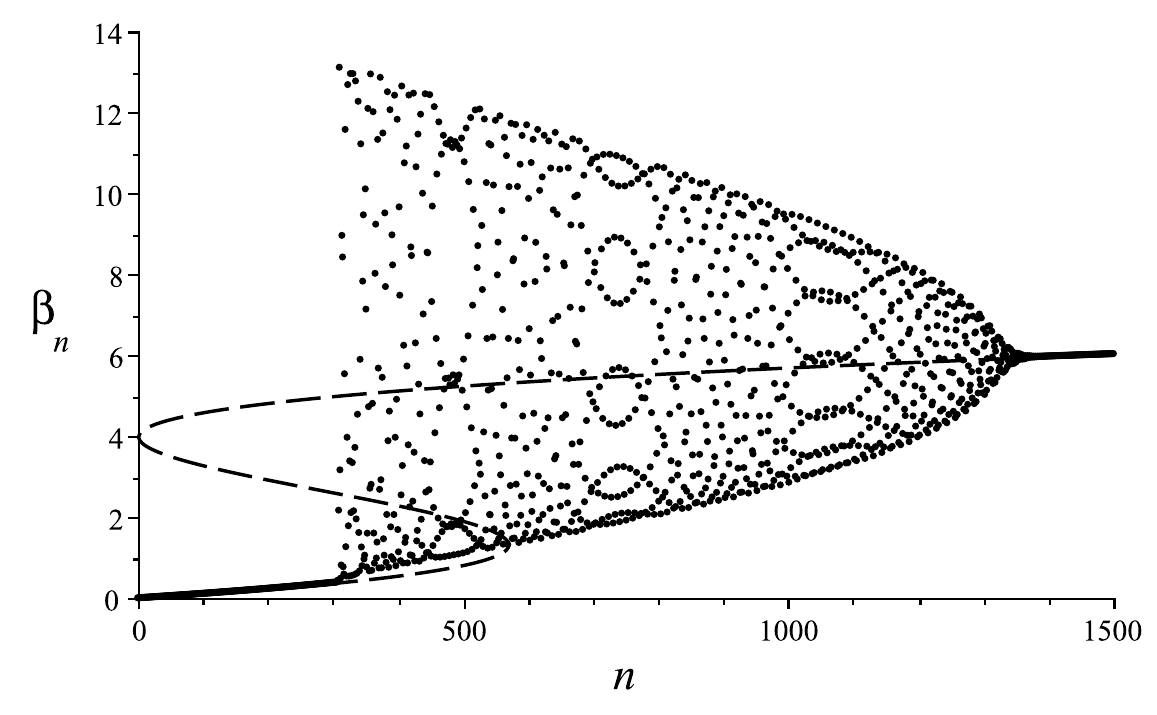}& \fig{2}{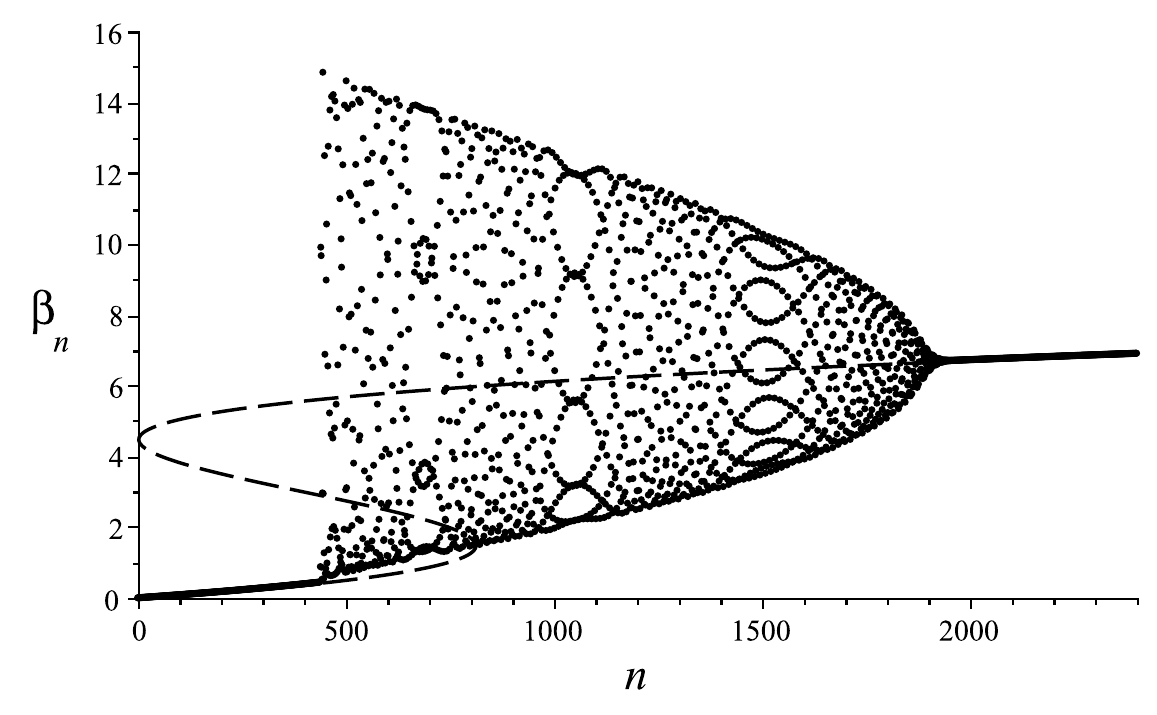} & \fig{2}{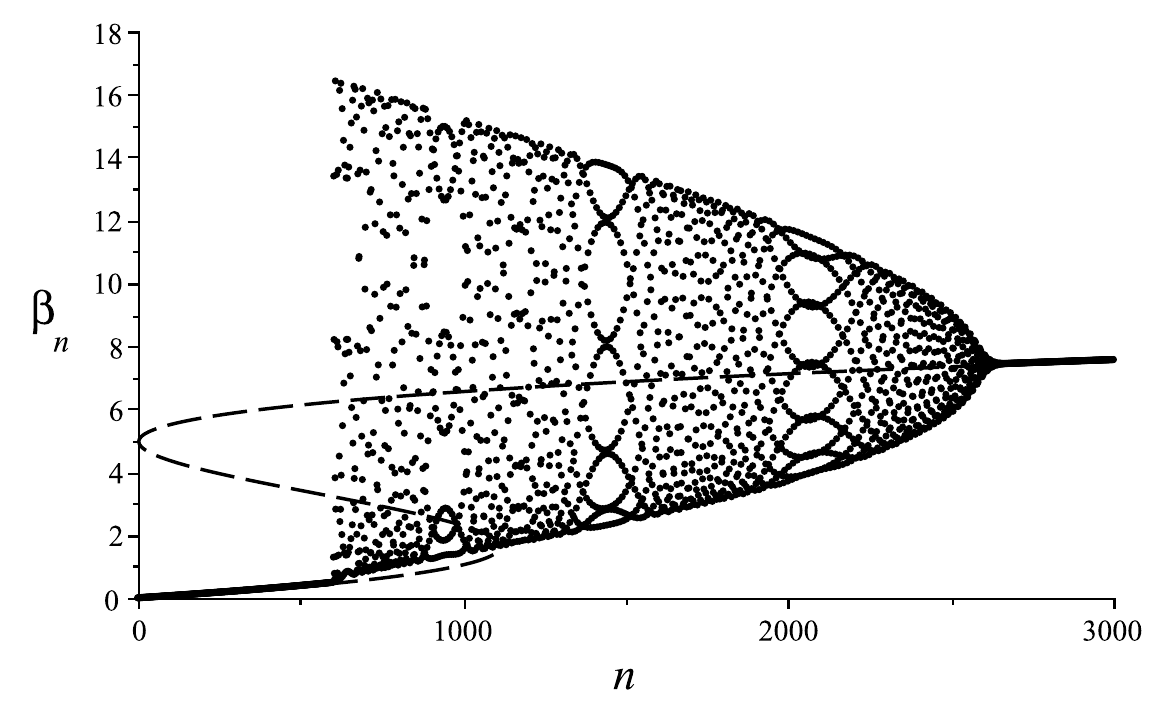}\\
\tau=40,\enskip \k=0.3 &\tau=45,\enskip \k=0.3 &\tau=50,\enskip \k=0.3\\[5pt]
\fig{2}{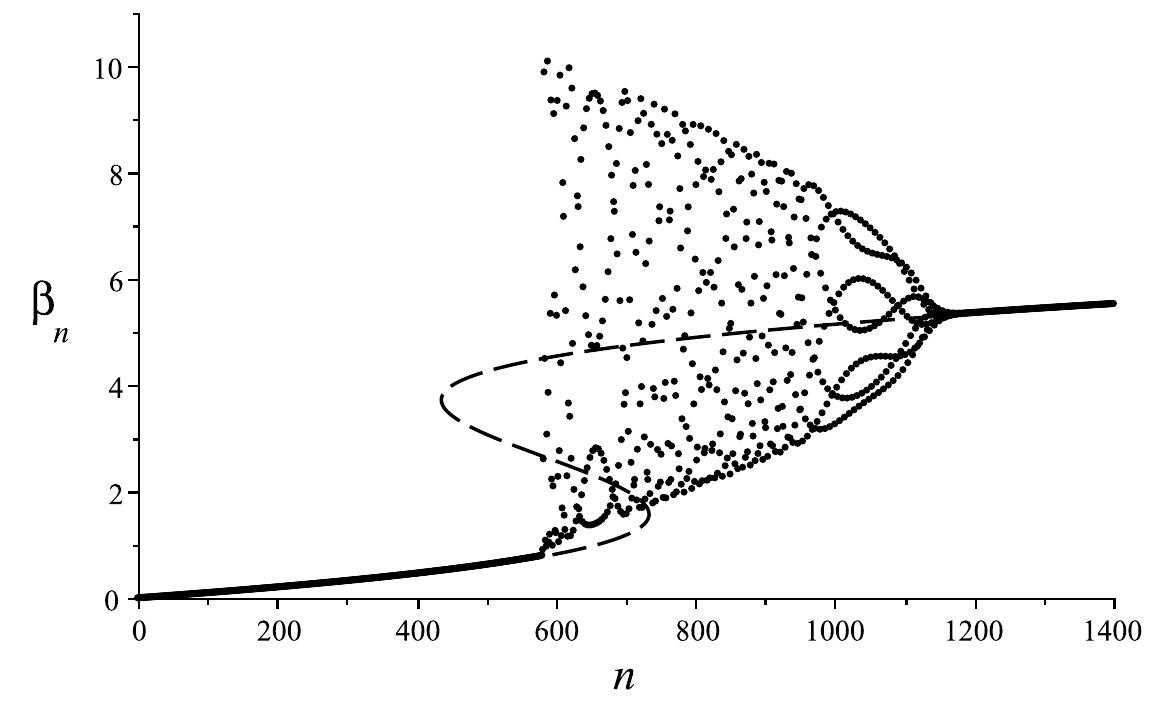}& \fig{2}{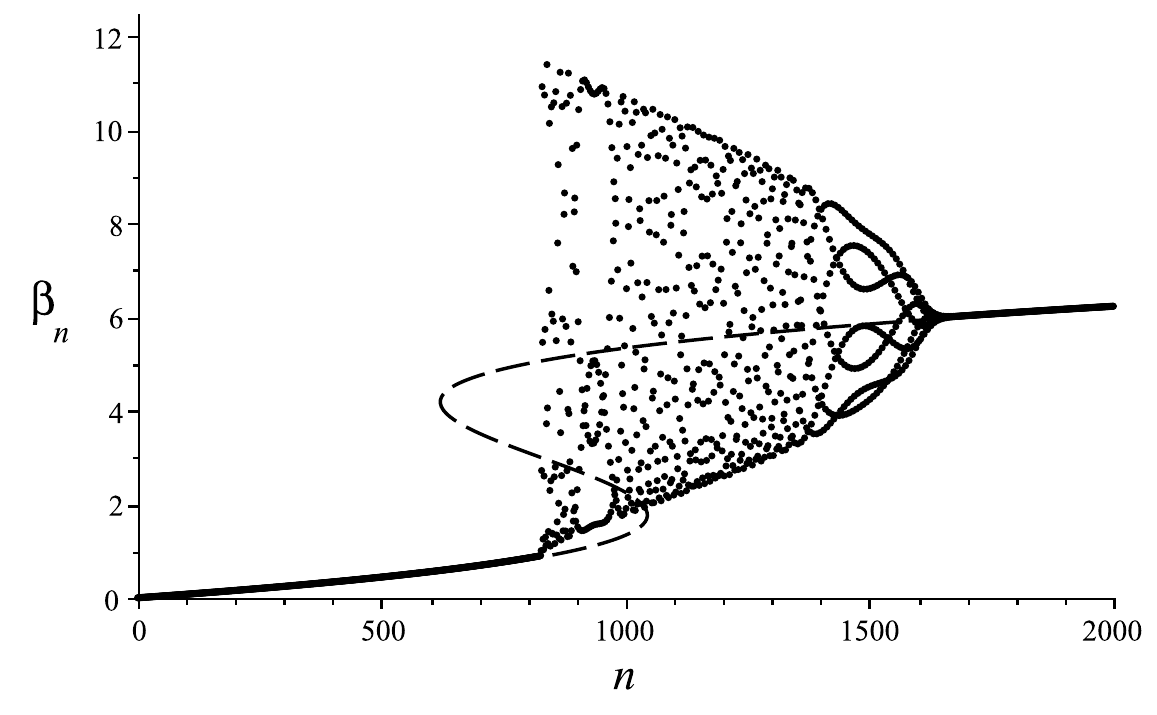} & \fig{2}{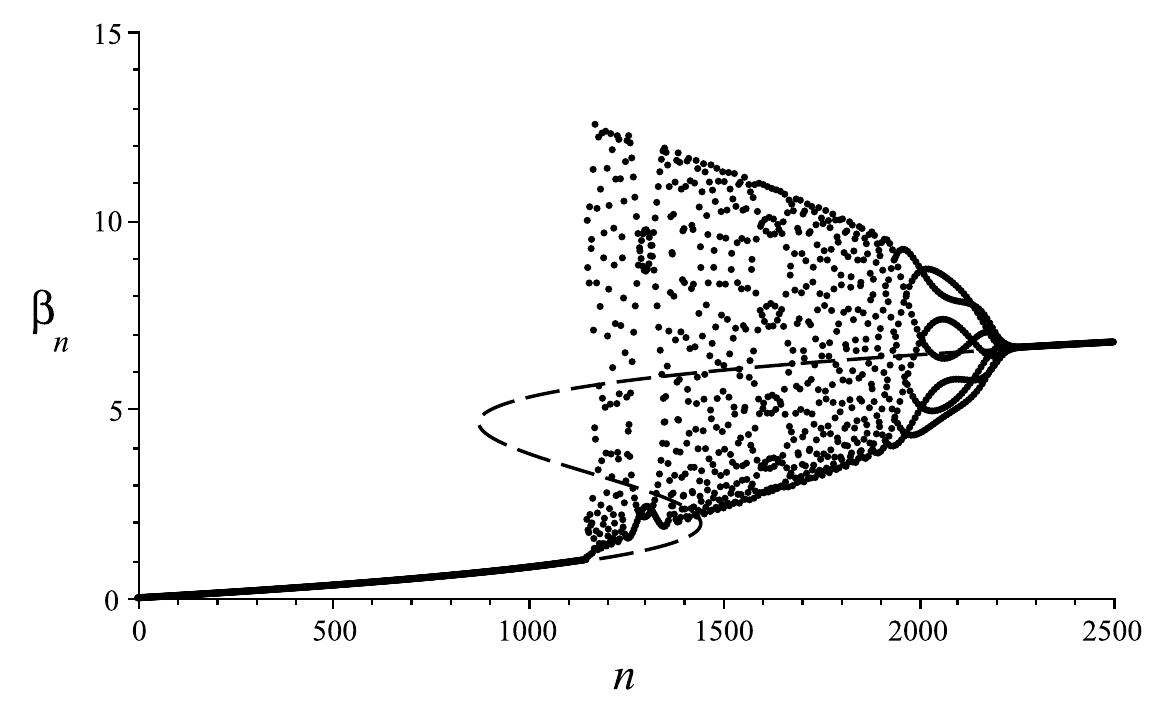}\\
\tau=40,\enskip \k=0.335&\tau=45,\enskip \k=0.335 &\tau=50,\enskip \k=0.335\\[5pt]
\fig{2}{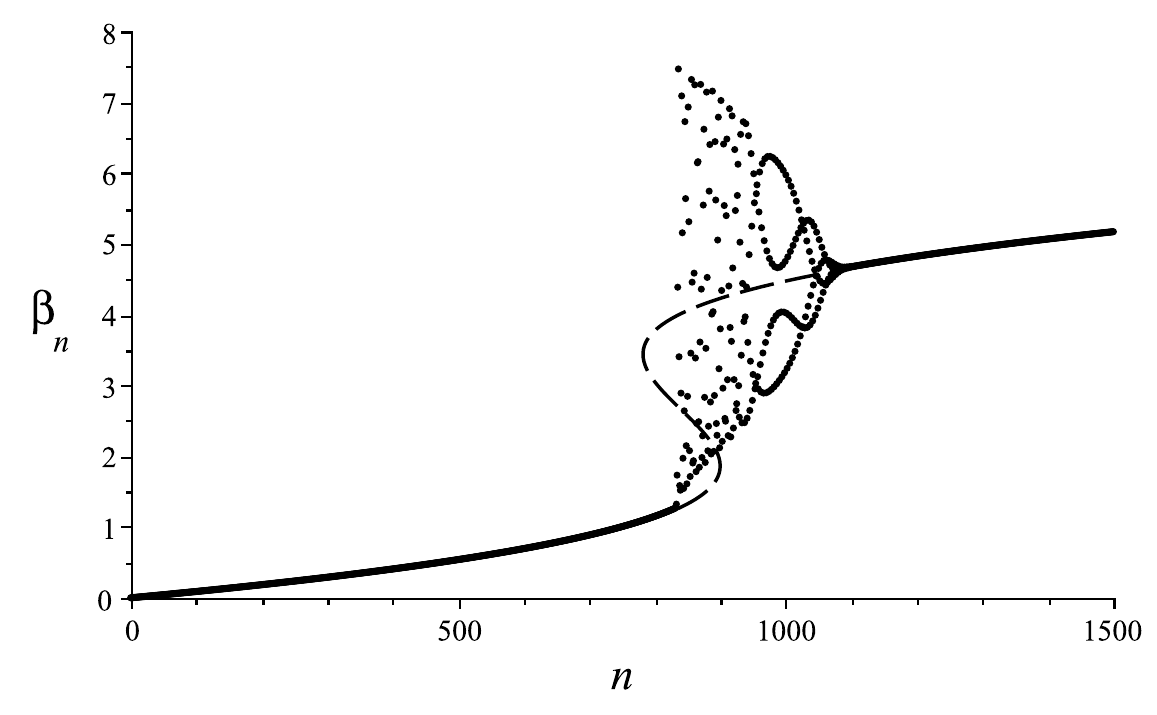}& \fig{2}{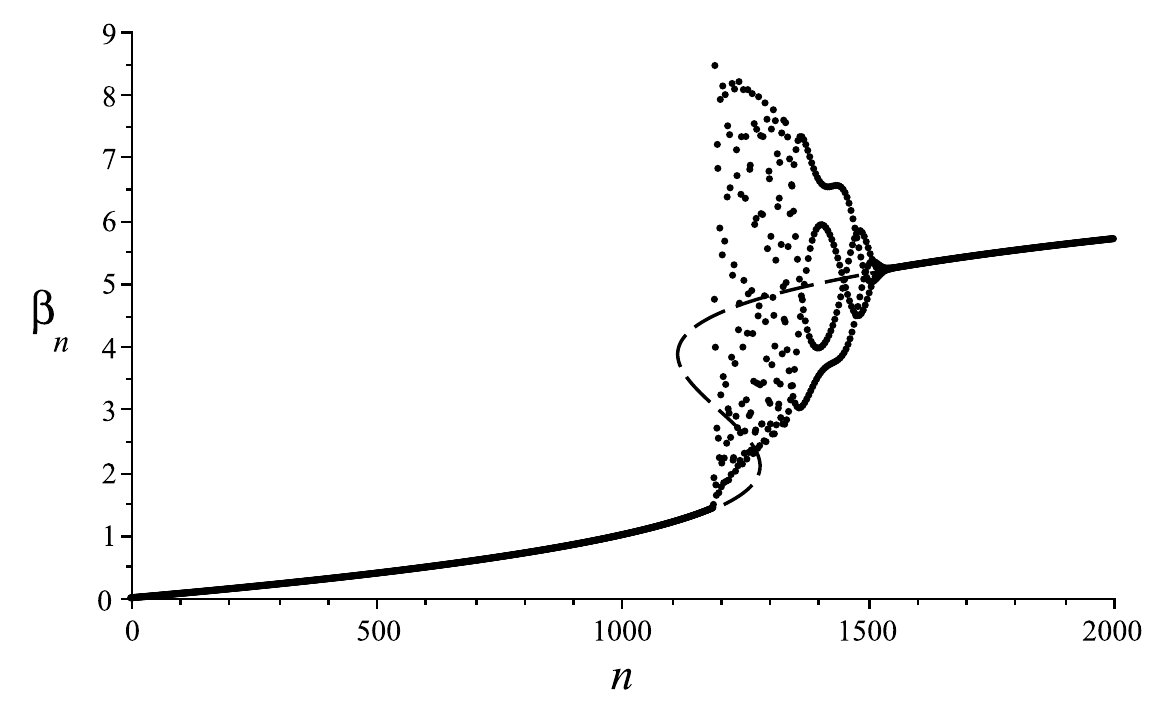}& \fig{2}{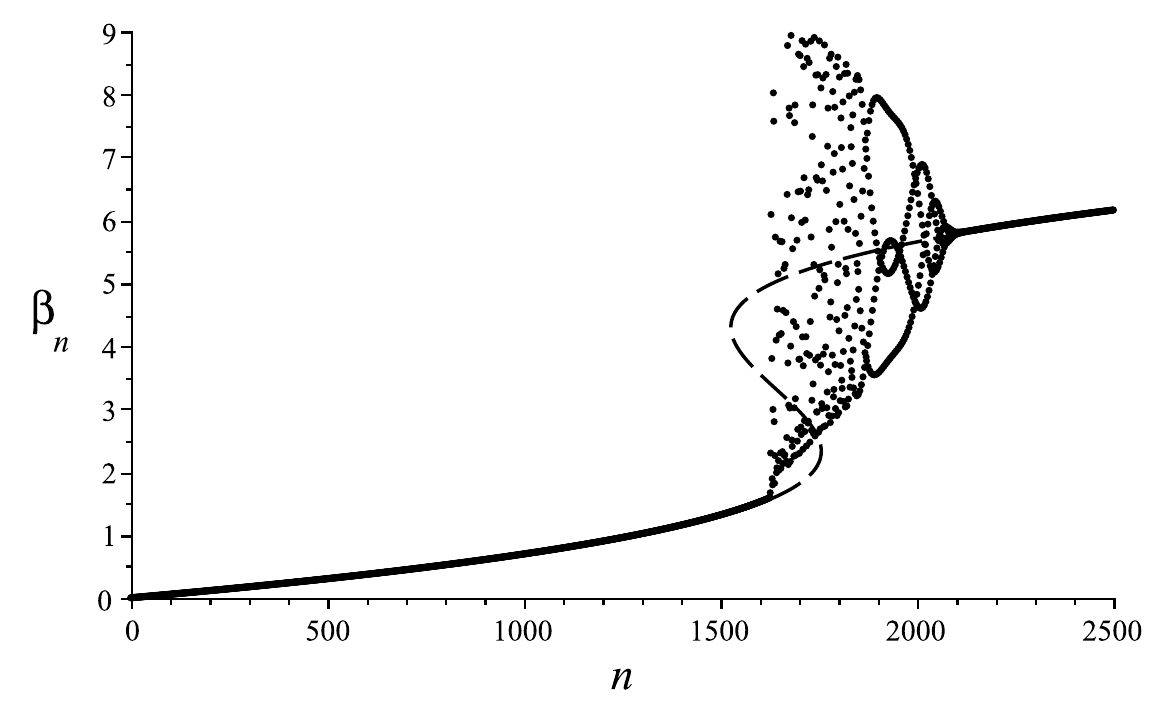}\\
\tau=40,\enskip \k=0.365 &\tau=45,\enskip \k=0.365 &\tau=50,\enskip \k=0.365
\end{array}\]
\caption{\label{Fig817}Plots of $\b_n$ for $\tau=40$, $45$, $50$, in the cases when $\k=0.275$, $\k=0.3$, $\k=0.335$ and $\k=0.365$.}
\end{figure}

\begin{figure}[ht!]
\[\begin{array}{c@{\quad}c@{\quad}c}
\fig{2}{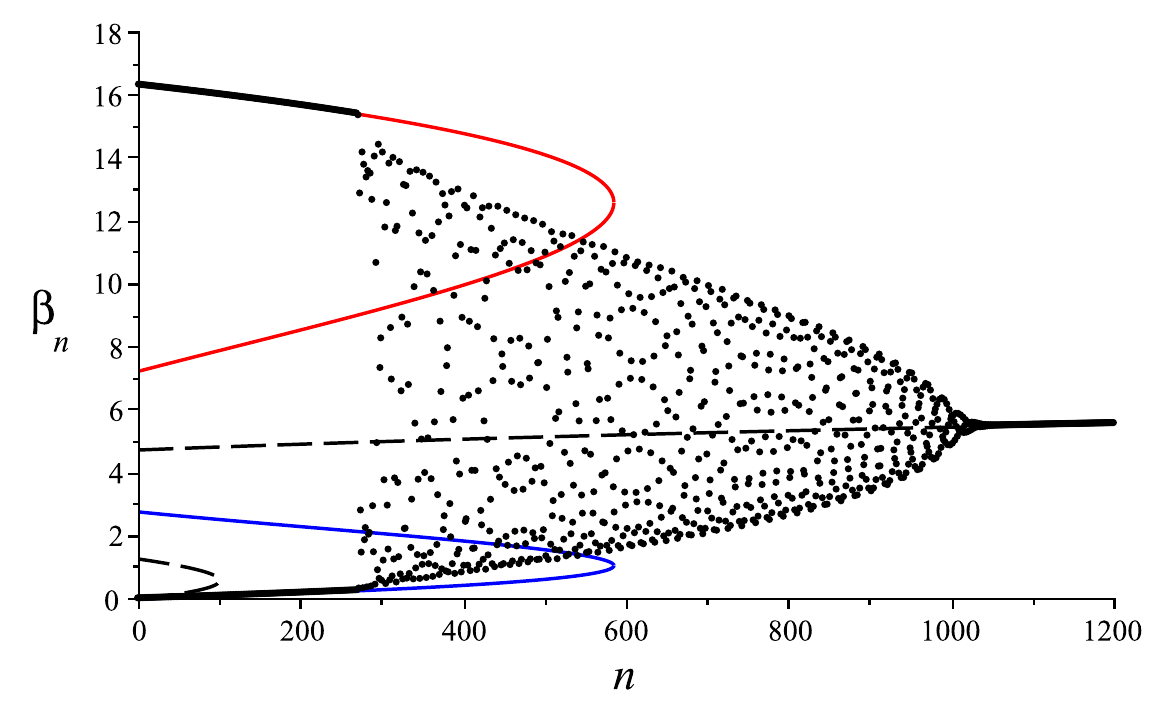}&\fig{2}{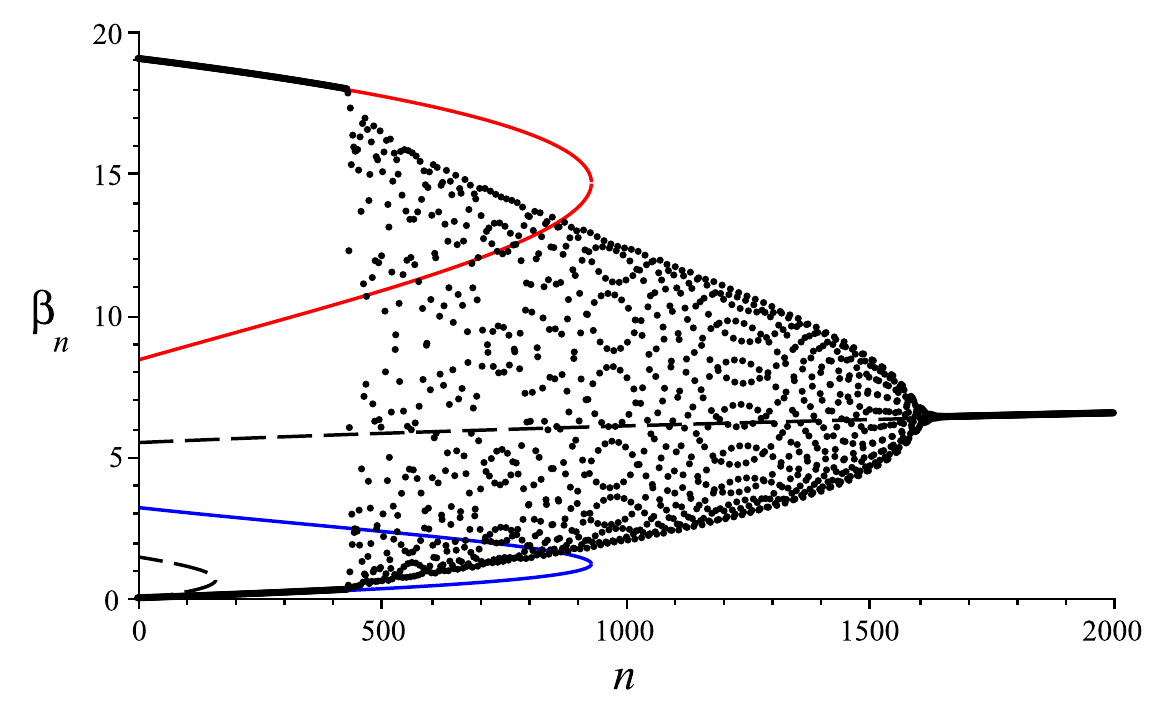} & \fig{2}{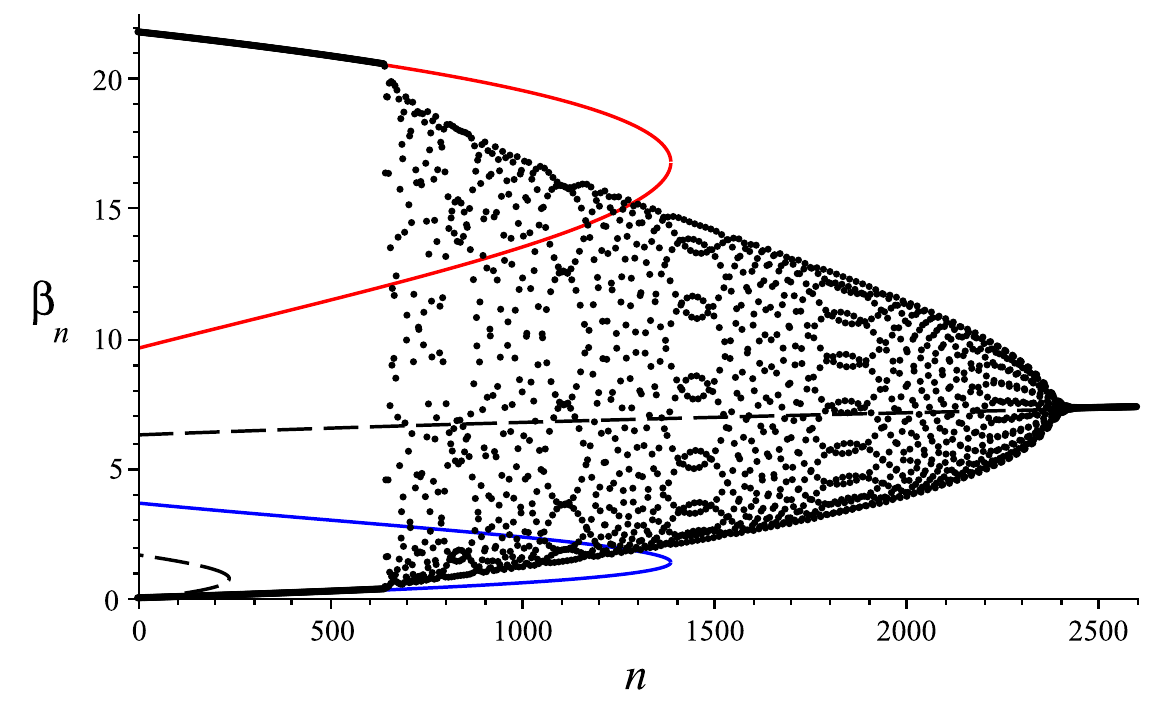}\\
\tau=30,\enskip \k=\tfrac{1}{5}&\tau=35,\enskip \k=\tfrac{1}{5} &\tau=40,\enskip \k=\tfrac{1}{5}\\
\fig{2}{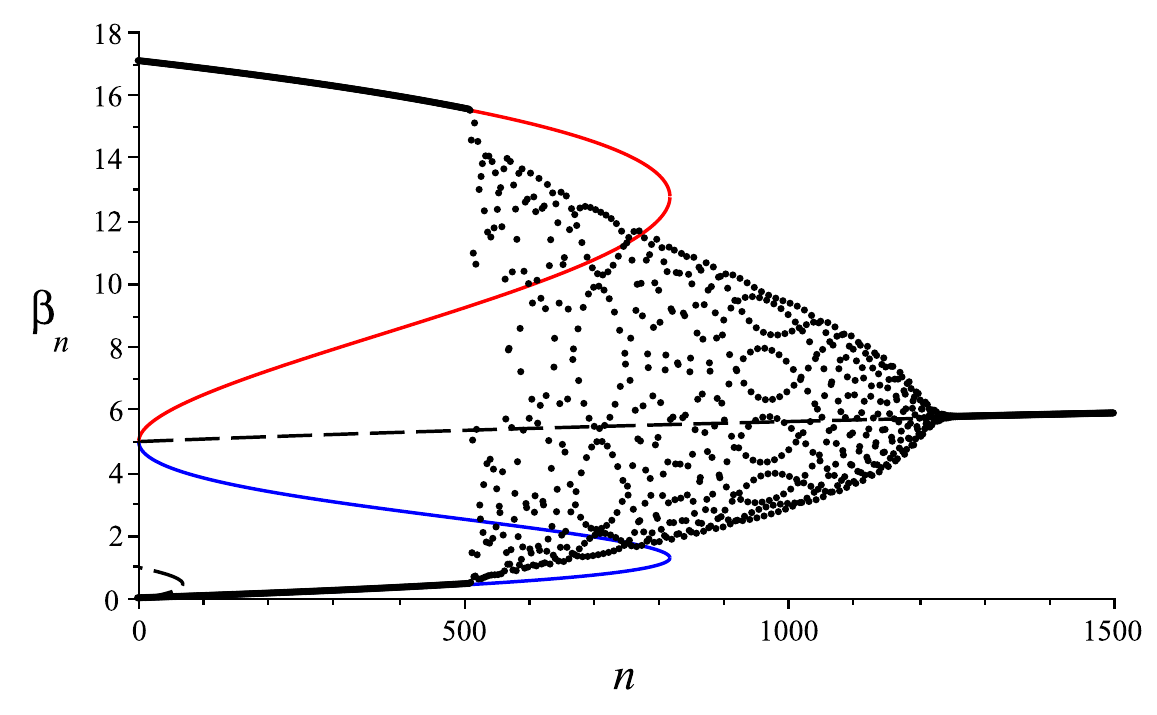}& \fig{2}{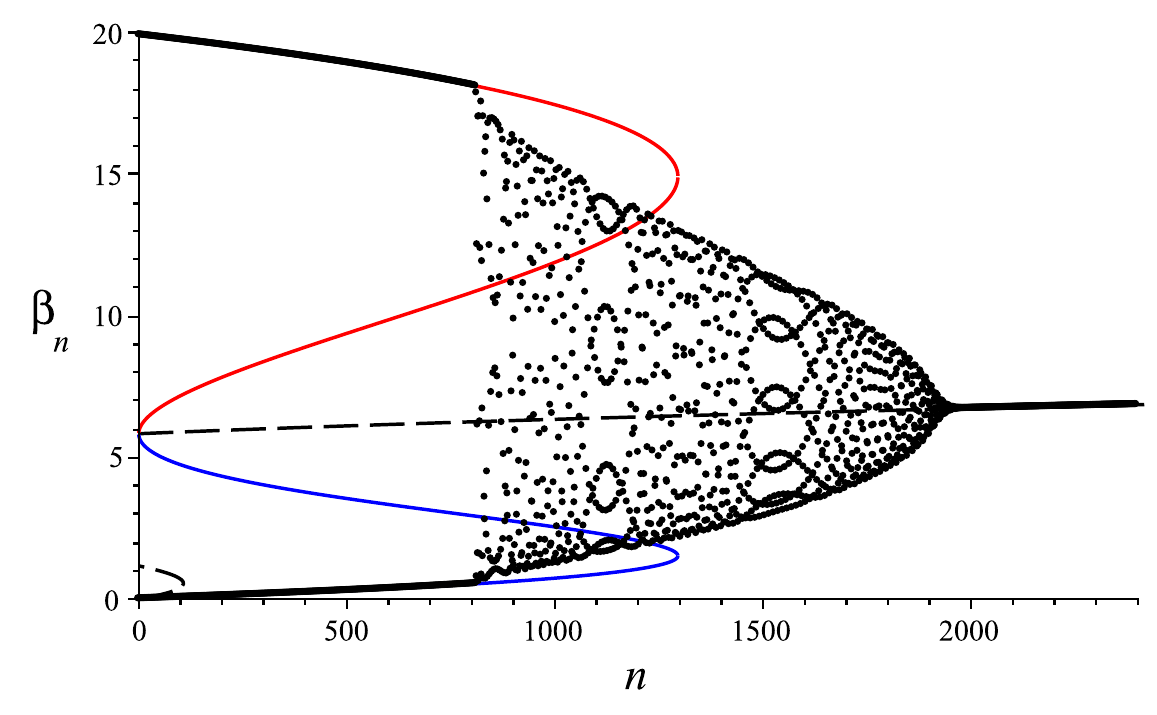} & \fig{2}{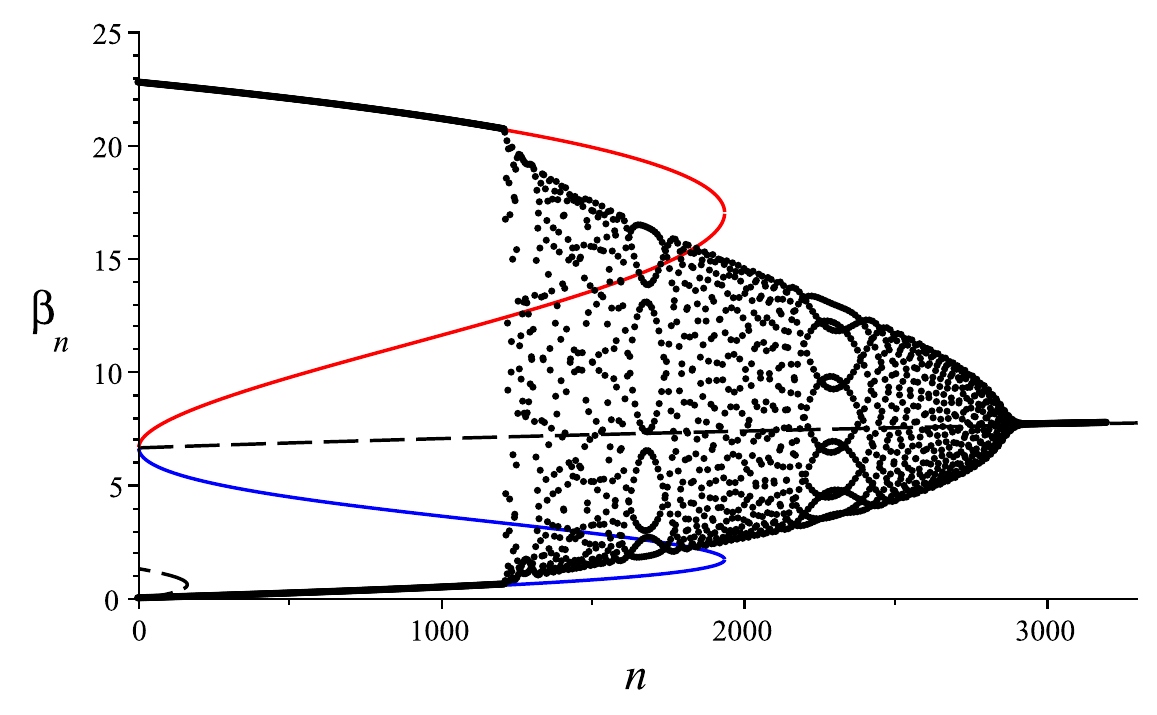}\\
\tau=30,\enskip \k=\tfrac{1}{6}&\tau=35,\enskip \k=\tfrac{1}{6} &\tau=40,\enskip \k=\tfrac{1}{6}\\
\fig{2}{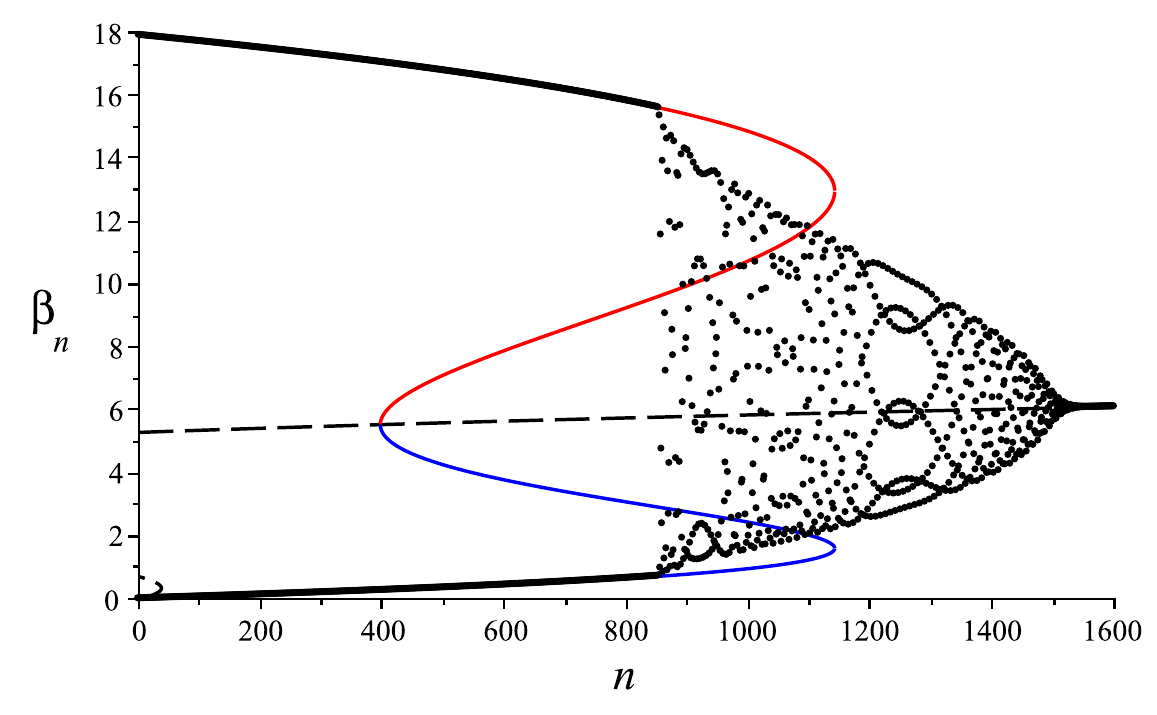} & \fig{2}{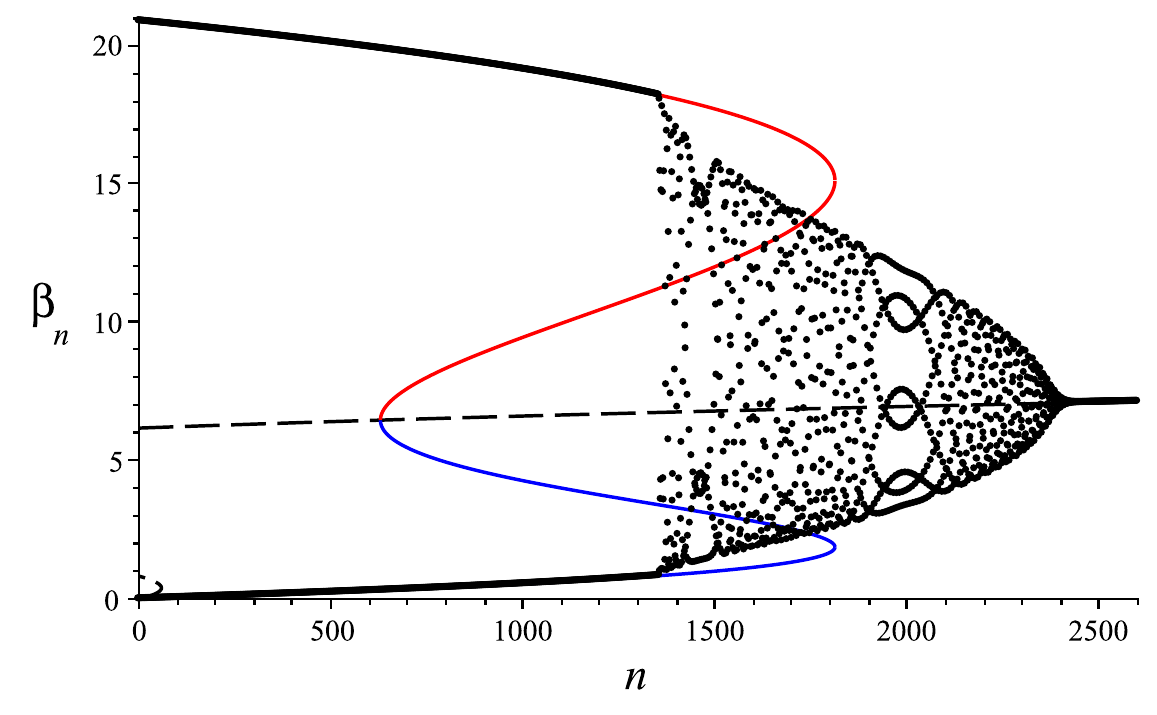} & \fig{2}{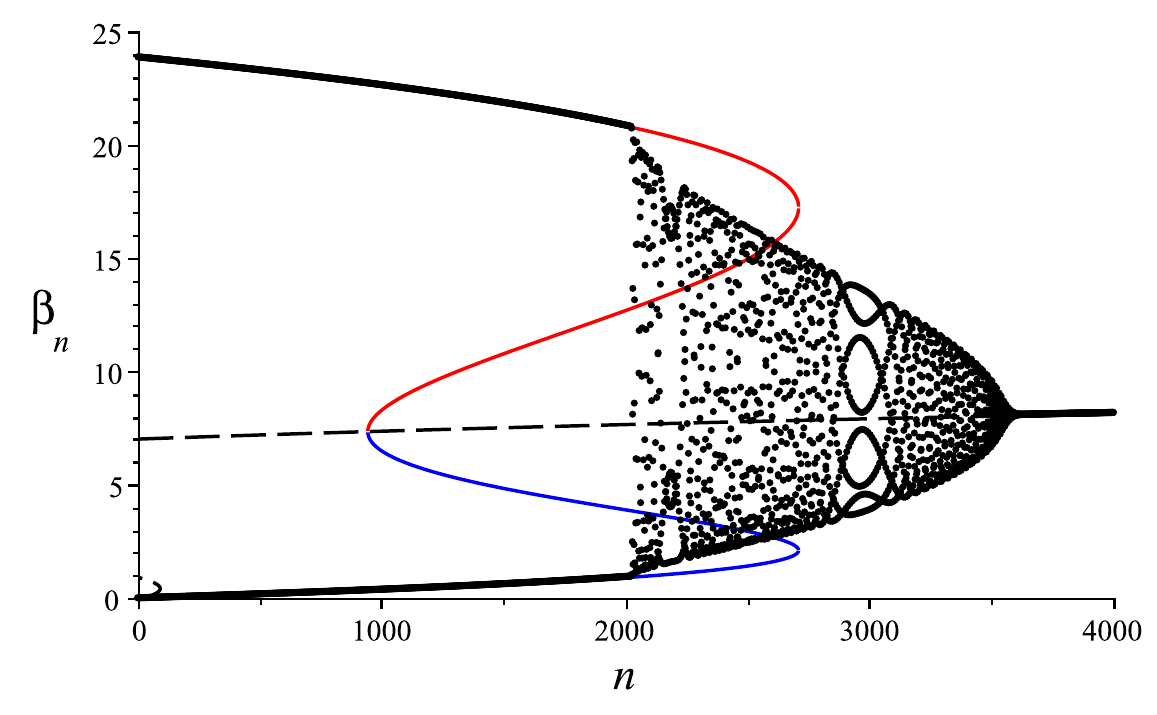}\\
\tau=30,\enskip \k=\tfrac{1}{8}&\tau=35,\enskip \k=\tfrac{1}{8} &\tau=40,\enskip \k=\tfrac{1}{8}
\end{array}\]
\caption{\label{Fig818}Plots of $\b_n$ for $\tau=30$, $35$, $40$, in the cases when $\k=\tfrac{1}{5}$, $\k=\tfrac{1}{6}$ and $\k=\tfrac{1}{8}$.}
\end{figure}
\begin{figure}[ht!]
\[\begin{array}{c@{\quad}c@{\quad}c}
\fig{2}{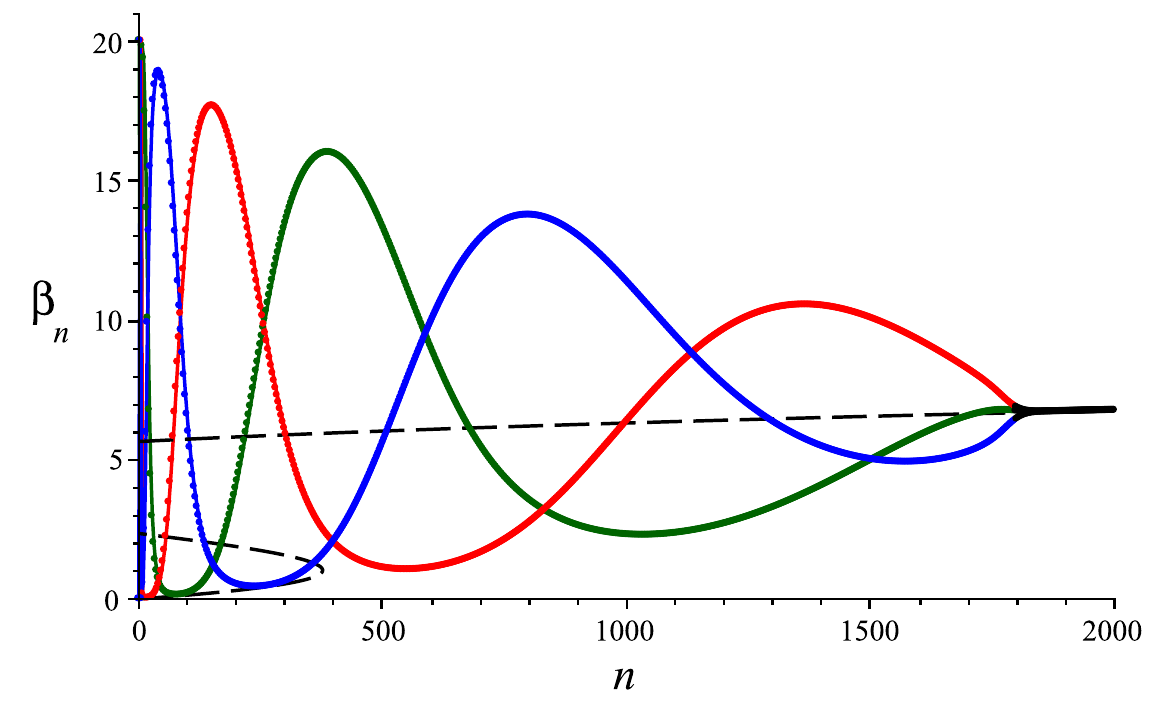} & \fig{2}{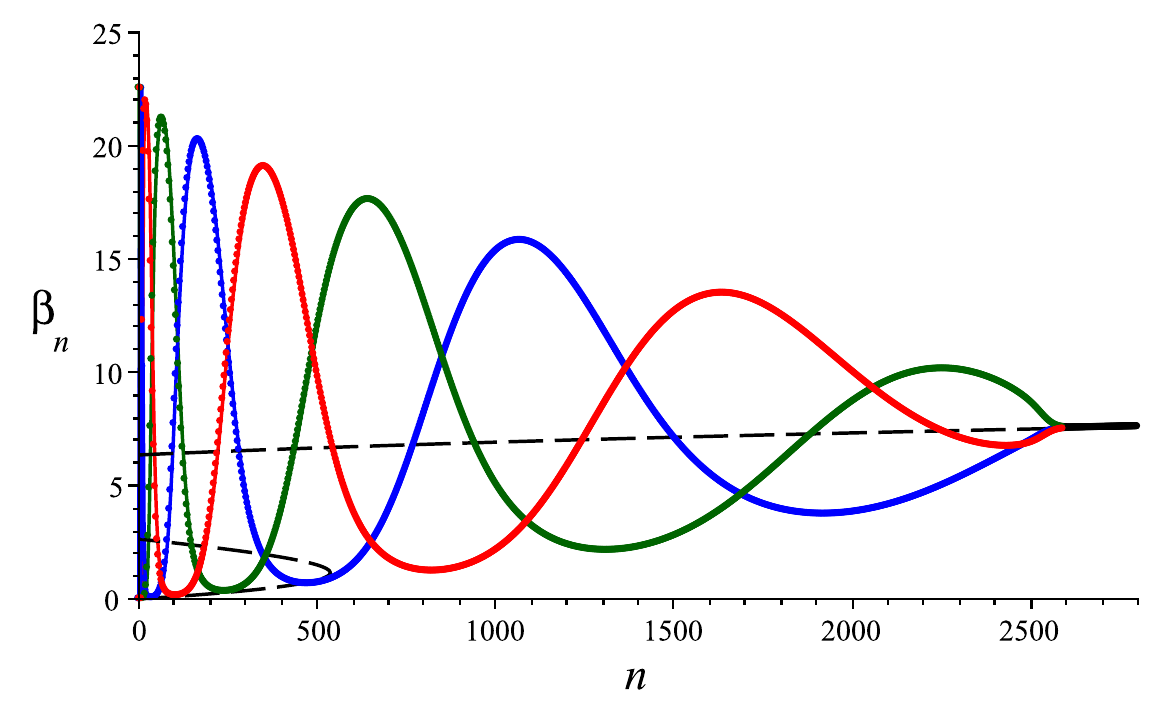} &
\fig{2}{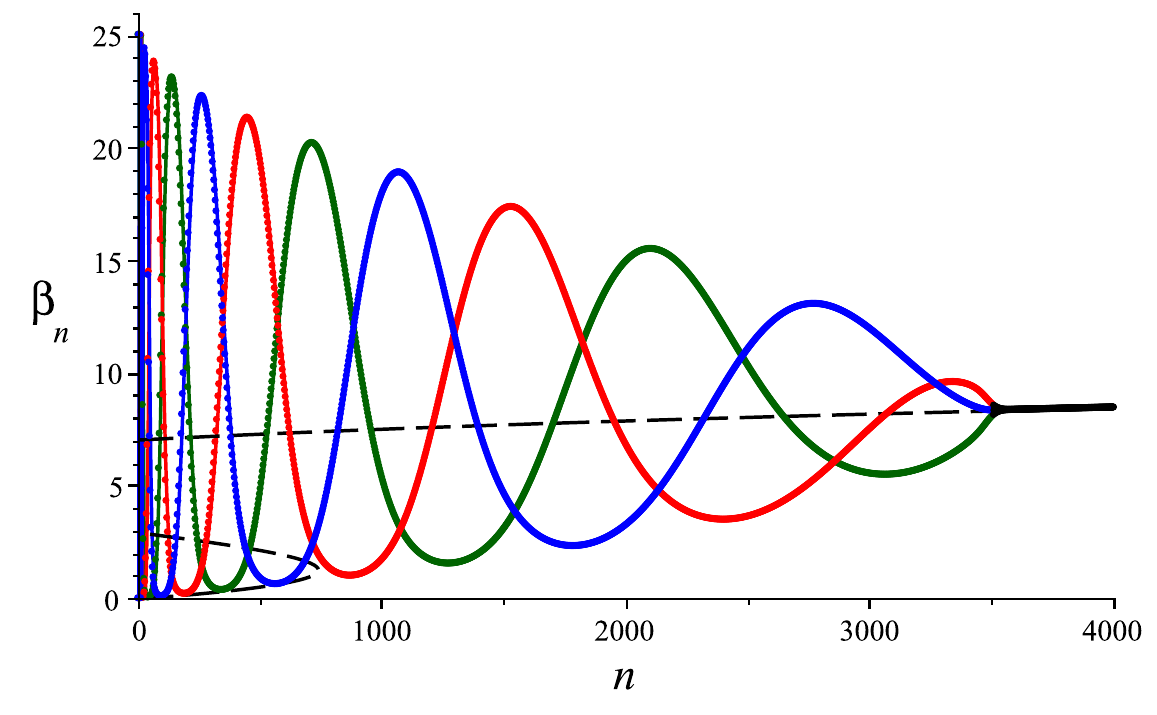} \\
\tau=40,\enskip\k=0.249 &\tau=45,\enskip\k=0.249 &
\tau=50,\enskip\k=0.249 \\
\fig{2}{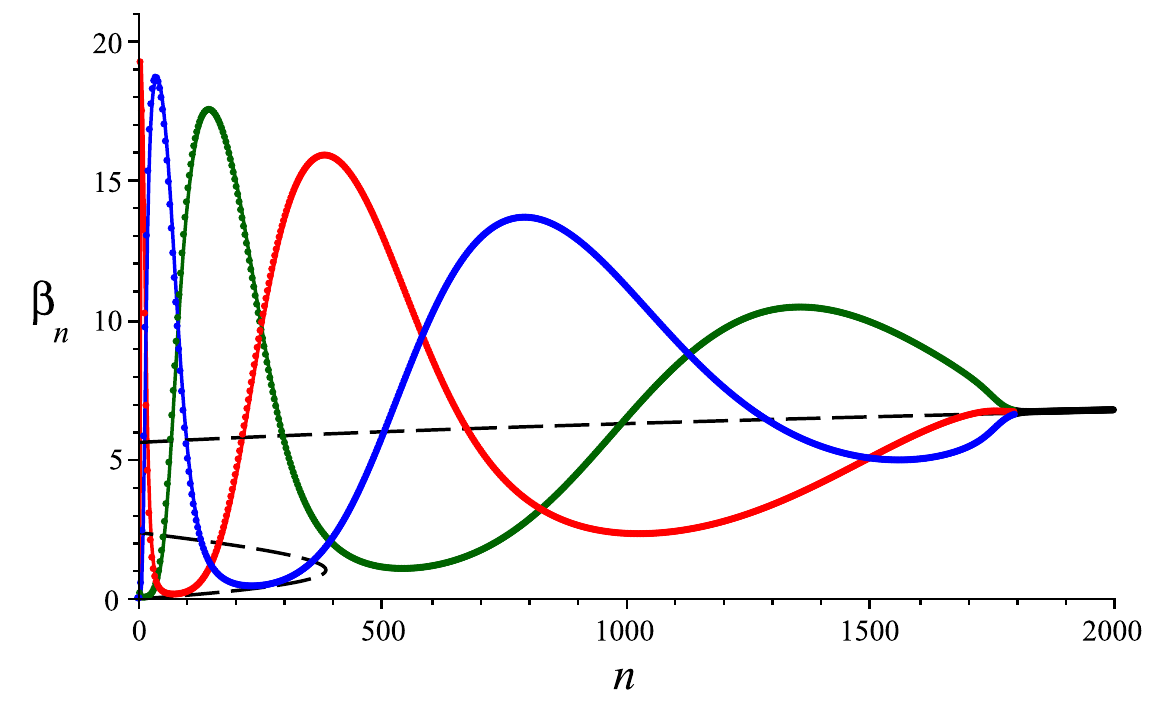} & \fig{2}{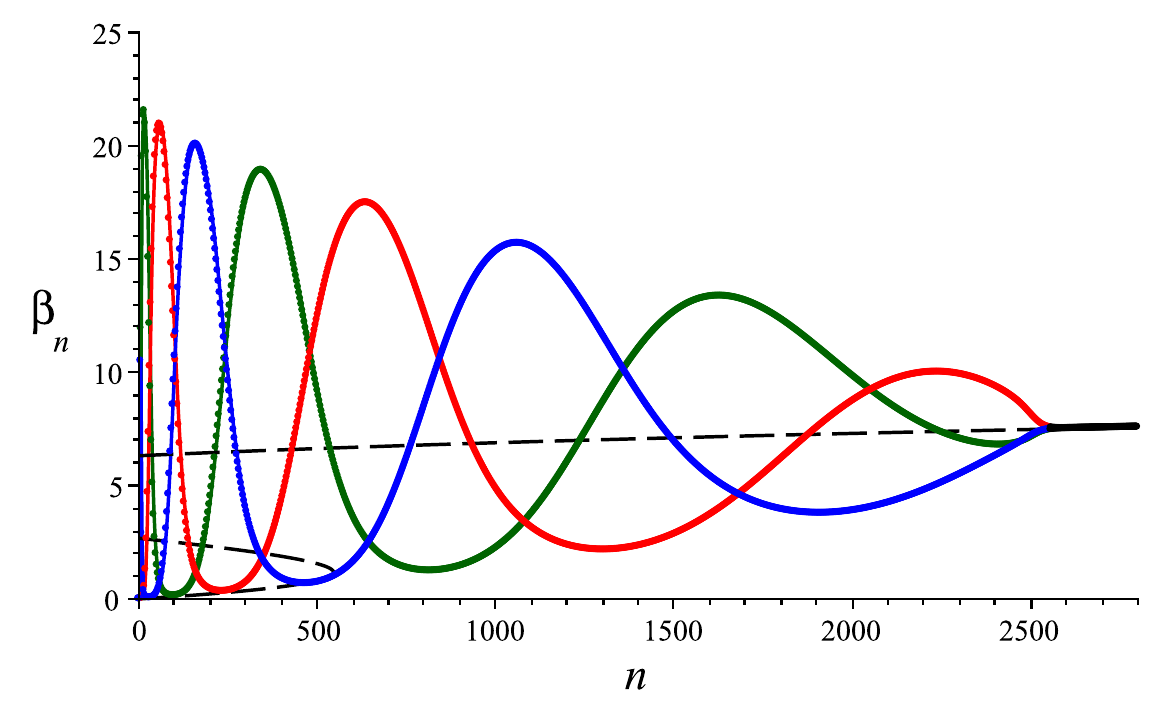} &
\fig{2}{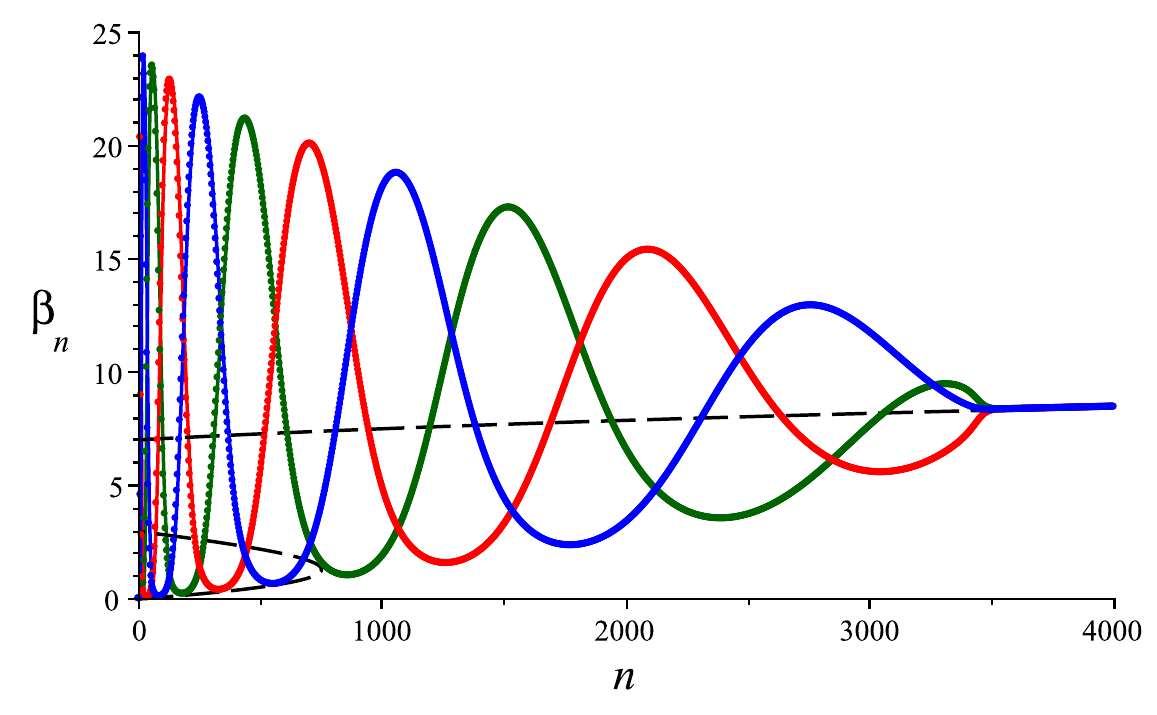} \\
\tau=40,\enskip\k=0.251 &\tau=45,\enskip\k=0.251 &
\tau=50,\enskip\k=0.251
\end{array}\]
\caption{\label{Fig819}Plots of $\b_n$ for $\tau=40$, $45$, $50$, in the cases when $\k=0.249$ and $\k=0.251$. The recurrence coefficients $\b_{3n}$ are plotted in \blue{blue}, {$\b_{3n+1}$} in \green{green} and $\b_{3n-1}$ in \red{red}. As $\tau$ increases we see that the number of oscillations increases and also that $\b_{3n+1}$ and $\b_{3n-1}$ interchange between the two cases.}
\end{figure}

\subsection{\label{subsec79}Comparison between numerical calculations and closed formulae}
{As we have seen, the recurrence coefficients $\b_{n}(\tau;\k)$ can be expressed as Hankel determinants of the moments using the expressions in Corollary \ref{Cor:2.3}. In \S \ref{sec:closed moments} we derived closed formula expressions for the moments in the cases when $\k=\tfrac{1}{4}$, $\k=\tfrac{1}{3}$ and $\k=0$ via special functions, though closed formula expressions for other values of the parameters $(\k,\tau)$ are not known at present. The computation of such Hankel determinants is a highly challenging computational problem, even for matrices of relatively modest dimension. Such difficulty is more compounded when the moments, which are the determinant entries, are special functions.  
We remark that for the quartic Freud weight \eqref{Freud42}, where the recurrence coefficients are expressed in terms of parabolic cylinder functions \cite{refCJ18,refCJK}, Iserles and Webb \cite{refIW} comment that ``these explicit coefficients, unfortunately, cannot be computed easily and rapidly". 
In the cases when $\k=\tfrac{1}{4}$, $\k=\tfrac{1}{3}$ and $\k=0$, the numerical evaluation of the initial conditions \eqref{beta_ics} agrees exactly with the explicit expressions, to the accuracy used. However there is a significant difference in the time taken to compute $\b_1$ and $\b_2$ as illustrated in Table \ref{table71}, in the case when $\k=\tfrac{1}{3}$, for $\tau=30$, $\tau=40$ and $\tau=50$.}

\begin{table}[h!]
\centering
{\begin{tabular}{|c@{\enskip}|@{\enskip}r@{\enskip}|@{\enskip}r|}
\hline
$\tau$ & Exact & Numerical\\ \hline
30 & 1.75 & 34.92\\
40 & 4.97 & 276.98\\
50 & 235.65 & 805.43 \\ \hline
\end{tabular}}
\caption{\label{table71}{A comparison of the time taken (in seconds) to compute $\b_1$ and $\b_2$ using the closed form expressions and numerically in the case when $\k=\tfrac{1}{3}$, for $\tau=30$, $\tau=40$ and $\tau=50$.}}\end{table}

\comment{To illustrate how sensitive the problem is to the initial conditions, in Figure \ref{fig:tau40_pert}, in the case when $\tau=40$ and $\k=\tfrac{1}{3}$ where the initial condition $\b_1$ is (a), $\b_1=\ifrac{\mu_2}{\mu_0}$, (b), $\b_1=\ifrac{\mu_2}{\mu_0}+10^{-5}$ and (c), $\b_1=\ifrac{\mu_2}{\mu_0}+10^{-10}$,
with the other initial conditions as in \eqref{beta_ics}.
We see that even a small perturbation of the initial condition leads to a completely different behaviour.
\begin{figure}[ht!]
\[\begin{array}{c@{\quad}c@{\quad}c}
\fig{2}{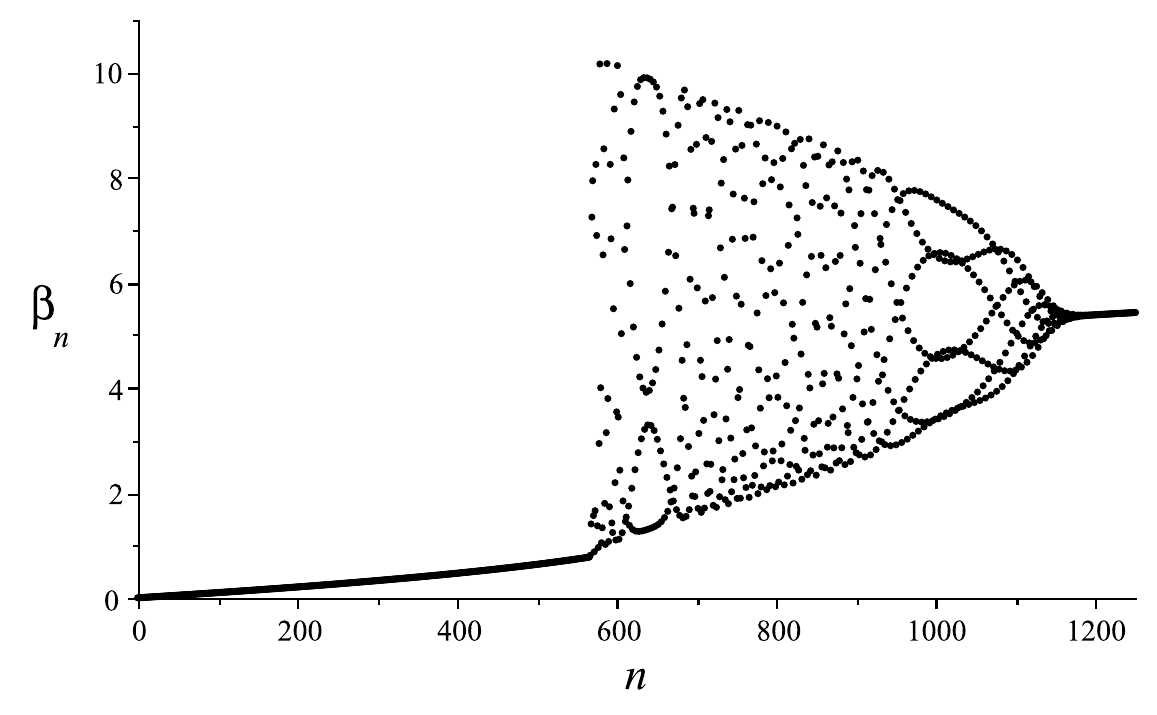} & \fig{2}{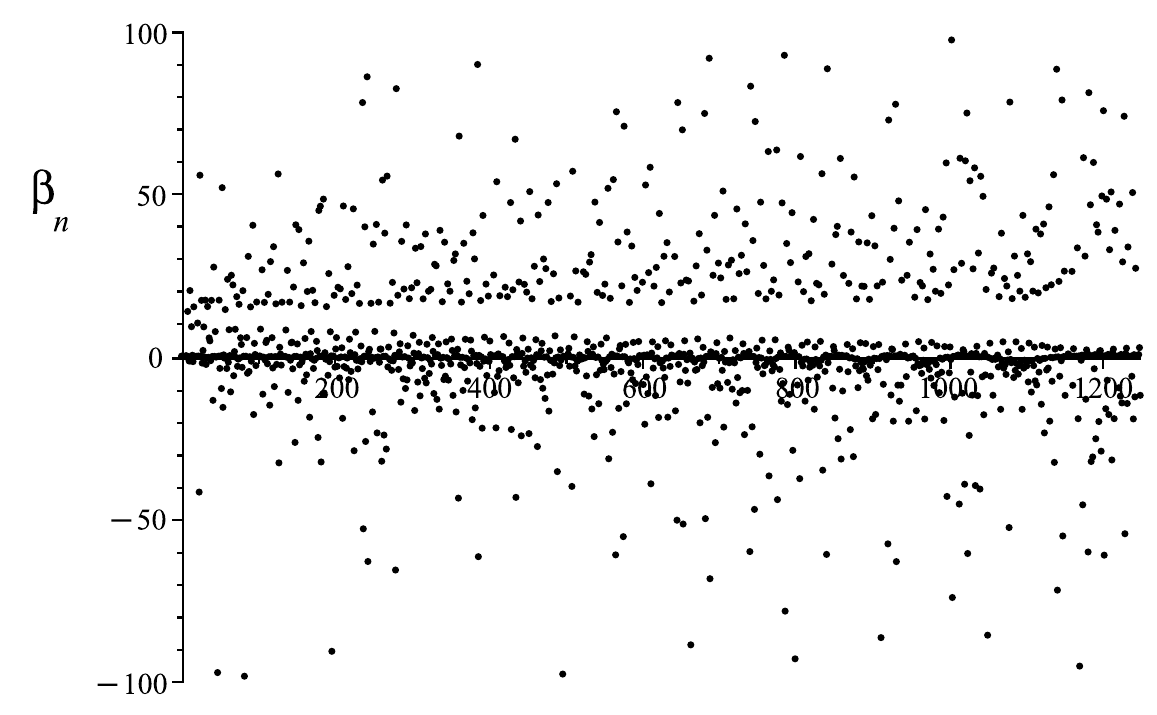} & \fig{2}{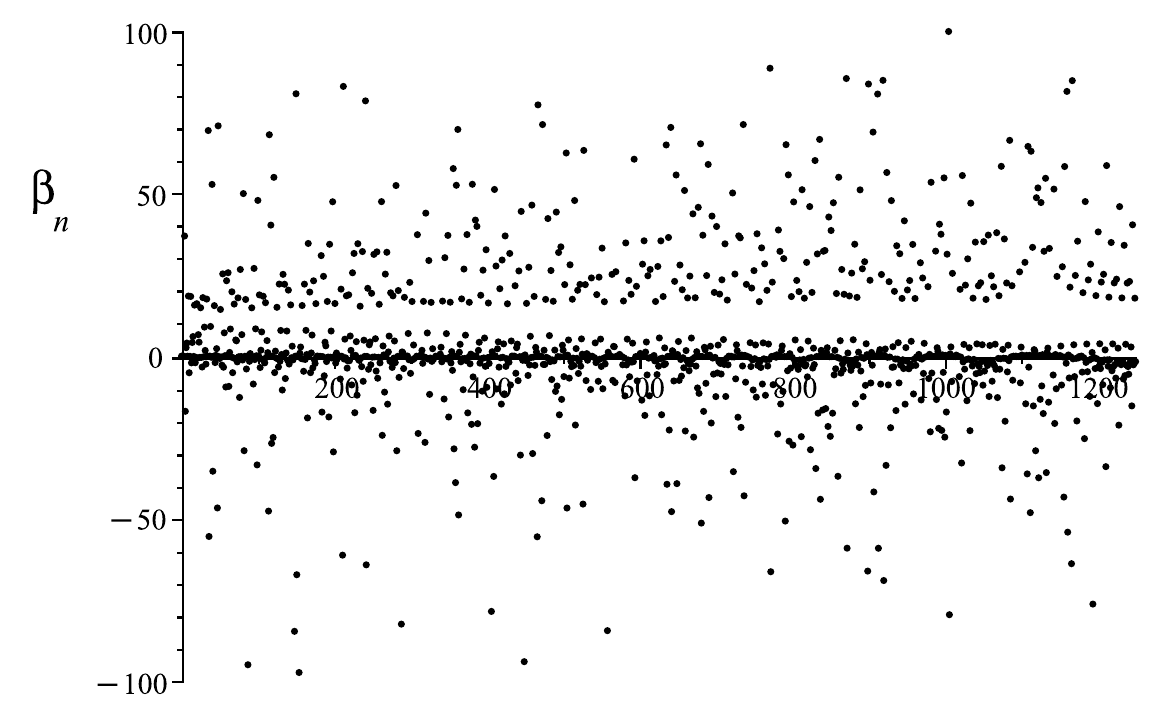}\\
\text{(a)},\enskip\b_1=\mu_2/\mu_0 & \text{(b)},\enskip\b_1=\mu_2/\mu_0+10^{-5} & 
\text{(c)},\enskip\b_1=\mu_2/\mu_0+10^{-10}
\end{array}\]
\caption{\label{fig:tau40_pert}{Recurrence coefficients for the sextic Freud weight with $\tau=40$ and $\k=\tfrac{1}{3}$. In (b) the value of $\b_1$ has been perturbed by $10^{-5}$ and in (c) by $10^{-10}$.}}
\end{figure}}

\subsection{\label{subsec710}Sensitivity to initial data: computational implications}
{We said earlier that the problem was highly sensitive to the initial conditions. 
A small perturbation of the initial conditions leads to a completely different behaviour, as illustrated in Figure \ref{fig:tau40_pert}(b) where $\b_1$ is perturbed by $10^{-10}$ with the other initial conditions as in \eqref{beta_ics} and the same number of digits as in \ref{fig:tau40_pert}(a).
Also if the number of the digits used is not sufficient then there is breakdown and some $\b_n$ are negative, for $n$ sufficiently large, as illustrated in Figure \ref{fig:tau40_pert}(c).}
\begin{figure}[ht!]
\[\begin{array}{c@{\quad}c@{\quad}c}
\fig{2}{Freud642_tau40_13exact} & \fig{2}{Freud642_tau40_13pert10}  & \fig{2}{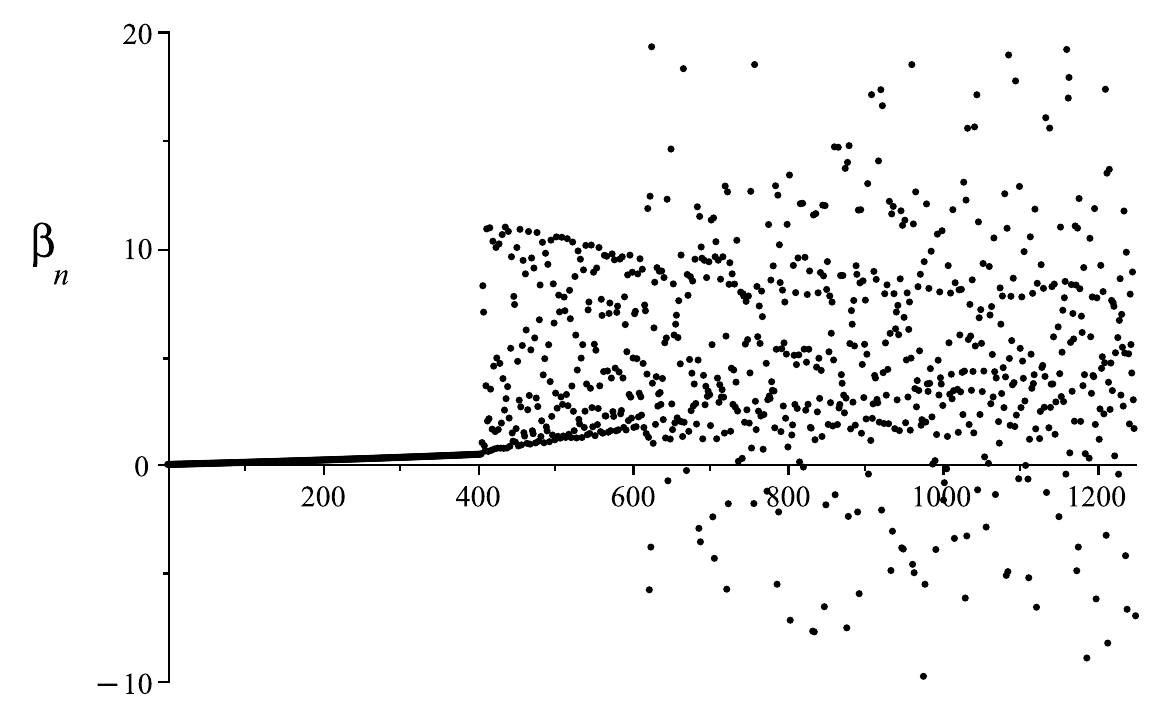}\\
\text{(a)},\enskip\b_1=\mu_2/\mu_0, \ \text{digits}=1200 & 
\text{(b)},\enskip\b_1=\mu_2/\mu_0+10^{-10} &
\text{(c)},\enskip\b_1=\mu_2/\mu_0, \ \text{digits}=800
\end{array}\]
\caption{\label{fig:tau40_pert}{Recurrence coefficients for the sextic Freud weight with $\tau=40$ and $\k=\tfrac{1}{3}$. 
In (b), the value of $\b_1$ has been perturbed by $10^{-10}$ and in (c), a smaller number of digits is used.}}
\end{figure}

{In fact, we believe that any choice for the initial conditions other than \eqref{beta_ics} will lead to a sequence $\big\{\b_n(\k,\tau)\big\}_{n\geq 0}$ with negative terms. More formally, we conjecture the following:
\begin{conjecture}{\rm
For all $\k,\tau\in\mathbb{R}$, there is a unique positive solution of the dP$_{\rm I}^{(2)}$ equation \eqref{eq:rr642} for which $\b_0(\k,\tau)=0$ and 
$\b_n(\k,\tau)>0$ for all $n\geq 1$, corresponding to the initial values \eqref{beta_ics}.
}\end{conjecture}
\begin{remark}{\rm It is well-known that the \dPI\ equation
\[\b_{n}(\b_{n+1}+\b_{n}+\b_{n-1} +K)=n,\]
with $K$ a constant and $\b_{0}=0$ has a unique positive solution \cite{refBNevai,refLewQ,refNevai83,refNevai84}.
}\end{remark}}

\section{Two-dimensional plots}\label{sec:2D}
In this section we analyse the evolution of the $\b_n$ compared to $\b_{n-1}$ as $n$ increases. The discussion in the previous section indicates how the plots in the $(\b_n,\b_{n-1})$-plane will behave when $n$ is small or $n$ is large. However, it is not so clear what the relationship between $\b_n$ and $\b_{n-1}$ is in the transition region. In fact, the plots of 
 $(\b_n,\b_{n-1})$ give further and different insight into the behaviour of the recurrence coefficient $\b_n$ as $n$ increases. 
 
\begin{figure}[ht!]
\[\begin{array}{c@{\qquad}c@{\qquad}c}
\tfig{4.4}{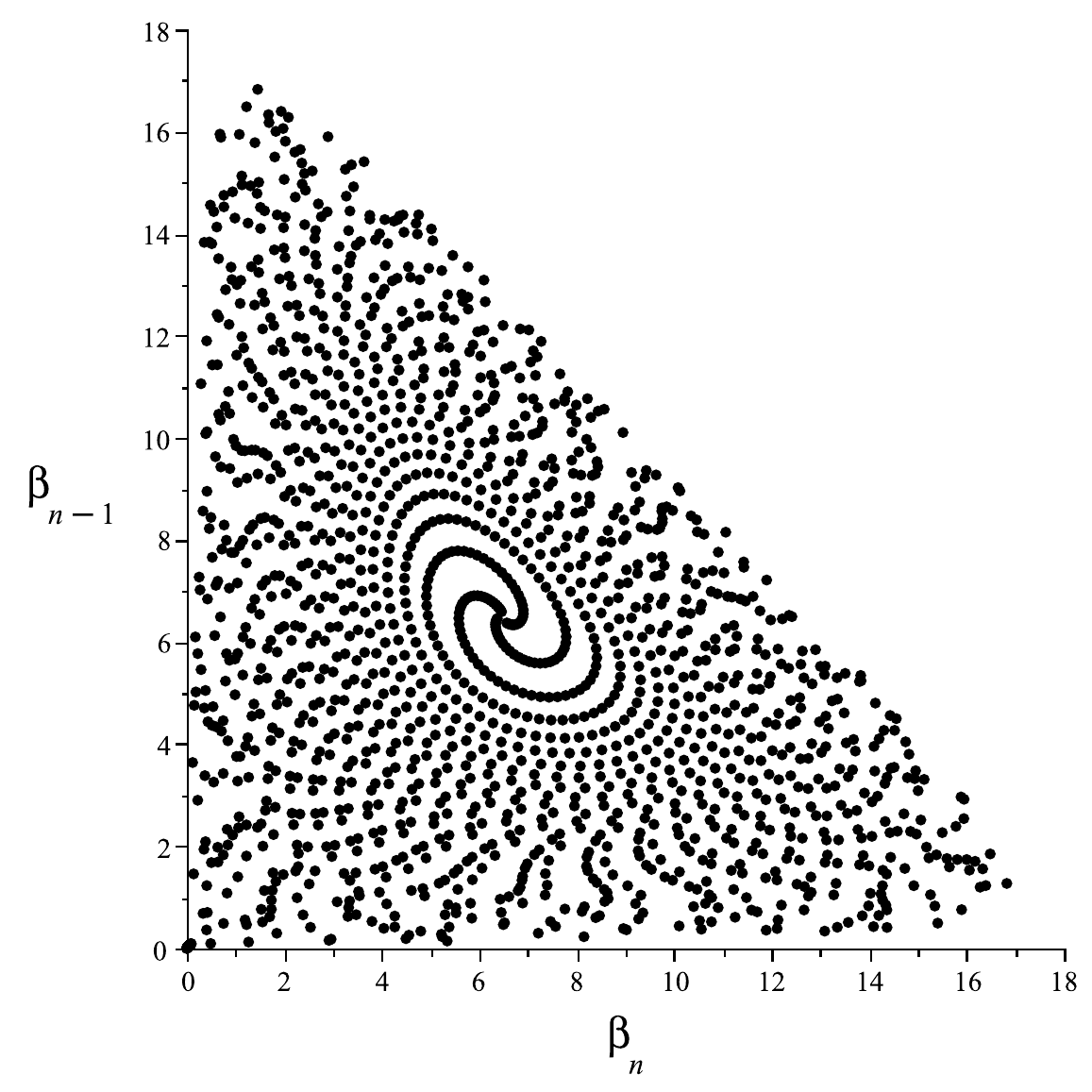} & \tfig{4.4}{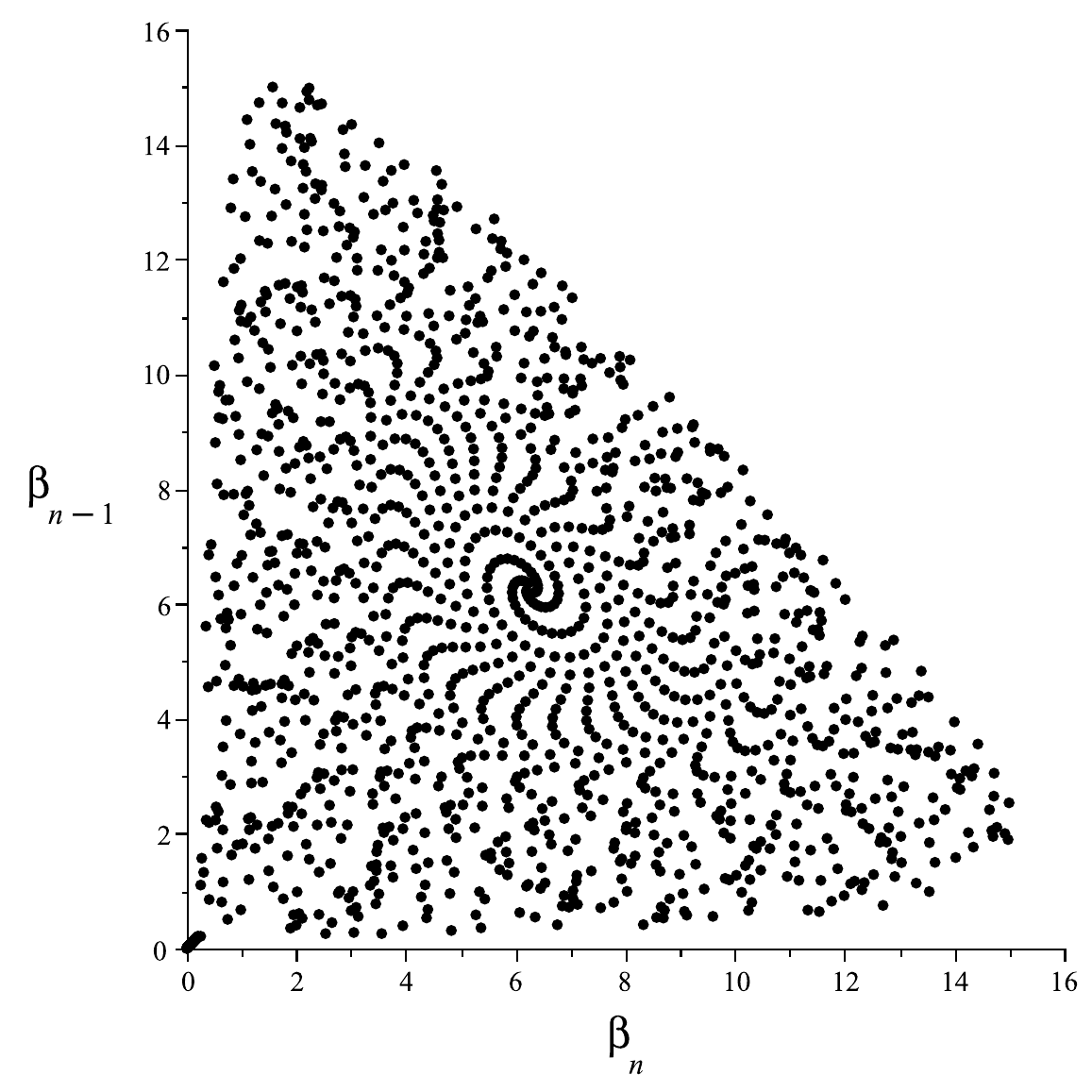} & \tfig{4.4}{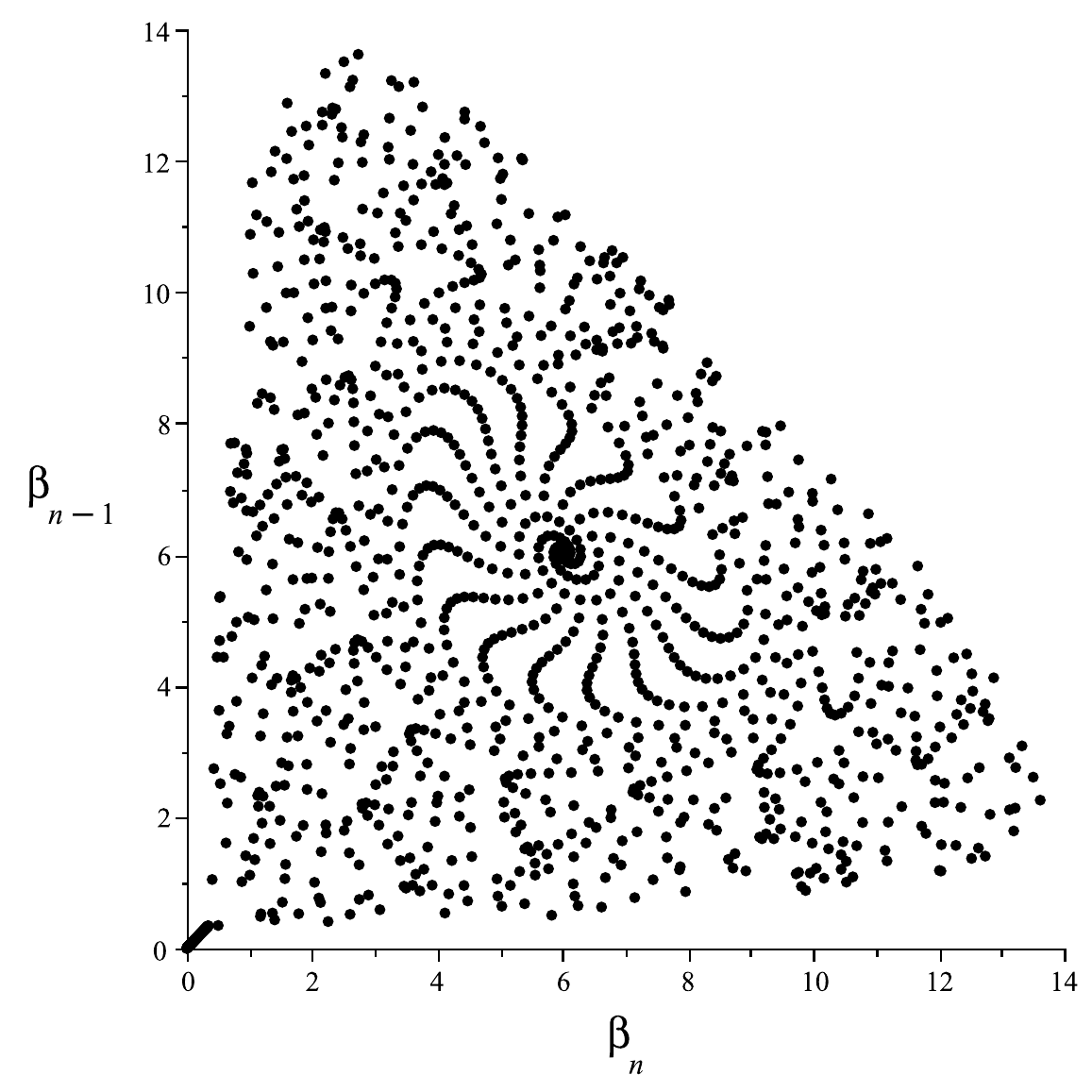} \\
\tau=40, \enskip\k=0.265 &\tau=40, \enskip\k=0.28 &\tau=40, \enskip\k=0.295 \\[5pt]
\tfig{4.4}{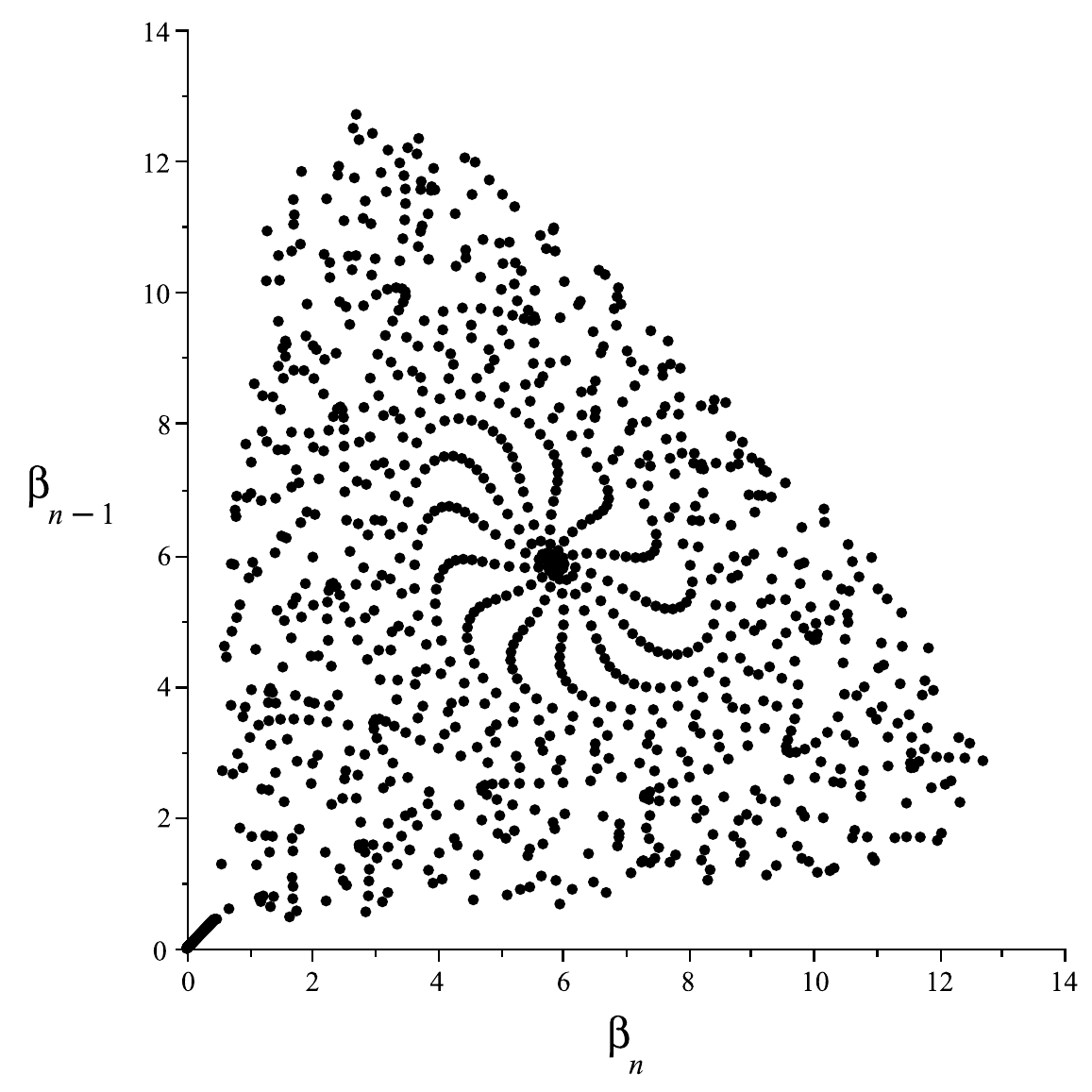} & \tfig{4.4}{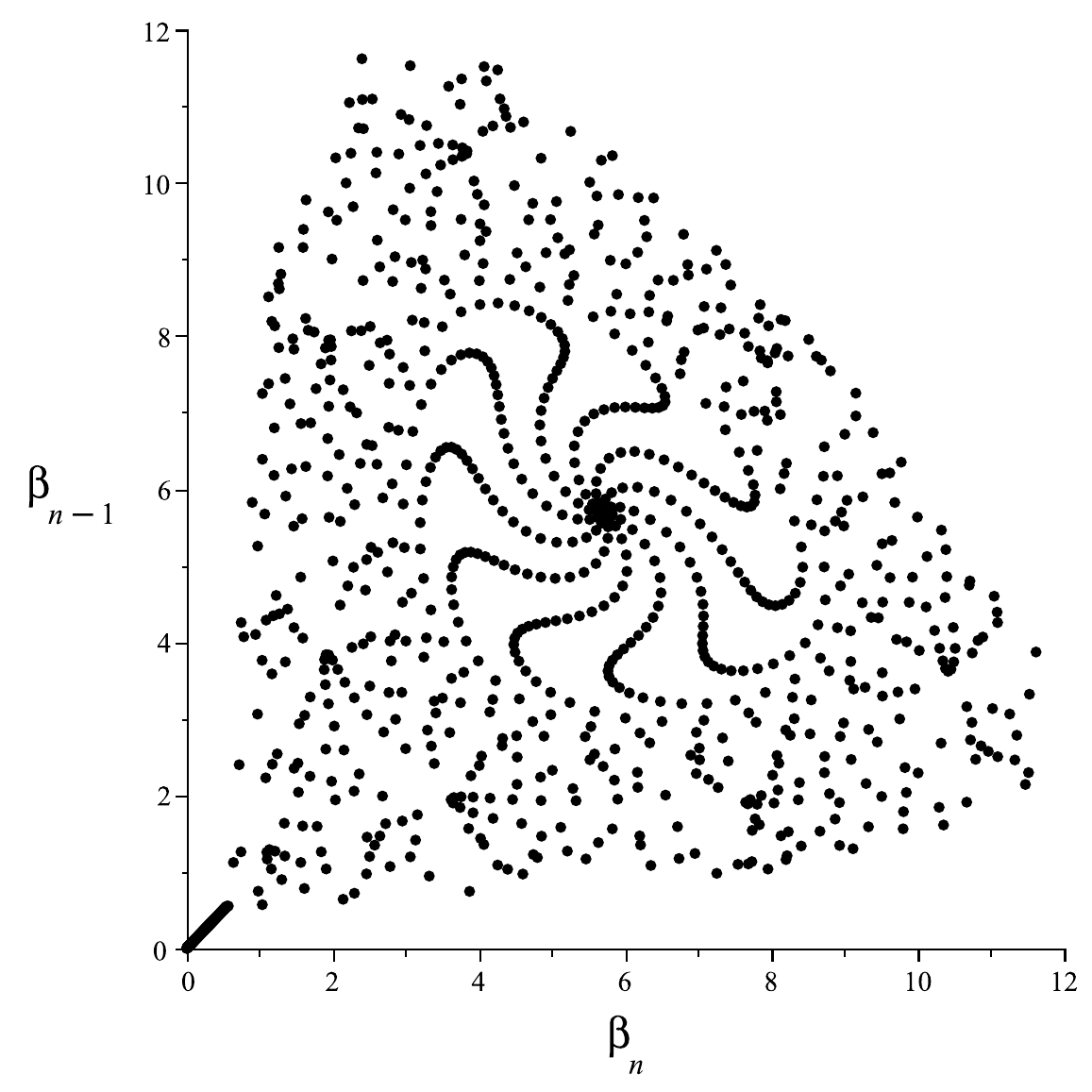} & \tfig{4.4}{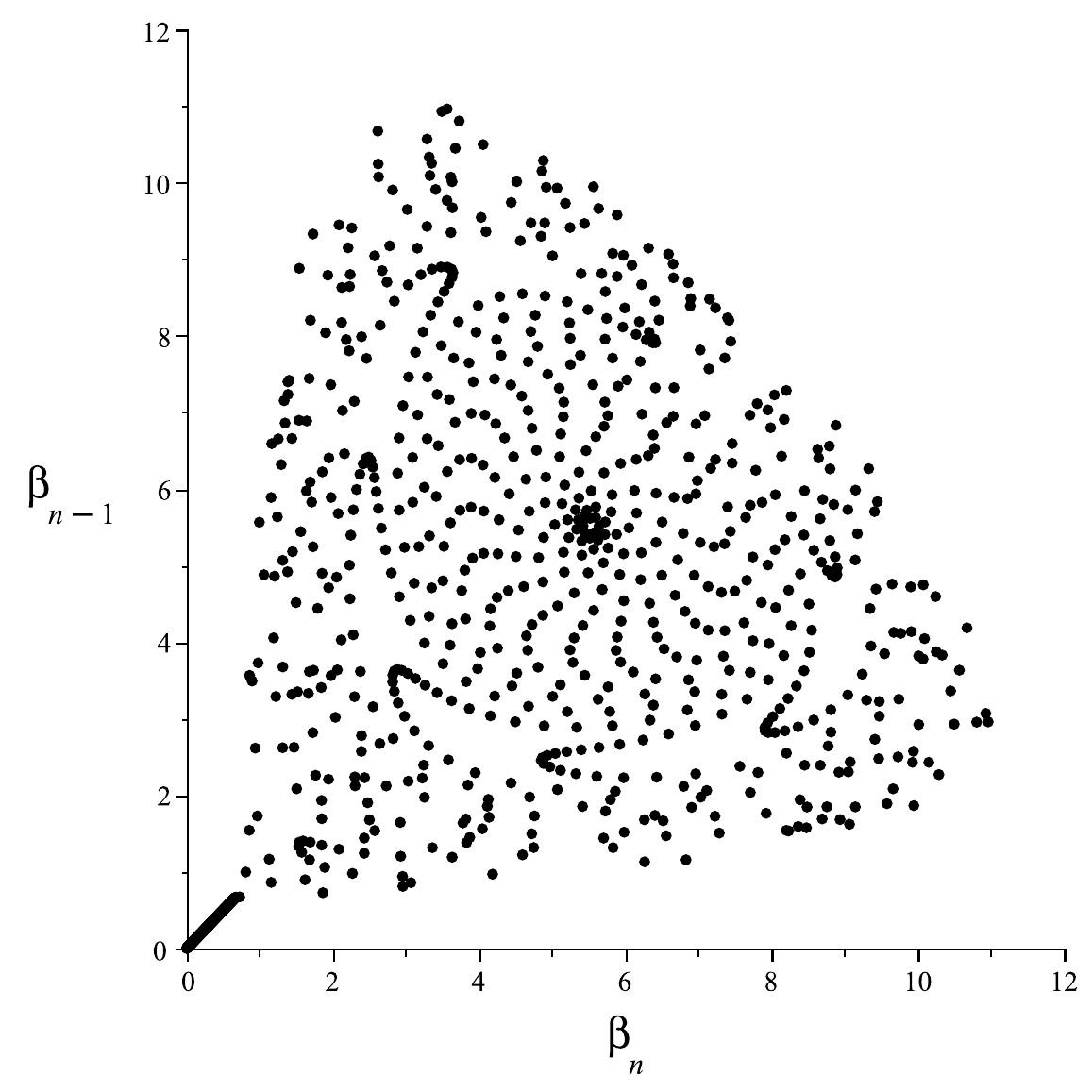} \\
\tau=40, \enskip\k=0.305 &\tau=40, \enskip\k=0.315 &\tau=40, \enskip\k=0.325\\[5pt]
\tfig{4.4}{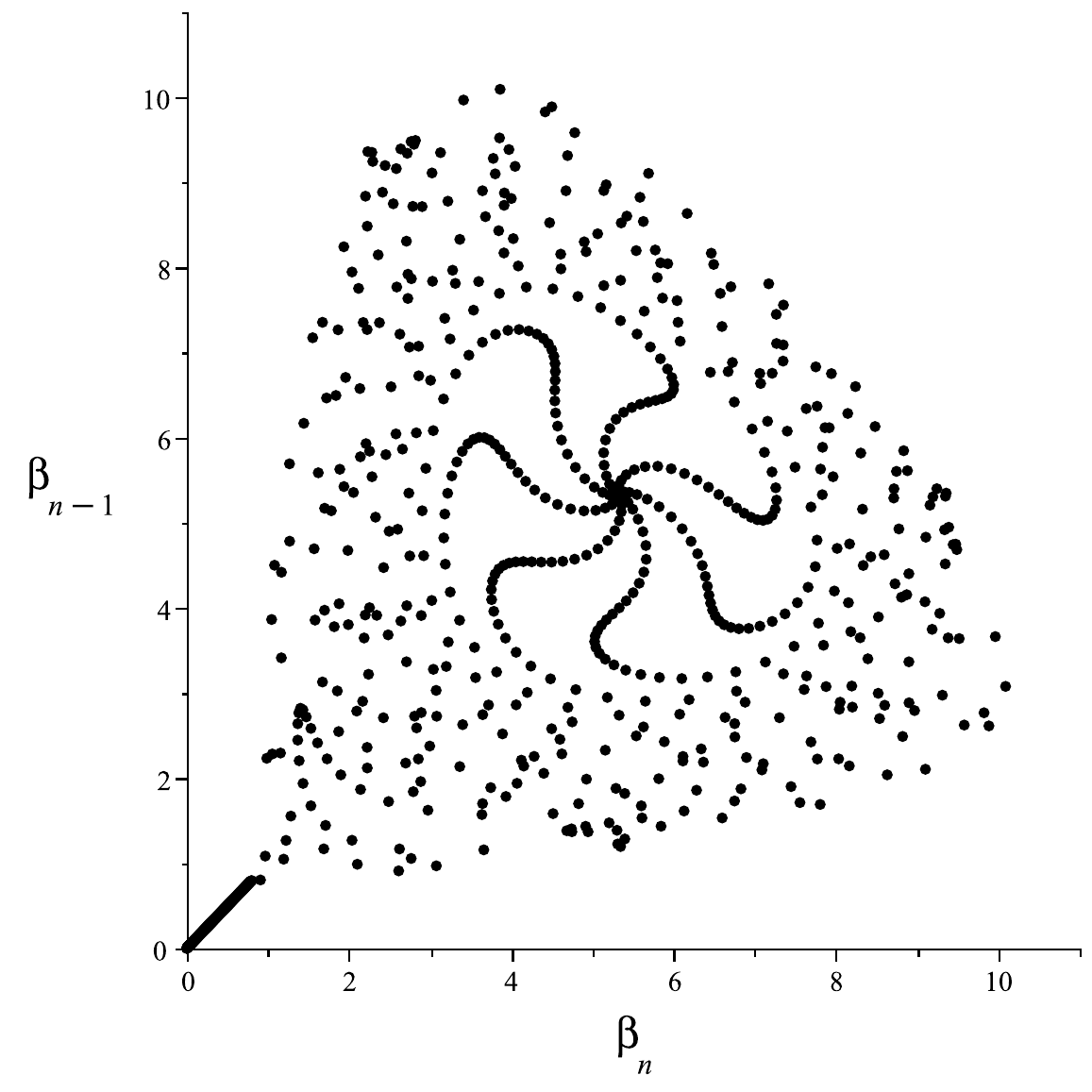} & \tfig{4.4}{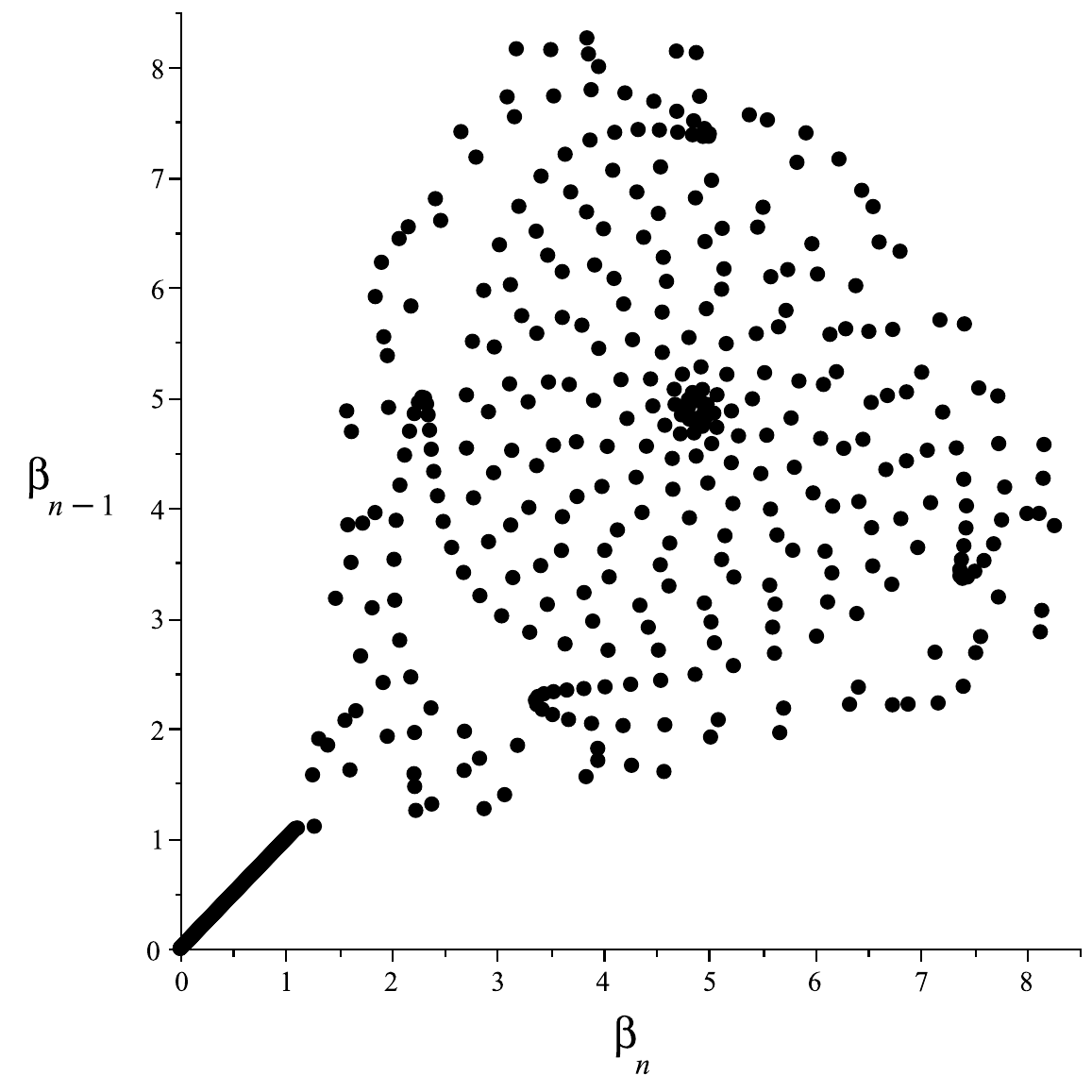} & \tfig{4.4}{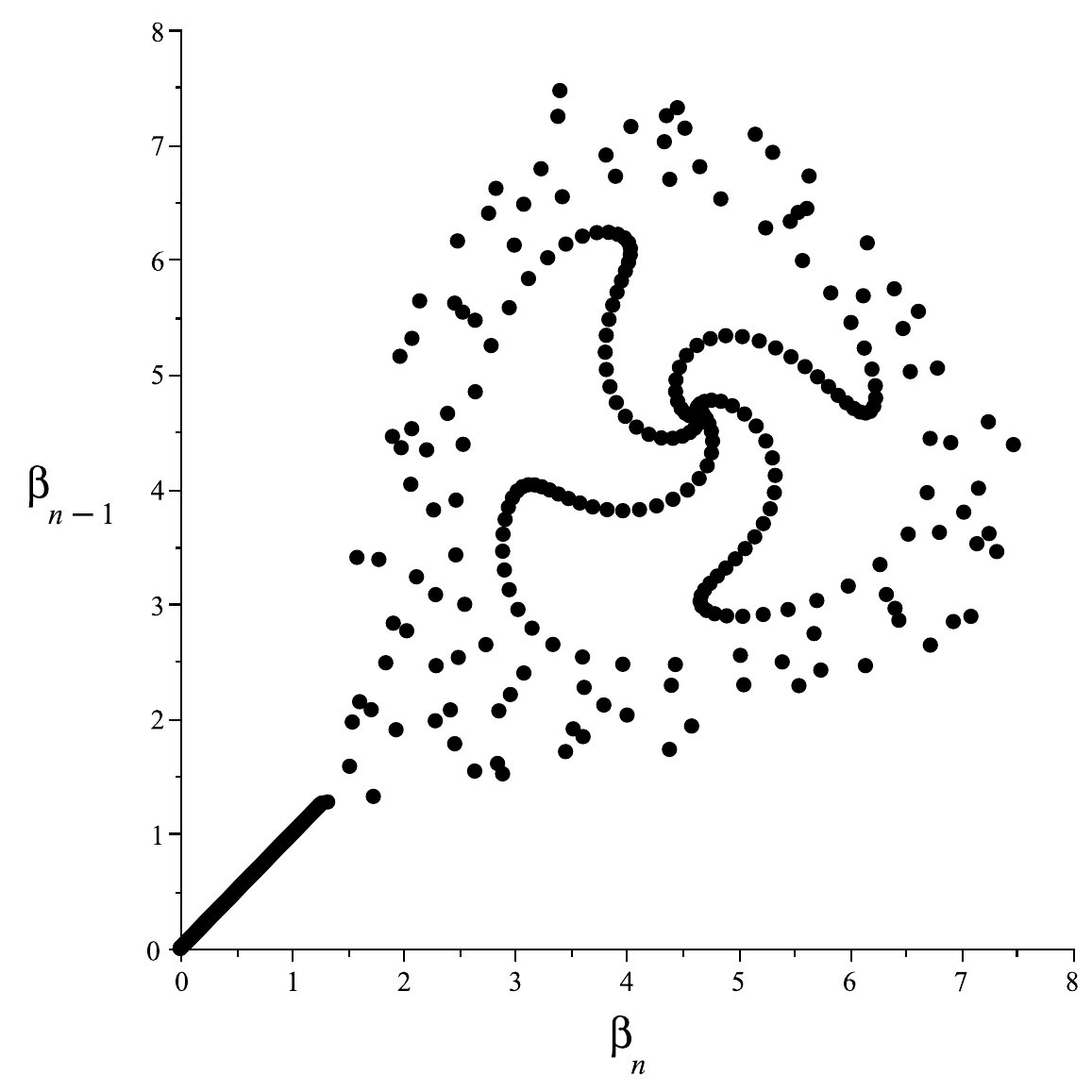} \\
\tau=40, \enskip\k=0.335 &\tau=40, \enskip\k=0.355 &\tau=40, \enskip\k=0.365
\end{array}\]
\caption{\label{Fig81}Plots of $(\b_n,\b_{n-1})$ for $\tau=40$ and $\tfrac{1}{4}<\k<\tfrac{2}{5}$.}
\end{figure}

{In Figure \ref{Fig81} we plot $(\b_n,\b_{n-1})$ for $\tau=40$ and various $\k$ with $\tfrac{1}{4}<\k<\tfrac{2}{5}$. In particular we note that the ``quasi-periodicity" for large $n$ varies for different $\k$. For example, when $\k=0.265$ it is three-fold,
when $\k=0.315$ it is ten-fold,
when $\k=0.335$ it is seven-fold and when $0.365$ it is four-fold. {One can compare with the plots in Figure \ref{Fig817} and observe this claim as $n$ increases.}}

In Figure \ref{Fig82}, we plot $(\b_n,\b_{n-1})$ for $\tau=40$, $\tau=50$ and $\tau=60$ with $\k=\tfrac{1}{3}$, illustrating what happens as $\tau$ increases. The ``quasi-periodic" region, which is seven-fold, in the centre becomes more prominent.
\begin{figure}[ht!]
\[\begin{array}{c@{\qquad}c@{\qquad}c}
\tfig{4.4}{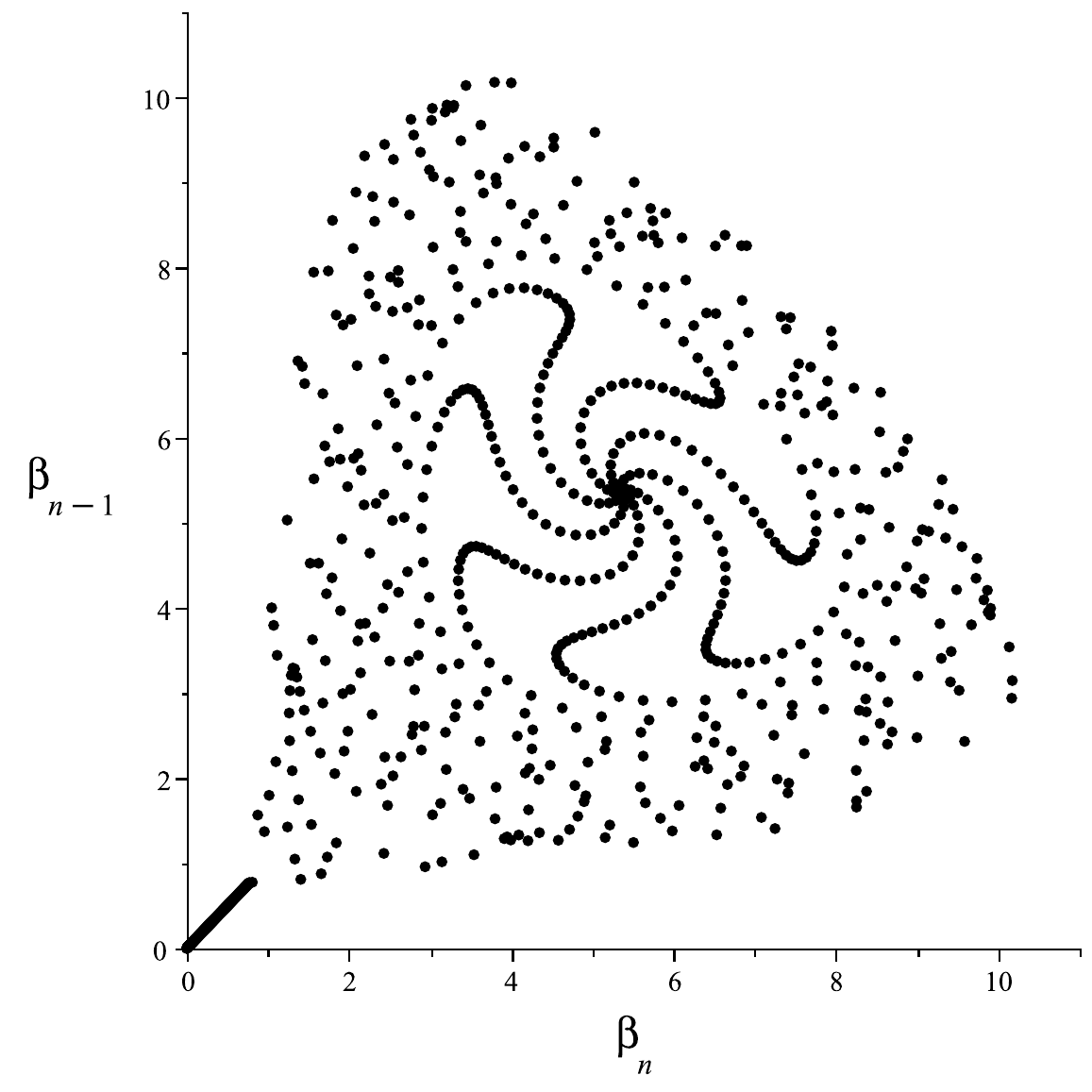} & \tfig{4.4}{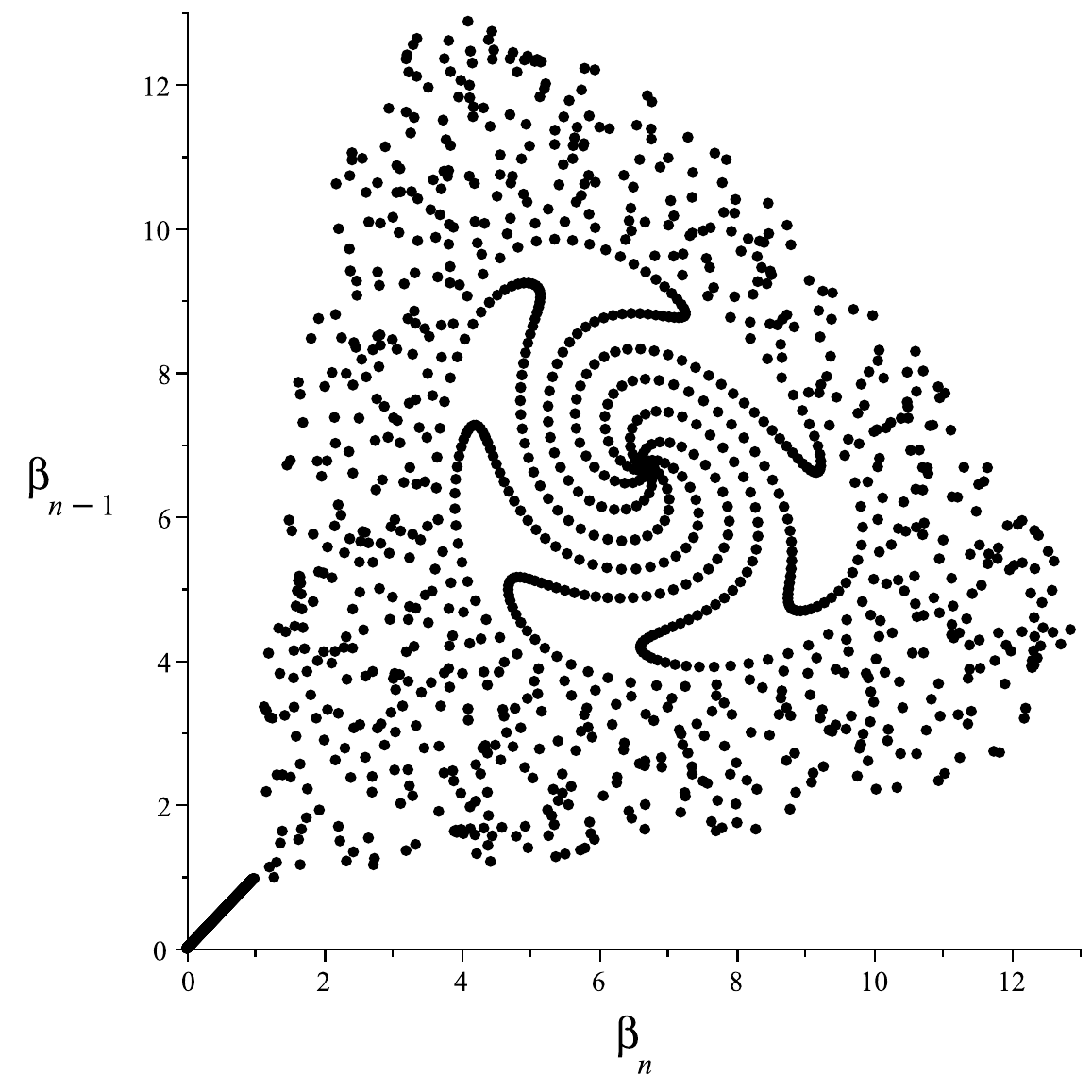} & \tfig{4.4}{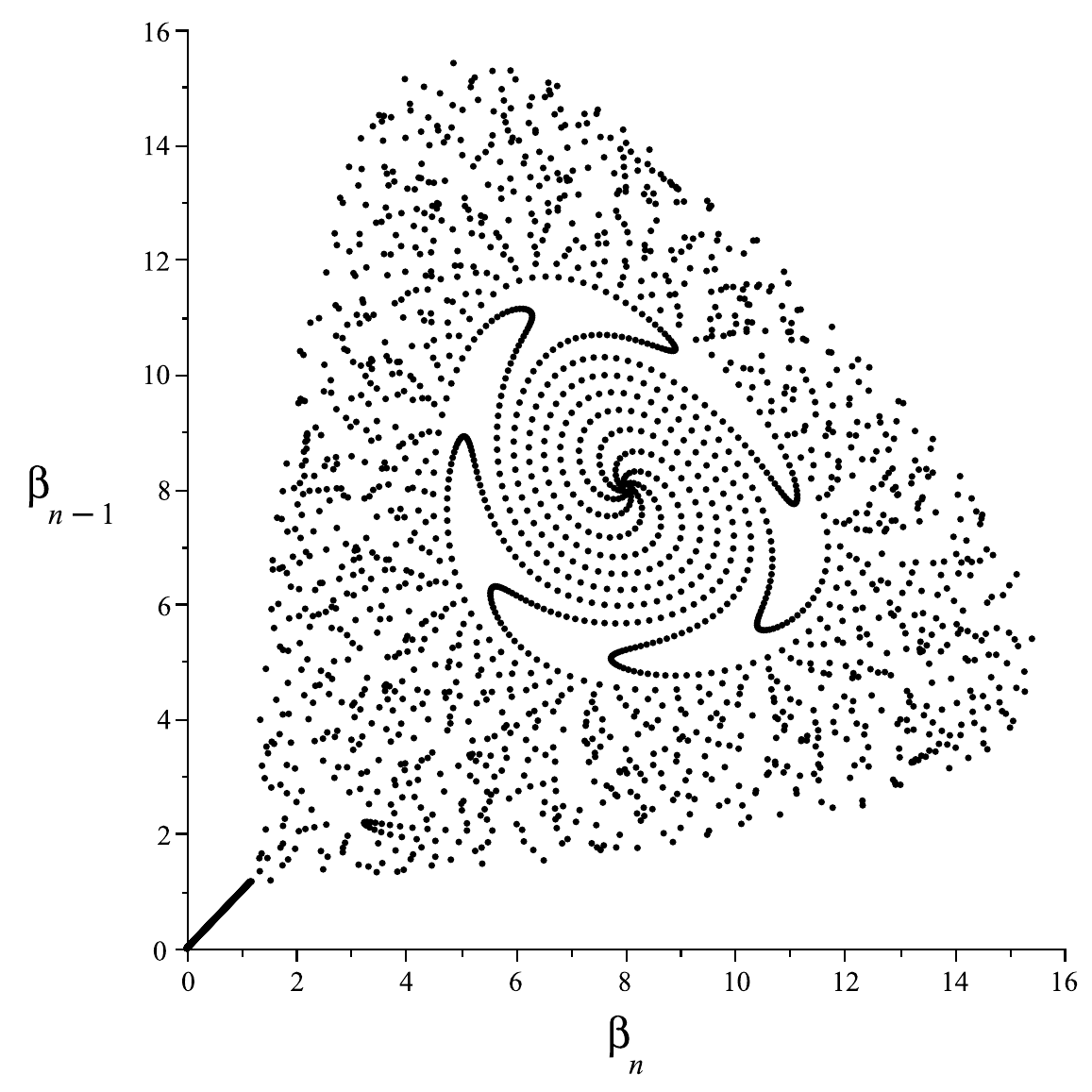} \\
\tau=40, \enskip\k=\tfrac{1}{3} &\tau=50, \enskip\k=\tfrac{1}{3} &\tau=60, \enskip\k=\tfrac{1}{3}
\end{array}\]
\caption{\label{Fig82}Plots of $(\b_n,\b_{n-1})$ for $\tau=40$, $50$, $60$ and $\k=\tfrac{1}{3}$.}
\end{figure}

Subtracting the equations in the system \eqref{UVsys}, and assuming $u\not=v$, yields
\beq 3(u^2+4uv+v^2)-2\tau(u+v)+\k\tau^2=0.\label{UVsys1a}\eeq
In Figure \ref{Fig83}, we plot $(\b_n,\b_{n-1})$ for $\tau=30$ and
$0\leq\k\leq0.1$, and the curve \eqref{UVsys1a}.
In Figure \ref{Fig84} plots of $(\b_n,\b_{n-1})$ for $\tau=40$ and $0.15\leq\k\leq0.24$ are given. As $\k$ increases the portion of the ``triangular region" is increasingly filled.
\begin{figure}[ht!]
\[\begin{array}{c@{\qquad}c@{\qquad}c}
\tfig{4.4}{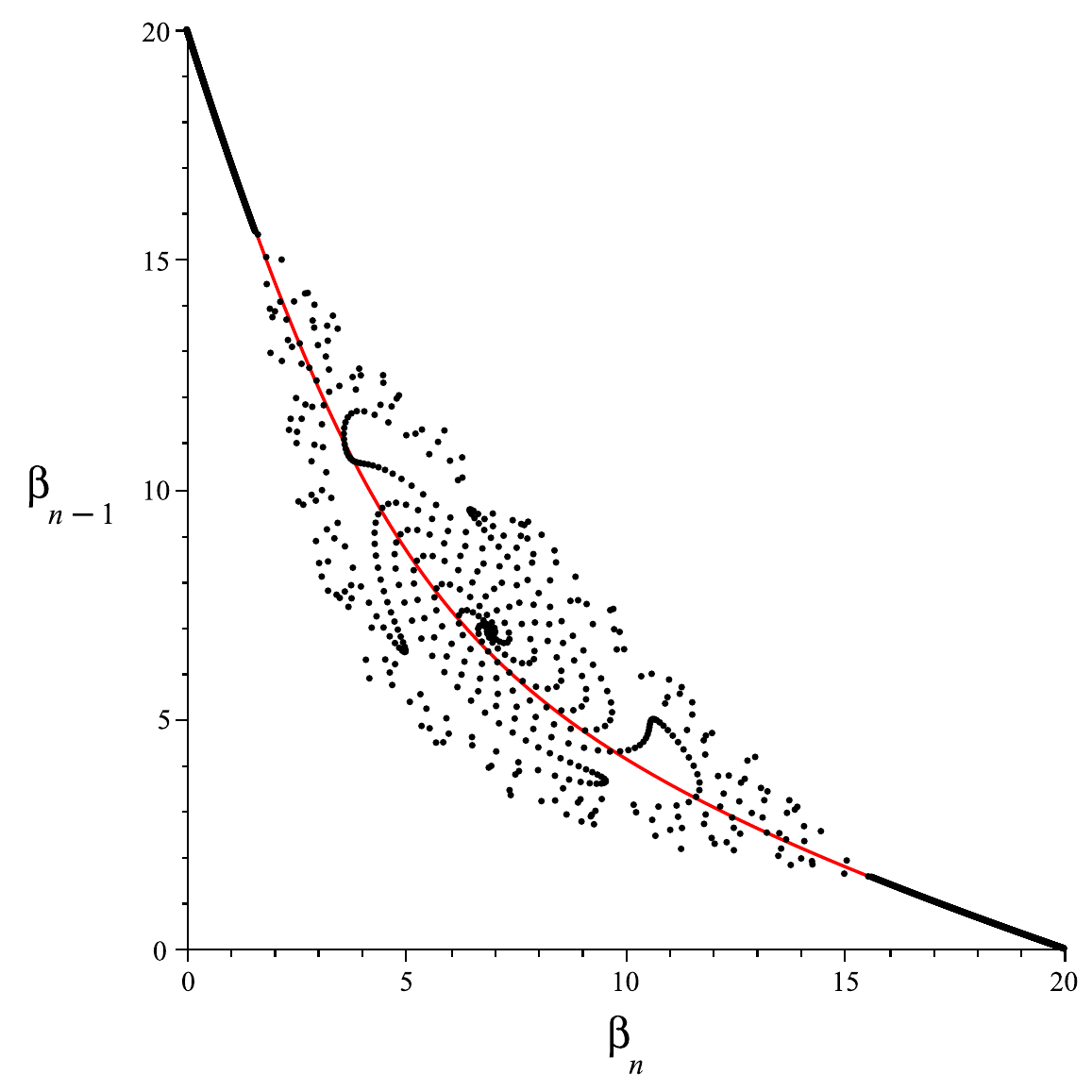} & \tfig{4.4}{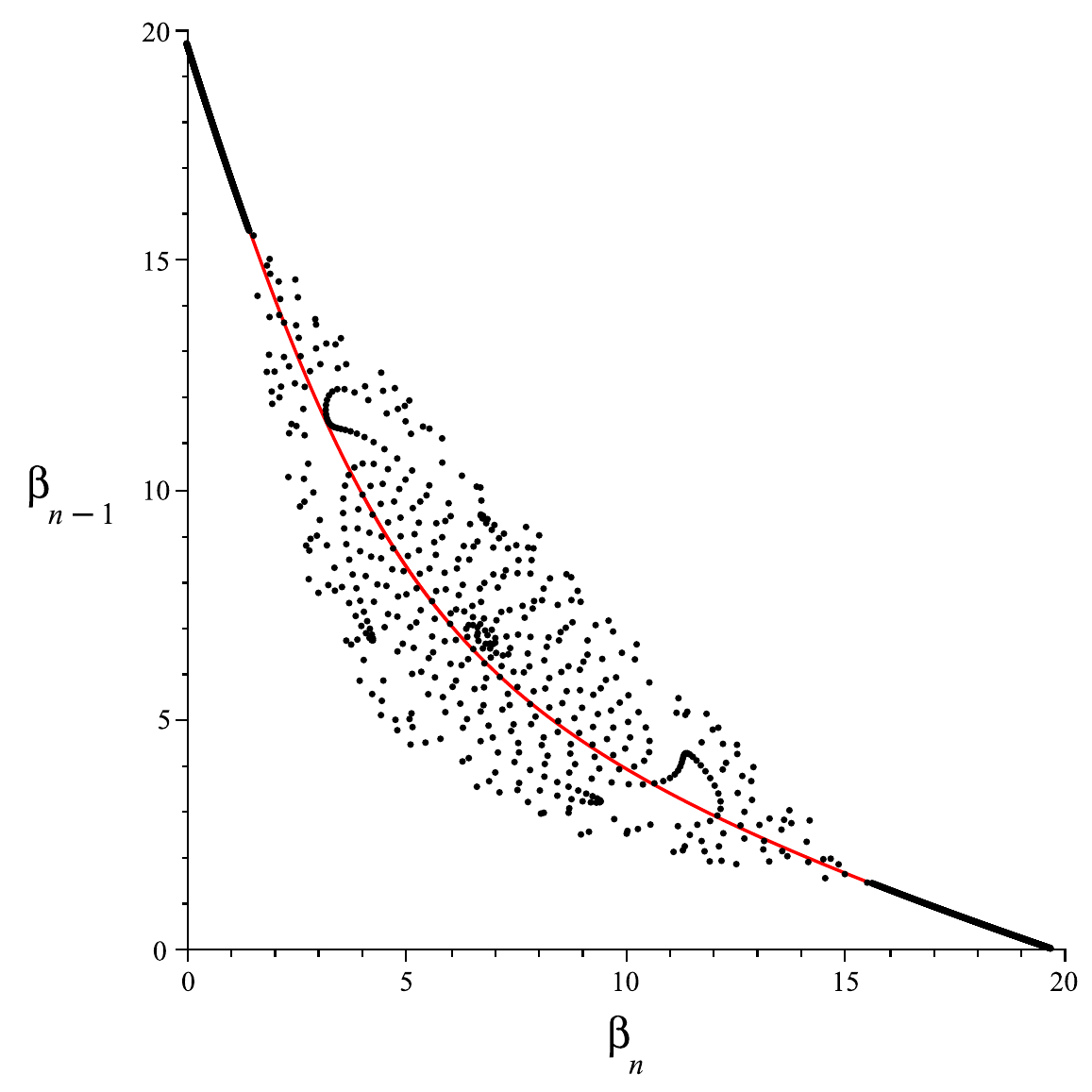} & \tfig{4.4}{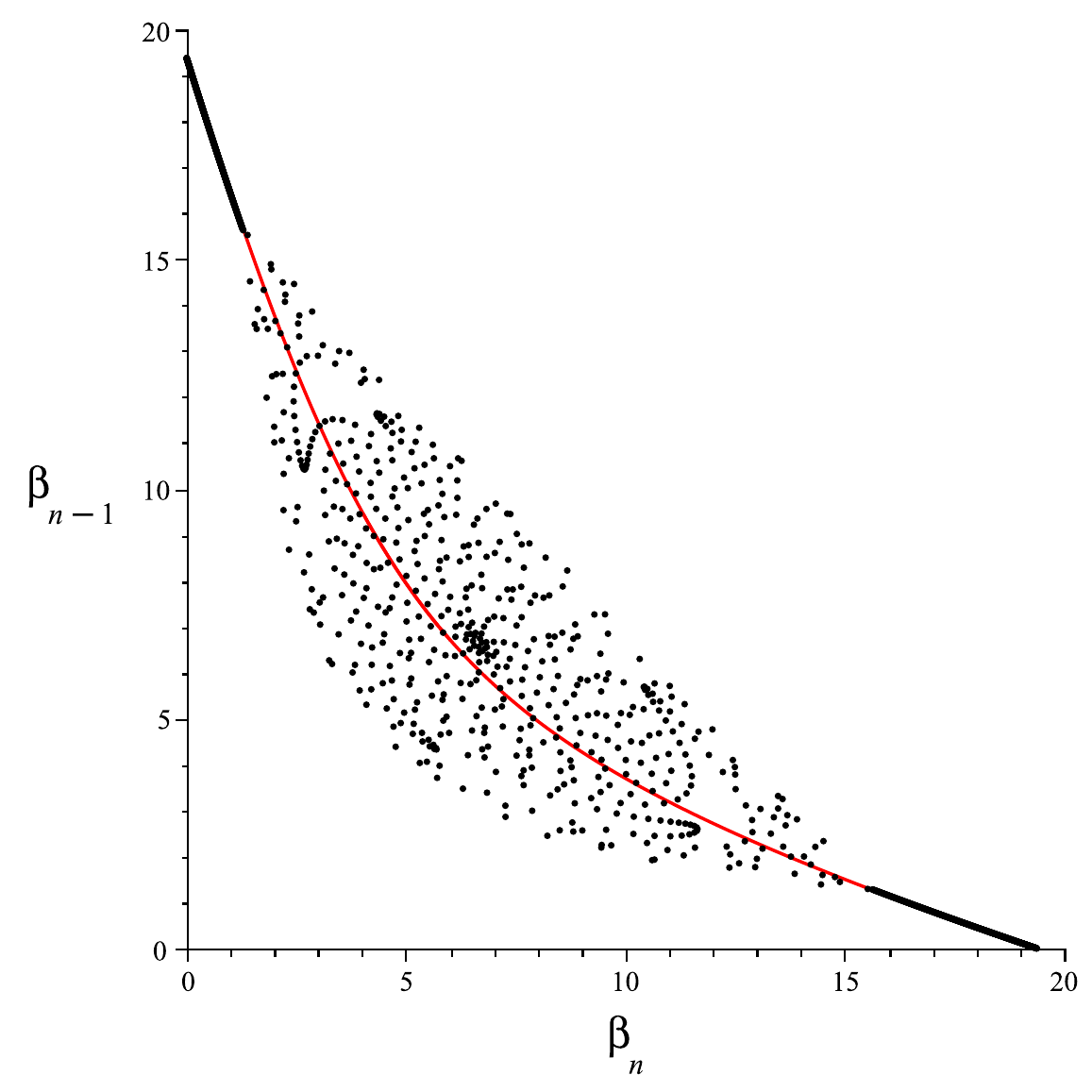}\\
\tau=30, \enskip\k=0 &\tau=30, \enskip\k=0.02 &\tau=30, \enskip\k=0.04 \\[5pt]
\tfig{4.4}{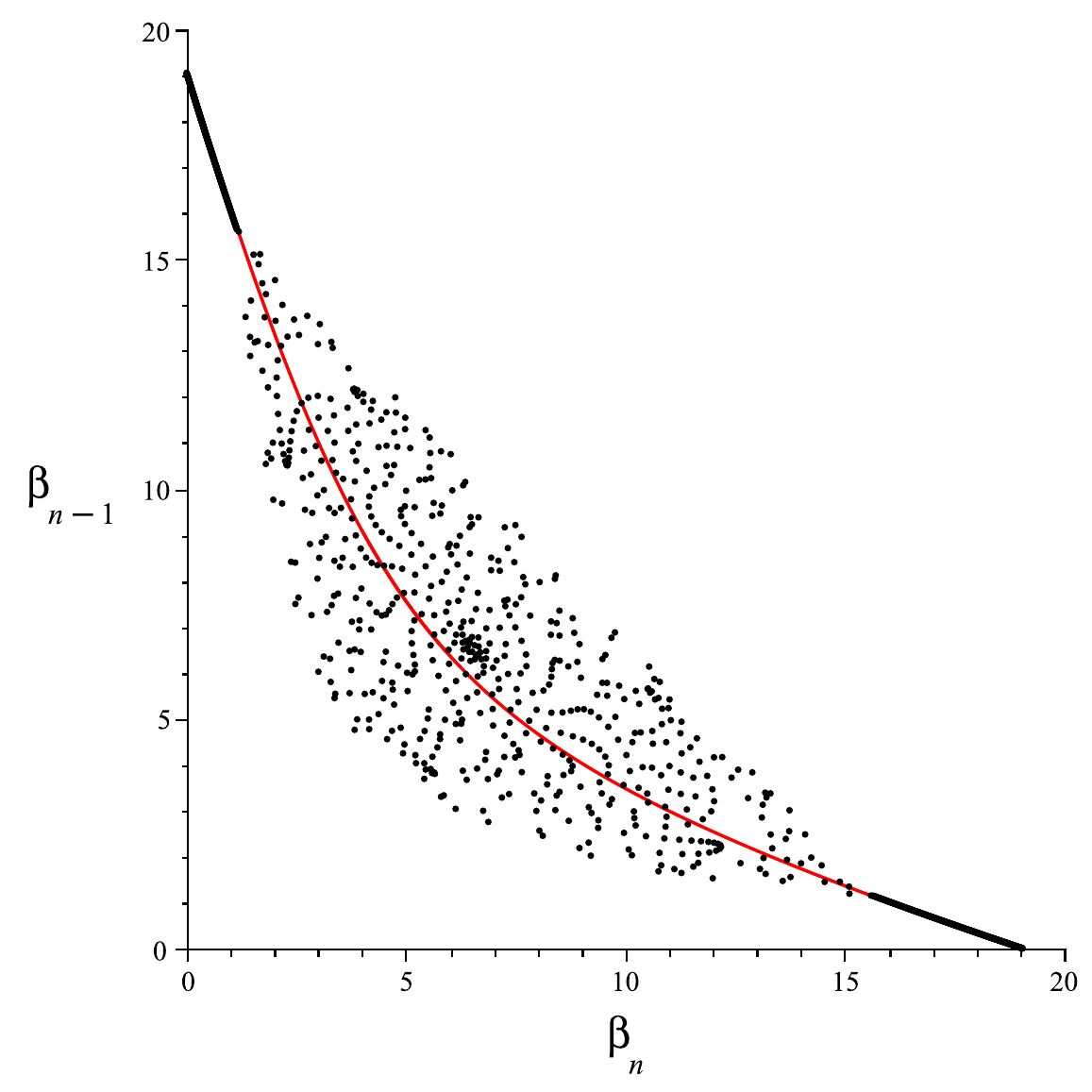} & \tfig{4.4}{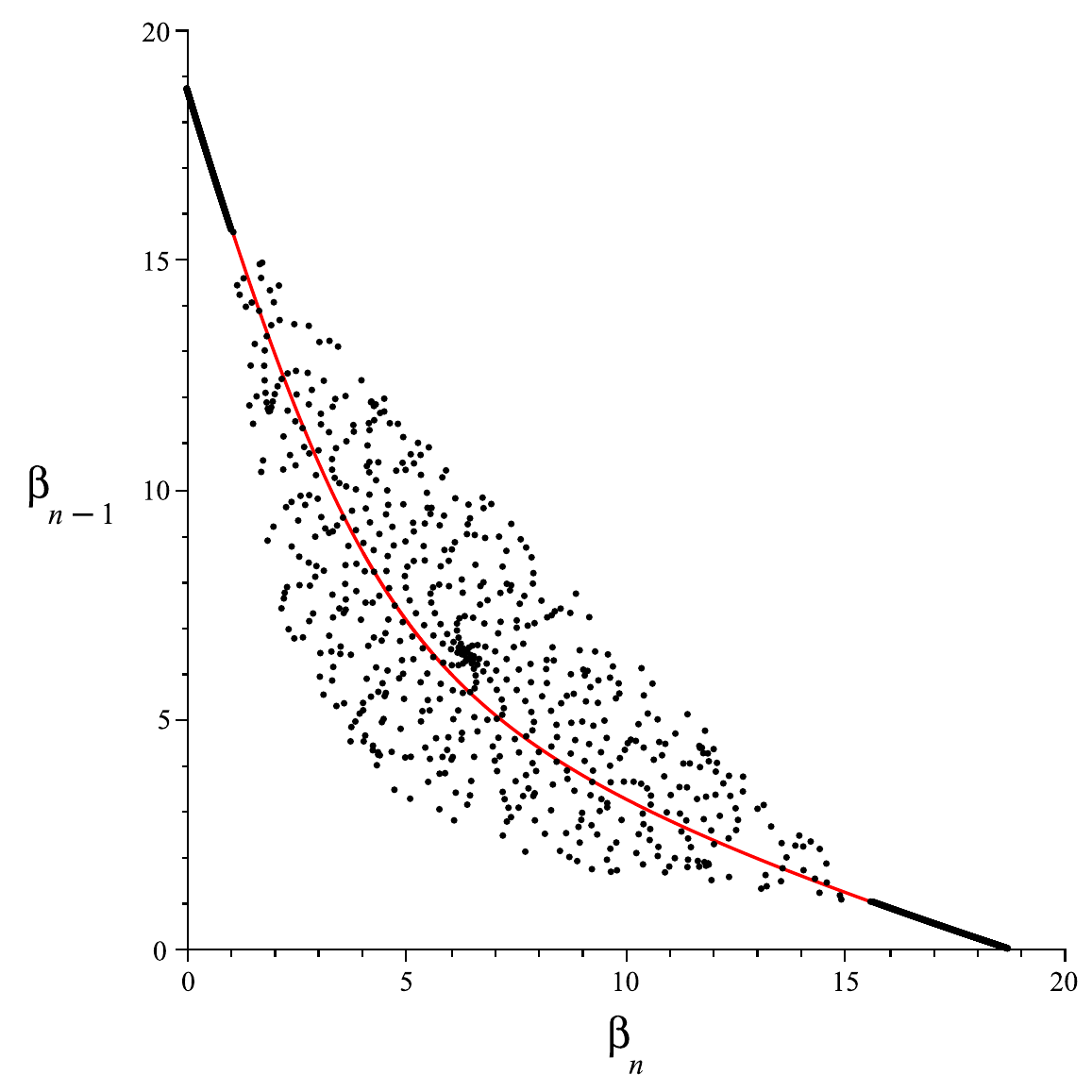} & \tfig{4.4}{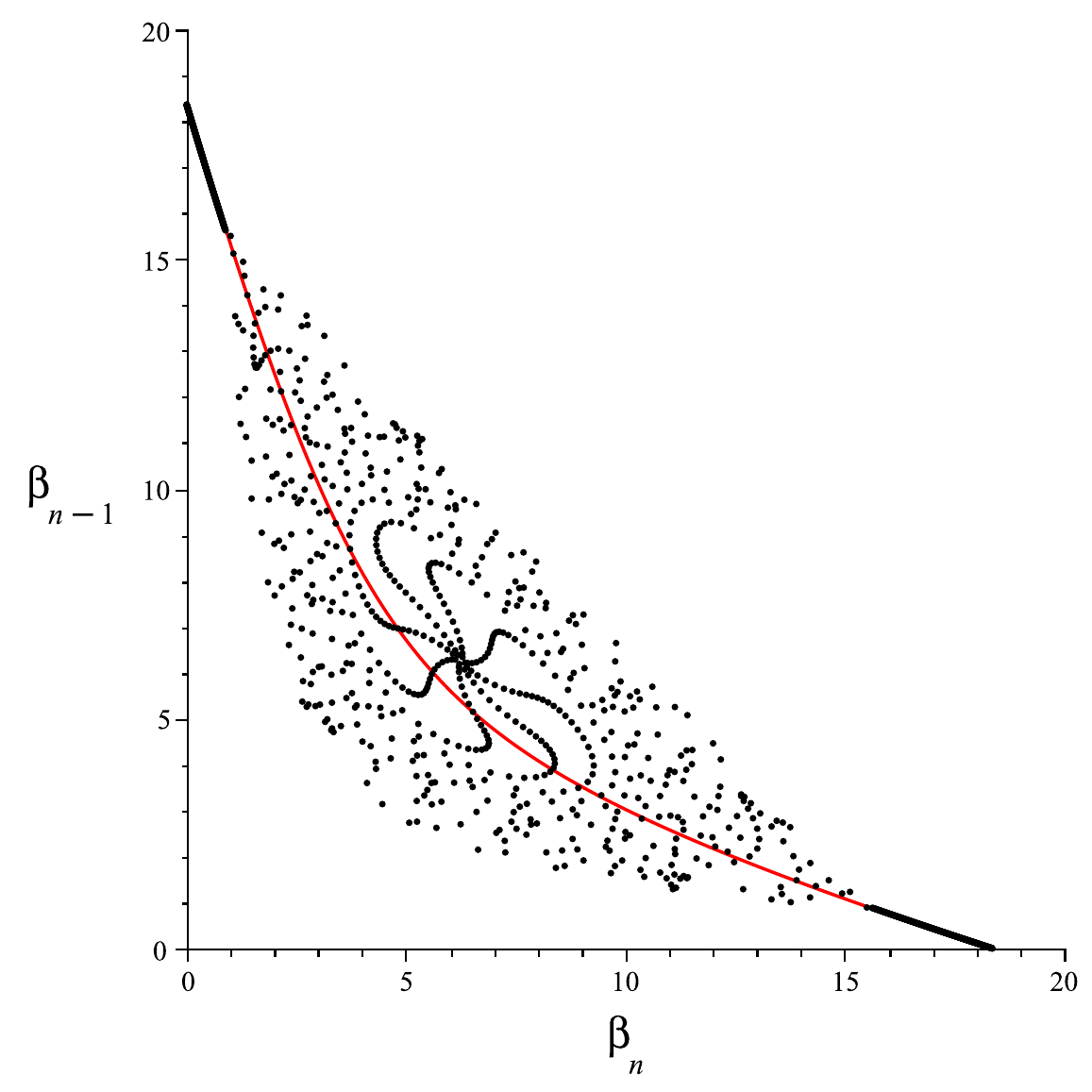}\\
\tau=30, \enskip\k=0.06 &\tau=30, \enskip\k=0.08 &\tau=30, \enskip\k=0.1 
\end{array}\]
\caption{\label{Fig83}Plots of $(\b_n,\b_{n-1})$ for $\tau=30$ and $0\leq\k\leq0.1$ together with the curve \eqref{UVsys1a}, which is the red line.}
\end{figure}

\begin{figure}[ht!]
\[\begin{array}{c@{\qquad}c@{\qquad}c}
 \tfig{4.4}{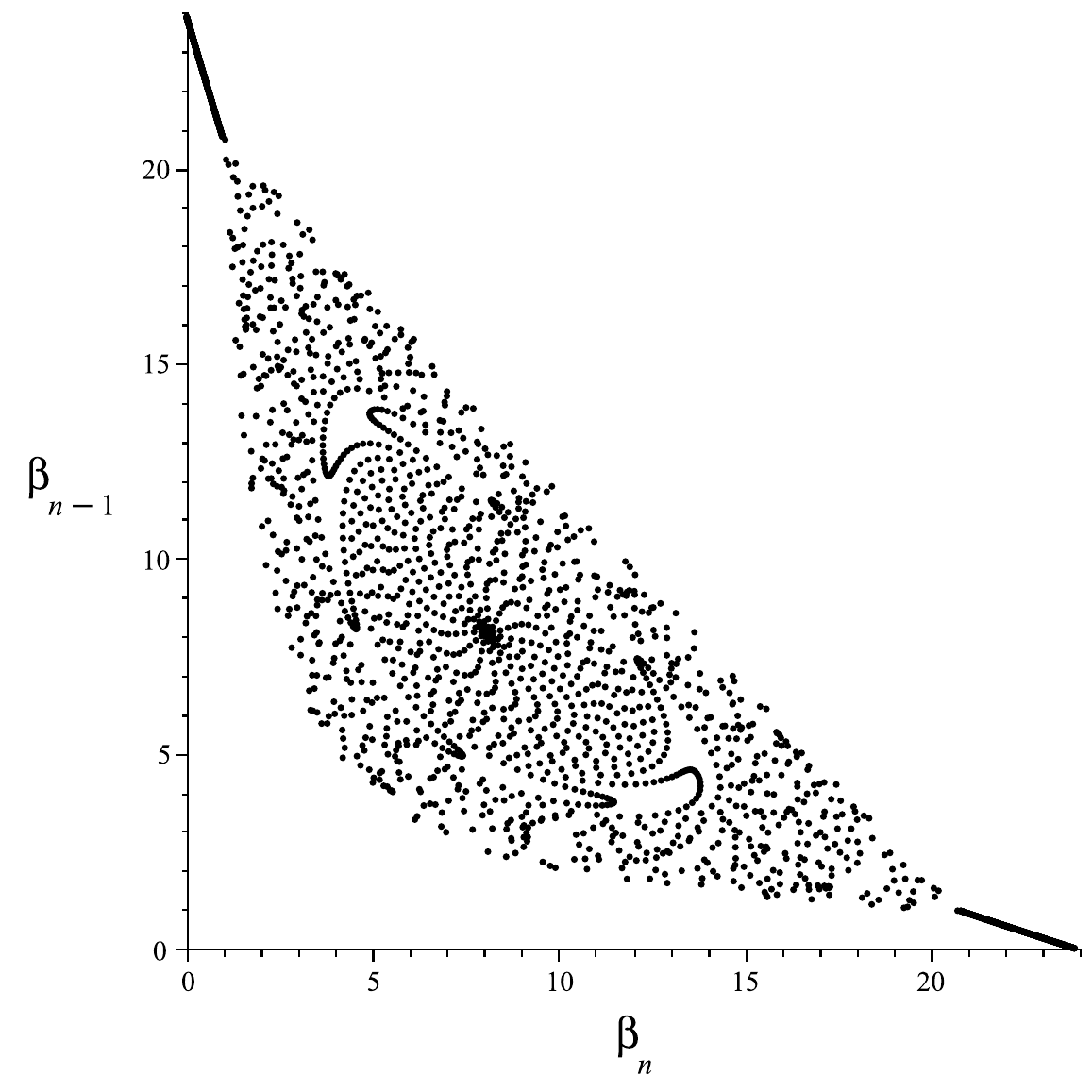} &\tfig{4.4}{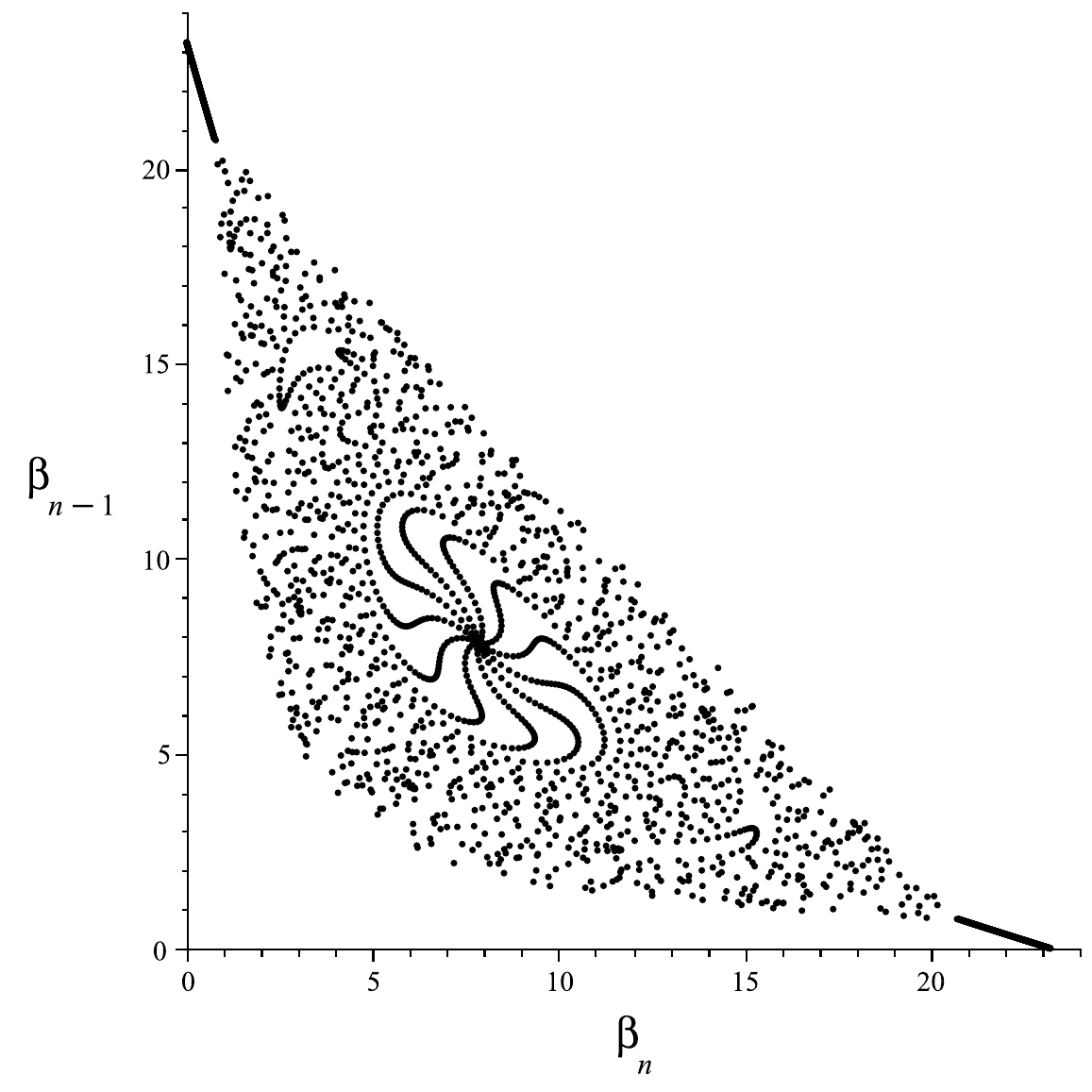} & \tfig{4.4}{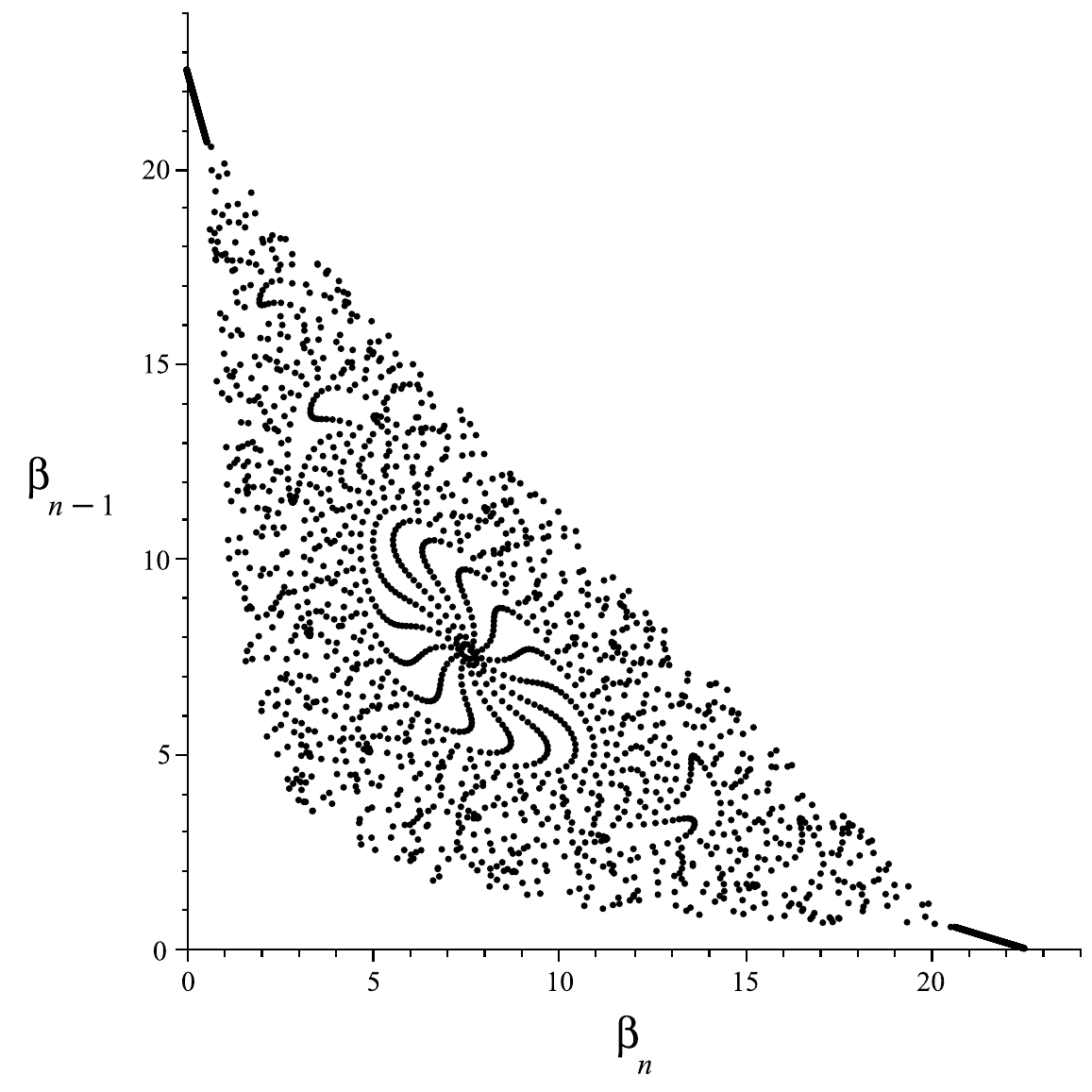} \\
\tau=40,\enskip \k=0.125 &\tau=40,\enskip \k=0.15 &\tau=40,\enskip \k=0.175 \\[5pt]
 \tfig{4.4}{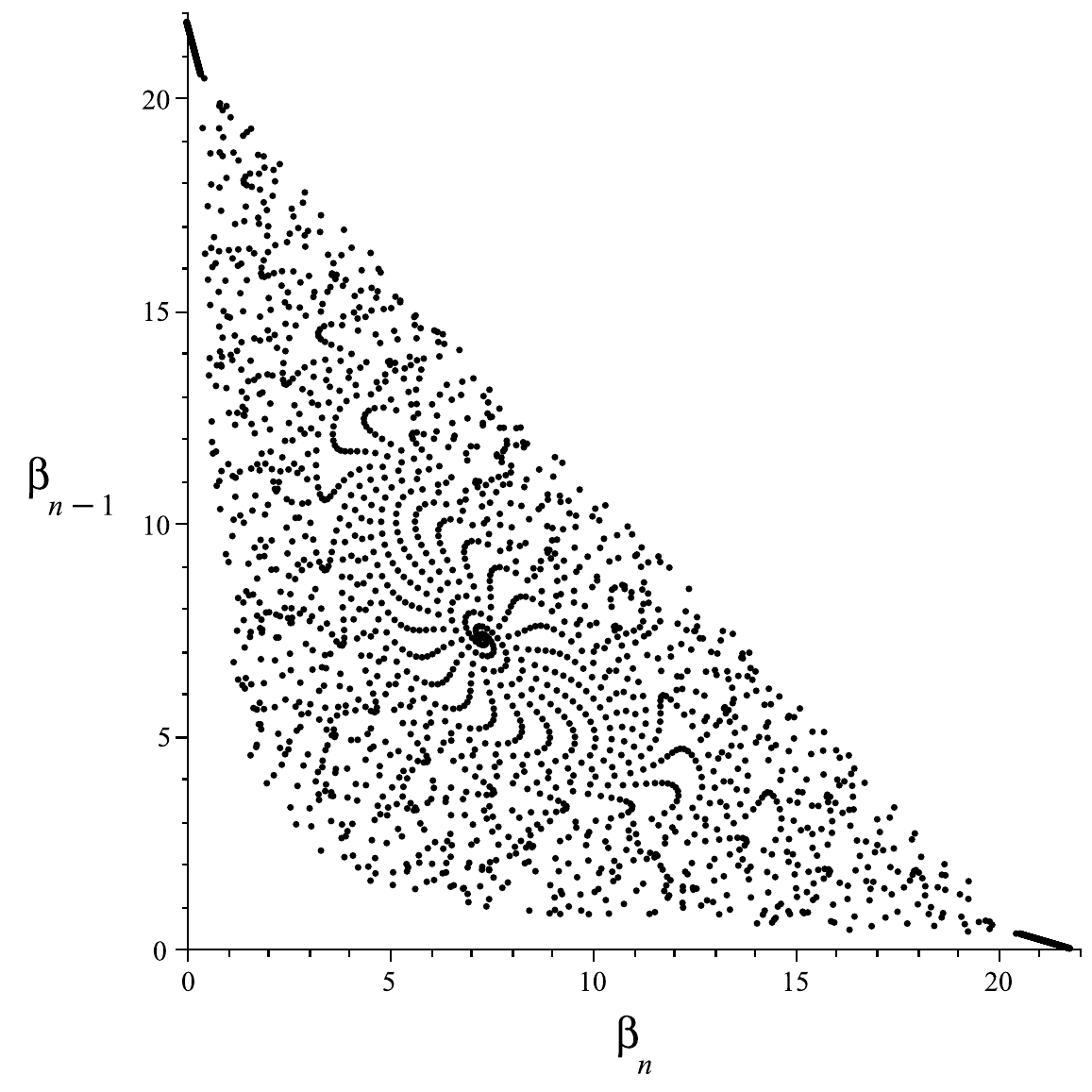} & \tfig{4.4}{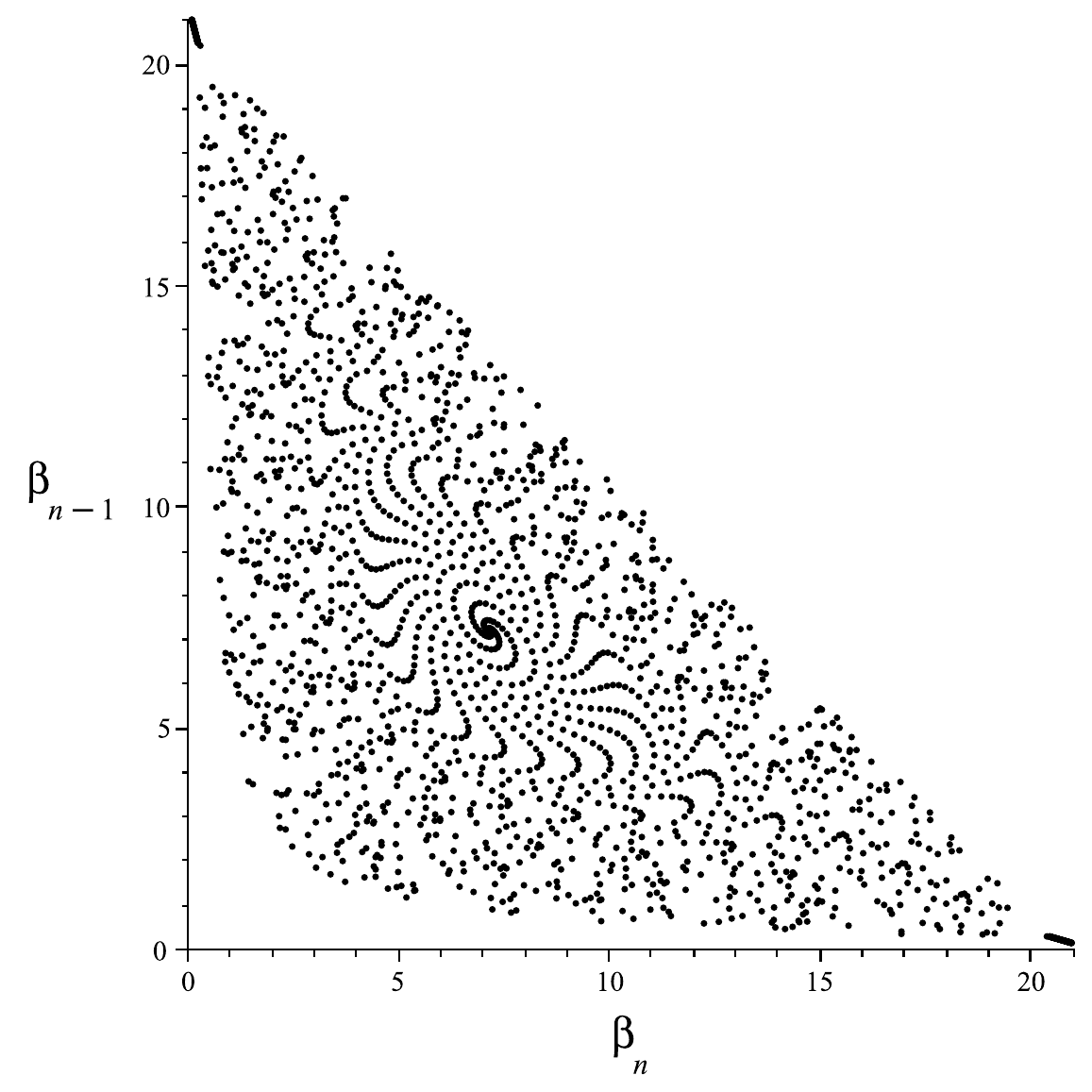} & \tfig{4.4}{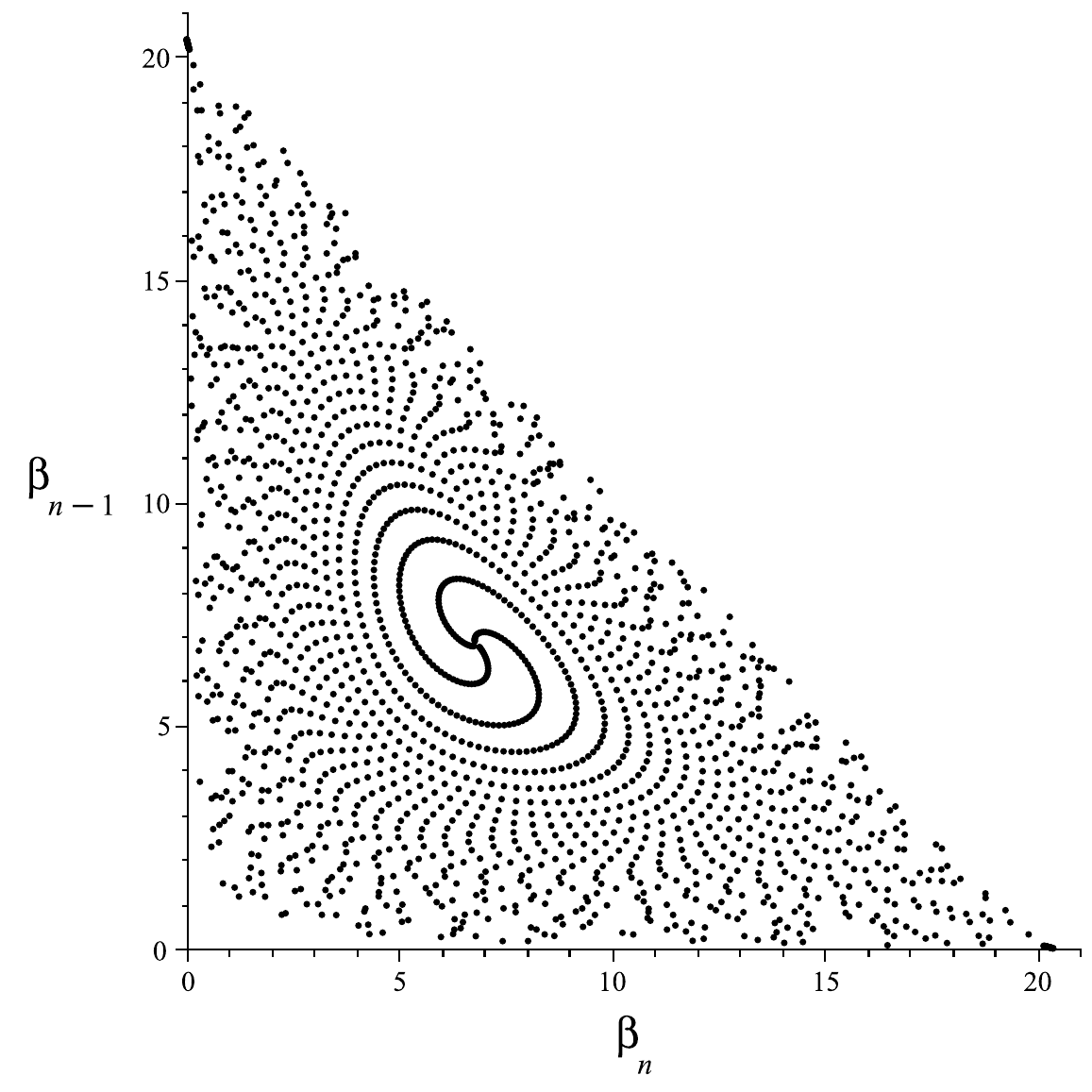}\\
\tau=40,\enskip \k=0.2 &\tau=40,\enskip \k=0.21 &\tau=40,\enskip \k=0.24 
\end{array}\]
\caption{\label{Fig84}Plots of $(\b_n,\b_{n-1})$ for $\tau=40$ and $0.15\leq\k\leq0.24$.}
\end{figure}

\begin{figure}[ht!]
\[\begin{array}{c@{\qquad}c@{\qquad}c}
 \tfig{4.4}{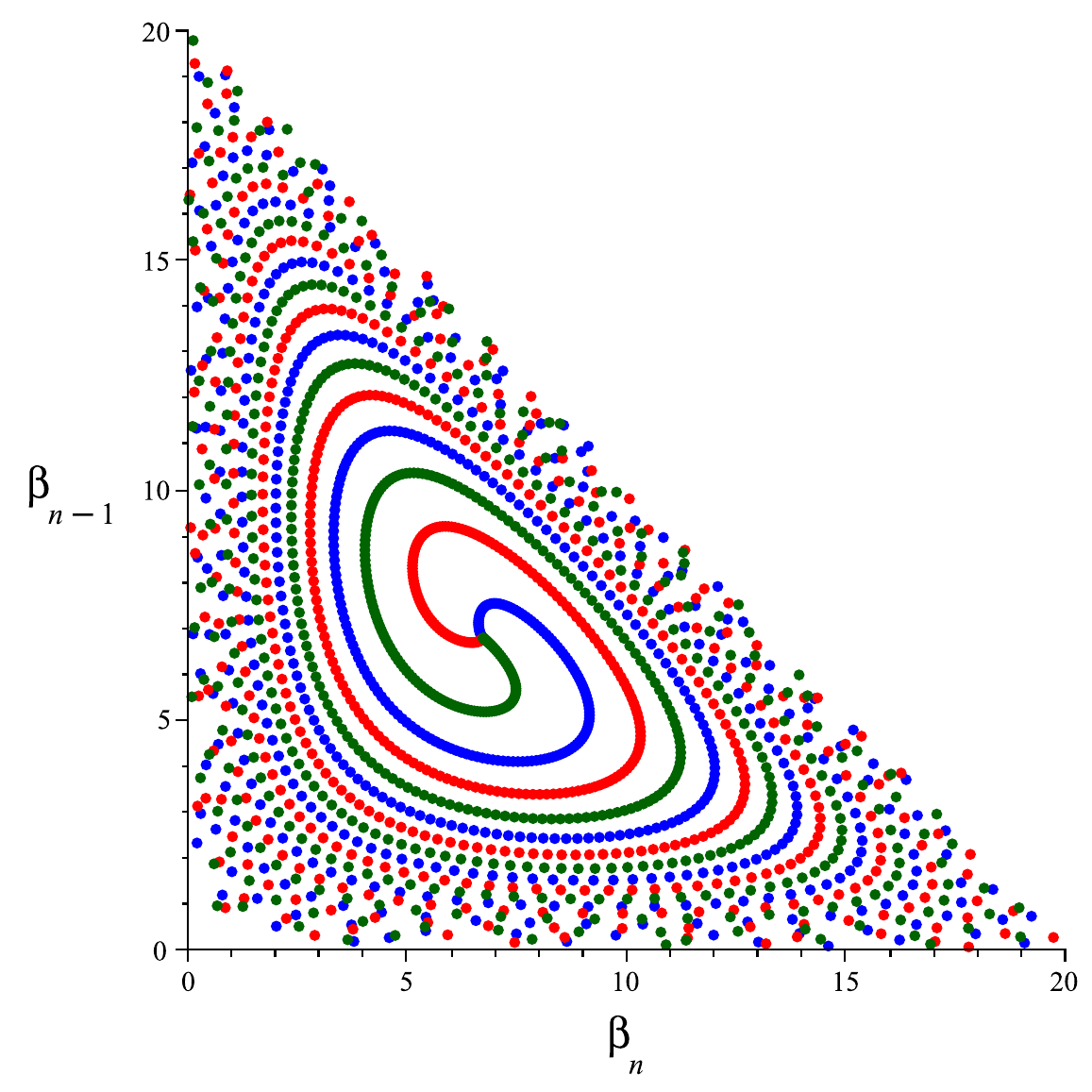} & \tfig{4.4}{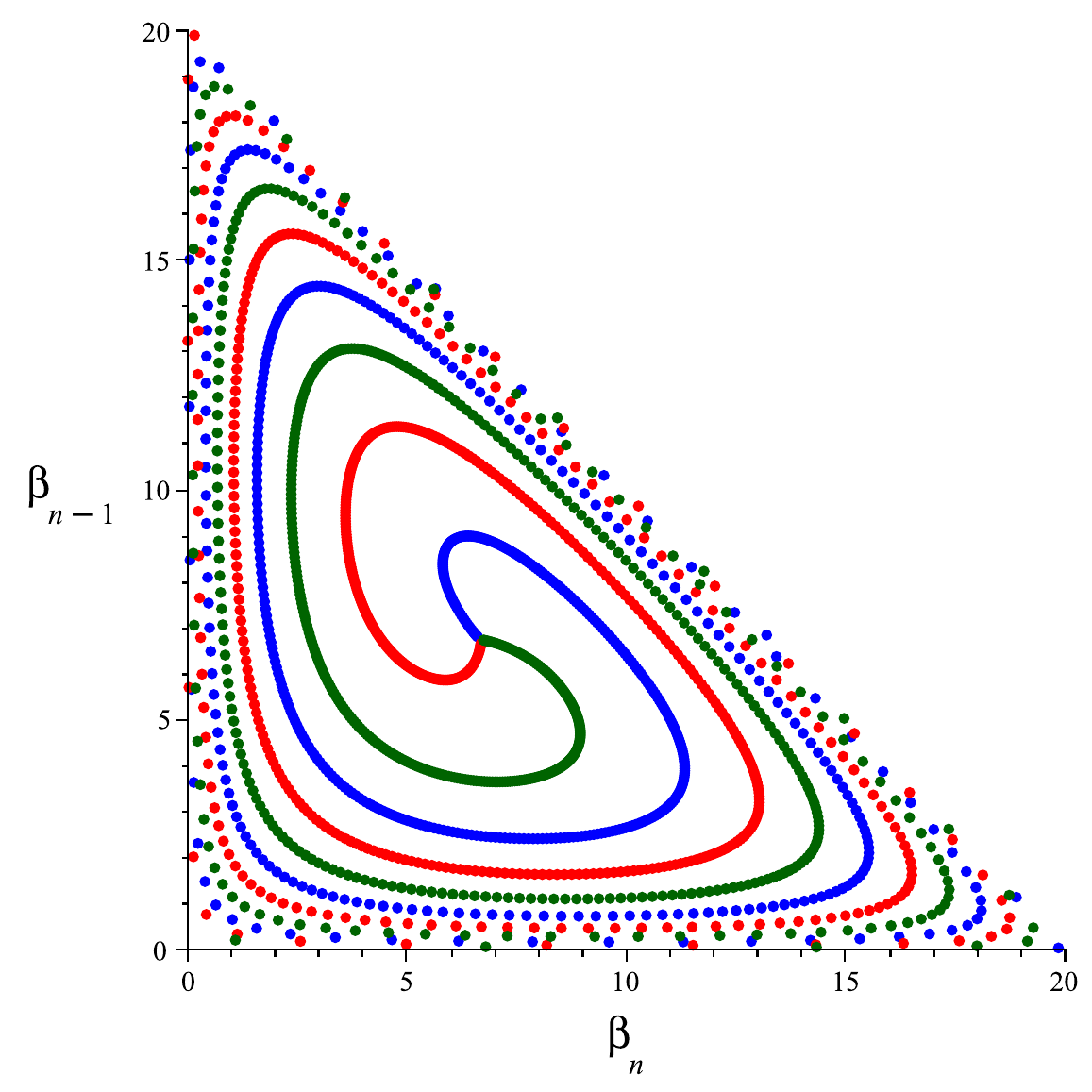} & \tfig{4.4}{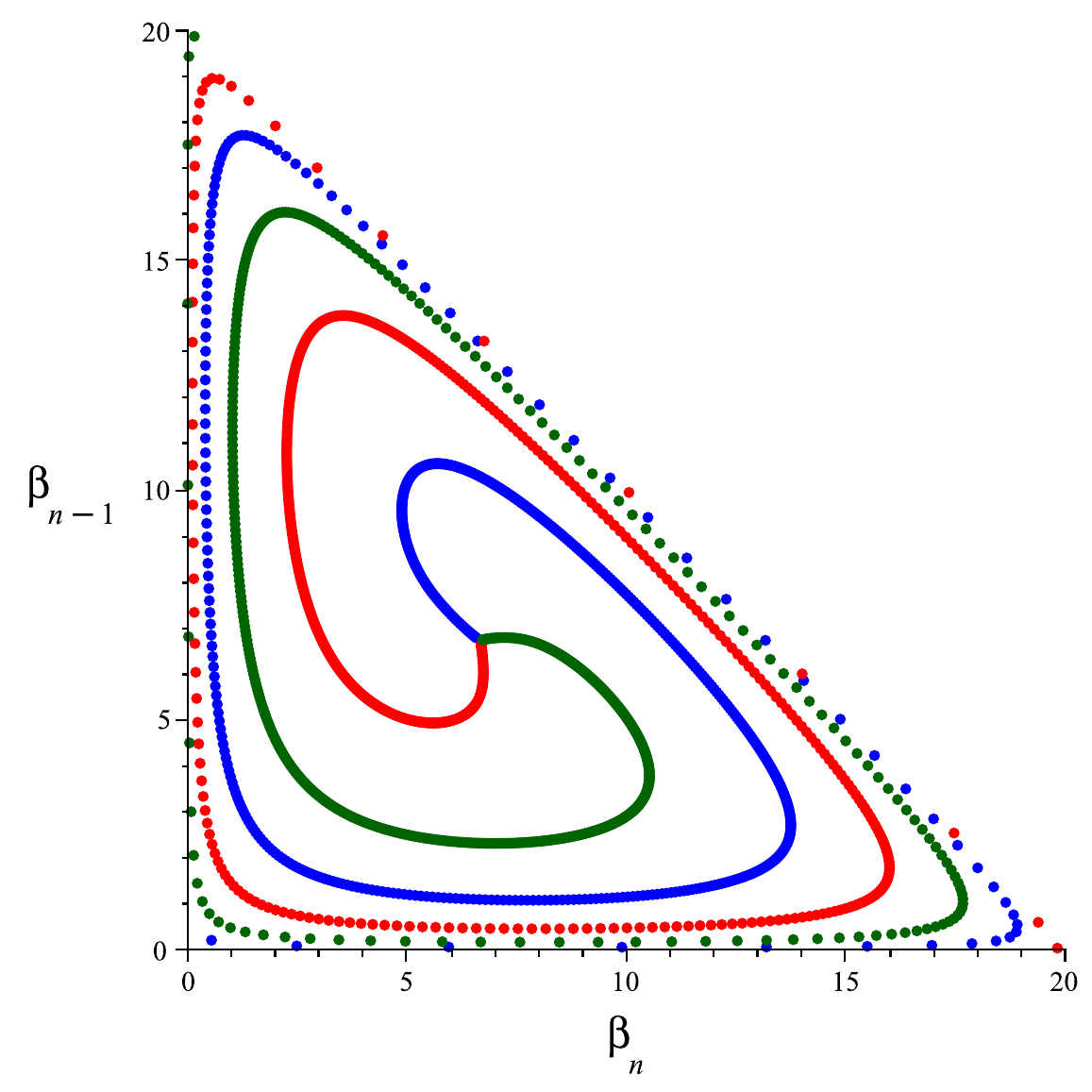}\\
\tau=40,\enskip \k=0.245 &\tau=40,\enskip \k=0.248 &\tau=40,\enskip \k=0.249\\[5pt]
 \tfig{4.4}{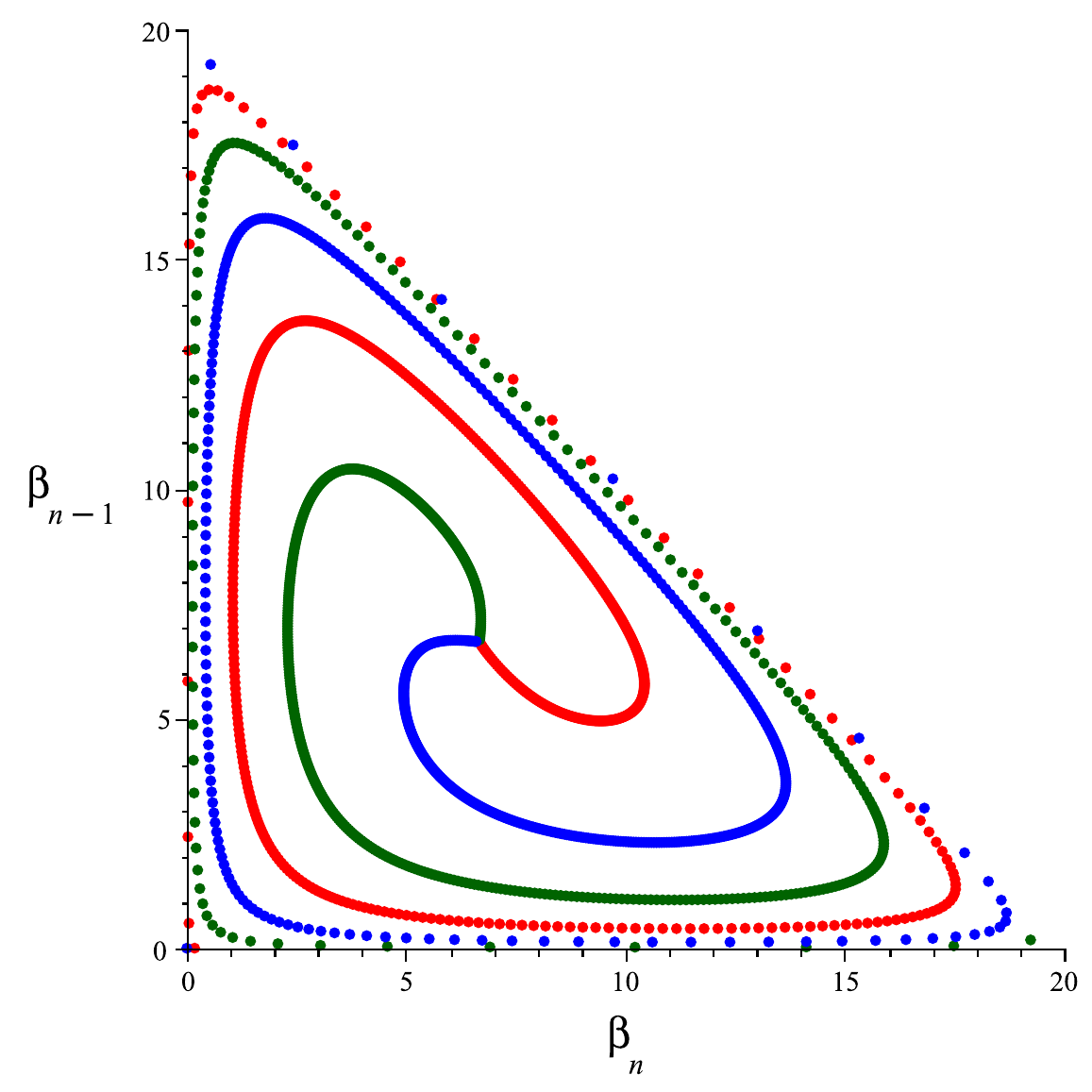} & \tfig{4.4}{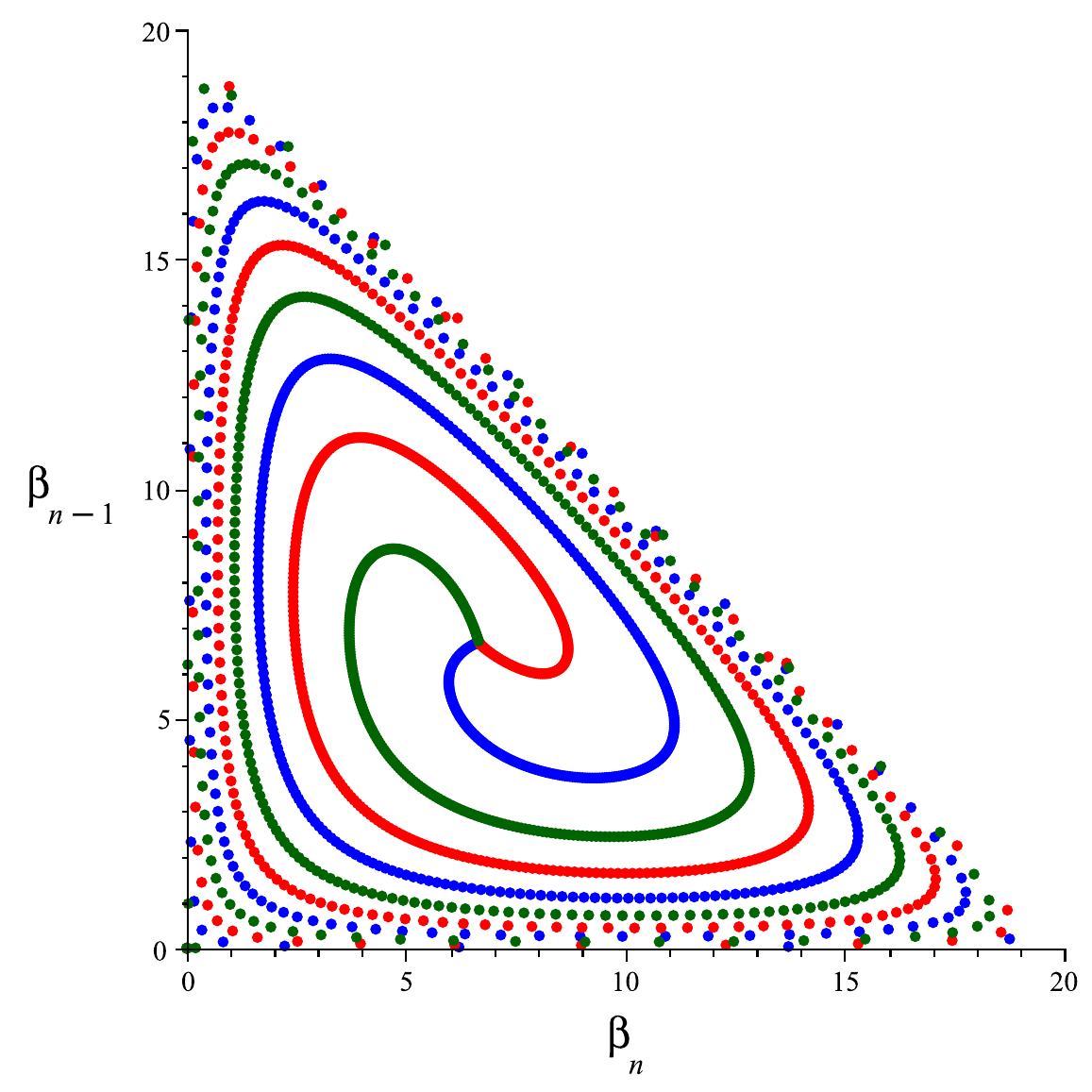} & \tfig{4.4}{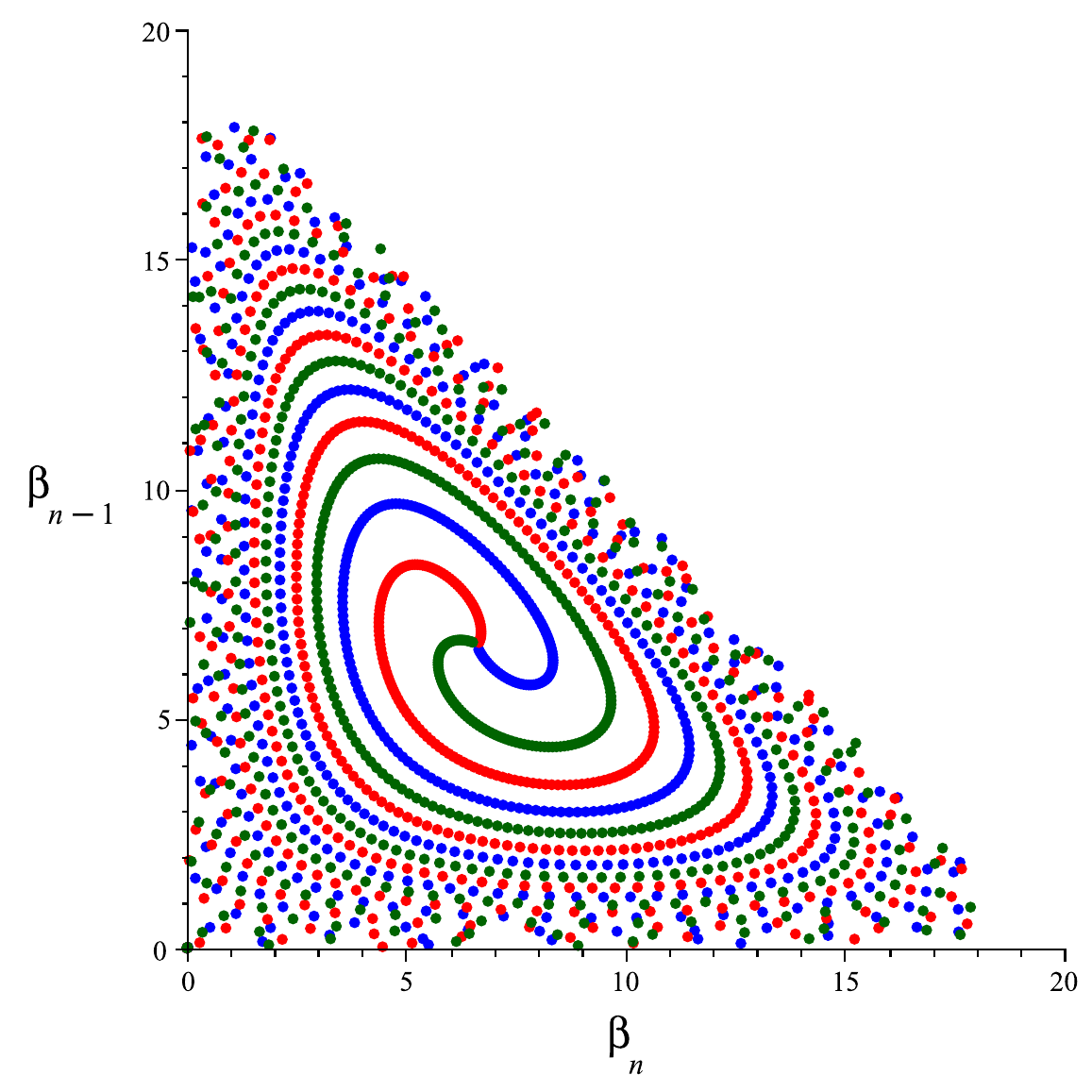}\\
\tau=40,\enskip \k=0.251 &\tau=40,\enskip \k=0.252 &\tau=40,\enskip \k=0.255
\end{array}\]
\caption{\label{Fig85}Plots of $(\b_n,\b_{n-1})$ for $\tau=40$ and $0.245\leq\k\leq0.255$, with $(\b_{3n},b_{3n-1})$ are plotted in \blue{blue}, $(\b_{3n+1},b_{3n})$ in \red{red} and $(\b_{3n+2},\b_{3n+1})$ in \green{green}.} 
\end{figure}
In Figure \ref{Fig85} plots of $(\b_n,\b_{n-1})$ for $\tau=40$ and $0.245\leq\k\leq0.255$ are given, with $(\b_{3n},\b_{3n-1})$ are plotted in \blue{blue}, $(\b_{3n+1},\b_{3n})$ in \red{red} and $(\b_{3n+2},\b_{3n+1})$ in \green{green}. We note that, as in Figure \ref{Fig810}, when $\k$ is close to $\tfrac{1}{4}$, then $\b_{3n+1}$ and $\b_{3n-1}$ essentially are interchanged as $\k$ passes through $\tfrac{1}{4}$. When $\k=\tfrac{1}{4}$ the plot of $(\b_n,\b_{n-1})$ just gives three lines, which seem to be straight and meet at a point. {This is entirely expected and follows directly from the discussion in \S\ref{caseiii}.}

\section{Volterra lattice hierarchy}\label{sec:Volterra}
The \textit{Volterra lattice hierarchy} is given by
\beq \pderiv{\b_n}{t_{2k}}= \b_n\Big(V_{n+1}^{(2k)}-V_{n-1}^{(2k)}\Big),\qquad k=1,2,\ldots,\label{eq:Volhier}\eeq
where $V_n^{(2k)}$ is a combination of various $\b_n$ evaluated at different points on the lattice, see, for example, \cite{refBDM25,refBM20} and the references therein. 

The first three flows $V_n^{(2k)}$ are given by
\begin{subequations}\label{VolHier}\begin{align}
V_{n}^{(2)} &= \b_{n},\\
V_{n}^{(4)} &=V_{n}^{(2)}\left(V_{n+1}^{(2)}+V_{n}^{(2)}+V_{n-1}^{(2)}\right) 
=\b_{n} \left(\b_{n+1}+\b_{n}+\b_{n-1}\right),\\
V_{n}^{(6)} 
 &= V_{n}^{(2)}\left( V_{n+1}^{(4)}+V_{n}^{(4)} + V_{n-1}^{(4)} + V_{n+1}^{(2)}V_{n-1}^{(2)} \right)\nonumber\\ 
&=\b_{n} \left( \b_{n+2}\b_{n+1}+\b_{n+1}^2+2\b_{n+1}\b_{n} + \b_{n}^2+2\b_{n}\b_{n-1} +\b_{n -1}^2+\b_{n-1}\b_{n-2} + \b_{n+1}\b_{n-1}\right).
\end{align}\end{subequations}
Higher order flows $V_n^{(2k)}$, with $k\geq 4$ can be obtained recursively based on the orthogonality of a polynomial sequence with respect to higher order Freud weights, as described in \cite[\S4(a)]{refCJK}. For the sake of illustration the next two flows are: 
\begin{align*}
V_{n\phantom{1}}^{(8)} &= V_{n\phantom{1}}^{(2)}\left( V_{n+1}^{(6)}+V_{n\phantom{1}}^{(6)} + V_{n-1}^{(6)} \right)+V_{n\phantom{1}}^{(4)}V_{n+1}^{(2)}V_{n-1}^{(2)} +V_{n+1}^{(2)}V_{n\phantom{1}}^{(2)}V_{n-1}^{(2)}\left(V_{n+2}^{(2)}+V_{n-2}^{(2)} \right),\\
 V_{n\phantom{1}}^{(10)}
 &= V_{n\phantom{1}}^{(2)} \left(V^{(8)}_{n+1}+V_{n\phantom{1}}^{(8)}+V^{(8)}_{n-1}\right)+V_{n\phantom{1}}^{(6)} V^{(2)}_{n+1} V^{(2)}_{n-1}
 +V^{(2)}_{n+1} V_{n\phantom{1}}^{(2)} V^{(2)}_{n-1} \left(V^{(4)}_{n+2}+V^{(4)}_{n-2}\right)\\
 &\qquad +V^{(2)}_{n+1} V^{(2)}_{n\phantom{1}} V^{(2)}_{n-1} \left\{\left(V^{(2)}_{n\phantom{1}}+V^{(2)}_{n-1}\right)V^{(2)}_{n+2}+\left(V^{(2)}_{n+1}+V^{(2)}_{n\phantom{1}}\right) V^{(2)}_{n-2}+V^{(2)}_{n+2} V^{(2)}_{n-2} \right\}.
\end{align*}
In fact, the Volterra lattice has an infinite hierarchy of commuting symmetries, as explained in \cite[\S2]{refCMW}. 

\begin{remark}{\rm
The discrete equation \eqref{eq:rr642} satisfied by the $\b_n$
can be written as
\[6V_n^{(6)}-4\tau V_n^{(4)} -2 t V_n^{(2)} = n,\]
and $\b_n$ satisfies the differential-difference equations
\begin{align*} \pderiv{\b_n}{t}&=\b_n \Big(V_{n+1}^{(2)}-V_{n-1}^{(2)}\Big)\\&=\b_n(\b_{n+1}-\b_{n-1}),\\[5pt] \pderiv{\b_n}{\tau}&=\b_n \Big(V_{n+1}^{(4)}-V_{n-1}^{(4)}\Big)\\&=\b_n\left[(\b_{n+2}+\b_{n+1}+\b_n)\b_{n+1} - (\b_{n}+\b_{n-1}+\b_{n-2})\b_{n-1}\right],
\end{align*}
recall \eqref{eq:langlat} and \eqref{eq:langlat4} in Lemma \ref{lem:34}, which are the first two equations in the Volterra hierarchy \eqref{eq:Volhier}.}
\end{remark}

Plots of $V_n^{(2)}(\tau;\k)$, $V_n^{(4)}(\tau;\k)$ and $V_n^{(6)}(\tau;\k)$, given by \eqref{VolHier}, are in Figure \ref{VolHier1} for $\tau=50$ and $\tfrac{1}{4}<\k<\tfrac{2}{5}$ 
and in Figure \ref{VolHier2}
for $\tau=40$ and $0<\k<\tfrac{1}{4}$.
These show that the plots of $V_n^{(2)}(\tau;\k)$, $V_n^{(4)}(\tau;\k)$ and $V_n^{(6)}(\tau;\k)$ for fixed $\tau$ and $\k$ have a very similar structure. The main difference is the asymptotics as $n\to\infty$ since 
from Lemma \ref{betan_asynp}, with $t=-\tau^2\k$, we have following formal series expansions
\[V_n^{(2)}(\tau;\k)=\b_n=\frac{n^{1/3}}{\ga}+\frac{\tau}{15}+\frac{(2-5\k)\tau^2\,\ga}{450\,n^{1/3}}
+\frac{2(4-15\k)\tau^3}{675\ga\,n^{2/3}}+\O\big(n^{-4/3}\big),\qquad\text{as}\quad n\to\infty,\]
with $\ga=\sqrt[3]{60}$, and so
\begin{align*}
V_n^{(4)}(\tau;\k)&\approx 3\b_n^2=\frac{\ga\,n^{2/3}}{20} +\frac{2\tau n^{1/3}}{5\,\ga}+\frac{(3-5\tau)\tau^2}{75}+\frac{2(1-3\k)\ga\tau^3}{675\,n^{1/3}}+\O\big(n^{-2/3}\big),\\ 
V_n^{(6)}(\tau;\k)&\approx 10\b_n^3=\frac{n}{6}+\frac{\tau\ga\,n^{2/3}}{30}+\frac{(4-5\k)\tau^2n^{1/3}}{15\,\ga}+\frac{(2-5\k)\tau^3}{75}+\O\big(n^{-1/3}\big),
\end{align*}
as $n\to\infty$.

\begin{figure}[ht!]
\[\begin{array}{c@{\quad}cc@{\quad}c}
V_n^{(2)} &V_n^{(4)} &V_n^{(6)} \\
\fig{2}{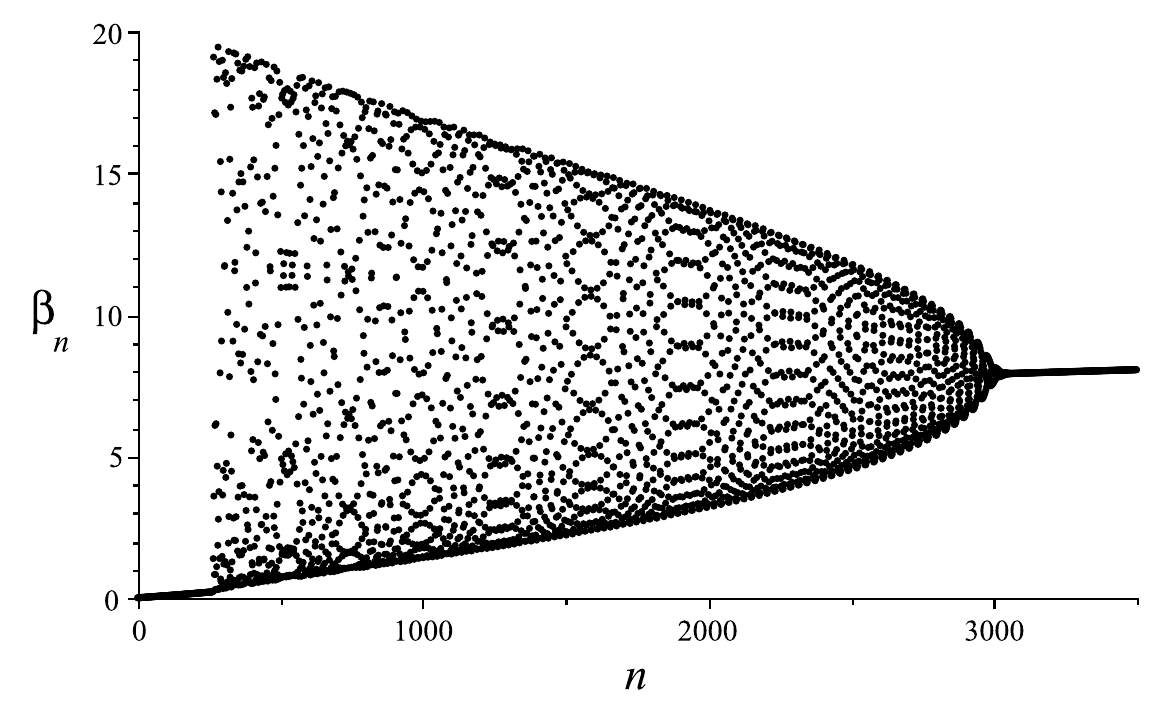} & \fig{2}{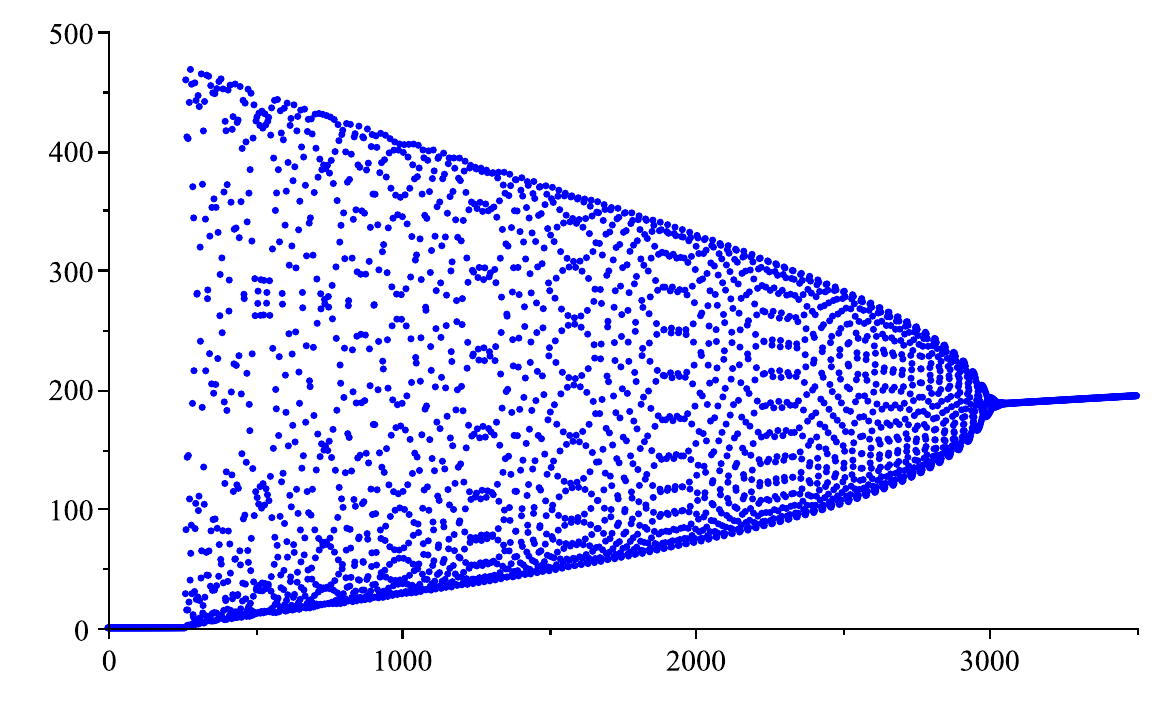} & 
\fig{2}{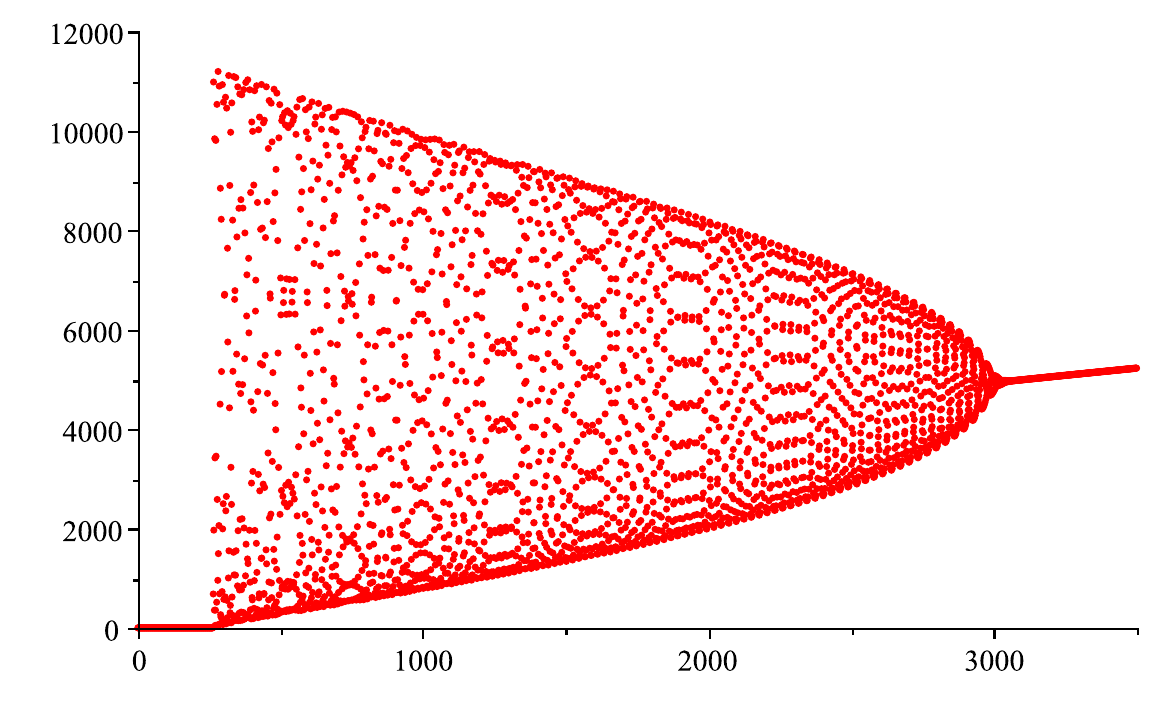}\\
\tau=50,\quad \k=0.275 &\tau=50,\quad \k=0.275 &\tau=50,\quad \k=0.275\\[5pt]
\fig{2}{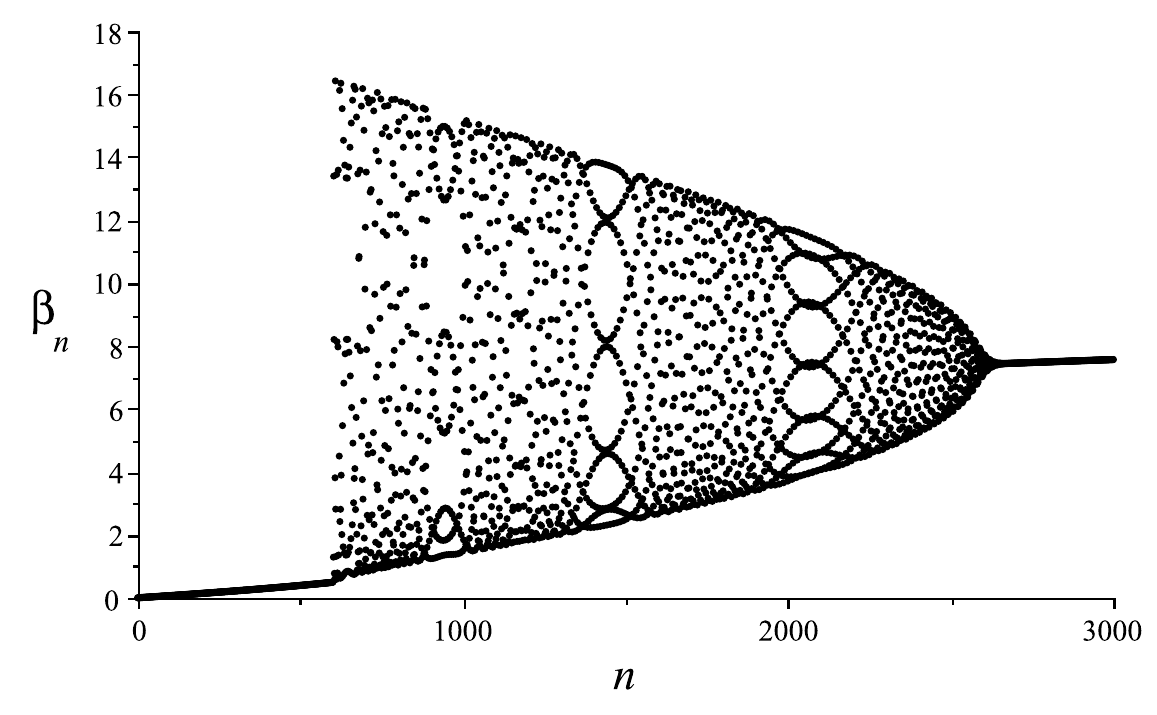} & \fig{2}{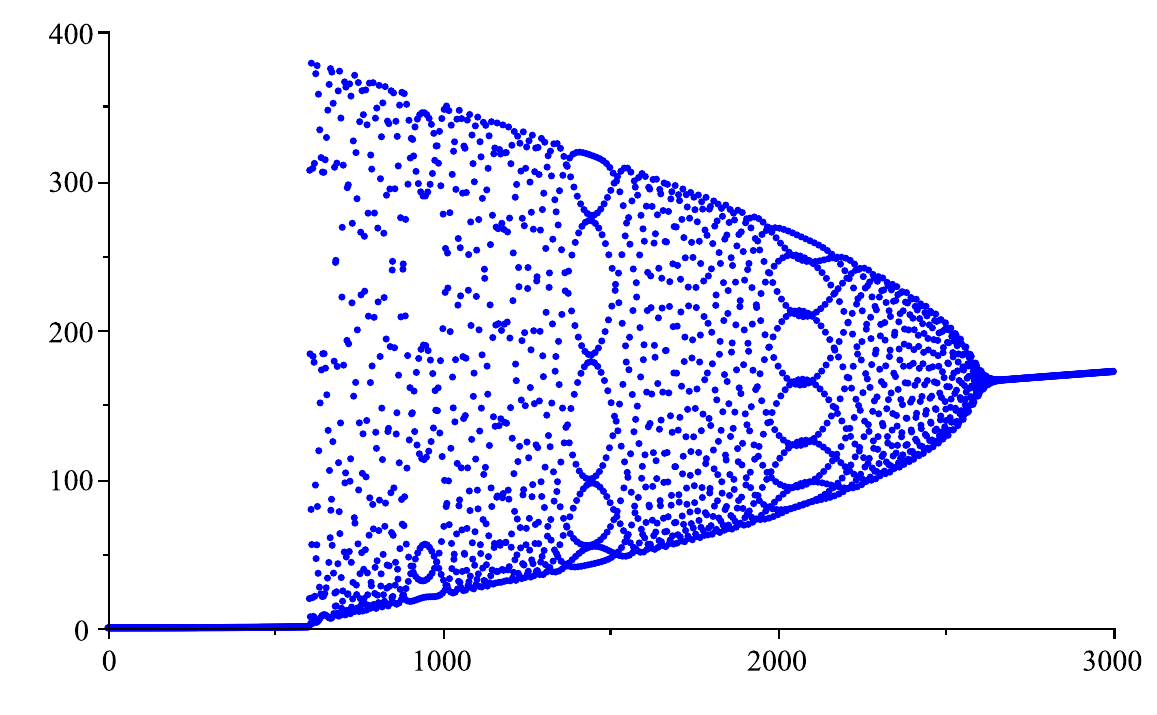} 
&\fig{2}{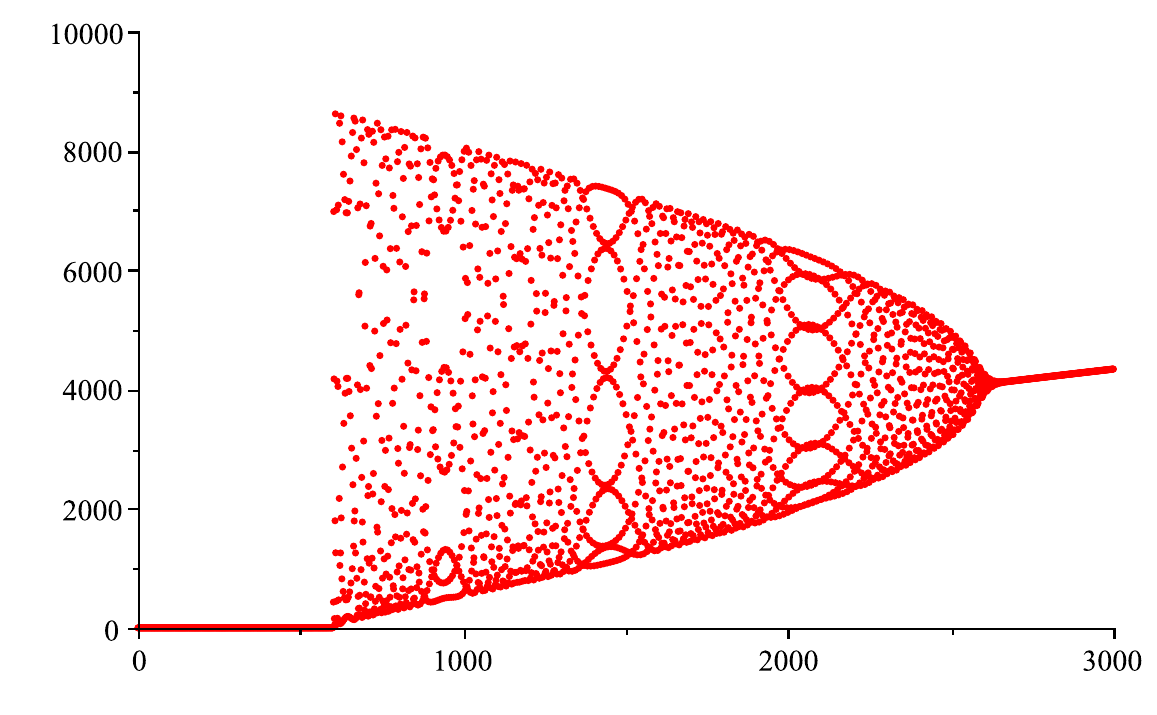}\\
\tau=50,\quad \k=0.3 &\tau=50,\quad \k=0.3 &\tau=50,\quad \k=0.3\\[5pt]
\fig{2}{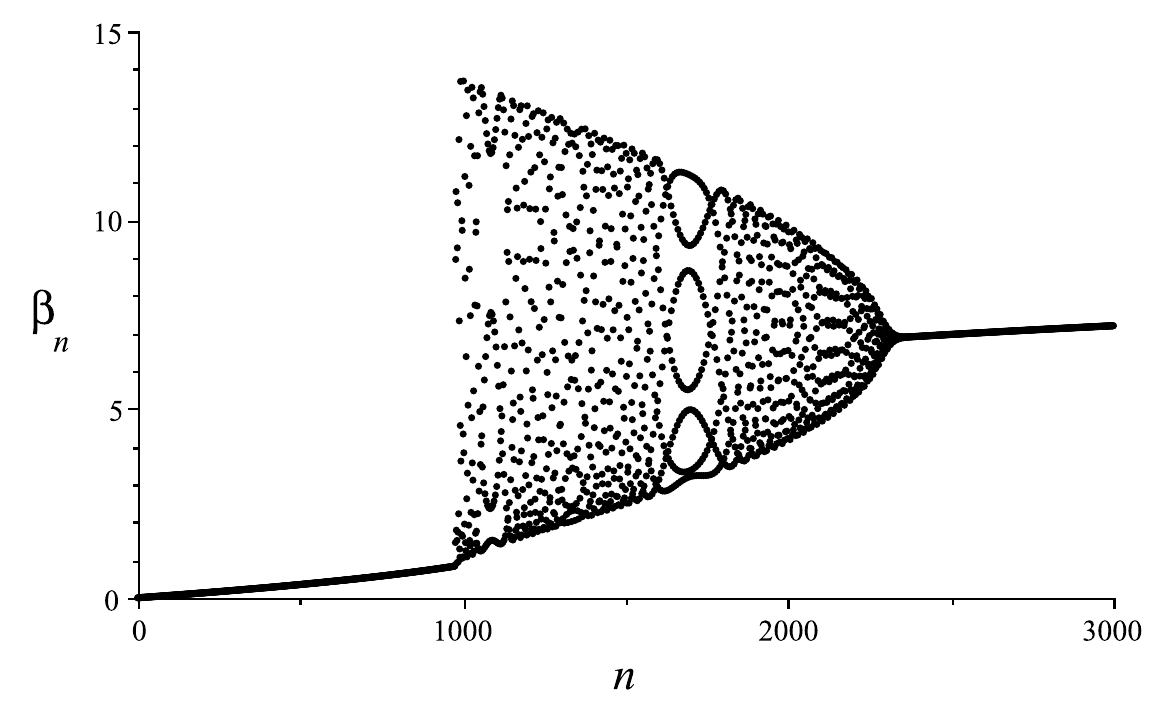} & \fig{2}{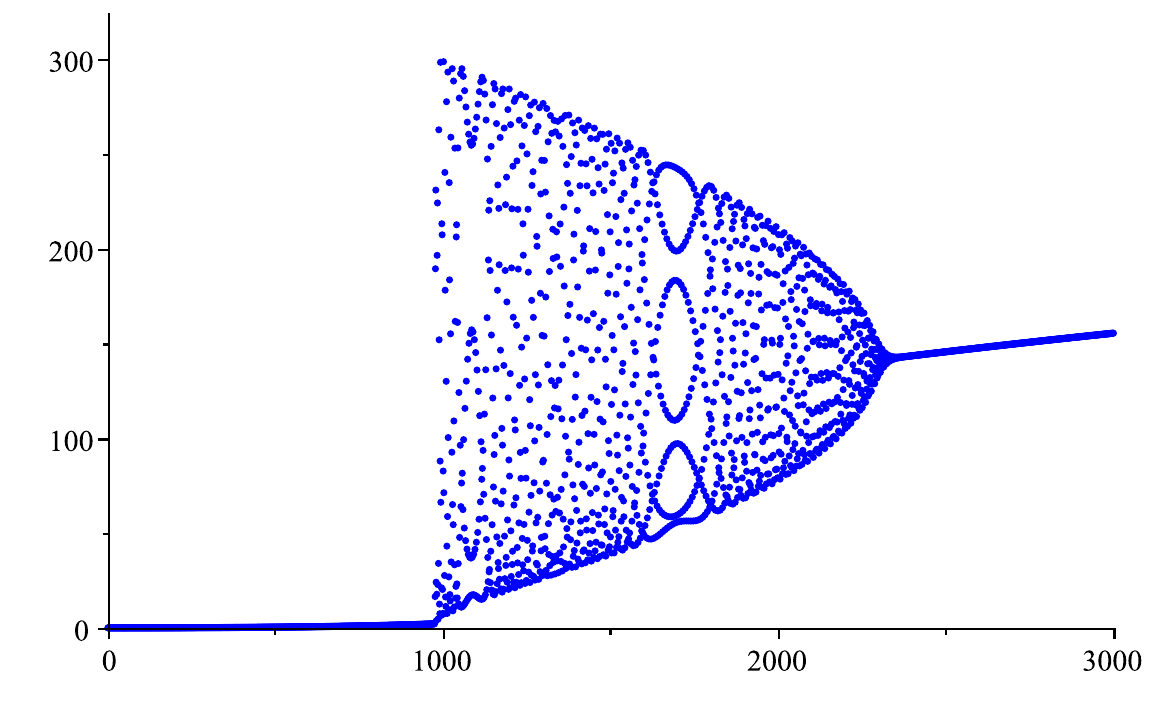} 
&\fig{2}{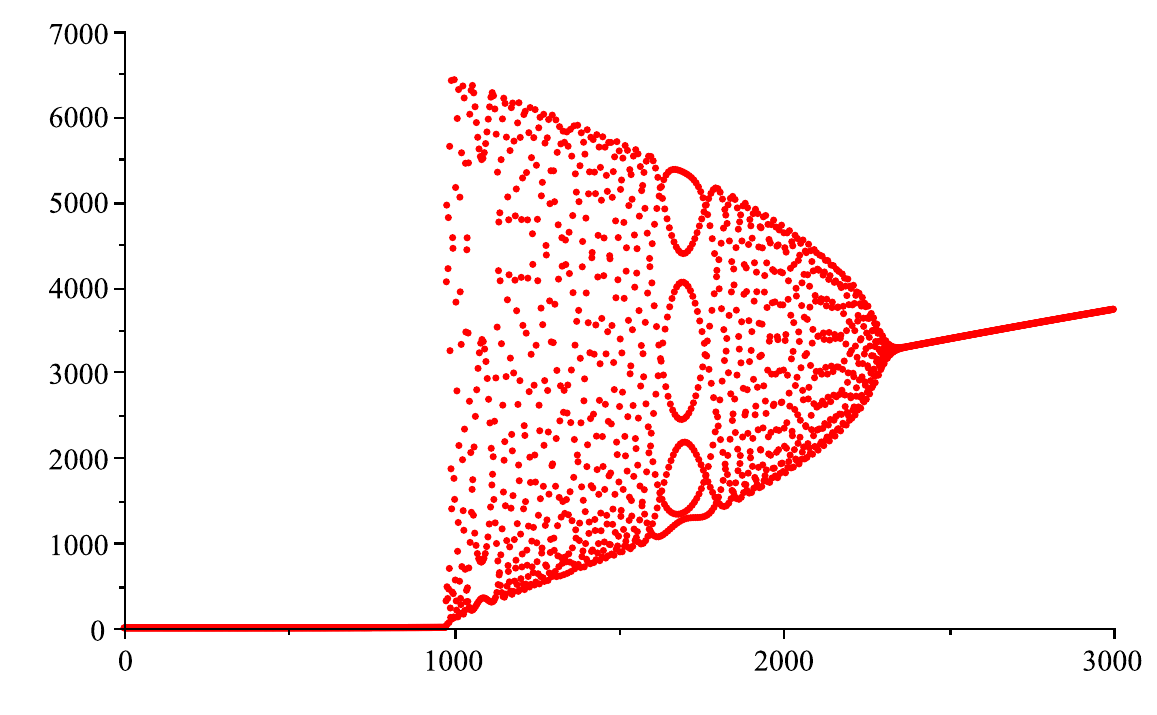}\\
\tau=50,\quad \k=0.325 &\tau=50,\quad \k=0.325 &\tau=50,\quad \k=0.325\\[5pt]
\fig{2}{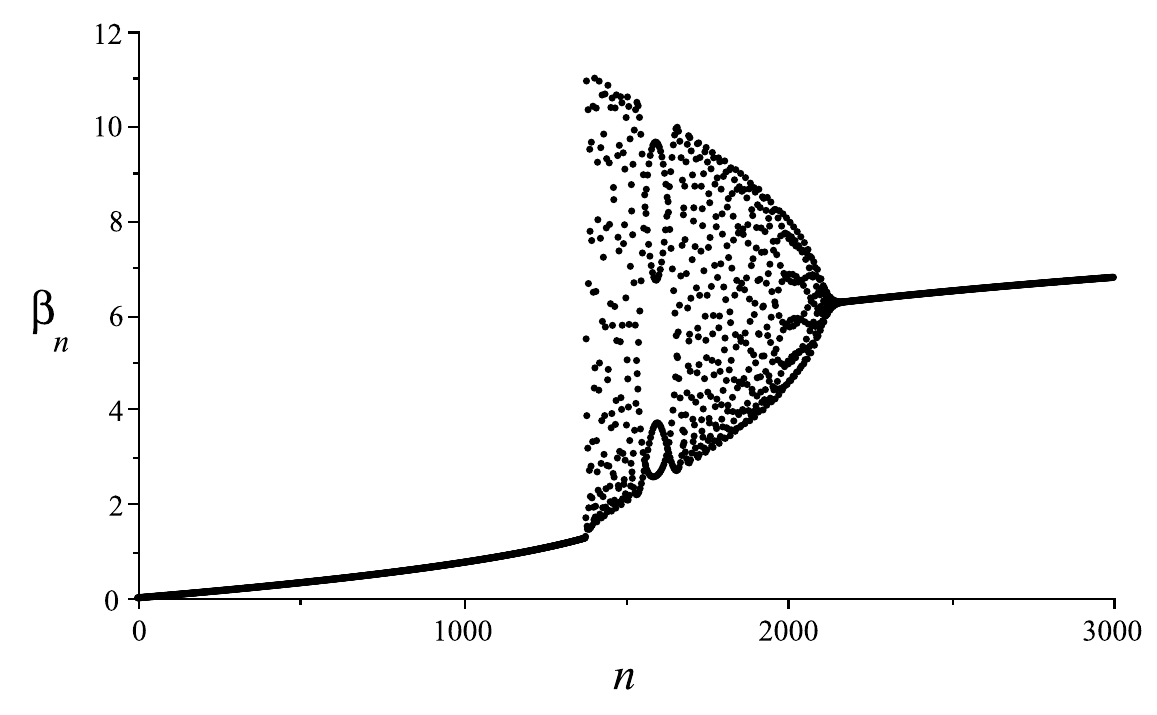} & \fig{2}{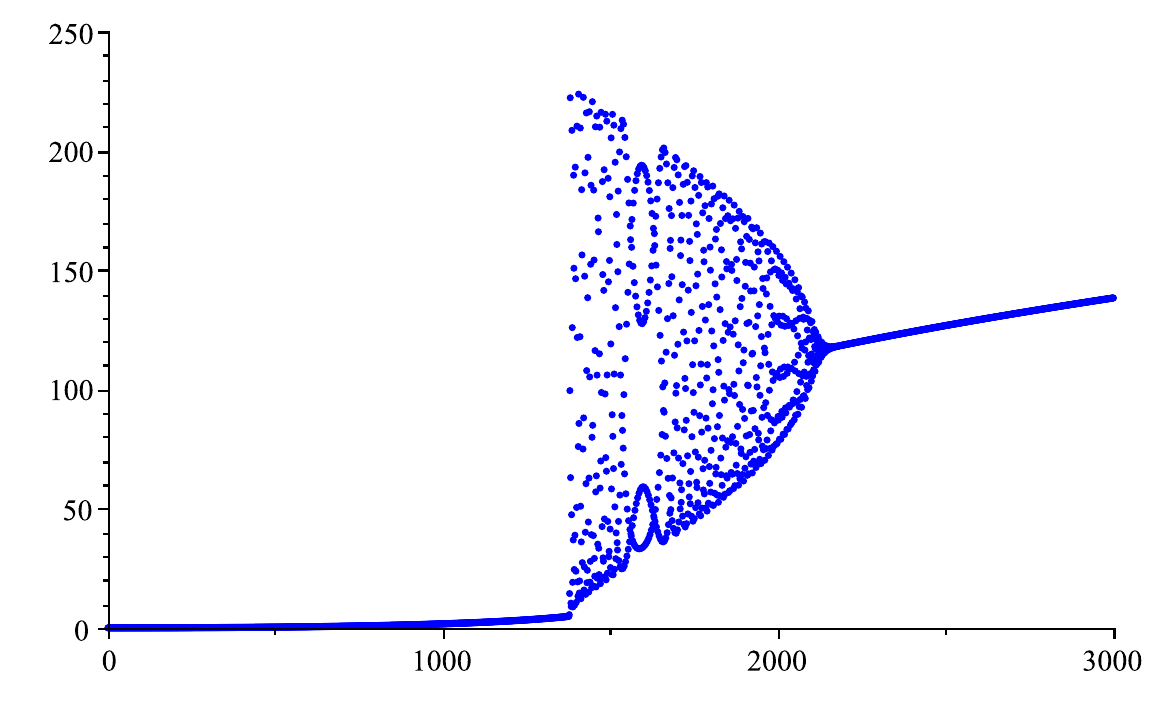} 
&\fig{2}{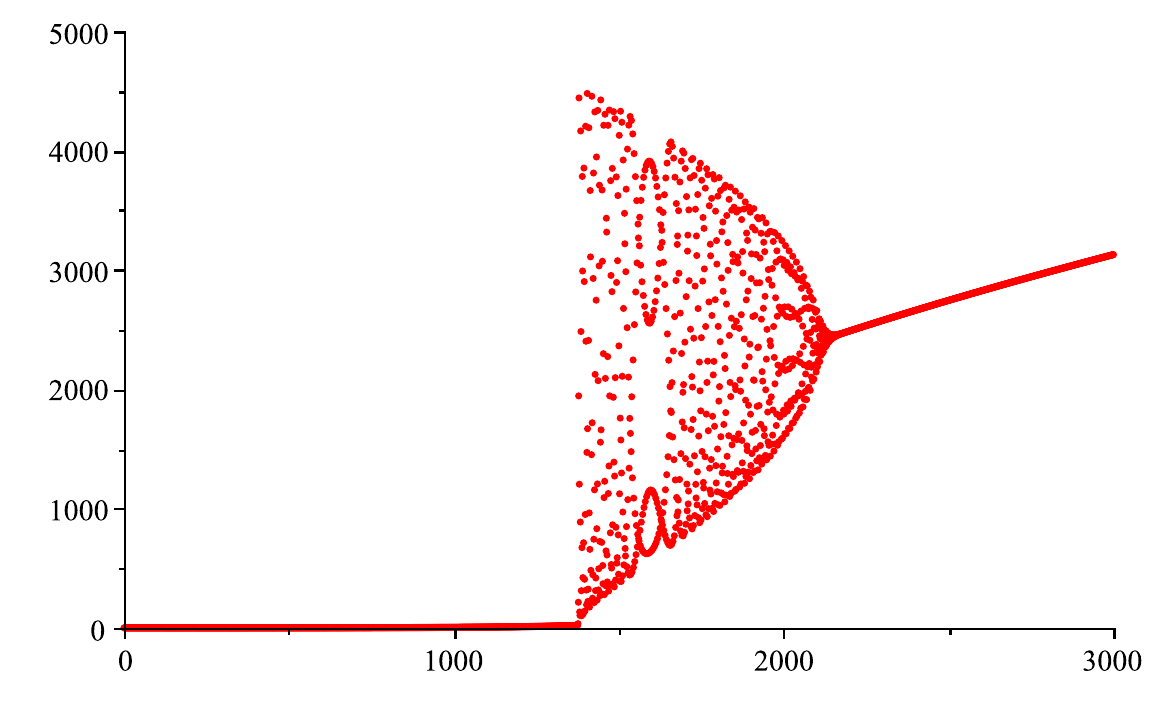}\\
\tau=50,\quad \k=0.35 &\tau=50,\quad \k=0.35 &\tau=50,\quad \k=0.35\\[5pt]
\fig{2}{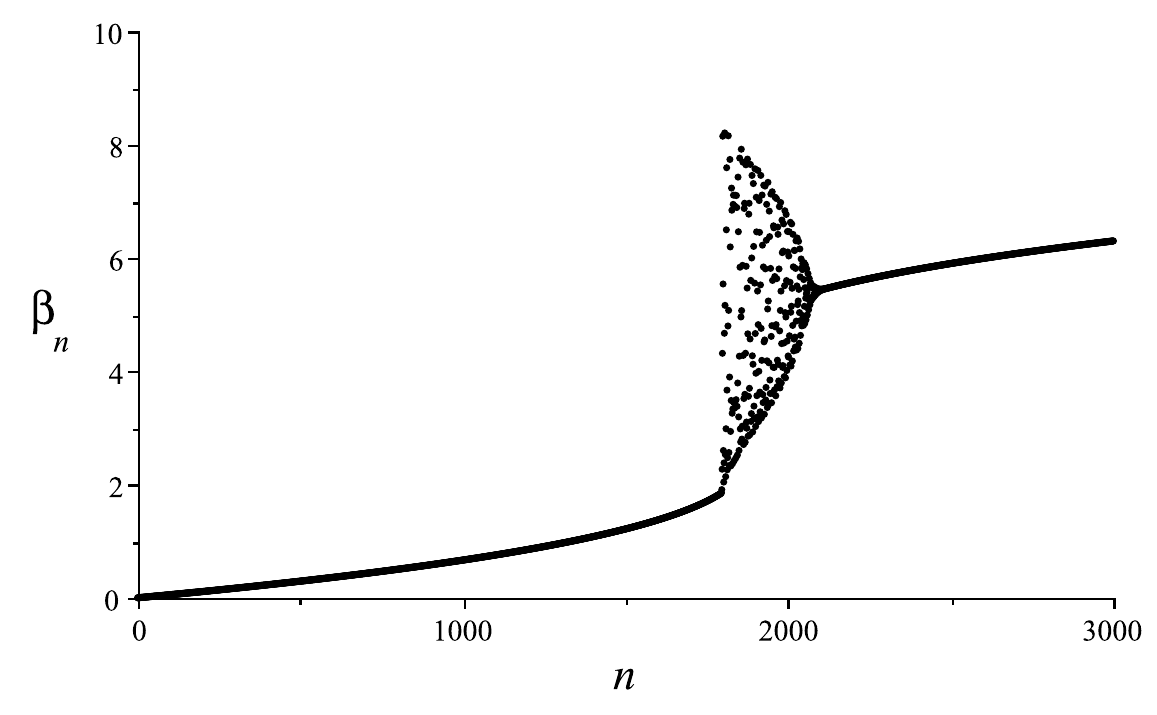} & \fig{2}{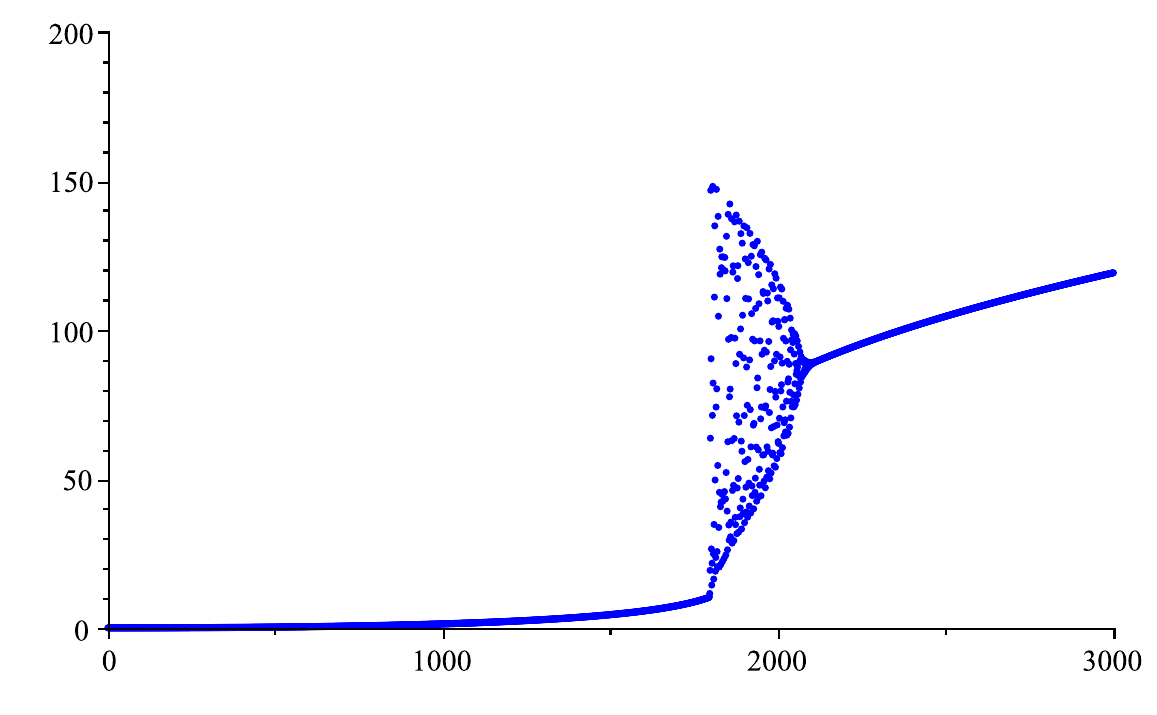} 
&\fig{2}{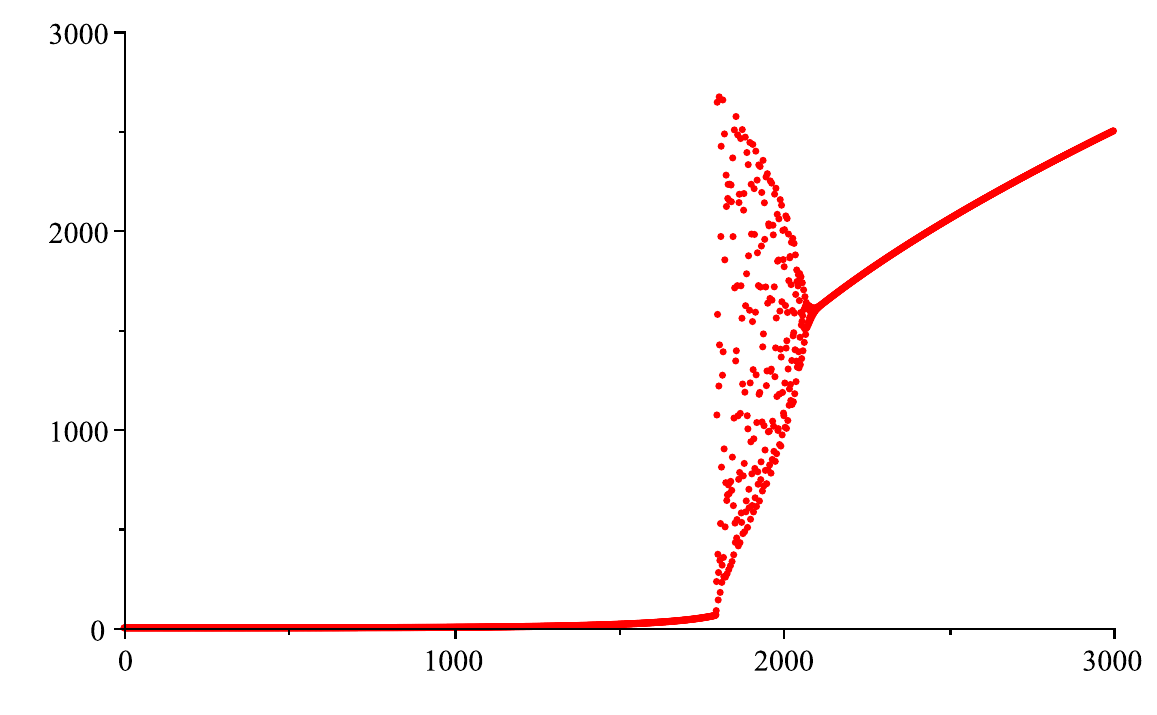}\\
\tau=50,\quad \k=0.375 &\tau=50,\quad \k=0.375 &\tau=50,\quad \k=0.375
\end{array}\]
\caption{\label{VolHier1}Plots of $V_n^{(2)}$ (black), $V_n^{(4)}$ (\blue{blue}) and $V_n^{(6)}$ (\red{red}) for $\tau=50$, with $\k=0.275$, $\k=0.3$, $\k=0.325$, $\k=0.35$
and $\k=0.375$.}
\end{figure}

\begin{figure}[ht!]
\[\begin{array}{c@{\quad}cc@{\quad}c}
V_n^{(2)} &V_n^{(4)} &V_n^{(6)} \\
\fig{2}{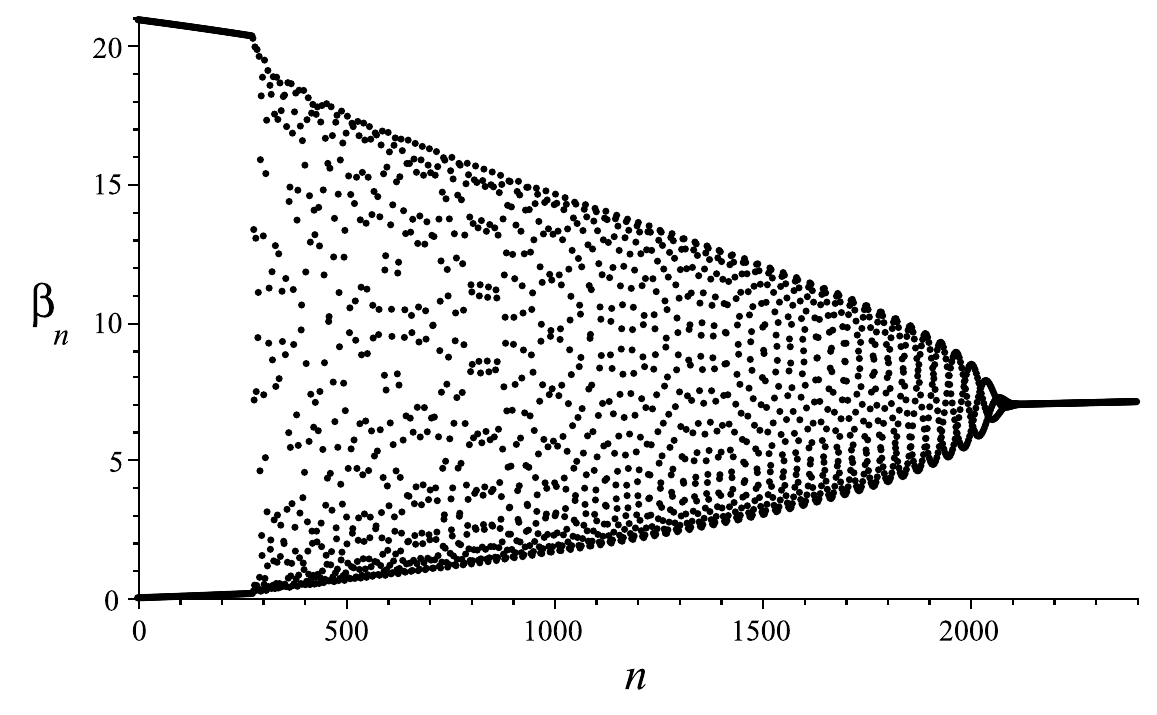}& \fig{2}{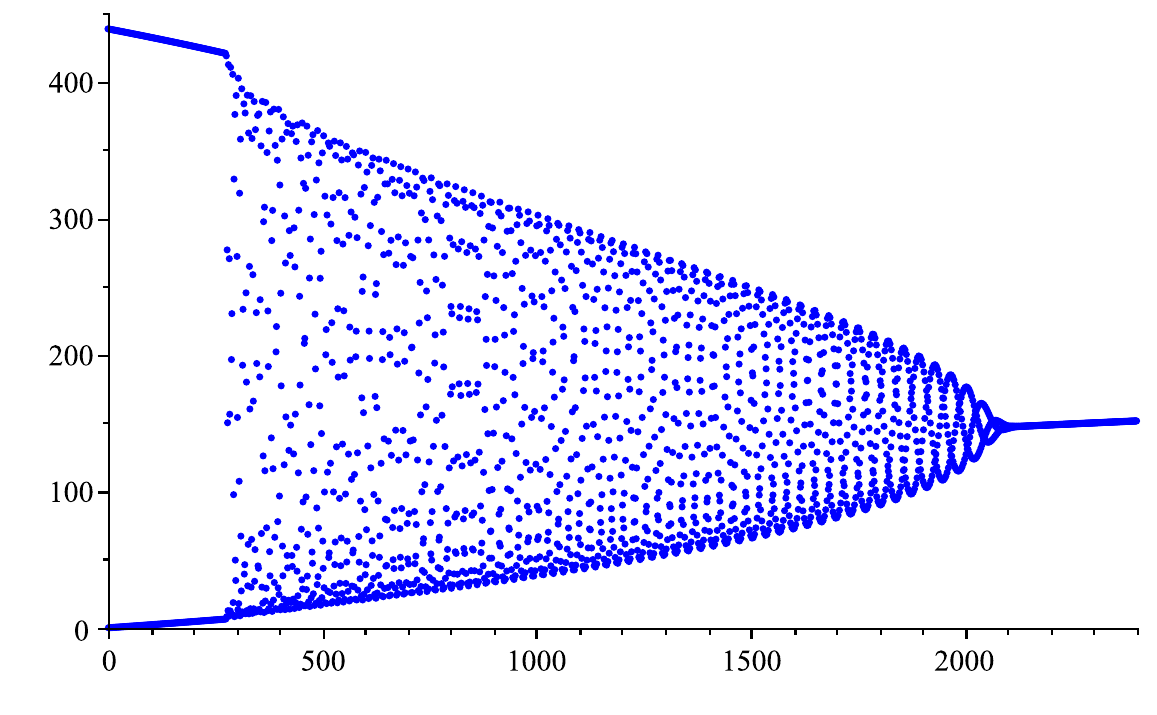}
& \fig{2}{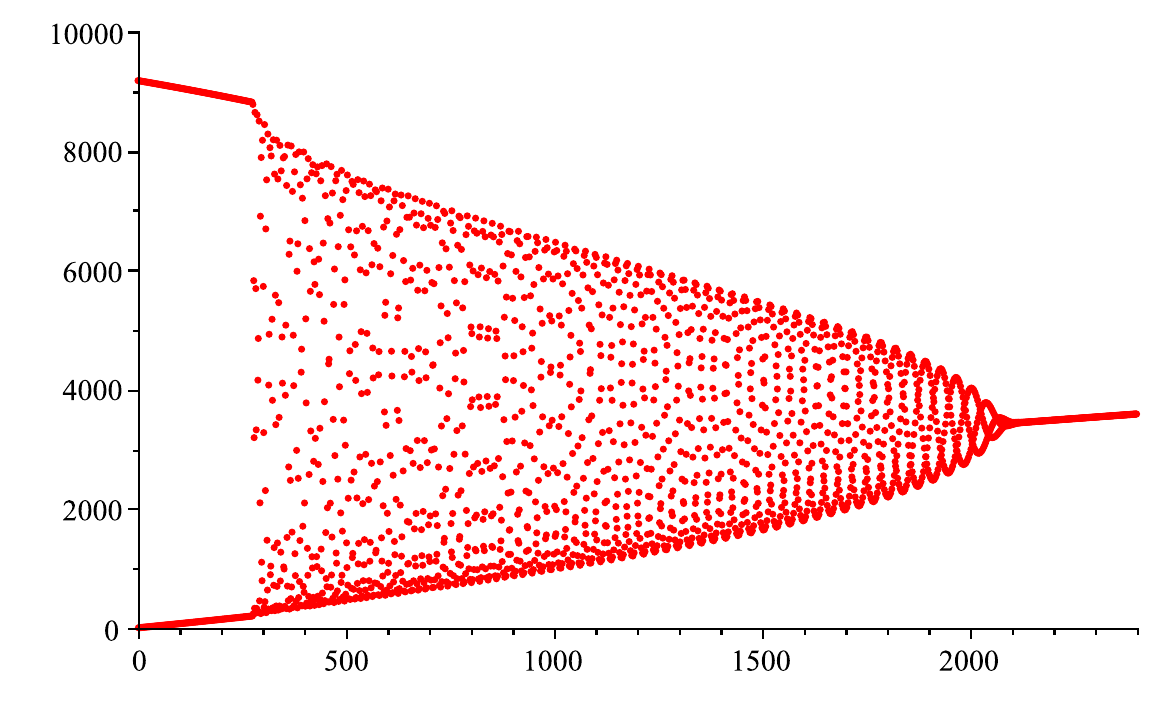}\\
\tau=40,\quad \k=0.225 &\tau=40,\quad \k=0.225 &\tau=40,\quad \k=0.225\\[5pt]
\fig{2}{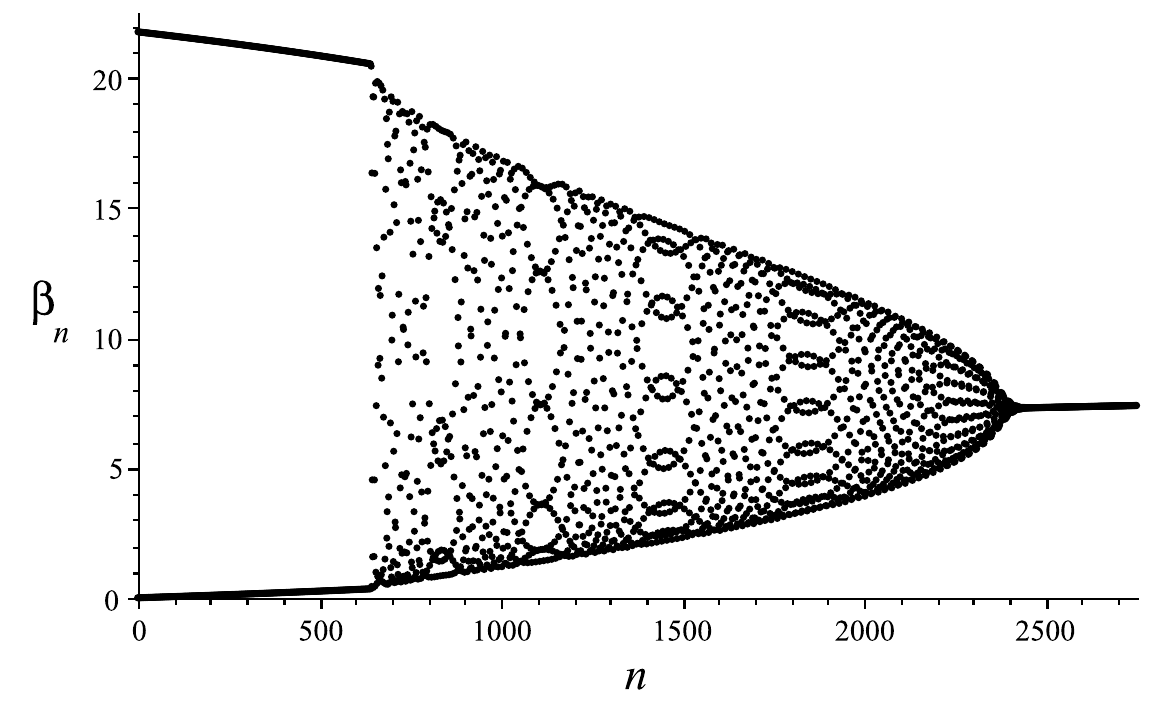}& \fig{2}{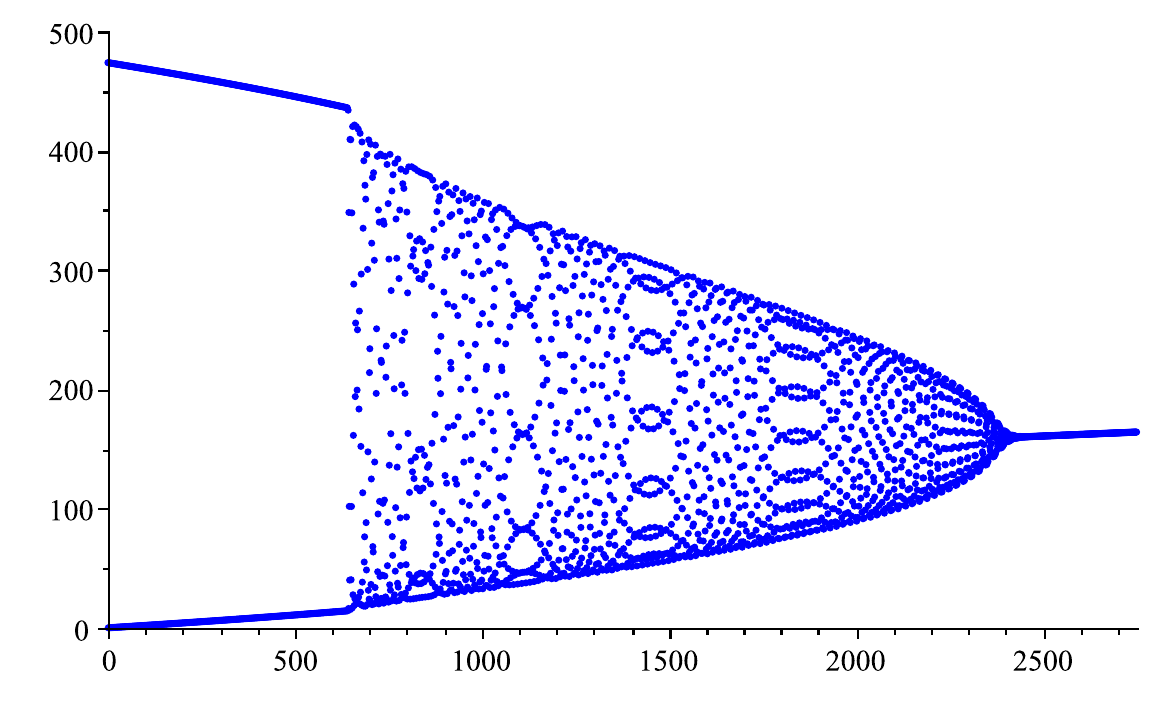}
& \fig{2}{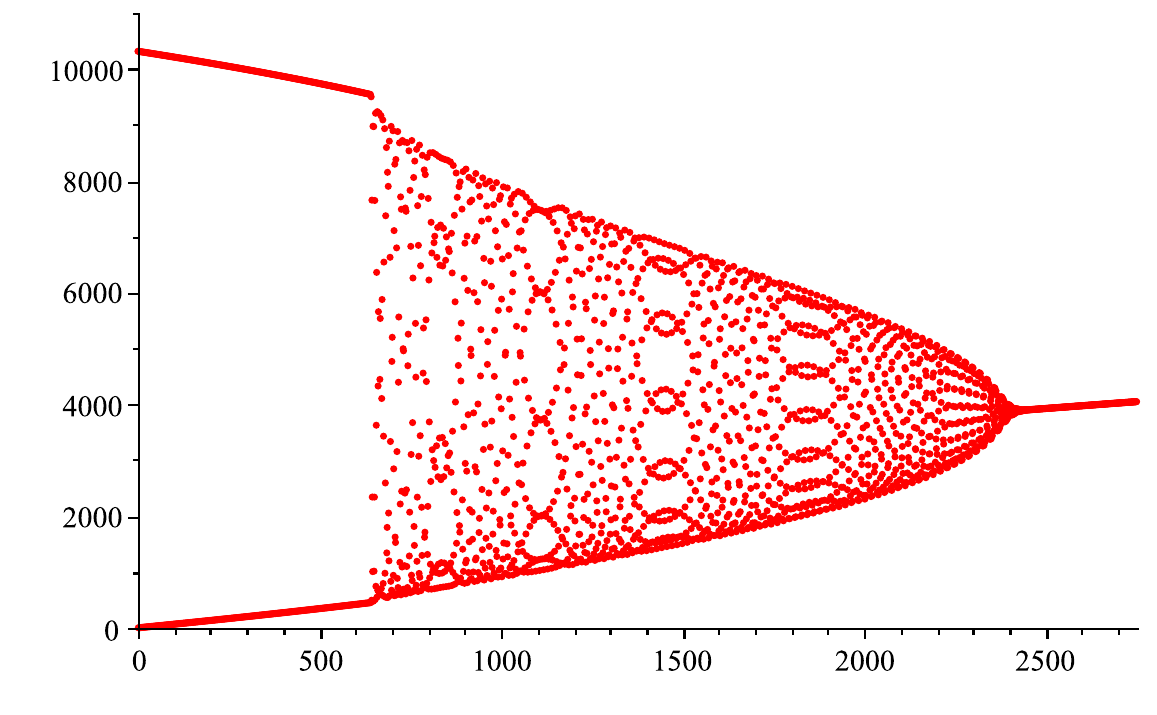}\\
\tau=40,\quad \k=0.2 &\tau=40,\quad \k=0.2 &\tau=40,\quad \k=0.2\\[5pt]
\fig{2}{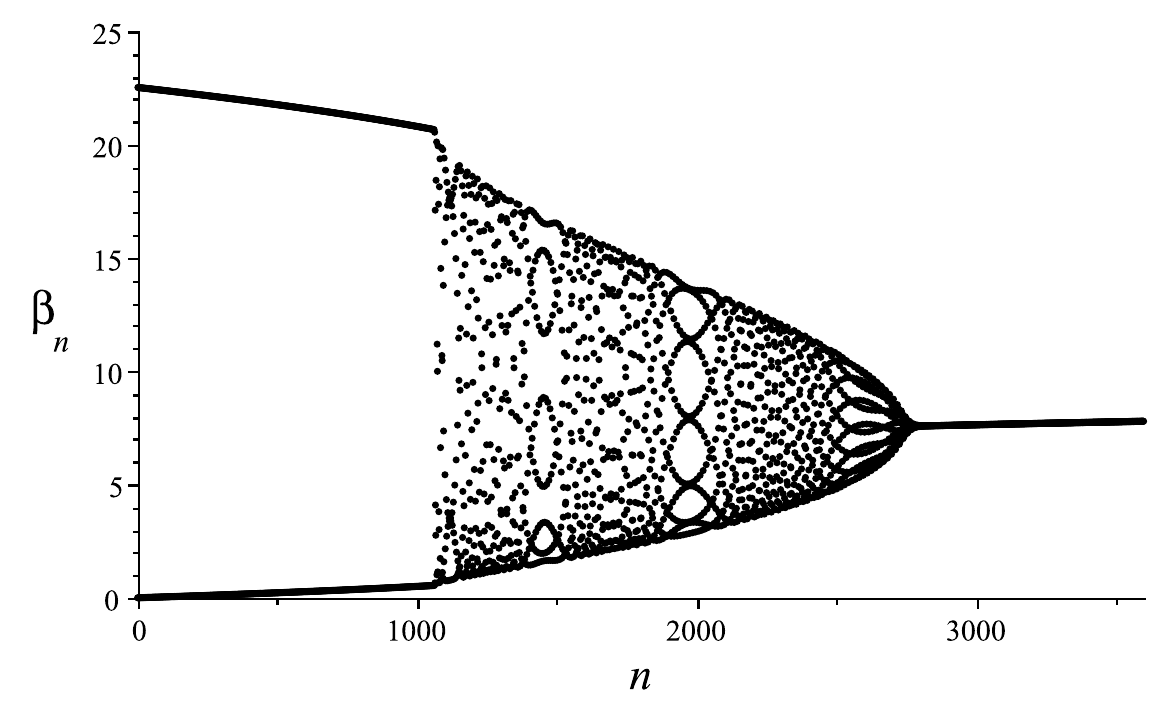}& \fig{2}{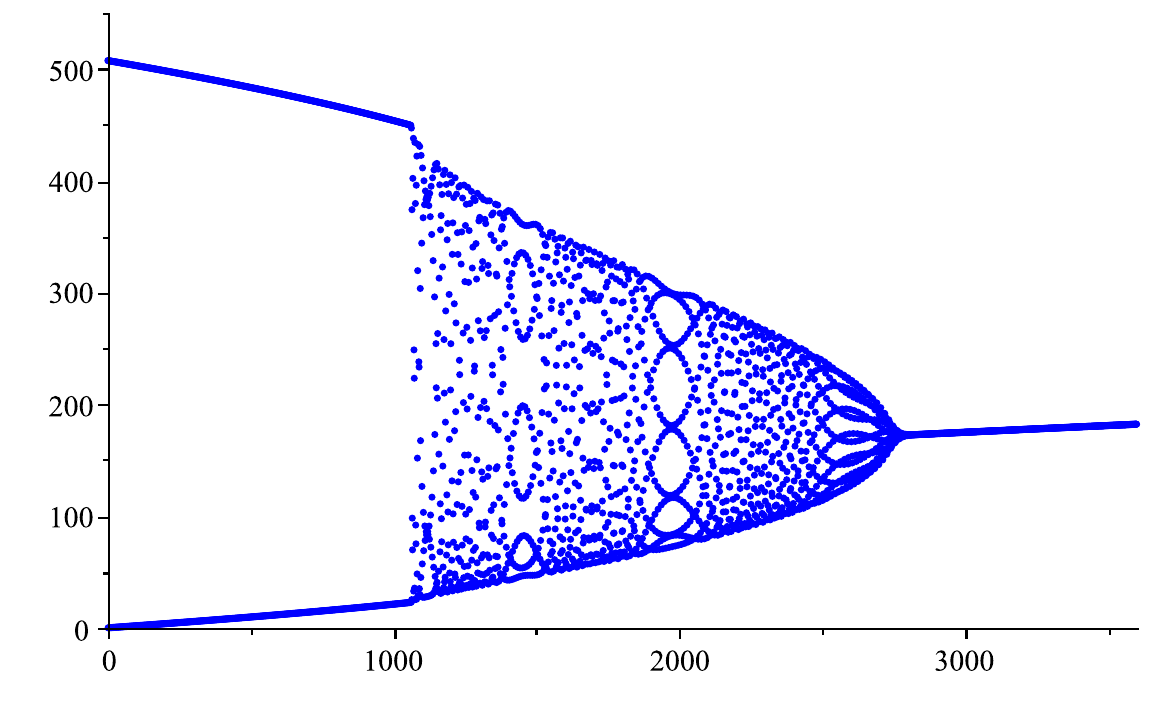}
& \fig{2}{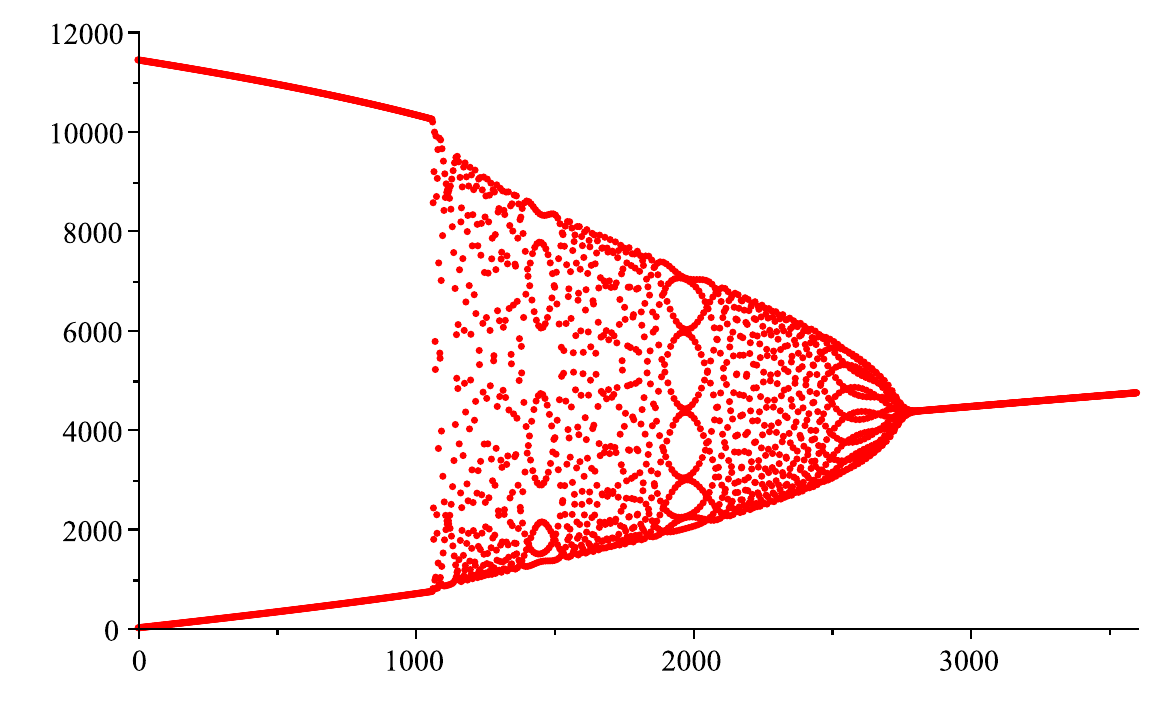}\\
\tau=40,\quad \k=0.175 &\tau=40,\quad \k=0.175 &\tau=40,\quad \k=0.175\\[5pt]
\fig{2}{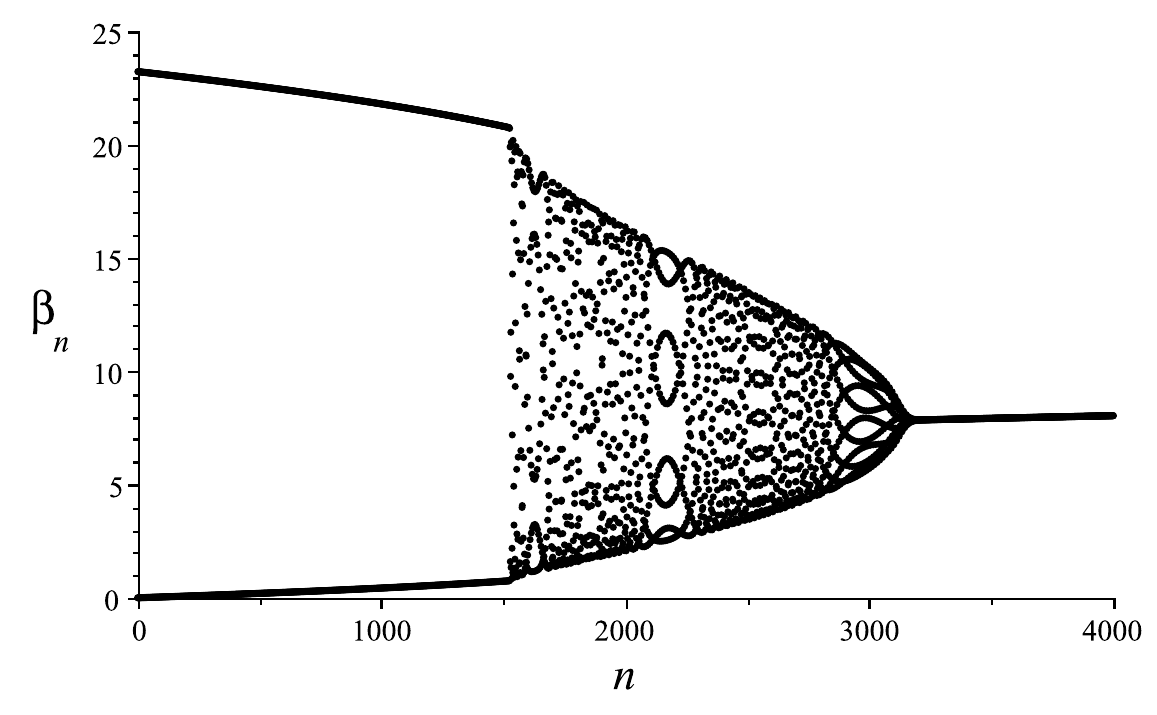}& \fig{2}{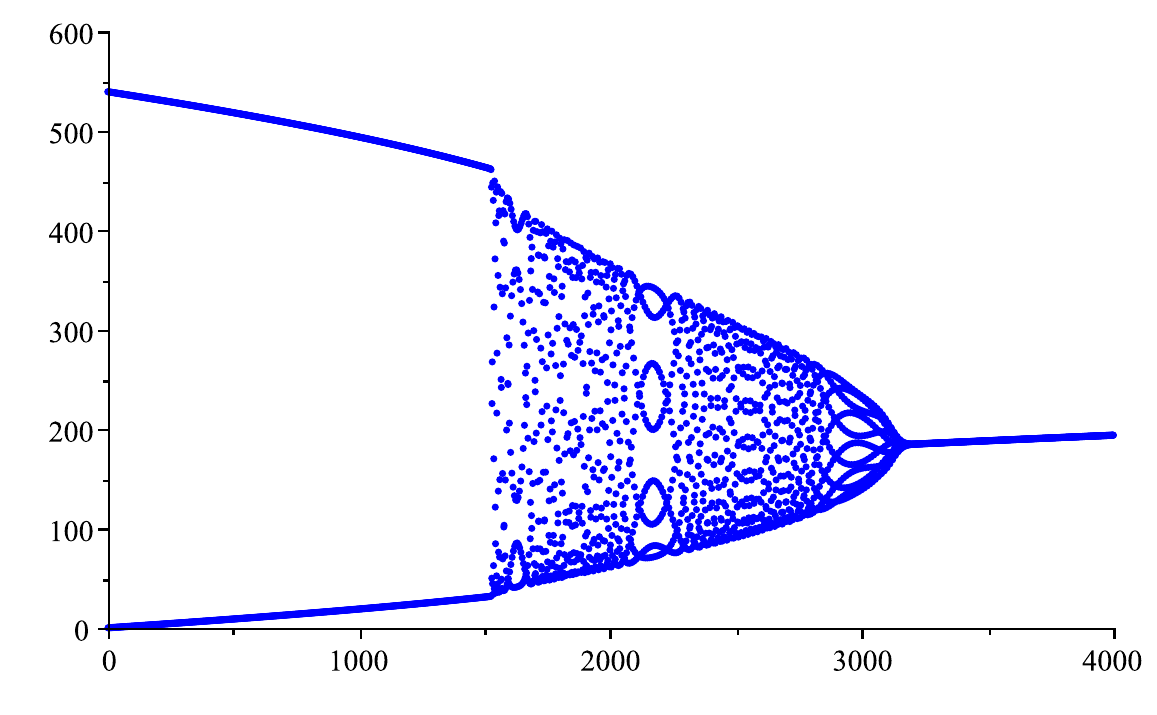}
& \fig{2}{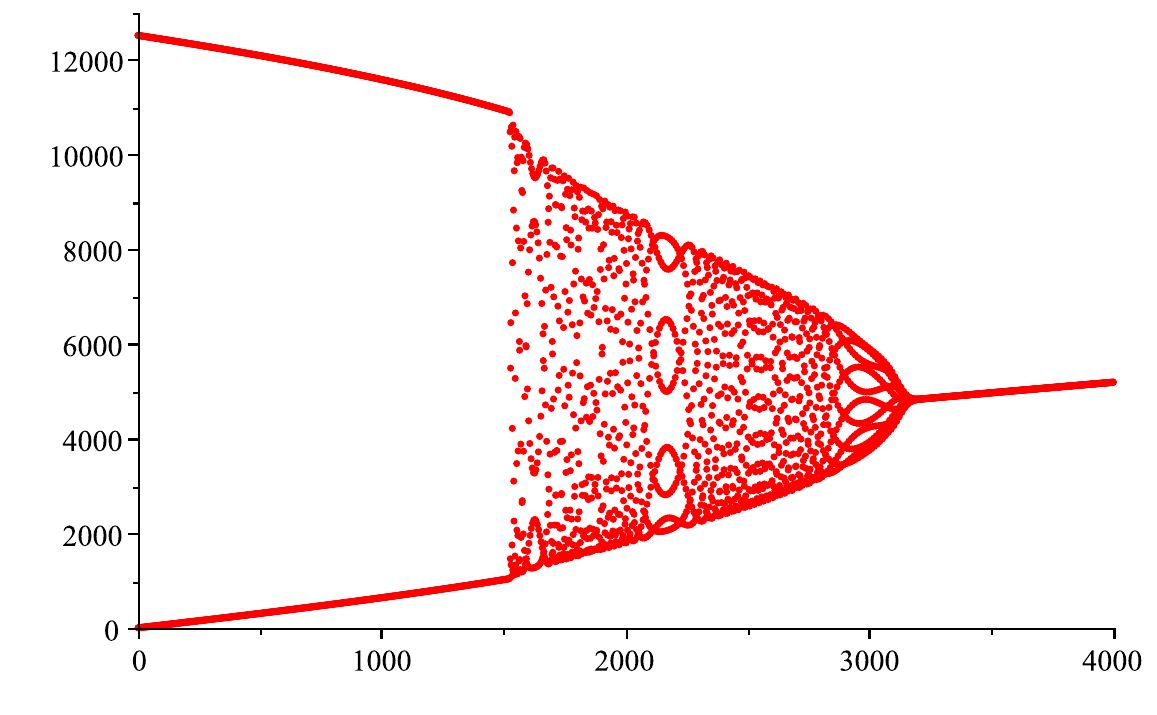}\\
\tau=40,\quad \k=0.15 &\tau=40,\quad \k=0.15 &\tau=40,\quad \k=0.15\\[5pt]
\fig{2}{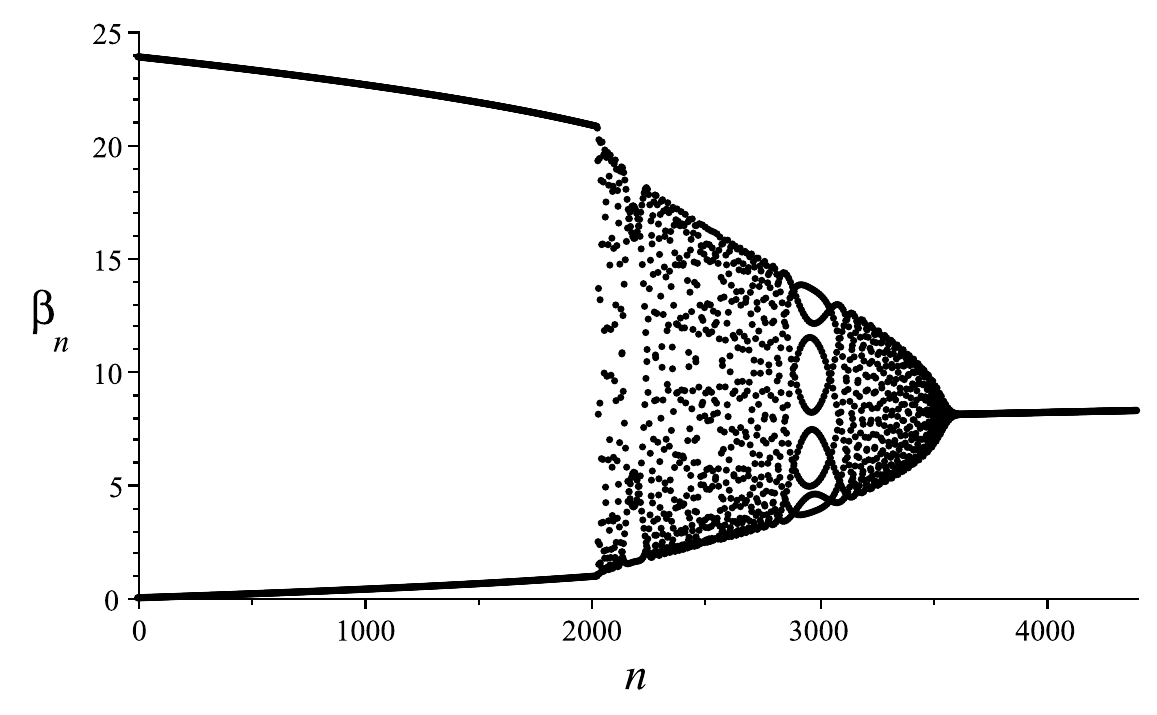}& \fig{2}{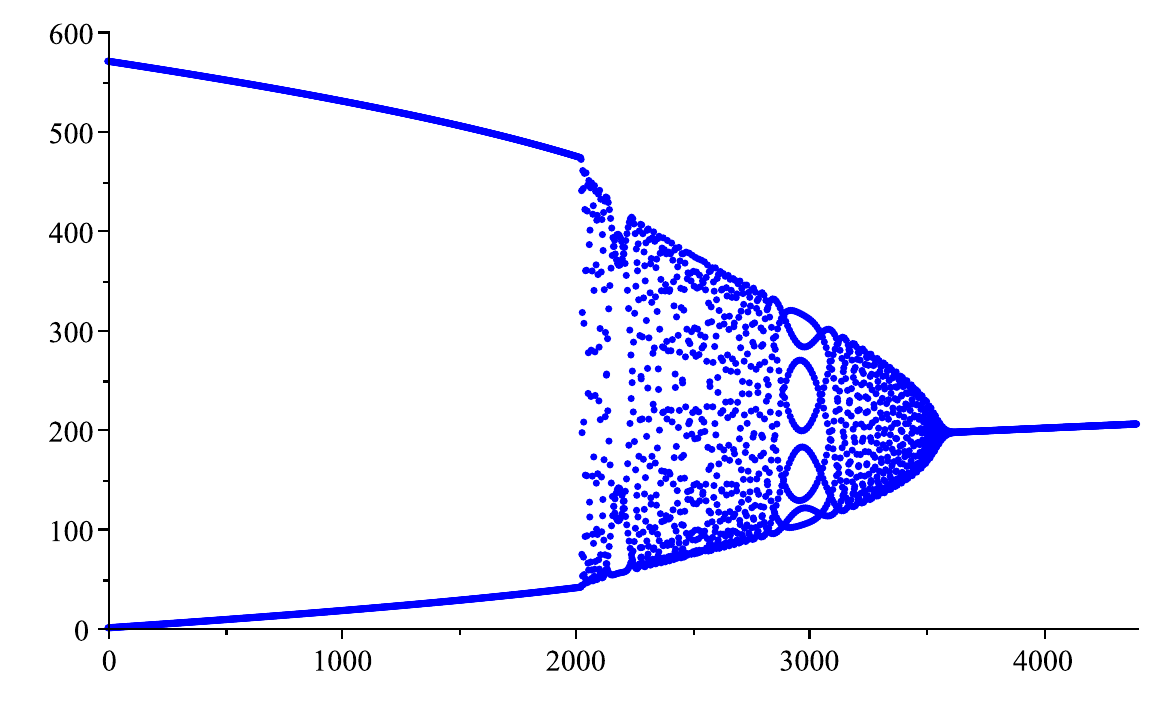}
& \fig{2}{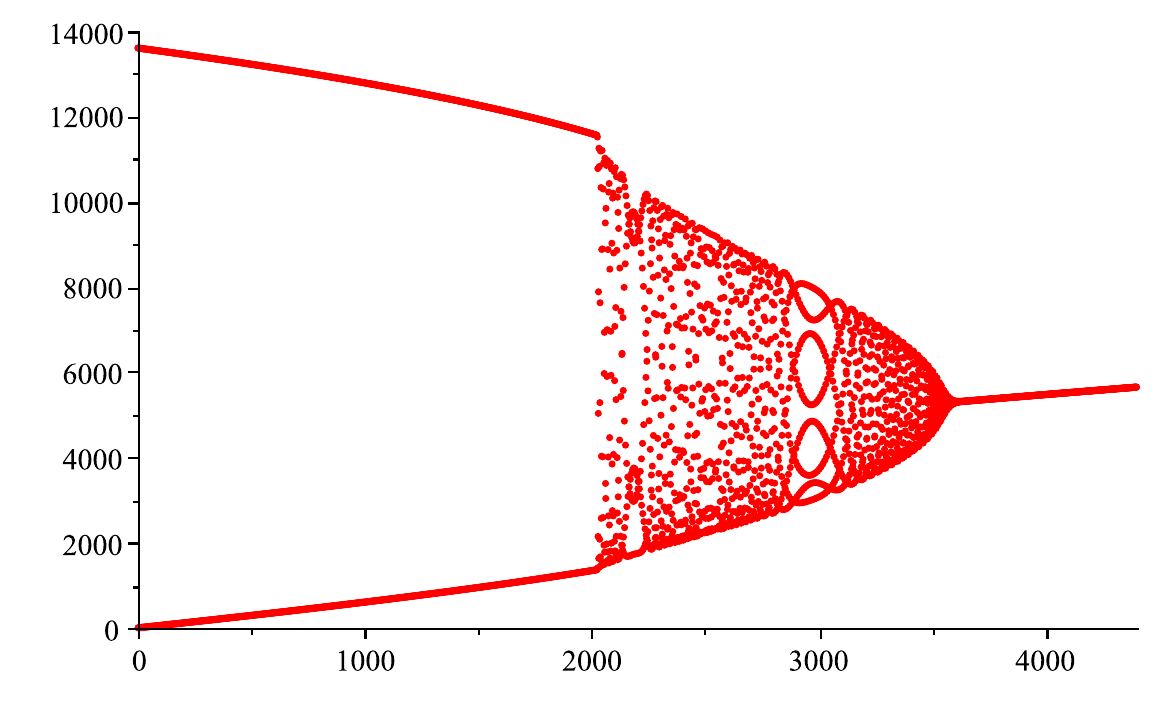}\\
\tau=40,\quad \k=0.125 &\tau=40,\quad \k=0.125 &\tau=40,\quad \k=0.125
\end{array}\]
\caption{\label{VolHier2}Plots of $V_n^{(2)}$ (black), $V_n^{(4)}$ (\blue{blue}) and $V_n^{(6)}$ (\red{red}) for $\tau=40$, with $\k=0.225$, $\k=0.2$ $\k=0.175$, $\k=0.15$ and $\k=0.125$.}
\end{figure}

\section{Discussion}
In this paper we have discussed the behaviour of the recurrence coefficient $\b_n$ in the three-term recurrence relation for the symmetric sextic Freud weight 
\beq \w(x;\tau,\k)=\exp\left\{-\left(x^6-\tau x^4+\k\tau^2x^2\right)\right\},\label{Freud642k1}\eeq
{supported on the real line, where $\tau$ and $\k$ are real parameters}. {In three cases, when $\tau>0$ and either $\k=0$, $\k=\tfrac{1}{4}$ or $\k=\tfrac{1}{3}$, we have obtained explicit expressions for the associated moments in terms generalised hypergeometric functions.}

The numerical computations show that there are three particular regions {in the $(\tau;\k)$-plane} of interest:
\begin{enumerate}[(i)]
\item if $\tau>0$ and $\tfrac{1}{4}<\k<\tfrac{2}{5}$ then the recurrence coefficient $\b_n$ approximately follows the cubic {curve}
\beq \label{eq:cubick} 
60\b^3-12\tau \b^2+2\k\tau^2\b -n=0,\eeq
both initially and for $n$ large;
\item if $\tau>0$ and $-\tfrac{2}{3}<\k<\tfrac{1}{4}$ then the recurrence coefficients $\b_{2n}$ and $\b_{2n+1}$ initially approximately follow the curves
\begin{align*}&6 u\!\left(u^{2}+6 u v +3 v^{2}\right)-4\tau u\!\left(u+2 v \right)+2\k\tau^2u=n,\\
&6 v\!\left(3 u^{2}+6 u v +v^{2}\right)-4\tau v\!\left(2 u +v \right)+2\k\tau^2 v =n,
\end{align*}
with $u=\b_{2n}$ and $v=\b_{2n+1}$, and {then} follow the cubic \eqref{eq:cubick} for $n$ large; 
\item if $\tau>0$ and $\k\approx \tfrac{1}{4}$ then the recurrence coefficients $\b_{3n}$ and $\b_{3n\pm1}$ approximately follow three curves initially and {then} follow the cubic {curve} \eqref{eq:cubick} for $n$ large. When $\tau>0$ and $\k= \tfrac{1}{4}$ then the recurrence coefficients $\b_{3n\pm1}$ approximately follow the same curve.
\end{enumerate}

In cases (i) and (ii) there is a ``transition region", which Jurkiewicz \cite{refJurk91} and S\'en\'echal \cite{refSen92} described as ``chaotic", though as discussed above, was subsequently shown not to be {the case} cf.~\cite{refBDE00,refBG24,refEynard09,refEynM11}. Our results support this point of view. The structure and size of the ``transition region", such as the value of $n$ when transition commences (an issue raised by S\'en\'echal \cite{refSen92}), appears to depend on both $\tau$ and $\k$ and it remains an open question to analytically describe this.

In case (iii), in the neighbourhood of $\k=\tfrac{1}{4}$, which is between the other cases, there is a ``three-fold" structure. The nature of this {mutation} between {the scenarios} (i) and (ii) is currently under investigation, and we do not pursue this further here.

Elsewhere the recurrence coefficient $\b_n$ either follows a curve, which is monotonically increasing, or $\b_{2n}$ and $\b_{2n+1}$ initially follow two curves which meet and then they follow the same curve, for example as shown in Figures \ref{Fig716} and \ref{Fig718}.

In a recent paper, Clarkson \etal\ \cite{refCDHM} studied the tertiary \dPI\ equation \cite{refGRP,refTGR}
\beq v_n(v_{n+1}+v_{n-1}+1) = (n+1)\ep,\qquad \ep>0, \label{tdpI} \eeq
with $v_{-1}=0$, which arises in quantum minimal surfaces \cite{refAHK,refHoppe} and found solutions in terms of the modified Bessel functions $I_{\pm1/6}(z)$ and $I_{\pm5/6}(z)$, i.e.\ the same modified Bessel functions that arise in the description of the moments of \eqref{Freud642k} when $\k=\tfrac14$. 
The tertiary \dPI\ \eqref{tdpI} also arises in connection with orthogonal polynomials in the complex plane with respect to the weight
\[w(z;t)=\exp\left\{-t|z|^2 +\tfrac13\rmi(z^3+\overline{z}^3)\right\},\]
with $t>0$ and $z=x+\rmi y\in\Com$ \cite{refFH}, with the \textit{same} unique solution as in \cite{refCDHM}.
Further Teodorescu \etal\ \cite{refTBAZW} show that
\eqref{tdpI} arises in the theory of random normal matrices \cite{refWZ}, which motivated the study in \cite{refFH}.
It is an interesting open question as to whether there is any relationship between these problems, though again we do not pursue this further here.

As remarked in \S\ref{Intro}, the symmetric sextic Freud weight \eqref{Freud642k1} is equivalent to the weight
\beq \label{Freud642N}
w(x)=\exp\left\{-N\!\left(g_6x^6+g_4 x^4+g_2x^2\right)\right\},\eeq
with parameters $N$, $g_2$, $g_4$ and $g_6>0$, which has been studied numerically by several authors, e.g.~\cite{refBM20,refBG13,refDDJT,refJurk91,refLech92,refLech92b,refLechRR,refSS91,refSen92}. For the weight \eqref{Freud642N} the critical quantity is $g_2g_6/g_4^2$ which {plays the role of $\k$} in \eqref{Freud642k1}. If $\tfrac{1}{4}<g_2g_6/g_4^2<\tfrac{2}{5}$, with $g_4<0$, then it is equivalent to the case discussed in \S\ref{casei} and if $0<g_2g_6/g_4^2<\tfrac{1}{4}$, with $g_4<0$, then it is equivalent to the case discussed in \S\ref{caseii}. The effect of increasing $N$ is equivalent to increasing $\tau$, {as discussed in \S\ref{subsec711} and illustrated in Figures \ref{Fig817}, \ref{Fig818} and \ref{Fig819}}.

\section*{Acknowledgements} 
We thank Ines Aniceto, Costanza Benassi, Marta Dell'Atti, Andy Hone, Chris Lustri, Antonio Moro, Walter Van Assche and Jing Ping Wang for their helpful comments and illuminating discussions.
PAC and KJ gratefully acknowledge the support of a Royal Society Newton Advanced Fellowship NAF$\backslash$R2$\backslash$180669. 
PAC and AL would like to thank the Isaac Newton Institute for Mathematical Sciences for support and hospitality during the programmes ``\textit{Applicable resurgent asymptotics: towards a universal theory}" and ``\textit{Dispersive hydrodynamics: mathematics, simulation and experiments, with applications in nonlinear waves}", supported by EPSRC grant number EP/R014604/1, and the satellite programme
``\textit{Emergent phenomena in nonlinear dispersive waves}", supported by EPSRC grant EP/V521929/1, when some of the work on this paper was undertaken. 
{We also thank the referees whose comments and suggestions have considerably improved the exposition.}

\subsection*{ORCID iDs} 
\begin{tabular}{ll}
Peter Clarkson & 0000-0002-8777-5284\\
Kerstin Jordaan & 0000-0002-1675-5366\\ 
Ana Loureiro & 0000-0002-4137-8822
\end{tabular}

\def\ams{American Mathematical Society}
\def\AAM{Acta Appl. Math.}
\def\ARMA{Arch. Rat. Mech. Anal.}
\def\AC{Acta Crystrallogr.}
\def\AM{Acta Metall.}
\def\ampa{Ann. Mat. Pura Appl. (IV)}
\def\AP{Ann. Phys., Lpz.}
\def\APNY{Ann. Phys., NY}
\def\APP{Ann. Phys., Paris}
\def\BAMS{Bull. Amer. Math. Soc.}
\def\CJP{Can. J. Phys.}
\def\CMP{Commun. Math. Phys.}
\def\CPAM{Commun. Pure Appl. Math.}
\def\CQG{Classical Quantum Grav.}
\def\CSF{Chaos, Solitons \&\ Fractals}
\def\DE{Diff. Eqns.}
\def\DU{Diff. Urav.}
\def\EJAM{Europ. J. Appl. Math.}
\def\FUNK{Funkcial. Ekvac.}
\def\IP{Inverse Problems}
\def\JAMS{J. Amer. Math. Soc.}
\def\JAP{J. Appl. Phys.}
\def\JCP{J. Chem. Phys.}
\def\JDE{J. Diff. Eqns.}
\def\JFM{J. Fluid Mech.}
\def\JJAP{Japan J. Appl. Phys.}
\def\JP{J. Physique}
\def\JPhCh{J. Phys. Chem.}
\def\JMAA{J. Math. Anal. Appl.}
\def\JMMM{J. Magn. Magn. Mater.}
\def\JMP{J. Math. Phys.}
\def\JNMP{J. Nonl. Math. Phys.}
\def\JPA{J. Phys. A: Math. Gen.}
\def\JPB{J. Phys. B: At. Mol. Phys.} 
\def\jpb{J. Phys. B: At. Mol. Opt. Phys.}
\def\JPC{J. Phys. C: Solid State Phys.} 
\def\JPCM{J. Phys: Condensed Matter} 
\def\JPD{J. Phys. D: Appl. Phys.}
\def\JPE{J. Phys. E: Sci. Instrum.}
\def\JPF{J. Phys. F: Metal Phys.}
\def\JPG{J. Phys. G: Nucl. Phys.} 
\def\jpg{J. Phys. G: Nucl. Part. Phys.}
\def\JSP{J. Stat. Phys.}
\def\JOSA{J. Opt. Soc. Am.}
\def\JPSJ{J. Phys. Soc. Japan}
\def\JQSRT{J. Quant. Spectrosc. Radiat. Transfer}
\def\LMP{Lett. Math. Phys.}
\def\LNC{Lett. Nuovo Cim.}
\def\NC{Nuovo Cim.}
\def\NIM{Nucl. Instrum. Methods}
\def\NL{Nonlinearity}
\def\NMJ{Nagoya Math. J.}
\def\NP{Nucl. Phys.}
\def\PL{Phys. Lett.}
\def\PMB{Phys. Med. Biol.}
\def\PR{Phys. Rev.}
\def\PRL{Phys. Rev. Lett.}
\def\PRS{Proc. R. Soc.}
\def\prsl{Proc. R. Soc. Lond. A}
\def\PRSL{Proc. R. Soc. A}
\def\PS{Phys. Scr.}
\def\PSS{Phys. Status Solidi}
\def\PTRS{Phil. Trans. R. Soc.}
\def\RMP{Rev. Mod. Phys.}
\def\RPP{Rep. Prog. Phys.}
\def\RSI{Rev. Sci. Instrum.}
\def\SAM{Stud. Appl. Math.}
\def\SSC{Solid State Commun.}
\def\SST{Semicond. Sci. Technol.}
\def\SUST{Supercond. Sci. Technol.}
\def\ZP{Z. Phys.}
\def\JCAM{J. Comput. Appl. Math.}

\def\OUP{Oxford University Press}
\def\CUP{Cambridge University Press}
\def\AMS{American Mathematical Society}
\def\SIGMA{SIGMA}

\def\bibitm{\vspace{-6pt}\bibitem}

\def\refjl#1#2#3#4#5#6#7{\bibitm{#1} \textrm{\frenchspacing#2}, #3, 
\textit{\frenchspacing#4}, \textbf{#5} (#6) #7.}
\def\refjnl#1#2#3#4#5#6#7{\bibitm{#1} \textrm{\frenchspacing#2}, 
\textit{\frenchspacing#4}, \textbf{#5} (#6) #7.}

\def\refpp#1#2#3#4{\bibitm{#1} \textrm{\frenchspacing#2}, \textrm{#3}, #4.}

\def\refjltoap#1#2#3#4#5#6#7{\bibitm{#1} \textrm{\frenchspacing#2}, \textrm{#3},
\textit{\frenchspacing#4} (#7). 
#6.}

\def\refbk#1#2#3#4#5{\bibitm{#1} \textrm{\frenchspacing#2}, ``\textit{#3}", #4, #5.}

\def\refcf#1#2#3#4#5#6#7{\bibitm{#1} \textrm{\frenchspacing#2}, \textrm{#3},
in ``\textit{#4}", {\frenchspacing#5}, pp.\ #7, #6.}

\end{document}

\def\bibitm{\vspace{-6pt}\bibitem}

\def\refjl#1#2#3#4#5#6#7{\bibitm{#1} \textrm{\frenchspacing#2} #6 #3 \textit{\frenchspacing#4} \textbf{#5} #7}

\def\refpp#1#2#3#4#5{\bibitm{#1} \textrm{\frenchspacing#2} #3 \textrm{#4} (#5)}

\def\refbk#1#2#3#4#5{\bibitm{#1} \textrm{\frenchspacing#2} #5 \textit{#3} #4}

\def\refcf#1#2#3#4#5#6#7{\bibitm{#1} \textrm{\frenchspacing#2} #6 \textrm{#3} \textit{#4} {\frenchspacing#5} pp~#7}

\end{document}